\documentclass[pra,onecolumn,11pt,tightenlines,superscriptaddress,nofootinbib]{revtex4}
\usepackage[a4paper, left=0.8in, right=0.8in, top=0.9in, bottom=0.9in]{geometry}

\usepackage{cmap} 
\usepackage[utf8]{inputenc}
\usepackage[english]{babel}
\usepackage[T1]{fontenc}
\usepackage{amsmath}
\usepackage{amsfonts}
\usepackage{bm}
\usepackage{xcolor}
\definecolor{blueviolet}{rgb}{0.2, 0.2, 0.6}
\definecolor{webgreen}{rgb}{0,.5,0}
\definecolor{webbrown}{rgb}{.6,0,0}
\usepackage[pdftex,
	bookmarks=false,
	colorlinks=true, 
	urlcolor=webbrown,
	linkcolor=blueviolet,
	citecolor=webgreen,
	pdfstartpage=1,
	pdfstartview={FitH},  
	bookmarksopen=false
	]{hyperref}
\allowdisplaybreaks
\usepackage{tikz}
\usepackage{braket}
\usepackage{natbib}
\usepackage{amsthm}
\usepackage{dsfont}
\usepackage{listings}

\DeclareFixedFont{\ttb}{T1}{txtt}{bx}{n}{9} 
\DeclareFixedFont{\ttm}{T1}{txtt}{m}{n}{9}  

\usepackage{amssymb}

\let\emptyset\varnothing

\usepackage{color}
\definecolor{deepblue}{rgb}{0,0,0.5}
\definecolor{deepred}{rgb}{0.6,0,0}
\definecolor{deepgreen}{rgb}{0,0.5,0}

\newcommand\pythonstyle{\lstset{
language=Python,
basicstyle=\ttm,
morekeywords={self},              
keywordstyle=\ttb\color{deepblue},
emph={MyClass,__init__},          
emphstyle=\ttb\color{deepred},    
stringstyle=\color{deepgreen},
frame=tb,                         
showstringspaces=false
}}

\lstnewenvironment{python}[1][]
{
\pythonstyle
\lstset{#1}
}
{}

\newif\ifdraft
\drafttrue

\newcommand\pythoninline[1]{{\pythonstyle\lstinline!#1!}}

\theoremstyle{plain}
\newtheorem{theorem}{Theorem}
\newtheorem{corollary}[theorem]{Corollary}
\newtheorem{lemma}[theorem]{Lemma}
\newtheorem{proposition}[theorem]{Proposition}

\theoremstyle{definition}
\newtheorem{definition}[theorem]{Definition}

\theoremstyle{remark}
\newtheorem{remark}{Remark}

\usepackage{titlesec}

\newcommand{\euler}{\mathrm{e}}
\newcommand{\rmi}{\mathrm{i}}

\DeclareMathOperator*{\sign}{sign}

\DeclareMathOperator{\Tr}{tr}
\DeclareMathOperator*{\E}{{\mathbb{E}}}
\DeclareMathOperator*{\Var}{\mathrm{Var}}

\newcommand{\bI}{{\mathbb{I}}}
\newcommand{\cI}{{\mathcal{I}}}

\newcommand{\cB}{{\mathcal{B}}}
\newcommand{\cC}{{\mathcal{C}}}
\newcommand{\cP}{{\mathcal{P}}}
\newcommand{\cW}{{\mathcal{W}}}
\newcommand{\cX}{{\mathcal{X}}}
\newcommand{\cY}{{\mathcal{Y}}}
\newcommand{\cZ}{{\mathcal{Z}}}

\newcommand{\cE}{{{\mathcal{E}}}}
\newcommand{\cF}{{\mathcal{F}}}
\newcommand{\cM}{{\mathcal{M}}}
\newcommand{\cN}{{{\mathcal{N}}}}

\newcommand{\cQ}{{\mathcal{Q}}}

\newcommand{\ketbra}[2]{\lvert #1 \rangle \! \langle #2 \rvert}
\newcommand{\norm}[1]{\left\lVert#1\right\rVert}

\usepackage[noend]{algpseudocode}
\usepackage{algorithm,algorithmicx}

\algrenewcommand\alglinenumber[1]{\sf\scriptsize\color{blue}{#1}}
\algrenewcommand\algorithmicrequire{\textbf{Input:}}
\algrenewcommand\algorithmicensure{\textbf{Output:}}

\begin{document}

\title{Foundations for learning from noisy quantum experiments}
\author{Hsin-Yuan Huang}
\affiliation{Institute for Quantum Information and Matter and\\ Department of Computing and Mathematical Sciences, Caltech, Pasadena, CA, USA}
\affiliation{AWS Center for Quantum Computing, Pasadena, CA, USA}
\author{Steven T. Flammia}
\affiliation{AWS Center for Quantum Computing, Pasadena, CA, USA}
\affiliation{California Institute of Technology, Pasadena, CA, USA}
\author{John Preskill}
\affiliation{Institute for Quantum Information and Matter and\\ Department of Computing and Mathematical Sciences, Caltech, Pasadena, CA, USA}
\affiliation{AWS Center for Quantum Computing, Pasadena, CA, USA}

\date{\today}

\begin{abstract}
Understanding what can be learned from experiments is central to scientific progress.
In this work, we use a learning-theoretic perspective to study the task of learning physical operations in a quantum machine when all operations (state preparation, dynamics, and measurement) are \emph{a priori} unknown.
We prove that, without any prior knowledge, if one can explore the full quantum state space by composing the operations, then every operation can be learned.
When one cannot explore the full state space but all operations are approximately known and noise in Clifford gates is gate-independent, we find an efficient algorithm for learning all operations up to a single unlearnable parameter characterizing the fidelity of the initial state.
For learning a noise channel on Clifford gates to a fixed accuracy, our algorithm uses quadratically fewer experiments than previously known protocols.
Under more general conditions, the true description of the noise can be unlearnable; for example, we prove that no benchmarking protocol can learn gate-dependent Pauli noise on Clifford+T gates even under perfect state preparation and measurement.
Despite not being able to learn the noise, we show that a noisy quantum computer that performs entangled measurements on multiple copies of an unknown state can yield a large advantage in learning properties of the state compared to a noiseless device that measures individual copies and then processes the measurement data using a classical computer.
Concretely, we prove that noisy quantum computers with two-qubit gate error rate $\epsilon$ can achieve a learning task using $N$ copies of the state, while $N^{\Omega(1/\epsilon)}$ copies are required classically.
\end{abstract}

\maketitle

Understanding what we can learn from experiments is central to many scientific fields.
By conducting experiments we obtain information about the physical world.
This information can be organized into knowledge allowing us to predict how the world would behave under different circumstances and to design complex systems with desired functionalities. In particular, the development of sophisticated quantum technologies is guided by knowledge acquired in experiments on quantum systems.

In this work, we build on computational learning theory \cite{kearns1994introduction, goodfellow2016deep, mohri2018foundations} to develop a rigorous theory for reasoning about what can and cannot be learned in a world governed by quantum mechanics.
We represent scientists and their classical algorithms abstractly as classical agents that conduct experiments by specifying actions that control a quantum system.
Actions include specifying which initial states are prepared, which CPTP maps are performed, and which POVM measurements are executed at the end of an experiment.
Classical agents do not have perfect knowledge about how these actions affect the quantum world, e.g., what state is prepared, or what quantum processes and measurements are actually implemented.
But classical agents can improve their understanding of the physical world through experiments.
We refer to the mapping from each action to the corresponding physical operation as a world model $\cW$.
The central question we would like to answer is: What can a classical agent learn about the true world model $\mathcal{W}_{\mathrm{true}}$ describing the underlying physical reality?

To answer this question, we make a distinction between two different viewpoints for assessing how well the classical agent learns.
The first viewpoint judges whether the classical agent accurately \emph{learns the intrinsic descriptions} of the physical operations.
This is the viewpoint commonly considered in quantum state or process tomography \cite{hradil1997quantum,gross2010quantum,blume2010optimal, d2003quantum, mohseni2008quantum, o2016efficient,o2017efficient,haah2017sample, guta2020fast, huang2020predicting, huang2021provably} or in Hamiltonian learning~\cite{anshu2021sample,haah2021optimal}.
For example, given a quantum computer, we might want the classical agent to characterize the noise afflicting the state preparations, maps, or measurements acting on various qubits.
Knowledge about the intrinsic descriptions of the physical operations is crucial for calibrating, controlling, and improving a complex quantum many-body system \cite{wiseman2009quantum, brif2010control, yang2014robust}.
If, for example, the classical agent finds the measurements to be particularly noisy, we should focus on improving our measurement procedure.

The second viewpoint examines whether the classical agent can \emph{predict the extrinsic behavior} of the quantum system.
In particular, given an experiment, we ask the classical agent to predict the experimental outcomes.
The classical agent does not need to learn the actual descriptions of the physical operations and is free to use any model as long as the predictions are accurate.
The second viewpoint is relevant when we want to control the quantum system to achieve specific tasks, such as mitigating errors in a particular computation \cite{temme2017error, endo2018practical, kandala2019error, strikis2021learning}.
Because the true physical descriptions are not learned, a model that can predict extrinsic behavior might not provide useful guidance regarding how to improve the quantum device.
While intrinsic descriptions are more informative than a model for predicting extrinsic behavior, intrinsic descriptions are also more challenging to learn.

In this work, we obtain fundamental results regarding what one can learn from noisy quantum experiments.
We also present case studies illustrating some practical implications of these foundational results.
In related work to appear elsewhere, we will present a versatile mathematical framework for developing rigorous neural network algorithms capable of learning a wide range of noise models in quantum many-body systems.


\begin{figure}
    \centering
    \includegraphics[width=0.95\textwidth]{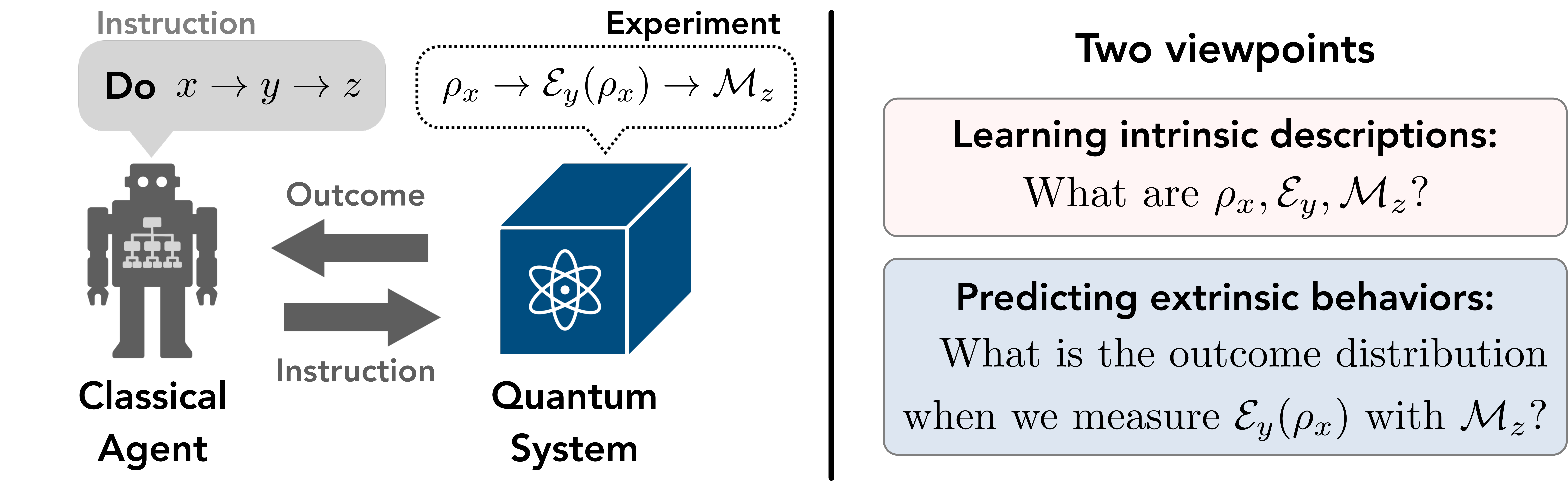}
    \caption{Illustration of a classical agent that can learn from quantum experiments. The classical agent specifies an instruction consisting of the actions $x, y, z$ for the experiment. The quantum system runs the experiment using a state $\rho_x$, a quantum process $\cE_y$, and a POVM measurement $\cM_z$ to produce an experimental outcome.
    The classical agent does not have prior knowledge on what $\rho_x, \cE_y, \cM_z$ are. The goal of the classical agent is to either learn $\rho_x, \cE_y, \cM_z$ (learning intrinsic descriptions) or predict the distribution of experimental outcomes (predicting extrinsic behaviors).}
    \label{fig:theme}
\end{figure}

\section{Foundations}

To elucidate the basic features of learning intrinsic descriptions and predicting extrinsic behaviors, we prove a series of fundamental results regarding what the classical agent can learn and how efficiently the classical agent can learn it.
We begin by establishing the formalism.
Then, we present results that elucidate the two viewpoints.

\subsection{A formalism for learning from quantum experiments}

We consider three sets of actions that the classical agent could perform: $\mathcal{X}$ is the set of actions for preparing different kinds of states, $\mathcal{Y}$ consists of actions for implementing certain physical evolutions, and $\mathcal{Z}$ is a set of possible POVM measurements.
A world model $\mathcal{W}$ is a mapping from the actions to the actual physical operations.
For example, each action $y \in \mathcal{Y}$ corresponds to a quantum channel $\mathcal{E}_y$ that is applied to the quantum system when the classical agent performs the action $y$.
We represent the world model by the collection of physical operations associated to all actions,
\begin{equation}
    \mathcal{W} = \left(\{\rho_x\}_{x \in \mathcal{X}}, \{\mathcal{E}_y\}_{y \in \mathcal{Y}}, \{\mathcal{M}_z\}_{z \in \mathcal{Z}} \right),
\end{equation}
where $\rho_x$ is a quantum state, $\mathcal{E}_y$ is a quantum channel, and $\mathcal{M}_z$ is a POVM measurement.
In Appendix~\ref{sec:qworld}, we provide formal definitions of world models and formulate some relations between world models.
Due to the intrinsic degeneracy in the formulation of quantum mechanics, two world models related by a global unitary or anti-unitary transformation describe the same physical reality~\cite{Wigner1931GruppentheorieUI}.

The classical agent conducts an experiment by instructing the quantum system to execute a given sequence of actions.
Each experiment $E$ begins with an action $x \in \mathcal{X}$ that prepares the state $\rho_x$, followed by a sequence of actions $y_1, \ldots, y_L \in \mathcal{Y}$ that evolves the state $\rho_x$ to $\mathcal{E}_{y_L}(\ldots \mathcal{E}_{y_2}(\mathcal{E}_{y_1}(\rho_x))\ldots)$, and the experiment ends with an action $z \in \mathcal{Z}$ that produces a measurement outcome after performing the POVM $\mathcal{M}_z$ on $\mathcal{E}_{y_L}(\ldots \mathcal{E}_{y_2}(\mathcal{E}_{y_1}(\rho_x))\ldots)$.
We suppose that the classical agent has prior knowledge that the true world model $\mathcal{W}_{\mathrm{true}}$ belongs to a class of candidate models $\mathcal{Q} = \{\mathcal{W}\}$.
The goal of the classical agent is to learn about the true world model $\mathcal{W}_{\mathrm{true}}$ by conducting experiments.
The set of models $\mathcal{Q}$ will be referred to as the \emph{model class}, following the nomenclature in classical learning theory~\cite{mohri2018foundations}.

Note that this framework is very general.
For example, the set $\mathcal{Q}$ of candidate models could be uncountably large, which might describe an unknown noise processes occurring in an experimental setup described by continuous parameters.
There could even be non-Markovian or time-dependent effects in this framework.
For example, in a real experiment a sequence of actions $x, y_1, y_2, y_3$ might prepare a state $\mathcal{E}_{y_3|y_2 y_1}( \mathcal{E}_{y_2|y_1}( \mathcal{E}_{y_1}(\rho_x)))$ where each physical operation depends on the history of previous actions.
Superficially, this appears to lie outside the framework.
However, these time-dependent and non-Markovian effects can be put into our framework by embedding into a new world model where state preparations contain an additional memory register, which is continually updated as subsequent actions are performed.
The details of this mapping are in Appendix~\ref{sec:time-depend}.
Thus, our framework can model many situations where a classical agent learns about a quantum system.

\subsection{Learning intrinsic descriptions of the world model}

We will say that a model class $\mathcal{Q}$ is \emph{learnable} if for any world model $\mathcal{W}\in \mathcal{Q}$,
the classical agent can determine all physical operations in $\mathcal{W}$ with an arbitrarily small error using a finite number of experiments, up to one global unitary or anti-unitary transformation. A formal definition of learnability is given in Appendix~\ref{sec:learning_theory_foundations}.
A similar criterion is often invoked in discussions of quantum tomography; however, in quantum state/process/measurement tomography, it is typically assumed that some actions are perfectly known, and the goal is to characterize other actions which are not perfectly known. For example, we might envision using perfect state preparations and perfect measurements to characterize noisy gates. We will consider instead a setting in which all actions are unknown \emph{a priori}, and our goal is to learn the true physical description of each action to arbitrary accuracy.

Learning the intrinsic physical description might be challenging or even impossible because the classical agent can perform only a limited set of actions.
For example, a classical agent cannot distinguish between the following two distinct physical realities in a single-qubit system with universal control using Hadamard ($H$) and $T$ ( $= Z^{1/4}$) gates,
\begin{align}
    \mathcal{W}^A: \quad \rho^A_0 &= I/2, & \mathcal{E}^A_H(\rho) &= H\rho H^\dagger, & \mathcal{E}^A_T(\rho) &= T\rho T^\dagger, & \mathcal{M}^A_0 &= \{\ketbra{0}{0}, \ketbra{1}{1}\}, \label{eq:worldA-maintext} \\
    \mathcal{W}^B: \quad \rho^B_0 &= I/2, & \mathcal{E}^B_H(\rho) &= I/2, & \mathcal{E}^B_T(\rho) &= I/2, & \mathcal{M}^B_0 &= \{\ketbra{0}{0}, \ketbra{1}{1}\}, \label{eq:worldB-maintext}
\end{align}
with actions $\mathcal{X} = \{0\}, \mathcal{Y} = \{H, T\}, \mathcal{Z} = \{0\}$.
In both world A and world B, no matter what sequence of actions drawn from $\{H,T\}$ is performed between the initial state preparation and the final measurement, the probability distribution governing the measurement outcomes is always the same, namely the uniform distribution over the two outcomes $\{0, 1\}$. Hence, although the maps labeled by $H$ and $T$ are certainly intrinsically different in world A than in world B, this difference cannot be perceived by the classical agent. It follows that any model class $\mathcal{Q}$ that contains both $\mathcal{W}^A$ and $\mathcal{W}^B$ is unlearnable, even if $\mathcal{Q}$ is uncountably large.

If the classical agent has no prior knowledge of the state preparations, maps, and measurements, is it possible to learn the intrinsic descriptions of all actions to arbitrary accuracy? Our first result shows that a model class $\mathcal{Q}$ is learnable if a suitable condition is met. Suppose that each world model $\mathcal{W}\in\mathcal{Q}$ contains a nontrivial measurement (one whose outcome depends on the state being measured), and that each $\mathcal{W}\in \mathcal{Q}$ contains actions that fully explore the quantum state space. Then $\mathcal{Q}$ is learnable. This observation motivates the following definition.

\begin{definition} \label{def:universal}
A world model $\mathcal{W} = \left(\{\rho_x\}_{x \in \mathcal{X}}, \{\mathcal{E}_y\}_{y \in \mathcal{Y}}, \{\mathcal{M}_z\}_{z \in \mathcal{Z}} \right)$ is \emph{universal} if (i) there is $x \in \mathcal{X}$, such that the state $\rho_x$ is pure; (ii) there are $y_1, \ldots, y_k \in \mathcal{Y}$, such that the maps $\mathcal{E}_{y_1}, \ldots \mathcal{E}_{y_k}$ constitute a universal set of unitary transformations; (iii) there is $z \in \mathcal{Z}$, such that the POVM $\mathcal{M}_z$ has at least one POVM element not proportional to the identity.
\end{definition}

When a world model is universal, we do not assume that the classical agent knows for which value of $x$ the state $\rho_x$ is pure, for which values of $y$ the maps form a universal set, or for which value of $z$ the measurement is nontrivial. Nor does the classical agent have any other prior knowledge of the state preparations, maps, and measurements. Nevertheless, the agent can learn all actions to arbitrary accuracy according to the following theorem.

\begin{theorem}[Learning intrinsic descriptions]
\label{thm:int-description}
Consider a (possibly uncountably large) model class $\cQ$ such that each candidate world model in $\cQ$ is universal as in Definition~\ref{def:universal}.
If the true world model is $\cW_{\mathrm{true}} \in \cQ$, then the classical agent can learn the description of every action in $\cW_{\mathrm{true}}$ to arbitrarily small error (up to one global unitary or anti-unitary transformation).
\end{theorem}
\begin{proof}[Proof idea]
Here, we present the general idea of the proof; further details are in Appendix~\ref{sec:int-description}.
First, we design a procedure for testing whether a composed CPTP map $\mathcal{E}_{y_k} \circ \ldots \circ \mathcal{E}_{y_1}$ is the identity map. Then, we use the fact that unitaries are the only reversible CPTP maps and the identity testing procedure to create a protocol for testing whether a CPTP map $\mathcal{E}_{y}$ is a unitary.
The unitary test allows the learning agent to identify every action that implements a unitary transformation to within a specified error tolerance.

Next, we show that a random composition of the identified unitaries forms an approximate Haar random unitary.
This is proven using a contraction theorem for random walks on compact semi-simple Lie groups \cite{Varj2012RandomWI}.
This result holds for any universal set of unitaries without the need to include inverses or have algebraic entries, but the theorem is weaker than the spectral gap theorem in \cite{bourgain2012spectral}.
Building on Corollary~7 in \cite{Varj2012RandomWI}, we can show that the expectation value of any Lipschitz continuous function over randomly composed unitaries (chosen from a universal set) is  approximately equal to the expectation value over Haar random unitaries.
Hence, although the learning agent doesn't know what each unitary is, the learning agent can still sample from approximately Haar-random unitaries.
In fact, it will suffice to sample approximately from a unitary two-design.
Up to this point, the learning agent has not learned the description for any of the actions.

We then prove the following: given a procedure that samples approximately from a unitary two-design and the availability of an unknown POVM with at least one POVM element not proportional to identity, we can create a procedure that estimates the  overlap $\Tr(\rho_1 \rho_2)$ for any two states $\rho_1, \rho_2$.
This procedure makes use of the two-design property of the random unitary ensemble.

Applying this overlap estimation procedure, the learning agent can determine for which value of $x$ the state $\rho_x$ is pure, and then reach other pure states by applying unitary circuits to $\rho_x$. Through further applications of the overlap estimation procedure, the learning agent can find a special set of states $\{\ket{\psi_k}\}$ with a particular geometry. Specifically, the special set of states corresponds to an orthonormal basis of pure states $\ket{e_1}, \ldots, \ket{e_d}$, superpositions of pairs of these basis states $\frac{1}{\sqrt{2}}(\ket{e_i} + \ket{e_j})$, $\frac{1}{\sqrt{2}}(\ket{e_i} + \rmi \ket{e_j})$ with a real or imaginary relative phase, and also superpositions of three basis states $\frac{1}{\sqrt{3}}(\ket{e_1} + \ket{e_i} + \ket{e_j})$, $\frac{1}{\sqrt{3}}(\ket{e_1} + \rmi \ket{e_2} + \rmi \ket{e_j})$, $\frac{1}{\sqrt{3}}(\ket{e_1} + \ket{e_i} + \rmi \ket{e_j})$.
The only ambiguity in this procedure is that all experiments would yield the same results if each state $\ket{\psi_k}$ were replaced by $U\ket{\psi_k}$ where $U$ is a fixed unitary transformation, and/or if $(i)$ were replaced by $(-i)$ in the superpositions of basis states (i.e., if a fixed antiunitary transformation were applied to each state $\ket{\psi_k}$).
This ambiguity is also present in the formulation of quantum mechanics as characterized by Wigner's theorem \cite{Wigner1931GruppentheorieUI}.
Without the superpositions of three basis states, there would be further ambiguity beyond the freedom to perform a global transformation.

Given the special set of states $\{\ket{\psi_k}\}$ and the procedure for estimating state overlaps, the classical agent can perform a version of quantum state tomography to learn the physical representation of $\rho_x$ for all $x\in\mathcal{X}$.
Similarly, the classical agent can learn the CPTP map $\mathcal{E}_y$ for all $y\in\mathcal{Y}$ by performing quantum state tomography on the states $\mathcal{E}_y(\ketbra{\psi_k}{\psi_k})$ for all special states $\ket{\psi_k}$.
The classical agent can also learn the POVM  $\mathcal{M}_z$ for all $z\in\mathcal{Z}$ from the outcome probability distribution when $\mathcal{M}_z$ is performed on the special states $\{\ket{\psi_k}\}$.
Hence, the classical agent has learned all the actions, up to a global unitary or antiunitary transformation.
\end{proof}

When we cannot explore the quantum state space completely, it may be impossible to learn the intrinsic description to arbitrarily small error even with an infinite number of experiments.
In  Appendix~\ref{sec:basiclearna}, we present some basic results establishing how modifying a model class affects the learnability of the model class.
Building on these results, we prove the unlearnability of various classes of quantum systems in Appendices~\ref{sec:cliffordT}~and~\ref{sec:noisystate}.

As a further example illustrating the concept of unlearnability, consider an $n$-qubit system with an initial state $\rho_0$, a set of unital CPTP maps ($\cE_y(I) = I$ for all $y$), and a POVM measurement $\cM_0$.
It is not hard to see that no learning algorithm can distinguish between the following two physical realities.
\begin{enumerate}
    \item The initial state is slightly depolarized, $\rho_0 = 0.9 \ketbra{0^n}{0^n} + 0.1 (I / 2^n)$. But the computational basis measurement is perfect, $\cM_0 = \{\ketbra{b}{b}\}_{b \in \{0, 1\}^n}$.
    \item The initial state is perfect, $\rho_0 = \ketbra{0^n}{0^n}$. But the computational basis measurement is slightly depolarized, $\cM_0 = \{ 0.9 \ketbra{b}{b} + 0.1( I / 2^n) \}_{b \in \{0, 1\}^n}$.
\end{enumerate}
\noindent This follows immediately from the property that unital maps preserve the maximally mixed state. We will present a less trivial example in Section~\ref{sec:unknown-error} and discuss some practical implications.

\subsection{Predicting extrinsic behavior of the world model}

Rather than insisting that the classical agent learn the intrinsic description of a world model, we might settle instead for a description of the world model's extrinsic behavior.
We give a formal definition of extrinsic behavior in Appendix~\ref{sec:behav-predict}.
Suppose we would like to predict the probability distribution of measurement outcomes for any experiment with $L$ CPTP maps, and suppose there are only finitely many actions.
How many experiments will the classical agent need to perform in order to make accurate predictions?
It would certainly suffice to perform each possible experiment many times, but because there are $|\mathcal{X}| |\mathcal{Y}|^L |\mathcal{Z}|$ possible experiments, this procedure is very costly for large $L$.
Is there a more efficient method?
In fact, in the worst case we can improve on the cost of this exhaustive procedure by at best a logarithmic factor.

 \begin{theorem}[Worst case complexity for predicting extrinsic behavior] \label{thm:worst-case-behav}
 To predict the probability of each experimental outcome to error $\epsilon$ for every experiment with $L$ maps, the classical agent has to perform at least $\Omega(|\mathcal{X}| |\mathcal{Y}|^L |\mathcal{Z}| / \epsilon^2)$ experiments in the worst case, and the classical agent can always achieve the task by running $\widetilde{\mathcal{O}}(|\mathcal{X}| |\mathcal{Y}|^L |\mathcal{Z}| / \epsilon^2)$ experiments\footnote{$\widetilde{\mathcal{O}}(\cdot)$ neglects the logarithmic factors.}.
 \end{theorem}

To derive the lower bound, we construct world models that behave like a special kind of maze. The classical agent has to navigate through the maze by performing a specific sequence of actions. Whenever the classical agent makes the wrong action, he/she fails.
We then combine this picture with a proof technique used in \cite{huang2021information, chen2021exponential} to establish the stated lower bound.
The detailed proof is given in Appendix~\ref{sec:hardness-behav}.

One way to avoid this daunting worst case complexity is to assume that we have found a set of composed states and POVM elements such that both sets span all possible states one can generate by performing actions in the world model.
The following theorem shows that in such cases a speed-up exponential in $L$ over the worst case can be achieved.
The full algorithm and the proof are given in Appendix~\ref{sec:thm-extrinsic-behav}.

\begin{theorem}[Predicting extrinsic behavior] \label{thm:ext-behavior}
Suppose we have found a set of unknown linearly independent states composed from $\{\rho_x,\mathcal{E}_y\}$ and a set of unknown POVM elements composed from $\{\cE_y, \cM_z\}$, such that both sets span all the states that can be prepared in world model $\mathcal{W}$.
Then, we can predict the probability of each experimental outcome to error $\epsilon$ for any experiment with $L$ maps after running $\widetilde{\mathcal{O}}( (|\mathcal{X}| + L^2 |\mathcal{Y}| + |\mathcal{Z}| ) / \epsilon^2)$ experiments.
\end{theorem}

Comparing Theorems \ref{thm:worst-case-behav} and \ref{thm:ext-behavior}, we infer that the spanning set of composed states and POVM elements assumed in Theorem \ref{thm:ext-behavior} cannot be found efficiently in the worst case.
But for many world models they can be found efficiently.

In Appendix~\ref{sec:related-work} we discuss gate set tomography from the perspective of our mathematical theory.
There we argue that existing gate set tomography protocols \cite{greenbaum2015introduction, blume2017demonstration, nielsen2020gate, brieger2021compressive} may be viewed as methods for predicting the extrinsic behavior of a world model.
Based on this understanding, the algorithm used in the proof of Theorem~\ref{thm:ext-behavior} establishes rigorous gate set tomography learning algorithms with provable prediction guarantees under suitable conditions.

\section{Case studies}

Let us now consider three case studies that illustrate the power and wide applicability of the formalism that we have developed.

\subsection{Learning under gate-independent noise on Clifford gates}

A common assumption used in the randomized benchmarking literature \cite{emerson2005scalable,knill2008randomized, magesan2011scalable, kimmel2014robust, roth2018recovering, helsen2021estimating} is that the noise processes afflicting Clifford gates are independent of the specific gate.
This assumption applies to all elements in the Clifford group and is used to ensure that randomized benchmarking can extract an accurate estimate for the average gate fidelity \cite{proctor2017randomized}, although weak gate-dependent noise can also be accommodated~\cite{Wallman2018,Merkel2021}.
More precisely, when the classical agent performs a Clifford gate $C$, the CPTP map implemented is
\begin{equation}
    \cE_C(\rho) =  \cN(C \rho C^\dagger),
\end{equation}
where $\cN$ is a CPTP map close to the identity that does not depend on the gate $C$.
In the benchmarking literature, it is also commonly assumed that there is a noisy zero state $\rho_0$ that is close to the all zero state $\ket{0^n}$, and there is a noisy computational basis measurement $\cM_0 = \{M_b\}_{b \in \{0, 1\}^n}$, where $M_b$ is close to $\ketbra{b}{b}$.
We will refer to these assumptions as bounded gate-independent noise on Clifford gates, and bounded noise on initial zero state preparation and computational basis measurement.
Typically, these assumptions are expected to hold (approximately) only for a subsystem consisting of a constant number of qubits in a many-qubit quantum computer.

Under these assumptions, we give a simple and practical algorithm for learning every physical operation up to one unknown parameter $f = \bra{0^n} \rho_0 \ket{0^n}$.
The parameter $f$ is the fidelity of the noisy zero state $\rho_0$, which is assumed to be close to one.
Using techniques presented in the appendices, it is straightforward to show that $f$ is unlearnable.
To see this, assume that the noise channel $\cN$ is the identity channel, and use the result in Appendix~\ref{sec:noisystate} showing that one cannot distinguish whether the initial state or the measurement is subjected to a depolarizing channel.
When the Clifford gate noise $\cN$ is not assumed to be the identity channel, the learning does not become any easier; hence $f$ is still unlearnable.

When the Clifford gate noise is unital, $\cN(I) = I$, the unlearnability of $f$ does not prevent
$\cN$ from being determined to arbitrarily small error.
Under this assumption of unital noise, several existing algorithms, robust to state preparation and measurement error, have been proposed to learn $\cN$ using information obtained from randomized benchmarking \cite{kimmel2014robust, roth2018recovering, helsen2021estimating}.
The best existing algorithm \cite{helsen2021estimating} learns $\cN$ using $\mathcal{O}(d^8 \log d)$ experiments, where $d = 2^n$.
Our proposed algorithm only requires $\mathcal{O}(d^4 \log d)$ experiments.
This number of experiments scales linearly, up to a logarithmic factor, with the number of parameters characterizing the noise channel.
The rigorous guarantees are expressed in the following theorem; the proof is in Appendix~\ref{sec:gate-indep-Clif}.

\begin{theorem}[Learning gate-independent noise on Clifford gates] \label{thm:learn-gate-indp}
Assume bounded gate-independent noise on Clifford gates, and bounded noise on initial zero state preparation and computational basis measurement.
All states, processes, and measurements $\{\rho_x, \cE_y, \cM_z\}$ can be learned up to a single unlearnable parameter $f = \bra{0^n} \rho_0 \ket{0^n}$.

Furthermore, if the Clifford gate noise channel $\cN$ obeys $\cN(I) = I$, then the Choi matrix for $\cN$ can be learned up to error $\epsilon$ in the Hilbert-Schmidt norm using $\mathcal{O}(d^4 \log d / \epsilon^2)$ experiments.
\end{theorem}
\begin{proof}[Proof idea]
We conduct two sets of randomized experiments to learn about the initial state $\rho_0$ (the noisy zero state), the Clifford gate noise $\cN$, and the measurement $\cM_0$ (the noisy computational basis measurement).
The first set of experiments prepares $\rho_0$, evolves by $\cE_C$ for a random Clifford $C$, and measures $\cM_0$.
The second set of experiments prepares $\rho_0$, evolves by $\cE_{C_1}$ for a random Clifford $C_1$, evolves by $\cE_{C_2}$ for a second random Clifford $C_2$, and measures $\cM_0$. (The division of the Clifford into $C_1$ and $C_2$, each with noise channel $\cN$, is used in the postprocessing of the measurement results.) The learning procedure exploits the unitary 2-design property of random Clifford gates and follows closely the classical shadow formalism~\cite{huang2020predicting, huang2022learning}, which is based on least squares estimation of the quantum objects~\cite{guta2020fast}.


From the first set of experiments we can learn, up to the unlearnable parameter $f$, the state $\rho_0$ and the POVM $\cN^\dagger(\cM_0)$, which corresponds to applying $\cN$ followed by the POVM $\cM_0$,
Then, the second set of experiments allows us to learn $\cN$. Once $\cN$ and $\cN^\dagger(\cM_0)$ are known, we can invert $\cN^\dagger$ to determine $\cM_0$.

With data obtained from the first and the second sets of randomized experiments, we can learn all physical operations.
To learn the state $\rho_x$, we repeatedly perform the randomized experiment: prepare $\rho_x$, evolve by $\cE_C$ for a random Clifford $C$, and measure $\cM_0$.
To learn the CPTP map $\cE_y$, we repeatedly perform the randomized experiment: prepare $\rho_0$, evolve by $\cE_{C_1}$ for a random Clifford $C_1$, evolve by $\cE_y$, evolve by $\cE_{C_2}$ for a second random Clifford $C_2$, and measure $\cM_0$.
To learn POVM $\cM_z$, we repeat the randomized experiment: prepare $\rho_0$, evolve by $\cE_C$ for a random Clifford $C$, and measure $\cM_z$.

The unitary 2-design property of random Clifford gates and standard concentration inequalities can be used to characterize the number of experiments needed to estimate all objects with a small error.
In particular, we show that if $\cN(I) = I$ then the proposed algorithm learns the Clifford gate noise $\cN$ to error $\epsilon$ from a total of $\mathcal{O}(d^4 \log d / \epsilon^2)$ experiments.
\end{proof}

\subsection{Bit-flip or phase-flip error?}
\label{sec:unknown-error}

In this case study, we relax the assumption of gate-independent noise and consider the task of learning the noise on Clifford and $T$ gates when the noise channel depends on the gate.
We assume that the noisy version of the Clifford or $T$ gate $G$ is
\begin{equation}
    \cE_G(\rho) = \cP_G(G \rho G^\dagger),
\end{equation}
where $\cP_G$ is a Pauli channel that depends on the gate $G$ and is close to the identity; tailoring a more general noise process using randomized compiling might result in such a channel \cite{Viola2005,Knill2005,wallman2016noise}.
In the following theorem, we show that it is impossible to learn the true gate-dependent Pauli noise even in a single-qubit system. To be concrete, we suppose that the classical agent can apply noisy versions of the $H$, $S$, and $T$ gates, and hopes to characterize the noise channel for each of these gates.

\begin{theorem}[Gate-dependent Pauli noise is unlearnable with Clifford+T gates] \label{prop:Clif+T}
Consider a qubit system. Suppose one can prepare $\ket{0}$ perfectly, measure in the computational basis perfectly, and apply $H$, $S$, and $T$ gates, where each gate is followed by an unknown gate-dependent Pauli noise channel close to the identity channel.
It is impossible for any algorithm to learn all the gate-dependent Pauli noise channels to arbitrarily small error.
\end{theorem}
\begin{proof}[Proof idea]
The theorem is established by proving that one is unable to determine whether a bit-flip ($X$) error is more likely to happen than a phase-flip ($Z$) error after the Hadamard gate $H$; these two possibilities define distinct physical realities that are not related by a global unitary or anti-unitary transformation.
Importantly, perfect preparation of $\ket{0}$ and perfect measurement in the computational basis are not sufficient by themselves for learning a general quantum channel via complete quantum process tomography.
Using a graphical method to track the Heisenberg-picture evolution of Pauli operators under circuits of noisy Clifford + $T$ gates, we show that the noise in the Hadamard gate is unlearnable if only the state $|0\rangle$ can be prepared.
Then using reduction methods developed in Appendix~\ref{sec:basiclearna}, we show that the unlearnability persists if noisy versions of arbitrary states can also be prepared.
The full proof is given in Appendix~\ref{sec:cliffordT}.
\end{proof}

This result may seem to contradict claims about known protocols, such as gate set tomography \cite{greenbaum2015introduction, blume2017demonstration, nielsen2020gate, brieger2021compressive} and ACES \cite{flammia2021averaged}.
To resolve the conflict, recall that gate set tomography learns an effective model that can be used to predict the extrinsic behavior.
While gate set tomography might be able to learn some intrinsic features of the true world model, it is still unable to learn all the gate-dependent Pauli noise channels due to the gauge freedom in the protocol.
In ACES \cite{flammia2021averaged}, it is assumed that one can prepare any tensor product of single-qubit stabilizer states perfectly.
By making this additional assumption, the no-go result in Theorem \ref{prop:Clif+T} can be evaded.

Theorem \ref{prop:Clif+T} is relevant in the setting where we would like to identify a noise process to facilitate the task of improving our hardware.
For example, one may want to know whether the Hadamard gate is experiencing a higher rate of bit-flip error or phase-flip error, so that the dominant error can be suppressed by modifying the device. Theorem~\ref{prop:Clif+T} shows that characterizing the noise process precisely is not always possible.

\subsection{Quantum advantage in learning from noisy experiments}
\label{sec:Qadvnoisy}

Recent theoretical results have established that \emph{quantum-enhanced experiments} can exhibit a substantial advantage compared to \emph{conventional experiments} in learning properties of an unknown quantum state or process \cite{aharonov2021quantum, huang2021information, chen2021quantum, chen2021exponential, chen2021hierarchy,huang2021quantum}. The prior works concerned the idealized setting in which the quantum device used for learning is perfectly known and perfectly controllable. But recent experiments using up to 40 qubits in the Sycamore processor have demonstrated that a notable quantum advantage can be achieved even with the relatively noisy quantum devices that are currently available \cite{huang2021quantum}. This experimental demonstration highlights the need for a theoretical analysis of quantum advantage when the device used in learning is noisy and initially uncharacterized. Here we establish a provable quantum advantage in this more realistic setting.

We focus on a particular task that was analyzed in \cite{huang2021information, chen2021exponential, huang2021quantum} and studied experimentally in \cite{huang2021quantum}.
In this task, the learning agent is provided with multiple copies of an unknown quantum state $\rho$.
In conventional experiments, the agent collects classical data by measuring each copy of $\rho$ separately. These measurements may be adaptive; that is, the measurement performed on each copy can depend on the results obtained in measurements on previous copies. Using the collected classical data, the agent learns a model of $\rho$. After the learning is complete, the agent is asked to predict a property of $\rho$ which is selected from a long list of incompatible properties.
In quantum-enhanced experiments, the agent can load copies of $\rho$ into a quantum memory, and perform collective measurements acting on multiple copies. It was shown in \cite{huang2021information, chen2021exponential, huang2021quantum} that in some cases the number of copies of $\rho$ needed to predict accurately is exponentially smaller in the (ideal) quantum-enhanced scenario than in the conventional scenario. Furthermore, this exponential quantum advantage can be achieved using quite simple two-copy measurements in which qubit pairs are measured in the Bell basis. These simple measurements can be executed in practice with reasonable fidelity.

We wish to show that the quantum advantage persists when we compare ideal conventional experiments with noisy quantum-enhanced experiments. We assume that the noisy quantum device can prepare a noisy and unknown initial product state, load a copy of the unknown physical state $\rho$ into memory, apply noisy and unknown single-qubit and two-qubit gates, and perform a noisy and unknown product measurement.
Multiple layers of gates can be performed, where each layer contains multiple non-overlapping gates applied in parallel. The noise in quantum gates may be highly correlated ---
the CPTP map implemented by each gate can depend arbitrarily on how other gates implemented in the same layer are chosen.
We \emph{do not assume} that each gate is affected by only a constant number of other gates in the same layer.
Furthermore the noisy CPTP maps are \emph{not} guaranteed to be close to the ideal unitary gates up to a constant error.
Under these assumptions, it is not possible to learn the intrinsic description of the device to arbitrary accuracy. In addition, exponentially many experiments are required to learn the extrinsic behaviors to high but constant accuracy because the noise processes are highly non-local.
Because the noise cannot be learned efficiently, it is not clear how to make use of standard error mitigation techniques \cite{temme2017error, endo2018practical, kandala2019error}.

Even though the device cannot be learned to arbitrarily small error, one can learn the intrinsic description of the device up to a certain accuracy floor.
To establish this, we adapt the strategy from the proof of Theorem~\ref{thm:int-description} to this modified setting.
Recall that in that proof we found a set of states with a specified geometry, and then built quantum tomography protocols making use of those states.
Under the assumptions of Theorem~\ref{thm:int-description}, we could find states with the desired geometry to arbitrary accuracy.
Under our current assumptions, this is no longer possible, but the geometry can be realized with a small error $\eta$, which suffices for characterizing two-qubit gates up to an error $\epsilon = \mathcal{O}(\eta)$.
Details of the procedure and proof are in Appendix~\ref{sec:algorithm-partial-learn}.

With these noisy two-qubit gates we can perform a noisy version of Bell measurement on qubit pairs, and hence execute a noisy version of the quantum-enhanced protocol that achieves an exponential quantum advantage in the ideal case \cite{huang2021information, chen2021exponential, huang2021quantum}.
For an $n$-qubit system this protocol calls for the evaluation of the parity of $\mathcal{O}(n)$ measurement outcomes, so that $(1-4 \epsilon)^{\mathcal{O}(n)}$ measurement repetitions suffice to extract a statistically useful result.
In contrast, $\Omega(2^n)$ copies of $\rho$ are needed to make accurate predictions in the conventional scenario, even in the case where the state $\rho$ is separable.
Therefore, the noisy quantum-enhanced protocol has a significant polynomial advantage over the optimal protocol in the ideal conventional scenario if $\epsilon$ is small.
This conclusion is expressed by the following theorem.

\begin{theorem}[Quantum advantage with noisy devices] \label{thm:qadv-noisy}
There exists a set of unentangled physical states $\rho$ and properties we would like to predict, such that if we need $N_{\mathrm{Q}}$ copies of $\rho$ in the noisy quantum-enhanced scenario to predict the property, then the required number of copies in noiseless conventional experiments must be
$N_{\mathrm{C}} = \Omega(N_{\mathrm{Q}}^{a})$, where $a = -\log(2) / (2 \log(1 - 4\epsilon)) = \mathcal{O}(1 / \epsilon)$ and $\epsilon$ is the error on each of the two-qubit operations.
\end{theorem}

\noindent If each of the two-qubit operations has an error $\epsilon$ of at most $0.5\%$, then we can obtain a separation of $N_{\mathrm{C}} = \Omega(N_{\mathrm{Q}}^{17.15})$.
The detailed proof of Theorem~\ref{thm:qadv-noisy} is given in Appendix~\ref{sec:qadv-learn-incomplete}.

\section{Conclusion}

We have developed a rigorous theory for reasoning about what can be learned from noisy quantum experiments, emphasizing the distinction between learning the extrinsic behavior of a noisy quantum system and the more challenging task of learning its intrinsic description.
We have applied our framework to several examples of learning tasks and to give a large polynomial quantum advantage for a learning task using noisy quantum devices.
While we have largely focused on issues of principle, more practical considerations, including a rigorous neural network algorithm for learning noisy quantum systems, will be discussed in future work.

This theory opens up several future directions.
In particular, the essential features of learnable and unlearnable model classes remain to be clarified.
In Appendix~\ref{sec:basiclearna} we presented some deductive rules relating the learnability of distinct model classes.
Might this framework be extended, so that the learnability or unlearnability of a wide variety of model classes can be inferred by applying such deductive rules to a small number of paradigmatic cases?
We have also seen that some model classes are very hard to learn, while others are relatively easy.
Can we deliberately engineer our quantum systems to make learning easier?
This question evokes an oft-cited principle of software design --- programs are more broadly useful if they are easier to troubleshoot.

More philosophically, it is remarkable that humans, assisted by our classical machines, can gain precise knowledge of the elusive quantum world.
What is the explanation for this ability, and what knowledge will remain beyond our grasp due to fundamental limitations on learning with classical agents?
Deepening our understanding of what can be learned from quantum experiments will not only facilitate the advance of quantum technology, but may also illuminate what is knowable and unknowable about the physical world.

\subsection*{Acknowledgments:}
\vspace{-0.5em}
{ The authors thank Sitan Chen, Jordan Cotler, Jerry Li, Richard Kueng, and Thomas Vidick for valuable input and inspiring discussions.
HH is supported by a Google PhD Fellowship.
JP acknowledges funding from  the U.S. Department of Energy Office of Science, Office of Advanced Scientific Computing Research, (DE-NA0003525, DE-SC0020290), and the National Science Foundation (PHY-1733907). The Institute for Quantum Information and Matter is an NSF Physics Frontiers Center.}

\newpage
\vspace{2.5em}
\appendix

\renewcommand{\appendixname}{APPENDIX}
\renewcommand{\thesubsection}{\MakeUppercase{\alph{section}}.\arabic{subsection}}
\renewcommand{\thesubsubsection}{\MakeUppercase{\alph{section}}.\arabic{subsection}.\alph{subsubsection}}
\makeatletter
\renewcommand{\p@subsection}{}
\renewcommand{\p@subsubsection}{}
\makeatother

\numberwithin{theorem}{section} 

\begin{center}\textbf{ \Large{}Appendices}{\Large\par}
\end{center}

\tableofcontents


\section{World models}
\label{sec:qworld}

We consider a general framework involving classical agents interacting with a quantum system.
This is a mathematical framework for reasoning about what experimentalists could learn from a finite-dimensional quantum system that they can interact with by various means.

\begin{definition}[$d$-dimensional world model]
Given sets $\mathcal{X}, \mathcal{Y}, \mathcal{Z}$ denoting the action space and a finite set~$\mathcal{B}$ denoting the possible measurement outcomes. A $d$-dimensional world model $\mathcal{W}$ is a tuple with three sets
\begin{equation}
    \mathcal{W} = \left(\{\rho_x\}_{x \in \mathcal{X}}, \{\mathcal{E}_y\}_{y \in \mathcal{Y}}, \{\mathcal{M}_z\}_{z \in \mathcal{Z}} \right),
\end{equation}
where $\rho_x$ is a $d$-dimensional density matrix, $\mathcal{E}_y$ is a $d$-dimensional CPTP map, $\mathcal{M}_z = \{M_{z b}\}_{b \in \mathcal{B}}$ is a POVM with finitely many elements indexed by $b \in \mathcal{B}$.
\end{definition}

Consider $d = 2$, which is equivalent to a qubit system.
We give an example to illustrate the above definition.
Let $\mathcal{X} = \{(\theta, \phi) \,\, | \,\, \theta \in [0, \pi], \phi \in [0, 2\pi]\}$ be an uncountably large set; we define
\begin{equation}
\rho_{\theta, \phi} = \frac{I + \sin(\theta) \cos(\phi) X + \sin(\theta) \sin(\phi) Y + \cos(\theta) Z}{2}.
\end{equation}
In this world, we can prepare any pure state on the single-qubit bloch sphere.
Let $\mathcal{Y} = \{h, t\}$ be a finite set consisting of two elements.
We consider $\mathcal{E}_h(\rho) = H\rho H^\dagger$ to be the Hadamard gate and $\mathcal{E}_t(\rho) = T\rho T^\dagger$ to be the $T$ gate.
Finally, we consider $\mathcal{Z} = \{0\}$ to be a singleton and $\mathcal{B} = \{0, 1\}$ to be a two-outcome space, where $\mathcal{M}_0 = \{\ketbra{0}{0}, \ketbra{1}{1}\}$ is the computational basis measurement.
$\mathcal{W}$ defines a single-qubit world where one can perform universal single-qubit quantum computation.
Alternatively, one could also consider $\mathcal{X}$ to be a set of $\vec{n} \in \mathbb{R}^3$ with $\norm{\vec{n}}_2 \leq 1$.
Or we could consider $\mathcal{Y}$ to be a set of sequences where each sequence is a pulse sequence the experimentalist could use to control the qubit system.
Intuitively, $\mathcal{X}, \mathcal{Y}, \mathcal{Z}$ contain descriptions of the actions an experimentalist could perform on the finite-dimensional quantum system, and $\mathcal{B}$ contains descriptions of the measurement outcomes.

\begin{remark}
In the above definition, we can have $x_1, x_2 \in \mathcal{X}$, such that $x_1 \neq x_2$ but $\rho_{x_1} = \rho_{x_2}$.
This construction encodes the intuition that there could be two different actions an experimentalist could perform that will result in the same initial state. For example, we can generate the state $\ket{1}$ from $\ket{0}$ by applying $\pi$ rotation along the X axis or the Y axis.
\end{remark}

The experimentalists could interact with the $d$-dimensional world model by performing experiments. The experimentalist selects a state $\rho_x$ to be prepared, composes various different evolutions $\mathcal{E}_{y_1}, \ldots, \mathcal{E}_{y_L}$, then reads out the final state through a chosen measurement apparatus $\mathcal{M}_z$.

\begin{definition}[Experiment] \label{def:experiment}
Given a $d$-dimensional world model $\mathcal{W}$.
An experiment is a list of finite elements given as
\begin{equation}
    E = (x \in \mathcal{X}, y_1 \in \mathcal{Y}, \ldots, y_L \in \mathcal{Y}, z \in \mathcal{Z}).
\end{equation}
Each experiment results in an outcome $b \in \mathcal{B}$ with probability $\Tr\left( M_{z b} \left(\mathcal{E}_{y_L} \circ \ldots \circ \mathcal{E}_{y_1}\right)(\rho_x) \right).$
\end{definition}

\subsection{Time dependent and non-Markovian models}
\label{sec:time-depend}

From the definition of experiment given in Def.~\ref{def:experiment}, we can see that applying the same action $y \in \cY$ results in the same CPTP map $\cE_{y}$.
In many practical settings, the quantum systems could have explicit time dependence ($\cE_{y}$ depends on the time it is applied) and could be non-Markovian (each physical operation depends on what operations are performed before).
In this subsection, we briefly describe how models with explicit time dependence and non-Markovian behavior can be reduced to a model that is time-independent and Markovian.

Suppose we have a model $\cW^{(\mathrm{exception})}$ where the sequence of actions $x, y_1, y_2, \ldots, y_L, z$ results in an outcome $b \in \mathcal{B}$ with probability
\begin{equation}
    \Tr\left( M_{z b | x, y_1, \ldots, y_L} \left(\mathcal{E}_{y_L | x, y_1, \ldots, y_{L-1}} \circ \ldots \circ \mathcal{E}_{y_1 | x}\right)(\rho_x) \right),
\end{equation}
and all the states, CPTP maps, POVMs are $d$-dimensional.
This model $\cW^{(\mathrm{exception})}$ is time-dependent and non-Markovian, hence is formally not covered by our definition of world model.
However, we can consider a world model $\cW$ that yields the same experimental outcome as $\cW^{(\mathrm{exception})}$ but is time-independent and Markovian.
We simply define $\cW' = = \left(\{\rho'_x\}_{x \in \mathcal{X}}, \{\mathcal{E}'_y\}_{y \in \mathcal{Y}}, \{\mathcal{M}'_z\}_{z \in \mathcal{Z}} \right)$ to be a world model, where the state space contains both the original $d$-dimensional state and memory state that stores the history of the actions.
The initial state $\rho'_x$ is given by
\begin{equation}
    \rho'_x = \rho_x \otimes \ketbra{\emptyset}{\emptyset},
\end{equation}
where the history is initialized to be empty.
The CPTP map $\cE'_{y}$ implements a CPTP based on the history stored in the memory state,
\begin{equation}
    \cE'_y( \rho \otimes \ketbra{\mathrm{history}}{\mathrm{history}} ) = \cE_{y | \mathrm{history}}(\rho) \otimes \ketbra{\mathrm{history}, y}{\mathrm{history}, y},
\end{equation}
where $\mathrm{history}$ is a sequence of the past actions $x, y_1, y_2, y_3, \ldots$.
Similarly, the POVM $\cM'_{z}$ in the newly defined world model $\cW'$ also depends on the history,
\begin{equation}
    \Tr( M'_{zb} ( \rho \otimes \ketbra{\mathrm{history}}{\mathrm{history}} ) ) = \Tr( M_{zb | \mathrm{history}} \rho).
\end{equation}
It is not hard to see that for the time-independent and Markovian world model $\cW'$, an experiment specified by the sequence of actions $x, y_1, y_2, \ldots, y_L, z$ results in an outcome $b \in \mathcal{B}$ with probability
\begin{align}
    &\Tr\left( M'_{z b} \left(\mathcal{E}'_{y_L} \circ \ldots \circ \mathcal{E}'_{y_1}\right)(\rho'_x) \right) \\
    = &\Tr\left( M_{z b | x, y_1, \ldots, y_L} \left(\mathcal{E}_{y_L | x, y_1, \ldots, y_{L-1}} \circ \ldots \circ \mathcal{E}_{y_1 | x}\right)(\rho_x) \right).
\end{align}
Together, we see that time-dependent and non-Markovian models can be covered by our definition at the expense of having a world model with much larger dimension.

We do mention a pathological case. Suppose that given a sequence of $x, y_1, \ldots, y_{\ell - 1}$, the CPTP map $\mathcal{E}_{y_\ell | x, y_1, \ldots, y_{\ell-1}}$ is uncomputable. This means that the size of the memory state for storing the history would have to grow unboundedly.
In this case, we need to reduce time-dependent and non-Markovian models to an infinite-dimensional time-independent and Markovian world model.
This pathological case is not practically relevant since there is no finite description of such a time-dependent and non-Markovian model.

\subsection{Basic properties and relations}

As the experimentalists improve their physical control (lasers, cavity, etc.), more initial states $\rho_x$ can be created, more evolutions $\mathcal{E}_y$ can be performed, and more types of qubit readout $\mathcal{M}_z$ can be achieved.
We could imagine an ideal case, where we are able to generate all states, perform all operations, and conduct all measurements. We consider such a world model to be complete.
A formal definition is given below. 

\begin{definition}[Completeness]
A $d$-dimensional $\mathcal{W} = \left(\{\rho_x\}_{x \in \mathcal{X}}, \{\mathcal{E}_y\}_{y \in \mathcal{Y}}, \{\mathcal{M}_z\}_{z \in \mathcal{Z}} \right)$
is complete if
\begin{itemize}
    \item for all states $\rho$, $\exists x \in \mathcal{X}$, $\rho_x = \rho$,
    \item for all CPTP maps $\mathcal{E}$, $\exists y \in \mathcal{Y}$, $\mathcal{E}_y = \mathcal{E}$,
    \item for all POVM $\mathcal{M}$ with outcomes indexed by $b \in \mathcal{B}$, $\exists z \in \mathcal{Z}$, $\mathcal{M}_z = \mathcal{M}$.
\end{itemize}
\end{definition}

We say the world model $\mathcal{W}$ has been extended to a richer world model $\mathcal{W}'$ if $\mathcal{W}'$ contains more actions corresponding to more initial states, quantum evolutions, and POVM measurements.
The formal definition is given below.

\begin{definition}[Extension] \label{def:extension}
A $d$-dimensional world model $\mathcal{W}' = \left(\{\rho_x'\}_{x \in \mathcal{X}'}, \{\mathcal{E}_y'\}_{y \in \mathcal{Y}'}, \{\mathcal{M}_z'\}_{z \in \mathcal{Z}'} \right)$ is an extension of a $d$-dimensional world model $\mathcal{W} = \left(\{\rho_x\}_{x \in \mathcal{X}}, \{\mathcal{E}_y\}_{y \in \mathcal{Y}}, \{\mathcal{M}_z\}_{z \in \mathcal{Z}} \right)$, denoted as $\mathcal{W}' \rhd \mathcal{W}$, if the following conditions hold
\begin{itemize}
    \item $\mathcal{X} \subseteq \mathcal{X}'$ and $\forall x \in \mathcal{X}, \rho_x = \rho_x'$ (State extension),
    \item $\mathcal{Y} \subseteq \mathcal{Y}'$ and $\forall y \in \mathcal{Y}, \mathcal{E}_y = \mathcal{E}_y'$ (CPTP map extension),
    \item $\mathcal{Z} \subseteq \mathcal{Z}'$ and $\forall z \in \mathcal{Z}, \mathcal{M}_z = \mathcal{M}_z'$ (POVM extension).
\end{itemize}
\end{definition}

We are now ready to define equivalence between two world models.
Before giving the formal definition, let us consider two $2$-dimensional worlds $\mathcal{W}_A, \mathcal{W}_B$ with the same spaces $\mathcal{X} = \{0\}, \mathcal{Y} = \{h, t\}, \mathcal{Z} = \{0\}, \mathcal{B} = \{0, 1\}$.
Furthermore, we consider the particular actions in the two world models $\mathcal{W}_A, \mathcal{W}_B$ to be given by
\begin{align}
    \rho^A_0 &= I/2, & \mathcal{E}^A_h(\rho) &= H\rho H^\dagger, & \mathcal{E}^A_t(\rho) &= T\rho T^\dagger, & \mathcal{M}^A_0 &= \{\ketbra{0}{0}, \ketbra{1}{1}\}, \label{eq:worldA} \\
    \rho^B_0 &= I/2, & \mathcal{E}^B_h(\rho) &= I/2, & \mathcal{E}^B_t(\rho) &= I/2, & \mathcal{M}^B_0 &= \{\ketbra{0}{0}, \ketbra{1}{1}\}.\label{eq:worldB}
\end{align}
It is not hard to show that we cannot distinguish between world A and B by performing experiments using the limited set of actions --- in both cases the measurement outcomes are sampled from the uniform distribution.
However, we can clearly see that the two world models are intrinsically different. In particular, in world A, the set of maps is a universal gate set that generates a dense subset of $\mathrm{SU}(2)$. But, in world B, all the maps are completely depolarizing channels.
Even though world A and B can not be distinguish using the limited set of actions $\mathcal{X} = \{0\}, \mathcal{Y} = \{h, t\}, \mathcal{Z} = \{0\}$, the two worlds are fundamentally different.
By adding new actions, such as the ability to prepare some non-trivial states, we can distinguish between the two world models by performing experiments.

To discuss these concepts in a rigorous manner, we formally define the following relations between two world models $\mathcal{W}_1, \mathcal{W}_2$.
We consider two world models to be equal $\mathcal{W}_1 = \mathcal{W}_2$ if all of the states, maps, and POVMs are equal.
And we say the two world models are different $\mathcal{W}_1 \neq \mathcal{W}_2$ if one of the states, maps, or POVMs is different.
In the above example, the two world models $\mathcal{W}_A, \mathcal{W}_B$ are different because the CPTP maps are different $\mathcal{E}^A_h \neq \mathcal{E}^B_h$ and $\mathcal{E}^A_t \neq \mathcal{E}^B_t$.

\begin{definition}[Equality] \label{def:equal}
Consider two $d$-dimensional world models $\mathcal{W}_A, \mathcal{W}_B$
with the same spaces $\mathcal{X}, \mathcal{Y}, \mathcal{Z}, \mathcal{B}$.
$\mathcal{W}_A$ is equal to $\mathcal{W}_B$, denoted as $\mathcal{W}_A = \mathcal{W}_B$, if all states are the same
$\rho^A_x = \rho^B_x, \forall x \in \mathcal{X},$ all CPTP maps are the same $ \mathcal{E}^A_y = \mathcal{E}^B_y, \forall y \in \mathcal{Y},$ and all POVMs are the same $M^A_{z b} = M^B_{z b}, \forall z \in \mathcal{Z}, b \in \mathcal{B}$.
\end{definition}

Then, we consider two world models to be weakly indistinguishable if they can not be distinguished using the set of actions in the world model.
In the example given in Equation~\eqref{eq:worldA}~and~\eqref{eq:worldB}, $\mathcal{W}_A, \mathcal{W}_B$ are weakly indistinguishable because the measurement outcomes are always uniformly distributed.

\begin{definition}[Weakly indistinguishable] \label{def:weakly-indist}
Consider two $d$-dimensional world models $\mathcal{W}_A, \mathcal{W}_B$
with the same spaces $\mathcal{X}, \mathcal{Y}, \mathcal{Z}, \mathcal{B}$.
$\mathcal{W}_A$ and $\mathcal{W}_B$ are weakly indistinguishable if for any experiment $E = (x \in \mathcal{X}, y_1 \in \mathcal{Y}, \ldots, y_L \in \mathcal{Y}, z \in \mathcal{Z})$ and outcome $b \in \mathcal{B}$, we have
\begin{equation} \label{eq:probeq}
    \Tr\left( M^{A}_{z b} \left(\mathcal{E}^{A}_{y_L} \circ \ldots \circ \mathcal{E}^{A}_{y_1}\right)(\rho^{A}_x) \right) = \Tr\left( M^{B}_{z b} \left(\mathcal{E}^{B}_{y_L} \circ \ldots \circ \mathcal{E}^{B}_{y_1}\right)(\rho^{B}_x) \right),
\end{equation}
i.e., the probabilities for obtaining the outcome $b$ in experiment $E$ are the same.
\end{definition}

We say two world models are strongly indistinguishable or equivalent to one another if they can not be distinguished by adding any set of actions.
World models $\mathcal{W}_A, \mathcal{W}_B$ are not equivalent because adding a non-completely-mixed state enables us to distinguish between $\mathcal{W}_A$ and $\mathcal{W}_B$.

\begin{definition}[Strongly indistinguishable / Equivalence] \label{def:equiv}
Consider two $d$-dimensional world models $\mathcal{W}_A, \mathcal{W}_B$
with the same spaces $\mathcal{X}, \mathcal{Y}, \mathcal{Z}, \mathcal{B}$.
$\mathcal{W}_A$ and $\mathcal{W}_B$ are equivalent or strongly indistinguishable, denoted as $\mathcal{W}_A \equiv \mathcal{W}_B$, if for all extensions of $\mathcal{W}_A$,
\begin{equation}
    \mathcal{W}_A' = \left(\{\rho^{A'}_x\}_{x \in \mathcal{X}'}, \{\mathcal{E}^{A'}_y\}_{y \in \mathcal{Y}'}, \{\mathcal{M}^{A'}_z\}_{z \in \mathcal{Z}'} \right) \rhd \mathcal{W}_A,
\end{equation}
there exists an extension of $\mathcal{W}_B$ with the same action space $\mathcal{X}', \mathcal{Y}', \mathcal{Z}'$,
\begin{equation}
    \mathcal{W}_B' = \left(\{\rho^{B'}_x\}_{x \in \mathcal{X}'}, \{\mathcal{E}^{B'}_y\}_{y \in \mathcal{Y}'}, \{\mathcal{M}^{B'}_z\}_{z \in \mathcal{Z}'} \right) \rhd \mathcal{W}_B,
\end{equation}
such that $\mathcal{W}_A'$ and $\mathcal{W}_B'$ are weakly indistinguishable.
\end{definition}

The above definition of equivalence has a natural characterization given by Theorem~\ref{prop:equiv}. Before stating the theorem, let us recall the definition of unitary and anti-unitary transformation $U$.
Given a $d \times d$ complex matrix $C$ with a chosen basis.
We define $\overline{C}$ to be the matrix where we take complex conjugation for all entries in $C$.
A unitary transformation $U$ is a $d \times d$ unitary matrix with $U^{-1} = U^{\dagger}$ that transforms $C$ to $U C U^{-1} = U C U^\dagger$.
An anti-unitary transformation $A$ is a product of a $d \times d$ unitary matrix $U$ and the complex conjugation operator $K$ that transforms $C$ to $A C A^{-1} = U \overline{C} U^\dagger$.
Theorem~\ref{prop:equiv} shows that equivalent world models are related by a unitary or anti-unitary transformation.

\begin{theorem}[A characterization of equivalence] \label{prop:equiv}
Consider two $d$-dimensional world models with the same spaces $\mathcal{X}, \mathcal{Y}, \mathcal{Z}, \mathcal{B}$,
\begin{align}
    \mathcal{W}_A &= \left(\{\rho^{A}_x\}_{x \in \mathcal{X}}, \{\mathcal{E}^{A}_y\}_{y \in \mathcal{Y}}, \{\mathcal{M}^{A}_z\}_{z \in \mathcal{Z}} \right),\\
    \mathcal{W}_B &= \left(\{\rho^{B}_x\}_{x \in \mathcal{X}}, \{\mathcal{E}^{B}_y\}_{y \in \mathcal{Y}}, \{\mathcal{M}^{B}_z\}_{z \in \mathcal{Z}} \right).
\end{align}
$\mathcal{W}_A \equiv \mathcal{W}_B$ if and only if there exists a unitary or anti-unitary transformation $U$, such that
\begin{align}
    &\rho^{B}_x = U \rho^{A}_x U^{-1}, &\,\, \forall x \in \mathcal{X},\\
    &\mathcal{E}^{B}_y(\cdot) = U \mathcal{E}^{A}_y( U^{-1} (\cdot) U) U^{-1}, &\,\, \forall y \in \mathcal{Y},\\
    &M^B_{z b} = U M^A_{z b} U^{-1}, &\,\, \forall z \in \mathcal{Z}, b \in \mathcal{B}.
\end{align}
\end{theorem}

We defer the proof to Appendix~\ref{sec:proofEquivalence}.
As an example, the following two world models with the action spaces $\mathcal{X} = \{0\}, \mathcal{Y} = \{h, s\}$ for Hadamard and Phase gates, $ \mathcal{Z} = \{0\},$ and outcome space $\mathcal{B} = \{0, 1\}$ are equivalent (related by an anti-unitary transformation $XK$):
\begin{align}
    \rho^A_0 &= \ketbra{0}{0}, & \mathcal{E}^A_h(\rho) &= H\rho H^\dagger, & \mathcal{E}^A_s(\rho) &= S\rho S^\dagger, & \mathcal{M}^A_0 &= \{\ketbra{0}{0}, \ketbra{1}{1}\}, \\
    \rho^B_0 &= \ketbra{1}{1}, & \mathcal{E}^B_h(\rho) &= H\rho H^\dagger, & \mathcal{E}^B_s(\rho) &= S\rho S^\dagger, & \mathcal{M}^B_0 &= \{\ketbra{1}{1}, \ketbra{0}{0}\}.
\end{align}
The possibility to describe the same physical world by two distinct descriptions arises from the intrinsic degeneracy in quantum mechanics: the freedom to choose an arbitrary basis in the Hilbert space (the unitary relation) or reverse the direction of time (the anti-unitary relation).

\section{A characterization of equivalence: Proof}
\label{sec:proofEquivalence}

We will focus on showing that $\mathcal{W}_A \equiv \mathcal{W}_B$ implies the existence of a unitary or anti-unitary transformation.
The other direction can be shown easily by noting that for all world model extensions of $\mathcal{W}_A$, we can extend $\mathcal{W}_B$ using the same unitary or anti-unitary transformation $U$.

First, we extend world model $\mathcal{W}_A$ to world model $\mathcal{W}'_A$ that comes with an expanded state preparation actions $\mathcal{X}' = \mathcal{X} \cup \Omega^{\mathrm{pure}}$ and an expanded measurement actions $\mathcal{Z}' = \mathcal{Z} \cup \Omega^{\mathrm{pure-POVM}}$.
In particular, $\rho^{A'}_{\xi}, \forall \xi \in \Omega^{\mathrm{pure}}$ consists of all the $d$-dimensional pure states, and $\mathcal{M}^{A'}_{\zeta}, \forall \zeta \in \Omega^{\mathrm{pure-POVM}}$ consists of all the POVMs such that a particular POVM element associated to $b^* \in \mathcal{B}$ is a pure state.
The definition of equivalence shows that there exists an extension $\mathcal{W}'_B$ with the same action space as $\mathcal{W}'_A$ such that $\mathcal{W}'_A$ and $\mathcal{W}'_B$ are weakly indistinguishable, i.e., all experiments yield the same distribution.
The above condition yields the following,
\begin{equation}\label{eq:weak-indistin-pure}
    \Tr(M^{A'}_{\zeta b^*} \rho^{A'}_{\xi}) = \Tr(M^{B'}_{\zeta b^*} \rho^{B'}_{\xi}), \forall \xi \in \Omega^{\mathrm{pure}}, \zeta \in \Omega^{\mathrm{pure-POVM}}.
\end{equation}
This concludes the basic construction of the extended world models.

Recall that $M^{A'}_{\zeta b^*}, \forall \zeta \in \Omega^{\mathrm{pure-POVM}}$ consists of all pure states.
Hence, for each $\xi \in \Omega^{\mathrm{pure}}$, there exists $\zeta(\xi) \in \Omega^{\mathrm{pure-POVM}}$ such that $M^{A'}_{\zeta b^*} = \rho^{A'}_{\xi}$ are the same pure state.
We extend $\rho^{A'}_{\xi}$ to an orthonormal set of basis consisting of $d$ pure states $\{ \rho^{A'}_{\xi_1}, \ldots, \rho^{A'}_{\xi_d} \}$, where $\xi_1 = \xi$.
We have the following from the above construction and Eq.~\eqref{eq:weak-indistin-pure},
\begin{equation}
    \Tr\left(M^{A'}_{\zeta(\xi_j) b^*} \rho^{A'}_{\xi_i}\right) = \delta_{i j} = \Tr\left( M^{B'}_{\zeta(\xi_j) b^*} \rho^{B'}_{\xi_i} \right), \forall i, j \in \{1, \ldots, d\},
\end{equation}
where $\delta_{i j} = 1$ if $i = j$ and $0$ otherwise.
We are now going to utilize the structure of the quantum states and POVM elements, in particular, a quantum state is a positive-semidefinite matrix with trace one, and a POVM element is a positive-semidefinite matrix with eigenvalues less than equal to one.
In particular, we use the following basic lemma that can be established by induction.

\begin{lemma} \label{lem:statePOVM-dim}
Consider $\ell \geq 1$. Given quantum states $\{\rho_i\}_{i = 1, \ldots, \ell}$ and POVM elements $\{F_j\}_{j=1, \ldots, \ell}$. If $\Tr(F_j \rho_i) = \delta_{ij}, \forall i, j \in \{1, \ldots, \ell\}$, then the collection of eigenvectors of $\rho_i$ with non-zero eigenvalues over all $i$ from 1 to $\ell$ span a subspace with dimension at least $\ell$.
\end{lemma}
\begin{proof}
Consider the eigenvectors of $\rho_i$ with non-zero eigenvalues to be $\{v^{i}_a \}_{a \in A_i}$ and the associated eigenvalues be $\{p^{i}_a\}_{a \in A_i}$, then $ \sum_{a \in A_i} p^{i}_a (v^{i}_a)^\dagger F_i v^{i}_a = 1$.
Since $0 \leq (v^{i}_a)^\dagger F_i v_a^{i} \leq 1$ (from the definition of POVM element), $p_a^{i} > 0$ (we only consider non-zero eigenvalues), $\sum_{a \in A_i} p_a^{i} = 1$ (from the definition of quantum state), and $\sum_{a \in A_i} p_a (v_a^{i})^\dagger F_i v_a^{i} = 1$, we have
\begin{equation} \label{eq:vFvone}
    (v_a^{i})^\dagger F_i v_a^{i} = 1, \forall a \in A_i.
\end{equation}
Similarly, for all $j \neq i$, $\sum_{a \in A_i} p_a (v_a^{i})^\dagger F_j v_a^{i} = 0$ implies that $(v_a^{i})^\dagger F_j v_a^{i} = 0, \forall a \in A_i$. Equivalently,
\begin{equation} \label{eq:vFvzero}
    \sqrt{F_j} v_a^{i} = 0, \forall a \in A_i,
\end{equation}
when $i \neq j$.
With the basic results given above, we are ready to prove the statement through induction. The base case $\ell=1$ is trivially true. Suppose the statement holds for $\ell-1$.
We assume that all the eigenvectors of $\rho_{\ell}$ with non-zero eigenvalues lie in the span of the eigenvectors with non-zero eigenvalues for $\rho_{i}$ with $i < \ell$.
Under this assumption, for all $a \in A_\ell$, there exists a set of coefficients $\{c_{i, a'}\}$ such that $v^{\ell}_a = \sum_{i < \ell} \sum_{a' \in A_i} c_{i, a'} v^{i}_{a'}$. This implies that
\begin{equation}
    1 = (v^{\ell}_a)^\dagger F_\ell v^{\ell}_a = \left(\sum_{i < \ell} \sum_{a' \in A_i} c_{i, a'} \sqrt{F_\ell} v^{i}_{a'}\right)^\dagger \left(\sum_{i < \ell} \sum_{a' \in A_i} c_{i, a'} \sqrt{F_\ell} v^{i}_{a'}\right) = 0.
\end{equation}
The first equality follows from Eq.~\eqref{eq:vFvone}. The last equality follows from Eq.~\eqref{eq:vFvzero}.
The contradiction shows that one of the eigenvectors of $\rho_{\ell}$ with non-zero eigenvalues is not in the span of the eigenvectors with non-zero eigenvalues for $\rho_i$ with $i < \ell$.
Therefore the statement holds for $\ell$.
\end{proof}

If there exists $k \in \{1, \ldots, d\}$, such that the rank of $M^{B'}_{\zeta(\xi_k) b^*}$ is greater than one, then the $d-1$ states $\{ \rho^{B'}_{\xi_i} \}_{i \neq k}$ must have their eigenvectors with non-zero eigenvalues span a $(d-2)$-dimensional subspace to ensure that $\Tr\left( M^{B'}_{\zeta(\xi_k) b^*} \rho^{B'}_{\xi_i} \right) = 0$.
However, from Lemma~\ref{lem:statePOVM-dim}, the eigenvectors of the $d-1$ states $\{ \rho^{B'}_{\xi_i} \}_{i \neq k}$ with non-zero eigenvalues must span at least a $(d-1)$-dimensional state space.
The contradiction implies that $M^{B'}_{\zeta(\xi_j) b^*}, \forall j \in \{1, \ldots, d\}$ must all be rank-one matrices.
The condition $\Tr\left( M^{B'}_{\zeta(\xi_1) b^*} \rho^{B'}_{\xi_1} \right) = 1$ then implies that $\rho^{B'}_{\xi_1}$ must be a pure state and $M^{B'}_{\zeta(\xi_1) b^*} = \rho^{B'}_{\xi_1}$.
We have now shown the following statement:
\begin{equation} \label{eq:statementMrho}
\mbox{
$\forall \xi \in \Omega^{\mathrm{pure}}$, $\rho^{B'}_{\xi}$ is a pure state and the POVM element $M^{B'}_{\zeta(\xi) b^*} = \rho^{B'}_{\xi}$.
}
\end{equation}
An implication of this result is that $\forall \xi_1, \xi_2 \in \Omega^{\mathrm{pure}}$, we have
\begin{equation} \label{eq:weak-indistin-symmetry}
    \Tr\left(\rho^{A'}_{\xi_1} \rho^{A'}_{\xi_2}\right) = \Tr\left(M^{A'}_{\zeta(\xi_1) b^*} \rho^{A'}_{\xi_2}\right) = \Tr\left(M^{B'}_{\zeta(\xi_1) b^*} \rho^{B'}_{\xi_2}\right) = \Tr\left(\rho^{B'}_{\xi_1} \rho^{B'}_{\xi_2}\right),
\end{equation}
where the second equation follows from Eq.~\eqref{eq:weak-indistin-pure}.

We can now construct a transformation $T$ over pure state space by considering $T(\rho^{A'}_{\xi}) = \rho^{B'}_{\xi}, \forall \xi \in \Omega^{\mathrm{pure}}.$
$T$ is a transformation that takes pure states to pure states that satisfies
\begin{equation}
    \Tr\left(\rho^{A'}_{\xi_1} \rho^{A'}_{\xi_2}\right) = \Tr\left(T(\rho^{A'}_{\xi_1}) T(\rho^{A'}_{\xi_2}) \right), \forall \xi_1, \xi_2 \in \Omega^{\mathrm{pure}},
\end{equation}
as a result of Eq.~\eqref{eq:weak-indistin-symmetry}.
Such a transformation $T$ is also known as a symmetry transformation.
By Wigner's theorem, $T$ must take the following form:
\begin{equation}
    T(\rho) = U \rho U^{-1},
\end{equation}
where $U$ is a unitary or an anti-unitary transformation.
For a proof of Wigner's theorem, see Appendix~A of Chapter~2 in \textit{The Quantum Theory of Fields, Vol. 1}, by Weinberg.
The above representation of the symmetry transformation $T$ yields
\begin{equation} \label{eq:state-rotation}
    \rho^{B'}_{\xi} = U \rho^{A'}_{\xi} U^{-1},\,\, \forall \xi \in \Omega^{\mathrm{pure}}.
\end{equation}
Using Eq.~\eqref{eq:statementMrho}, we also have the following relation for a subset of POVM elements,
\begin{equation} \label{eq:povm-rotation}
    M^{B'}_{\zeta(\xi) b^*} = U M^{A'}_{\zeta(\xi) b^*} U^{-1}, \,\, \forall \xi \in \Omega^{\mathrm{pure}}.
\end{equation}
Intuitively, we will now use the pure states $\rho_\xi, \forall \xi \in \Omega^{\mathrm{pure}}$ to perform quantum POVM tomography. Then use the rank-one POVM elements $M_{\zeta(\xi) b^*}, \forall \xi \in \Omega^{\mathrm{pure}}$ to perform quantum state tomography.
Finally, we use both $\rho_\xi$ and $M_{\zeta(\xi) b^*}$ to perform quantum process tomography. Together, we have established the statement of Theorem~\ref{prop:equiv}.

\paragraph{POVMs:}
    $\forall z \in \mathcal{Z}, b \in \mathcal{B}$, we have $\Tr(M^{A'}_{z b} \rho^{A'}_\xi) = \Tr(M^{B'}_{z b} \rho^{B'}_\xi) = \Tr(M^{B'}_{z b} U \rho^{A'}_{\xi} U^{-1}), \forall \xi \in \Omega^{\mathrm{pure}}$. The first equality follows from the weak indistinguishability between $\mathcal{W}_A'$ and $\mathcal{W}_B'$.
    The second equality follows from Eq.~\eqref{eq:state-rotation}.
    We consider a measure $\mu_\xi$ over $\xi$ such that $\rho^{A'}_\xi$ forms the Haar measure over the pure state space.
    The Haar integration formulas $\int d\mu \ketbra{\psi}{\psi} = I / d$ and $\int d\mu \ketbra{\psi}{\psi}^{\otimes 2} = (I + \mathrm{SWAP}) / (d (d+1))$ give us
    \begin{align}
        \frac{\Tr(M^{A'}_{z b})}{d} = \int d\mu_\xi \Tr(M^{A'}_{z b} \rho^{A'}_{\xi}) &= \int d\mu_\xi \Tr(M^{B'}_{z b} U \rho^{A'}_{\xi} U^{-1}) = \frac{\Tr(M^{B'}_{z b})}{d},\\
        \frac{\Tr(M^{A'}_{z b}) I + M^{A'}_{z b}}{d(d+1)} = \int d\mu_\xi \Tr(M^{A'}_{z b} \rho^{A'}_{\xi}) \rho^{A'}_{\xi}
        &= \int d\mu_\xi \Tr(M^{B'}_{z b} U \rho^{A'}_{\xi} U^{-1}) \rho^{A'}_{\xi} \\
        &= \frac{\Tr(M^{B'}_{z b}) I + U^{-1} M^{B'}_{z b} U}{d(d+1)}.
    \end{align}
    Therefore, $\forall z \in \mathcal{Z}, b \in \mathcal{B}, M^B_{z b} = M^{B'}_{z b} = U M^{A'}_{z b} U^{-1} = U M^{A}_{z b} U^{-1}$ as stated in Theorem~\ref{prop:equiv}.

\paragraph{States:}
    $\forall x \in \mathcal{X}$, we have $\Tr(M^{A'}_{\zeta(\xi) b^*} \rho^{A'}_x) = \Tr(M^{B'}_{\zeta(\xi) b^*} \rho^{B'}_x) = \Tr(U M^{A'}_{\zeta(\xi) b^*} U^{-1} \rho^{B'}_{x}), \forall \xi \in \Omega^{\mathrm{pure}}$. The first equality follows from the weak indistinguishability between $\mathcal{W}_A'$ and $\mathcal{W}_B'$.
    The second equality follows from Eq.~\eqref{eq:povm-rotation}.
    We consider a measure $\mu_\xi$ over $\xi$ such that $M^{A'}_{\zeta(\xi) b^*}$ forms the Haar measure over the pure state space.
    The Haar integration formulas give us
    \begin{align}
        \frac{ I + \rho^{A'}_x}{d(d+1)} &= \int d\mu_\xi \Tr(M^{A'}_{\zeta(\xi) b^*} \rho^{A'}_x) M^{A'}_{\zeta(\xi) b^*}\\
        &= \int d\mu_\xi \Tr(U M^{A'}_{\zeta(\xi) b^*} U^{-1} \rho^{B'}_{x}) M^{A'}_{\zeta(\xi) b^*}
        = \frac{ I + U^{-1} \rho^{B'}_x U}{d(d+1)}.
    \end{align}
    Therefore, $\forall x \in \mathcal{X}, \rho^{B}_{x} = \rho^{B'}_{x} = U \rho^{A'}_{x} U^{-1} = U \rho^{A}_{x} U^{-1}$ as stated in Theorem~\ref{prop:equiv}.

\paragraph{CPTP maps:}
    $\forall y \in \mathcal{Y}$, we have $\forall \xi_1, \xi_2 \in \Omega^{\mathrm{pure}},$
    \begin{align}
        \Tr(M^{A'}_{\zeta(\xi_2) b^*} \mathcal{E}^{A'}_y (\rho^{A'}_{\xi_1}) ) &= \Tr(M^{B'}_{\zeta(\xi_2) b^*} \mathcal{E}^{B'}_y (\rho^{B'}_{\xi_1}))
        = \Tr(U M^{A'}_{\zeta(\xi) b^*} U^{-1} \mathcal{E}^{B'}_y ( U \rho^{A'}_{\xi_1} U^{-1} ) ).
    \end{align}
    From the same analysis for states, we have
    \begin{equation}
        \mathcal{E}^{A'}_y (\rho^{A'}_{\xi_1}) = U^{-1} \mathcal{E}^{B'}_y ( U \rho^{A'}_{\xi_1} U^{-1} ) U.
    \end{equation}
    Because $\rho^{A'}_{\xi_1}$ can be any pure state, we have
    \begin{equation}
        U \mathcal{E}^{A}_y( U^{-1}(\cdot) U ) U^{-1} = U \mathcal{E}^{A'}_y(U^{-1}(\cdot) U) U^{-1} = \mathcal{E}^{B'}_y(\cdot) = \mathcal{E}^{B}_y(\cdot)
    \end{equation}
    as stated in Theorem~\ref{prop:equiv}.


\section{Foundations for learning intrinsic descriptions}
\label{sec:learning_theory_foundations}

The goal of learning is to conduct experiments to gain knowledge about the actual world model among a collection of potential world models.
Learning theory provides a formal language to study such an information gathering process.
For learning about a quantum-mechanical world, we would like to know what quantum-mechanical operations each action $x \in \mathcal{X}, y \in \mathcal{Y}, z \in \mathcal{Z}$ corresponds to assuming the true world model belongs to a specified set of world models.

One should think of the learning process as follows. A classical agent is given the premise that the true world model belongs to some set of world models. Then, the classical agent conducts experiments to learn what the true model is.
The set of possible models is called the concept class (often when each model is a classical Boolean function), hypothesis class (often when each model is a function from space $X$ to space $Y$), or model class in machine learning.
We will use the less loaded word, model class, to refer to a set of potential world models.
In many branches of mathematics, including learning theory \cite{mohri2018foundations}, a class is a set of mathematical objects. We will follow the same convention here.

\begin{definition}[Model class for $d$-dimensional world model]
Given sets $\mathcal{X}, \mathcal{Y}, \mathcal{Z}$ denoting the action spaces and set $\mathcal{B}$ denoting the outcome space.
A $d$-dimensional model class $\mathcal{Q}$ over $\mathcal{X}, \mathcal{Y}, \mathcal{Z}, \mathcal{B}$ is a set $\{\mathcal{W}\}$ of 
$d$-dimensional world model models with the same spaces $\mathcal{X}, \mathcal{Y}, \mathcal{Z}, \mathcal{B}$.
\end{definition}

Recall that some world models are equivalent to one another (describe the same physical reality) while not being equal (the mathematical description looks nominally different), i.e., $\mathcal{W}_1 \neq \mathcal{W}_2$ but $\mathcal{W}_1 \equiv \mathcal{W}_2$.
This is intrinsic to the description of quantum mechanics as we have shown in Theorem~\ref{prop:equiv}.
When a model class contains two world models that are nominally different but physically the same, we say the model class is redundant.
A formal definition is given below.

\begin{definition}[Redundant model class] \label{def:redundant}
Given sets $\mathcal{X}, \mathcal{Y}, \mathcal{Z}, \mathcal{B}$ and a model class $\mathcal{Q} = \{\mathcal{W}\}$ over $\mathcal{X}, \mathcal{Y}, \mathcal{Z}, \mathcal{B}$.
$\mathcal{Q}$ is redundant if $\exists \, \mathcal{W}_1 \neq \mathcal{W}_2 \in \mathcal{Q}$, $\mathcal{W}_1 \equiv \mathcal{W}_2$.
\end{definition}

Another basic concept about model classes is that two model classes could be equivalent to one another as a result of the equivalence of world models.
For example, if we have three equivalent world models $\mathcal{W} \equiv \mathcal{W}_1 \equiv \mathcal{W}_2$, then the model class $\mathcal{Q} = \{ \mathcal{W} \}$ contains the same set of equivalent world models as $\tilde{\mathcal{Q}} = \{ \mathcal{W}_1, \mathcal{W}_2 \}$.
Hence, we say the two model classes are equivalent.
We give the formal definition of equivalent model classes below.

\begin{definition}[Equivalent model classes] \label{def:equiv-modelclass}
Given sets $\mathcal{X}, \mathcal{Y}, \mathcal{Z}, \mathcal{B}$.
Model classes $\mathcal{Q}, \tilde{\mathcal{Q}}$ over $\mathcal{X}, \mathcal{Y}, \mathcal{Z}, \mathcal{B}$ are equivalent if and only if $\forall \mathcal{W} \in \mathcal{Q}, \exists \tilde{\mathcal{W}} \in \tilde{\mathcal{Q}}, \mathcal{W} \equiv \tilde{\mathcal{W}}$, and $\forall \tilde{\mathcal{W}} \in \tilde{\mathcal{Q}}, \exists \mathcal{W} \in \mathcal{Q}, \mathcal{W} \equiv \tilde{\mathcal{W}}$.
\end{definition}

A redundant model class is not preferred as the same physical reality is described by two different representations.
However, the following basic proposition shows that a redundant model class is equivalent to a non-redundant model class.
For example, if $\mathcal{Q} = \{\mathcal{W}_1, \mathcal{W}_2, \mathcal{W}_3\}$, where $\mathcal{W}_1 \equiv \mathcal{W}_2$ and $\mathcal{W}_1 \not\equiv \mathcal{W}_3$, then we have $\mathcal{Q}$ is equivalent to $\mathcal{Q}' = \{\mathcal{W}_1, \mathcal{W}_3\}$.

\begin{proposition} \label{prop:can-find-non-redundant}
Given sets $\mathcal{X}, \mathcal{Y}, \mathcal{Z}, \mathcal{B}$ and a model class $\mathcal{Q} = \{\mathcal{W}\}$ over $\mathcal{X}, \mathcal{Y}, \mathcal{Z}, \mathcal{B}$.
There exists a non-redundant model class $\tilde{\mathcal{Q}}$ over $\mathcal{X}, \mathcal{Y}, \mathcal{Z}, \mathcal{B}$, such that $\tilde{\mathcal{Q}}$ is equivalent to $\mathcal{Q}$.
\end{proposition}
\begin{proof}
We partition all world models in $\mathcal{Q}$ into equivalence classes, where all models in the same equivalence class are equivalent to one another, and those that are in different equivalence classes are not equivalent.
We choose one representative from each equivalence class of world models in $\mathcal{Q}$.
We define $\tilde{\mathcal{Q}}$ to be the set of the representatives.
$\tilde{\mathcal{Q}}$ is equivalent to ${\mathcal{Q}}$ and $\tilde{\mathcal{Q}}$ is non-redundant.
\end{proof}

After defining and illustrating some basic properties of model classes, we consider the learnability of a model class.
We say a model class is learnable if for any world model $\mathcal{W}$ in the model class, the classical agent can identify the physical operations to an arbitrarily small error up to a global unitary or anti-unitary transformation using a finite number of experiments.
The unitary or anti-unitary transformation $U$ is necessary as Theorem~\ref{prop:equiv} states that two world models related by $U$ are equivalent and describe the same physical reality.
The transformation $U$ corresponds to a change of basis in the quantum Hilbert space and potentially followed by a complex conjugation operation.

\begin{definition}[Learnability of a model class] \label{def:learnability}
Given sets $\mathcal{X}, \mathcal{Y}, \mathcal{Z}, \mathcal{B}$ and a model class $\mathcal{Q} = \{\mathcal{W}\}$ for $d$-dimensional world models over $\mathcal{X}, \mathcal{Y}, \mathcal{Z}, \mathcal{B}$.
The model class $\mathcal{Q}$ is learnable if $\forall \, \epsilon, \delta > 0, \forall \mathcal{W} \in \mathcal{Q}$, there exists a unitary or anti-unitary transformation $U$,
\begin{itemize}
    \item $\forall \, x \in \mathcal{X}$, with probability $\geq 1 - \delta$, a classical agent can conduct finitely many experiments as defined in Def.~\ref{def:experiment} to output $\tilde{\rho}_x$ satisfying $\norm{\rho_x - U \tilde{\rho}_x U^{-1} }_1 \leq \epsilon$,
    \item $\forall \, y \in \mathcal{Y}$, with probability $\geq 1 - \delta$, a classical agent can conduct finitely many experiments as defined in Def.~\ref{def:experiment} to output $\tilde{\mathcal{E}}_y$ satisfying $\norm{\mathcal{E}_y - U \tilde{\mathcal{E}}_y( U^{-1} (\cdot) U ) U^{-1} }_{\diamond} \leq \epsilon$.
    \item $\forall \, z \in \mathcal{Z}$, $b \in \mathcal{B}$, with probability $\geq 1 - \delta$, a classical agent can conduct finitely many experiments as defined in Def.~\ref{def:experiment} to output $\tilde{M}_{z b}$ satisfying $\norm{M_{z b} - U \tilde{M}_{z b} U^{-1} }_1 \leq \epsilon$.
\end{itemize}
\end{definition}

In many scenarios, it is too much to ask for the ability to learn everything about a world model, i.e., all initial states, CPTP maps, and measurements.
For example, we might only want to predict a property of one of the possible initial states, such as its purity.
Predicting properties can often be significantly more efficient than learning the full description \cite{huang2020predicting}.
Furthermore, even if a model class $\mathcal{Q}$ is unlearnable according to the above definition, we may still be able to predict some properties.

\begin{definition}[Predictability of a property] \label{def:predictability}
Given sets $\mathcal{X}, \mathcal{Y}, \mathcal{Z}, \mathcal{B}$, a model class $\mathcal{Q} = \{\mathcal{W}\}$ for $d$-dimensional world models over $\mathcal{X}, \mathcal{Y}, \mathcal{Z}, \mathcal{B}$, and a function $f$ that maps a world model $\mathcal{W}$ to a property represented by a value in $\mathbb{R}$.
The property $f$ is predictable in the model class $\mathcal{Q}$ if $\forall \, \epsilon, \delta > 0, \forall \mathcal{W} \in \mathcal{Q}$, with probability $\geq 1-\delta$, a classical agent can conduct finitely many experiments as defined in Def.~\ref{def:experiment} to output $\tilde{o} \in \mathbb{R}$ satisfying $|f(\mathcal{W}) - \tilde{o}| \leq \epsilon$.
\end{definition}

\section{A general theorem for learning intrinsic physical descriptions}
\label{sec:int-description}

The goal of this appendix is to prove the following theorem.
Here, we consider a model class $\mathcal{Q}$ such that for any candidate world model $\mathcal{W}$ in $\mathcal{Q}$, there exists an action to prepare a pure state, a set of actions for implementing a universal set of unitaries, and an action for implementing a nontrivial POVM.
A trivial POVM produces a measurement outcome independent of the input state.
The actions that satisfy the above conditions could be different for different candidate world model $\mathcal{W}$.
The classical agent has no knowledge of what these actions are and what the corresponding physical operations are.
Furthermore the model class $\mathcal{Q}$ could contain uncountably many candidate world models.
The theorem states that even without knowing what any action is,
the classical agent can learn the intrinsic descriptions of all actions when the actions enable the exploration of the entire quantum state space.


\begin{theorem}[Restatement of Theorem~\ref{thm:int-description}] \label{thm:int-description-ditto}
Given finite sets $\mathcal{X}, \mathcal{Y}, \mathcal{Z}, \mathcal{B}$.
Consider a $d$-dimensional model class $\mathcal{Q}$ over action spaces $\mathcal{X}, \mathcal{Y}, \mathcal{Z}$ and outcome space $\mathcal{B}$. Suppose that for all $\mathcal{W} = ( \{\rho_x\}_{x \in \mathcal{X}}, \{\mathcal{E}_y\}_{y \in \mathcal{Y}}, \{\mathcal{M}_x\}_{x \in \mathcal{Z}} ) \in \mathcal{Q}$,
\begin{itemize}
    \item $\exists x \in \mathcal{X}$, $\rho_x$ is pure.
    \item $\exists y_1, \ldots, y_k \in \mathcal{Y}$, $\mathcal{E}_{y_1}, \ldots \mathcal{E}_{y_k}$ constitute a universal set of unitary transformations.
    \item $\exists z \in \mathcal{Z}$, $\mathcal{M}_z$ has at least one element not proportional to identity.
\end{itemize}
Then, $\mathcal{Q}$ is learnable.
\end{theorem}

To prove Theorem~\ref{thm:int-description-ditto}, we give a learning algorithm such that for each world model $\mathcal{W}$ in $\mathcal{Q}$, the algorithm learns a world model $\tilde{\mathcal{W}}$ that satisfies $\tilde{\mathcal{W}} \equiv {\mathcal{W}}$ approximately, i.e., all physical descriptions of the actions in $\mathcal{X}, \mathcal{Y}, \mathcal{Z}$ are related by a global unitary or anti-unitary transformation (see Theorem~\ref{prop:equiv}).
The approximation error can be made arbitrarily close to zero as the algorithm conducts more experiments.
From Definition~\ref{def:learnability} on learnability of model class, we have $\mathcal{Q}$ is learnable.

In the following, we present an important lemma used in the proof. Then, we separate each step of the learning algorithm to learn the actions in the world model into subsections.
The proof will consider the dimension $d$ to be a constant.

\subsection{A lemma on generating Haar-random unitaries}

The proof of Theorem~\ref{thm:int-description-ditto} relies on the following lemma about the generation of approximate Haar measure using a universal set of unitaries.
Here, we say $U_1, \ldots, U_k$ forms a universal set of unitaries if the set $\mathcal{U} = \{U_1, \ldots, U_k, U_1^{-1}, \ldots, U_k^{-1}\}$ generates a dense subgroup of the special unitary group. We follow the standard terminology, where the subgroup generated by the set $\mathcal{U}$ is the group consisting of elements that can be written as a product of elements in $\mathcal{U}$.

\begin{definition}[Real-valued Lipschitz function]
A real-valued Lipschitz function $\phi$ over the special unitary group satisfies $|\phi(U) - \phi(V)| \leq \norm{U - V}_F$ for all unitaries $U, V$.
\end{definition}

\begin{lemma}[Random unitaries approximately form Haar measure] \label{lem:random-unit-Haar}
Given $k$ unitaries $U_1, \ldots, U_k$ that form a universal set of unitaries.
For any $\epsilon > 0$ and any real-valued Lipschitz function $\phi$ over the special unitary group, there exists $L > 0$, such that $\forall \ell \geq L$,
\begin{equation}
    \left| \frac{1}{k^{\ell}} \sum_{i_1 = 1}^{k} \ldots \sum_{i_\ell = 1}^{k} \phi(U_{i_1} \ldots U_{i_\ell}) - \int d\mu_{\mathrm{Haar}}(U) \phi(U) \right| < \epsilon,
\end{equation}
where $\mu_{\mathrm{Haar}}$ is the Haar measure (uniform distribution) over the special unitary group.
\end{lemma}

We prove Lemma~\ref{lem:random-unit-Haar} using a theorem given in \cite{Varj2012RandomWI}, which is a corollary of a spectral gap theorem regarding semi-simple compact connected Lie group.
Before stating the theorem, we give a few definitions.
Consider a semi-simple compact connected Lie group $G$ endowed with the bi-invariant Riemannian metric.
The bi-invariant Riemannian metric gives a distance $d(g, h), \forall g, h \in G$.
We define $\mathrm{Lip}(G)$ to be the set of functions $\{\phi\}: G \rightarrow \mathbb{R}$ such that $\forall \phi \in \mathrm{Lip}(G)$, $\sup_{g \neq h \in G} \frac{|\phi(g) - \phi(h)|}{d(g, h)} < \infty$.
For $\phi \in \mathrm{Lip}(G)$, we consider $\norm{\phi}_{\mathrm{Lip}} = \sup_{g \neq h \in G} \frac{|\phi(g) - \phi(h)|}{d(g, h)}$.
We also define $\mu_\mathrm{Haar}$ to be the Haar measure over $G$.
The root second moment 
of a function $\phi: G \rightarrow \mathbb{R}$ is given by $\norm{\phi}_2 = \sqrt{\int |\phi(x)|^2 d\mu_{\mathrm{Haar}}(x)}$.
For a probability measure $\mu$, consider $\tilde{\mu}$ to be the probability measure such that
\begin{equation}
\int f(x) d \tilde{\mu}(x) = \int f(x^{-1}) d\mu(x)
\end{equation}
for all continuous function $f$.
One can think of $\tilde{\mu}$ as the probability distribution over the inverse of the probability measure $\mu$.
For example, a uniform distribution over $U_1, \ldots, U_k$ yields a uniform distribution over $U_1^{-1}, \ldots, U_k^{-1}$.
Consider the convolution between two probability measure $\mu$ and $\nu$ to be a probability measure,
\begin{equation}
    (\mu \ast \nu) (g) = \sum_{h \in G} \mu(g h^{-1}) \nu(h).
\end{equation}
One can interpret $\mu \ast \nu$ as the probability distribution of $h \ell$ when we sample $h$ according to $\mu$ and $\ell$ according to $\nu$.
For example, a convolution of the uniform distribution over $U_1, \ldots, U_k$ with itself would be a uniform distribution over $U_{i} U_{j}, \forall i, j = 1, \ldots, k$.
Also, for a probability measure $\mu$, we consider the support of $\mu$, denoted as $\mathrm{supp}(\mu)$, to be the intersection of every set $A$ such that $\mu(A^c) = 0$.

\begin{theorem}[Corollary 7~in~\cite{Varj2012RandomWI}] \label{thm:varju2012}
Let $G$ be a semi-simple compact connected Lie group endowed with the bi-invariant Riemannian metric, and $\mu$ be a probability measure on $G$. If $\mathrm{supp}(\tilde{\mu} \ast \mu)$ generates a dense subgroup in $G$, then $\forall \psi_A \in \mathrm{Lip}(G)$ with $\norm{\psi_A}_2 = 1$ and $\int \psi_A(x) d\mu_\mathrm{Haar}(x) = 0$, we have
\begin{equation}
    \norm{\psi_B}_2 < 1 - \frac{c}{\log^A(\norm{\psi_A}_{\mathrm{Lip}} + 2)},
\end{equation}
where $\psi_B(g) = \int \psi_A(h^{-1} g)d \mu(h)$, $A \leq 2$ depends on $G$, and $c > 0$ depends only on $\mu$.
\end{theorem}

We are now ready to prove Lemma~\ref{lem:random-unit-Haar}.
We apply Theorem~\ref{thm:varju2012} by considering $G$ to be the special unitary group with a representation in the vector space of matrices, and consider the bi-invariant Riemannian metric to be the Euclidean metric in the matrix space.
The distance $d(U, V)$ on $G$ is the the length of the shortest path from $U$ to $V$ on the special unitary group.
Because we consider a constant dimensional special unitary group, we have
\begin{equation}
    \norm{U-V}_F \leq d(U, V) \leq C \norm{U-V}_F,
\end{equation}
for a constant $C \geq 1$, where $\norm{X}_F = \sqrt{\Tr(X^\dagger X)}$.
Hence, any real-valued function $\phi$ satisfying $|\phi(U) - \phi(V)| \leq \norm{U - V}_F$ for all unitaries $U, V$ is in $\mathrm{Lip}(G)$ with the Lipschitz norm $\norm{\phi}_{\mathrm{Lip}} \leq 1$.
Because $\phi$ is Lipschitz continuous, $\norm{\phi}_2 < \infty$.

\begin{proof}[Proof for Lemma~\ref{lem:random-unit-Haar}]
For the edge case where $\phi$ is a constant function, the lemma trivially holds.
Let $\mu$ be the uniform distribution over $U_1^{-1}, \ldots, U_k^{-1}$.
The probability distribution $\tilde{\mu}$ over the inverse of $\mu$ is the uniform distribution over $U_1, \ldots, U_k$.
$\tilde{\mu} \ast \mu$ is the uniform distribution over $U_i U_j^{-1}, \forall i, j = 1, \ldots, k$, and $\mathrm{supp}(\tilde{\mu} \ast \mu)$ is $\{ U_i U_j^{-1}, \forall i, j = 1, \ldots, k \}$, which generates a dense subgroup of the special unitary group.
Because $\phi$ is not a constant function,
we can let $\psi_1$ be
\begin{equation}
    \psi_1(g) = \frac{\phi(g) - \int \phi(U) d\mu_{\mathrm{Haar}}(U) }{\norm{ \phi - \int \phi(U) d\mu_{\mathrm{Haar}}(U) }_2}, \,\,\forall g \in G,
\end{equation}
which satisfies $\norm{\psi_1}_2 = 1$, $\int \psi_1(x) d\mu_\mathrm{Haar}(x) = 0$, and
\begin{equation}
    \norm{\psi_1}_{\mathrm{Lip}} \leq \norm{ \phi - \int \phi(U) d\mu_{\mathrm{Haar}(U)} }_2 \norm{\phi}_{\mathrm{Lip}} \leq \norm{ \phi - \int \phi(U) d\mu_{\mathrm{Haar}}(U) }_2.
\end{equation}
For any $\ell > 1$, we define
\begin{equation}
    \psi_\ell(g) = \frac{\frac{1}{k^\ell} \sum_{i_1=1}^k \ldots \sum_{i_\ell=1}^k  \phi(U_{i_1} \ldots U_{i_\ell} g) - \int \phi d\mu_{\mathrm{Haar}}}{\norm{ \phi - \int \phi(U) d\mu_{\mathrm{Haar}} }_2}, \,\,\forall g \in G,
\end{equation}
which satisfies $\int \psi_\ell(x) d\mu_\mathrm{Haar}(x) = 0$ and
\begin{align}
    |\psi_\ell(g_1) - \psi_\ell(g_2)| &\leq \int \left|\psi_{\ell -1}(h^{-1} g_1) - \psi_{\ell -1}(h^{-1} g_2) \right| d\mu(h)\\
    &\leq \norm{\psi_{\ell -1}}_{\mathrm{Lip}} \int d( h^{-1} g_1, h^{-1} g_2 )  d\mu(h) \leq \norm{\psi_{\ell -1}}_{\mathrm{Lip}} d(g_1, g_2).
\end{align}
Hence, $\norm{\psi_{\ell}}_{\mathrm{Lip}} \leq \norm{\psi_{\ell -1}}_{\mathrm{Lip}} \leq \norm{\psi_{1}}_{\mathrm{Lip}}$ for any $\ell > 1$.

We can apply Theorem~\ref{thm:varju2012} with the probability measure $\mu$  defined above, $\psi_A = \psi_{\ell - 1} / \norm{\psi_{\ell - 1}}_2$, and $\psi_B = \psi_{\ell} / \norm{\psi_{\ell - 1}}_2$ to obtain
\begin{align}
    \norm{\psi_{\ell}}_2 &\leq \left(1 - \frac{c}{\log^A(\norm{\psi_{\ell - 1}}_\mathrm{Lip} / \norm{\psi_{\ell - 1}}_2 + 2)} \right) \norm{\psi_{\ell - 1}}_2 \\
    &\leq \left(1 - \frac{c}{\log^A(\norm{\psi_{1}}_\mathrm{Lip} / \norm{\psi_{\ell - 1}}_2 + 2)} \right) \norm{\psi_{\ell - 1}}_2.
\end{align}
The second inequality follows from $\norm{\psi_{\ell}}_{\mathrm{Lip}} \leq \norm{\psi_{1}}_{\mathrm{Lip}}$.
By solving the above recursive relation, given any $\tilde{\epsilon} > 0$, there exists a large enough $L$ such that
\begin{equation}
    \norm{\psi_{\ell}}_2 \leq \tilde{\epsilon}, \,\, \forall \ell \geq L.
\end{equation}
From the definition of $\psi_{\ell}$ and the fact that $\phi$ is Lipschitz-continuous over the special unitary group, $\forall \epsilon > 0$, there exists $L > 0$, such that $\forall \ell \geq L$, we have
\begin{equation}
    \left| \frac{1}{k^\ell} \sum_{i_1=1}^k \ldots \sum_{i_\ell=1}^k  \phi(U_{i_1} \ldots U_{i_\ell} g) - \int \phi d\mu_{\mathrm{Haar}} \right| \leq \epsilon, \,\, \forall g \in G.
\end{equation}
By setting $g = I$, we conclude the proof of Lemma~\ref{lem:random-unit-Haar}.
\end{proof}

\subsection{The precision parameter}

Consider the precision parameter $\eta$ to be a fixed number at the start of the learning algorithm.
Each subsection below corresponds to a set of subroutines in the learning algorithm that depends on $\eta$.
After completion of all subroutines, the algorithm restarts with $\eta \leftarrow \eta / 2$.

\subsection{Testing identity and unitarity}

We first present the subroutine for finding all actions $y_1, \ldots, y_{k'}$ that correspond to unitary transformations, which is presented in the following lemma.

\begin{lemma}[Unitary identification]
For a sufficiently small $\eta$, there is a subroutine that returns $y_1, \ldots, y_{k'}$ such that $\cE_{y_1}, \ldots, \cE_{y_{k'}}$ are all the unitary transformations in $\{\cE_y\}_{y \in \mathcal{Y}}$.
\end{lemma}
\begin{proof}
Consider an arbitrary norm $\norm{\cdot}$ over the space of maps over quantum states.
Assume that we have a subroutine for estimating how close a composition of CPTP maps $\mathcal{E}_{y_1'} \circ \ldots \circ \mathcal{E}_{y_\ell'}$ is to an identity under the norm $\norm{\cdot}$.
The subroutine will be presented in Lemma~\ref{lem:identity-test}.
Recall that in the finite set $\mathcal{Y}$, there are some CPTP maps $\mathcal{E}_{y_1}, \ldots, \mathcal{E}_{y_k}$ that correspond to unitary transformations.
We define
\begin{equation}
    \mathcal{Y}_{\mathrm{unitary}} = \left\{ \cE_y^\dagger \cE_y = I | \forall y \in \mathcal{Y} \right\}.
\end{equation}
We now present a subroutine that returns $\mathcal{Y}_{\mathrm{unitary}}$.
The proof of this lemma relies on a basic geometric fact about unitary: the only CPTP maps with some CPTP maps as their inverse are unitaries.

For each $y \in \mathcal{Y}$, we consider composing $\mathcal{E}_y$ with $\mathcal{E}_{y_1'} \circ \ldots \circ \mathcal{E}_{y_{\ell-1}'}$ for $y_1', \ldots, y_{\ell-1}' \in \mathcal{Y}$ and $\ell \leq 1 / \eta$.
The subroutine returns all actions $y \in \mathcal{Y}$ such that
\begin{equation}
    \min_{y_1', \ldots, y_{\ell-1}'} \norm{\mathcal{E}_{y_1'} \circ \ldots \circ \mathcal{E}_{y_{\ell-1}'} \circ \mathcal{E}_y - I} \leq \eta.
\end{equation}
If $\mathcal{E}_y$ is a unitary transformation, we can find some $y_1', \ldots, y_{\ell-1}'$ such that $\mathcal{E}_{y_1'} \circ \ldots \circ \mathcal{E}_{y_{\ell-1}'} \circ \mathcal{E}_y$ is arbitrarily close to the identity under $\norm{\cdot}$ as $\eta$ goes to zero.
If $\mathcal{E}_y$ is not a unitary transformation, then for all $y_1', \ldots, y_{\ell-1}'$ and $\eta > 0$, there is a lower bound to how close $\mathcal{E}_{y_1'} \circ \ldots \circ \mathcal{E}_{y_{\ell-1}'} \circ \mathcal{E}_y$ could be to the identity under $\norm{\cdot}$.
As a result, when $\eta$ becomes small enough, the set of actions returned by the learning algorithm will be equal to $\mathcal{Y}_{\mathrm{unitary}}$.
\end{proof}

\begin{lemma}[Identity testing] \label{lem:identity-test}
For a sufficiently small $\eta$ and any $\epsilon > 0$, there exists a norm $\norm{\cdot}$ over the space of maps over quantum states and a subroutine that takes in $y_1', \ldots, y_{\ell}'$ and returns an estimate for
\begin{equation}
    \norm{\mathcal{E}_{y_1'} \circ \ldots \circ \mathcal{E}_{y_{\ell}'} - I}
\end{equation}
up to an additive error $\epsilon$ with high probability.
\end{lemma}
\begin{proof}
The central property is that an identity map $\mathcal{E}_{y_1'} \circ \ldots \circ \mathcal{E}_{y_\ell'}$ satisfies
\begin{align}
    &\Tr( M_{zb} \left((\mathcal{E}_{y_{2k}} \circ \ldots \circ \mathcal{E}_{y_{k+1}}) \circ (\mathcal{E}_{y_1'} \circ \ldots \circ \mathcal{E}_{y_\ell'}) \circ (\mathcal{E}_{y_{k}} \circ \ldots \circ \mathcal{E}_{y_{1}})\right)(\rho_x) )\\
    &= \Tr( M_{zb} \left((\mathcal{E}_{y_{2k}} \circ \ldots \circ \mathcal{E}_{y_{k+1}}) \circ (\mathcal{E}_{y_{k}} \circ \ldots \circ \mathcal{E}_{y_{1}})\right)(\rho_x) ),
\end{align}
for all $x \in \mathcal{X}, y_1, \ldots y_{2k} \in \mathcal{Y} \cup \{ \mathrm{NULL} \}, z \in \mathcal{Z}, b \in \mathcal{B}$.
We denote the first expression as $h(x, y_1, \ldots, y_{2k}, z, b)$ and the second expression as $h_0(x, y_1, \ldots, y_{2k}, z, b)$.
Here, the action $y = \mathrm{NULL}$ corresponds to not implementing the action in the experiment defined in Def.~\ref{def:experiment}.
The learning algorithm considers all possible compositions of maps with $2k = \lceil 1 / \eta \rceil$, where $\eta$ is the precision parameter.

The learning algorithm can obtain estimates for $h(x, y_1, \ldots, y_{2k}, z, b)$ and $h_0(x, y_1, \ldots, y_{2k}, z, b)$ by running the corresponding experiments with $K$ repetitions.
By Hoeffding's inequality, with $K = \mathcal{O}(\log(1/\delta) / \epsilon^2)$, the algorithm can estimate $h(x, y_1, \ldots, y_{2k}, z, b)$ and $h_0(x, y_1, \ldots, y_{2k}, z, b)$ to $\epsilon$-error with probability at least $1 - \delta$.
We consider $K$ to be large enough such that the algorithm outputs an estimate for the quantity $A$ defined as
\begin{align}
    A &= \max_{x \in \mathcal{X}, y_1, \ldots y_{2k} \in \mathcal{Y} \cup \{ \mathrm{NULL} \}, z \in \mathcal{Z}, b \in \mathcal{B}} \left|h(x, y_1, \ldots, y_{2k}, z, b) - h_0(x, y_1, \ldots, y_{2k}, z, b) \right|\\
    &= \max_{x \in \mathcal{X}, y_1, \ldots y_{2k} \in \mathcal{Y} \cup \{ \mathrm{NULL} \}, z \in \mathcal{Z}, b \in \mathcal{B}}\\
    &\quad\quad\quad\quad\quad \left| \Tr( M_{zb} \left((\mathcal{E}_{y_{2k}} \circ \ldots \circ \mathcal{E}_{y_{k+1}}) \circ (\mathcal{E}_{y_1'} \circ \ldots \circ \mathcal{E}_{y_\ell'} - I) \circ (\mathcal{E}_{y_{k}} \circ \ldots \circ \mathcal{E}_{y_{1}})\right)(\rho_x) ) \right|. \label{def:A-norm}
\end{align}
up to $\epsilon$-error with high probability.
We can interpret $A$ as a norm $\norm{\cdot}$ over the space of maps over quantum states when $\eta$ is sufficiently small,
\begin{equation}
    A = \norm{\mathcal{E}_{y_1'} \circ \ldots \circ \mathcal{E}_{y_\ell'} - I}.
\end{equation}
Positive definiteness follows from the fact that $A$ is zero when $\mathcal{E}_{y_1'} \circ \ldots \circ \mathcal{E}_{y_\ell'}$ is an identity; and $A$ must be greater than zero for a sufficiently small $\eta$ if $\mathcal{E}_{y_1'} \circ \ldots \circ \mathcal{E}_{y_\ell'}$ is not equal to an identity from Lemma~\ref{lem:characterA}.
The two other conditions, absolute homogeneity and subadditivity, both follow from the definition of $A$ in Eq.~\eqref{def:A-norm}.
\end{proof}

\begin{lemma}[Characterization of $A$] \label{lem:characterA}
For a sufficiently small $\eta$, we have
\begin{align}
    A &= \max_{x \in \mathcal{X}, y_1, \ldots y_{2k} \in \mathcal{Y} \cup \{ \mathrm{NULL} \}, z \in \mathcal{Z}, b \in \mathcal{B}}\\
    &\quad\quad\quad\quad\quad \left| \Tr( M_{zb} \left((\mathcal{E}_{y_{2k}} \circ \ldots \circ \mathcal{E}_{y_{k+1}}) \circ (\mathcal{E}_{y_1'} \circ \ldots \circ \mathcal{E}_{y_\ell'} - I) \circ (\mathcal{E}_{y_{k}} \circ \ldots \circ \mathcal{E}_{y_{1}})\right)(\rho_x) ) \right| > 0
\end{align}
if $\mathcal{E}_{y_1'} \circ \ldots \circ \mathcal{E}_{y_\ell'}$ is not equal to an identity.
\end{lemma}
\begin{proof}
This claim follows from the assumption that there exists a universal set of unitaries and a pure state in the action space.
Hence, we can generate a pure state $\left(\mathcal{E}_{y_{k}} \circ \ldots \circ \mathcal{E}_{y_{1}})\right)(\rho_x)$ such that
\begin{equation}
    (\mathcal{E}_{y_1'} \circ \ldots \circ \mathcal{E}_{y_\ell'})( (\mathcal{E}_{y_{k}} \circ \ldots \circ \mathcal{E}_{y_{1}})(\rho_x)) \neq (\mathcal{E}_{y_{k}} \circ \ldots \circ \mathcal{E}_{y_{1}})(\rho_x).
\end{equation}
From the assumption, we can also find a POVM $\mathcal{M}_z$ such that one of the POVM element $M_{zb}$ is not proportional to the identity.
There always exists a unitary $U_1$ that diagonalizes the Hermitian matrix
\begin{equation}
    D \equiv (\mathcal{E}_{y_1'} \circ \ldots \circ \mathcal{E}_{y_\ell'})( (\mathcal{E}_{y_{k}} \circ \ldots \circ \mathcal{E}_{y_{1}})(\rho_x)) - (\mathcal{E}_{y_{k}} \circ \ldots \circ \mathcal{E}_{y_{1}})(\rho_x),
\end{equation}
such that the eigenvalues $\{\lambda^D_i\}$ are sorted from a greater value to a smaller value.
Also, there exists a unitary $U_2$ that diagonalizes $M_{zb}$ with eigenvalues $\{\lambda^{M_{zb}}_i\}$ sorted from large to small.
We have
\begin{equation}
    \Tr( U_2 M_{zb} U_2^\dagger U_1  D U_1^\dagger) = \sum_{i=1}^d \lambda^{M_{zb}}_i \lambda^{D}_i.
\end{equation}
Because $D \neq 0$ and $\Tr(D) = 0$, the largest eigenvalue $\lambda^D_1 > 0$, the smallest eigenvalue $\lambda^D_d < 0$, and $\sum_{i=1}^d \lambda^D_i = 0$.
We can thus find an index $k \geq 2$ such that $\lambda^D_i > 0, \forall i < k$ and $\lambda^D_i \leq 0, \forall i \geq k$.
Since $M_{zb} \succeq 0$ is not proportional to identity, we have the largest eigenvalue $\lambda^{M_{zb}}_1 > \lambda^{M_{zb}}_d \geq 0$.
\begin{equation}
    \sum_{i=1}^d \lambda^{M_{zb}}_i \lambda^{D}_i = \sum_{i=1}^d (\lambda^{M_{zb}}_i - \lambda^{M_{zb}}_k) \lambda^{D}_i > 0.
\end{equation}
When the precision parameter $\eta$ is smaller enough, $(\mathcal{E}_{y_{2k}} \circ \ldots \circ \mathcal{E}_{y_{k+1}})$ can approximate any unitary with a properly chosen $y_{k+1}, \ldots, y_{2k}$ because there exists a universal set of unitaries by choosing the proper actions.
Hence, there exists $x \in \mathcal{X}, y_1, \ldots y_{2k} \in \mathcal{Y} \cup \{ \mathrm{NULL} \}, z \in \mathcal{Z}, b \in \mathcal{B}$ such that $h(x, y_1, \ldots, y_{2k}, z, b) > h_0(x, y_1, \ldots, y_{2k}, z, b)$.
Together, we see that $A$ must be greater than zero if $\mathcal{E}_{y_1'} \circ \ldots \circ \mathcal{E}_{y_\ell'}$ is not equal to an identity.
\end{proof}

\subsection{Estimating state overlaps}
\label{sec:state-overlap}

The learning algorithm has now identified a set $\{y_1, \ldots, y_{k'}\}$ of unitary transformation.
The algorithm randomly composes the identified unitary transformations $\cE_{y_1}, \ldots, \cE_{y_{k'}}$.
In particular, the algorithm randomly selects $1 / \eta$ unitaries with replacement and composes them to form an approximate Haar random unitary.
Using Lemma~\ref{lem:random-unit-Haar}, as $\eta$ becomes smaller, we can obtain a better approximation to the Haar random unitary.
The learning algorithm has no additional information other than the randomly composed operation is approximately Haar random.

The ability to generate Haar random unitary enables the learning algorithm to estimate state overlaps.
Given two states $\rho_1, \rho_2$ which can be obtained by composing some initial states and CPTP maps, we have a randomized measurement procedure that guarantees the following.

\begin{lemma}[State overlap estimation with a fixed POVM] \label{lem:state-overlap-fixedPOVM}
Given two states $\rho_1, \rho_2$, $\epsilon > 0$, $z \in \cZ, b \in \cB$, and a sufficiently small $\eta > 0$, there is a subroutine that estimates
\begin{equation}
    f_{M_{zb}}(\rho_1, \rho_2) = \alpha_{M_{zb}} \Tr(\rho_1 \rho_2) + \beta_{M_{zb}}
\end{equation}
up to $\epsilon$ additive error, where $\alpha_{M_{zb}}, \beta_{M_{zb}}$ depends on POVM element $M_{zb}$.
\end{lemma}
\begin{proof}
Consider $R$ repetitions. For repetition $r \in \{1, \ldots, R\}$, the subroutine performs:
\begin{enumerate}
    \item Randomly compose $1/\eta$ actions in $\mathcal{Y}_{\mathrm{unitary}}$ to generate a random CPTP map $\mathcal{E}$.
    \item Measure the POVM $\mathcal{M}_z$ on $\mathcal{E}(\rho_1)$ and check if the measurement outcome is $b$.
    \item Record a binary variable $C_r \in \{0, 1\}$ indicating if the outcome is $b$.
    \item Measure the POVM $\mathcal{M}_z$ on $\mathcal{E}(\rho_2)$ and check if the measurement outcome is $b$.
    \item Record a binary variable $D_r \in \{0, 1\}$ indicating if the outcome is $b$.
\end{enumerate}
From Lemma~\ref{lem:char-hatX}, we can show that $\hat{X} = \frac{1}{R} \sum_{r=1}^R C_r D_r$ is an accurate estimate for
\begin{align}
    f_{M_{zb}}(\rho_1, \rho_2) &= \frac{1}{d^2 - 1} \left( (\Tr(M_{zb})^2 - \Tr(M_{zb}^2) / d) + \Tr(\rho_1 \rho_2) (\Tr(M_{zb}^2) - \Tr(M_{zb})^2 / d) \right),\label{eq:two-design_M} \\
    &= \alpha_{M_{zb}} \Tr(\rho_1 \rho_2) + \beta_{M_{zb}},
\end{align}
up to $\epsilon$ error when $\eta$ is sufficiently small.
Hence this lemma can be established.
\end{proof}

\begin{lemma}[Characterization of $\hat{X}$] \label{lem:char-hatX}
Given two states $\rho_1, \rho_2$, $\epsilon > 0$, $z \in \cZ, b \in \cB$, a sufficiently large $R$ and a sufficiently small $\eta > 0$, we have
\begin{equation}
\left|\hat{X} - f_{M_{zb}}(\rho_1, \rho_2) \right| < \epsilon
\end{equation}
with high probability.
\end{lemma}
\begin{proof}
The expectation value of $\hat{X} = \frac{1}{R} \sum_{r=1}^R C_r D_r$ is equal to
\begin{align}
    \E \left[\frac{1}{R} \sum_{r=1}^R C_r D_r\right] &= \E_{\mathcal{E}} \Tr(M_{zb} \mathcal{E}(\rho_1)) \Tr(M_{zb} \mathcal{E}(\rho_2))\\
    &\approx \int_U d\mu_{\mathrm{Haar}}(U) \Tr(M_{zb} U \rho_1 U^\dagger) \Tr(M_{zb} U \rho_2 U^\dagger) \label{eq:approxHaar}\\
    &= \int_U d\mu_{\mathrm{Haar}}(U) \Tr( (M_{zb} \otimes M_{zb}) (U \otimes U) (\rho_1 \otimes \rho_2) (U^\dagger \otimes U^\dagger) ) \\
    &= \frac{1}{d^2 - 1} \left( (\Tr(M_{zb})^2 - \Tr(M_{zb}^2) / d) + \Tr(\rho_1 \rho_2) (\Tr(M_{zb}^2) - \Tr(M_{zb})^2 / d) \right). \label{eq:Haarintegration}
\end{align}
Eq.~\eqref{eq:approxHaar} is the consequence of the fact that a random composition of universal set of unitaries approximately forms a Haar random unitary.
In particular, using the Lipschitz continuity of the function $\phi(U) \equiv \Tr(M_{zb} U \rho_1 U^\dagger) \Tr(M_{zb} U \rho_2 U^\dagger)$, Lemma~\ref{lem:random-unit-Haar} shows that the approximation error can be made arbitrarily small as the number of composed unitary becomes sufficiently large (i.e., $\eta$ sufficiently small).
Because $C_r D_r$ is a random variable bounded by one, using Hoeffding's inequality, we can choose $R = \mathcal{O}(\log(1/\delta))/\epsilon^2$, such that $\frac{1}{R} \sum_{r=1}^R C_r D_r$ equals to $\E_{\mathcal{E}} \Tr(M_{zb} \mathcal{E}(\rho_1)) \Tr(M_{zb} \mathcal{E}(\rho_2))$ up to error $\epsilon / 2$ with probability at least $1-\delta$.
Eq.~\eqref{eq:Haarintegration}, on the other hand, uses the second moment Haar integration formula over special unitary group. In particular, for $\mathrm{SU}(d)$ and $X \in \mathbb{C}^{(d\times d) \times (d \times d)}$, we have
\begin{equation}
    \int_U d\mu_{\mathrm{Haar}}(U) (U \otimes U) X (U^\dagger \otimes U^\dagger) = \frac{1}{d^2-1} \left( I \Tr(X) + S \Tr( S X) - \frac{1}{d} S \Tr(X) - \frac{1}{d} I \Tr(S X) \right),
\end{equation}
where $S$ is the swap operator over the tensor product space.
Hence, when the precision parameter $\eta$ is small enough and the number $R$ of randomized experiments is large enough, $\hat{X} = \frac{1}{R} \sum_{r=1}^R C_r D_r$ is an accurate estimate for
\begin{equation}
    f_{M_{zb}}(\rho_1, \rho_2) = \frac{1}{d^2 - 1} \left( (\Tr(M_{zb})^2 - \Tr(M_{zb}^2) / d) + \Tr(\rho_1 \rho_2) (\Tr(M_{zb}^2) - \Tr(M_{zb})^2 / d) \right),
\end{equation}
with an additive error at most $\epsilon$. This establishes the claim.
\end{proof}

We are now ready to combine the two lemmas above to establish the main result of this subsection.

\begin{lemma}[State overlap estimation]
Given two states $\rho_1, \rho_2$, $\epsilon > 0$, a sufficiently small $\eta > 0$, and the existence of a non-identity $M_{zb}$ for some $z \in \cZ, b \in \cB$.
There is a subroutine that estimates $\Tr(\rho_1 \rho_2)$ up to $\epsilon$ additive error.
\end{lemma}
\begin{proof}
The learning algorithm utilizes the procedure in Lemma~\ref{lem:state-overlap-fixedPOVM} to build the subroutine achieving the claim of this lemma.
Using the fact that for a $d$-dimensional vector $x$, $d \norm{x}_2^2 \geq \norm{x}_1^2$, we have
\begin{equation}\label{eq:Mzb-bound}
\Tr(M_{zb}^2) - \Tr(M_{zb})^2 / d \geq 0.
\end{equation}
Furthermore, equality holds in Eq.~\eqref{eq:Mzb-bound} if and only if all eigenvalues of $M_{zb}$ are equal, which implies that $M_{zb}$ is proportional to identity.
If $M_{zb}$ is proportional to identity, $f_{M_{zb}}(\rho_1, \rho_2)$ will be a constant function independent of $\rho_1, \rho_2$.
In contrast, if $M_{zb}$ is not proportional to identity, then for some $x \in \mathcal{X}, y \in \mathcal{Y}$, $f_{M_{zb}}(\rho_1, \rho_2)$ will be distinct between the following two pairs of states,
\begin{equation} \label{eq:rho1rho2}
\rho_1 = \rho_x, \rho_2 = \mathcal{E}_y(\rho_x) \,\, \mbox{ and } \,\, \rho_1 = \rho_x, \rho_2 = \rho_x.
\end{equation}
In particular, this is true if we choose the $x$ such that $\rho_x$ is pure and $y$ such that $\mathcal{E}_y$ is one of the universal set of unitaries such that $\mathcal{E}_y(\rho_x) \neq \rho_x$.

From the assumption on the true world model (exists actions corresponding to preparation of a pure state, a universal set of unitaries, and a POVM element not proportional to identity), there always exists $z, b, x, y$ such that $f_{M_{zb}}(\rho_1, \rho_2)$ is distinct under the two pairs of states in Eq.~\eqref{eq:rho1rho2}.
Hence, as $\eta$ goes to zero, if $M_{zb}$ is proportional to the identity, then the largest difference for the estimate of $f_{M_{zb}}(\rho_1, \rho_2)$ maximized over $x \in \mathcal{X}, y \in \mathcal{Y}$ will approach zero.
In contrast, if $M_{zb}$ is not proportional to the identity, the largest difference for the estimate of $f_{M_{zb}}(\rho_1, \rho_2)$ maximized over $x \in \mathcal{X}, y \in \mathcal{Y}$ will be greater than a positive value.
Hence, we can consider an algorithm that finds the pair of $z \in \mathcal{Z}, b \in \mathcal{B}$ that yields the largest difference for the estimate of $f_{M_{zb}}(\rho_1, \rho_2)$ maximized over $x \in \mathcal{X}, y \in \mathcal{Y}$.
The deduction above guarantees that the algorithm would find a POVM element $M_{zb}$ that is not proportional to the identity under a sufficiently small $\eta$.

After finding a pair of $z, b$ such that $M_{zb}$ is not proportional to identity, we can now describe the procedure that estimates $\Tr(\rho_1 \rho_2)$ for any state $\rho_1, \rho_2$.
Recall from Lemma~\ref{lem:state-overlap-fixedPOVM} and Eq.~\ref{eq:Mzb-bound} that
\begin{equation} \label{eq:f-M-zb}
f_{M_{zb}}(\rho_1, \rho_2) = \alpha_{M_{zb}} \Tr(\rho_1\rho_2) + \beta_{M_{zb}},
\end{equation}
where $\alpha_{M_{zb}} > 0$.
Because $\alpha_{M_{zb}} > 0$, when $\rho_1 = \rho_2$, we can see that $f_{M_{zb}}(\rho_1, \rho_1)$ is maximized when $\rho_1$ is a pure state.
The maximum value of $f_{M_{zb}}(\rho_1, \rho_2)$ is $\alpha_{M_{zb}} + \beta_{M_{zb}}$, and the minimum value is $\beta_{M_{zb}}$.
If $\rho_1$ is not a pure state, we can see that $f_{M_{zb}}(\rho_1, \rho_2) < \alpha_{M_{zb}} + \beta_{M_{zb}}$.
The subroutine would hence go through all $x \in \mathcal{X}$ and find an $x^*$ such that the estimate for $f_{M_{zb}}(\rho_x, \rho_x)$ is maximized.
Recall from the assumption of the true world model, there exists an action $x$ that prepares a pure state.
The gap between the finite set of actions that prepare pure states and those that prepare mixed states allows us to guarantee that the action $x^*$ we found prepares a pure state $\rho_{x^*}$ when $\eta$ is sufficiently small.
Because $\rho_{x^*}$ is a pure state, we have $\Tr(\rho_{x^*}^2) = 1$ and hence from Eq.~\eqref{eq:two-design_M},
\begin{equation} \label{eq:f_Mzb-same}
    f_{M_{zb}}(\rho_{x^*}, \rho_{x^*}) = \frac{1}{(d+1)d}(\Tr(M_{zb}^2) + \Tr(M_{zb})^2).
\end{equation}
The learning algorithm can obtain an estimate for  $f_{M_{zb}}(\rho_{x^*}, \rho_{x^*})$.
After that, it is only necessary to estimate  $\Tr(M_{zb})$ in order to determine $\Tr(M_{zb}^2)$, which suffices to specify the two values $\alpha_{M_{zb}}$ and $\beta_{M_{zb}}$.

An estimate for $\Tr(M_{zb})$ can be obtained by reusing the randomized measurement data from the procedure described in Lemma~\ref{lem:state-overlap-fixedPOVM}.
We can show that $\hat{Y} = \frac{1}{R} \sum_{r=1}^R C_r$ is an accurate estimate for $\Tr(M_{zb}) / d$.
Using the following first-moment Haar integration formula over the special unitary group,
\begin{equation}
\int_U d\mu_{\mathrm{Haar}} \, U X U^\dagger = \Tr(X) I/d,
\end{equation}
and the standard concentration inequality, we have that $\frac{1}{R} \sum_{r=1}^R C_r$ gives an estimate for
$\frac{\Tr(M_{zb})}{d}$
up to a small error $\epsilon$ for sufficiently large $R$ and sufficiently small $\eta$.
Along with an estimate for Eq.~\eqref{eq:f_Mzb-same}, the learning algorithm can determine both $\Tr(M_{zb}^2)$ and $\Tr(M_{zb})^2$, and hence $\alpha_{M_{zb}}$ and $\beta_{M_{zb}}$.
Together, the learning algorithm can produce an accurate estimate for quantum state overlap $\Tr(\rho_1 \rho_2)$ from an estimate for $f_{M_{zb}}(\rho_1, \rho_2)$ given in Lemma~\ref{lem:state-overlap-fixedPOVM} and the estimates for $\alpha_{M_{zb}}$ and $\beta_{M_{zb}}$.
This concludes the proof of this lemma.
\end{proof}

\subsection{Learning descriptions of a special set of states}
\label{sec:learning-special-set}

At this point, the algorithm still has not learned a description for any of the actions. However, the algorithm has identified several important features of the actions.
The algorithm has found $x^* \in \mathcal{X}$ where $\rho_{x^*}$ is a pure state.
The algorithm has also discovered $\mathcal{Y}_{\mathrm{unitary}} = \{y_1, \ldots, y_k\} \in \mathcal{Y}$ that forms a universal set of unitaries, which we will now denote as $U_{y_1}, \ldots, U_{y_k} \in \mathrm{SU}(d)$.
Furthermore, the algorithm has now obtained a subroutine that provides accurate estimate for state overlap $\Tr(\rho_1 \rho_2)$.
The learning algorithm can now utilize these tools to construct the entire structure of quantum state space.

More precisely, the algorithm will find a special set of pure quantum states $\{ \rho_i \}$ that satisfies a certain geometry.
The algorithm generates the special set of pure states by applying compositions of the unitaries $U_y, \forall y \in \mathcal{Y}_{\mathrm{unitary}}$ to the pure state $\rho_{x^*}$.
We will limit the algorithm to consider a composition of at most $1 / \eta$ unitaries.
This means that the algorithm will only find a collection of states that satisfies the geometry \emph{approximately}.
However, an approximate geometry with small error implies that the learned descriptions will only be subject to a small error.
As $\eta$ goes to zero, the geometry and the learned description will become accurate to an arbitrarily small error.
The geometry enables us to provide the intrinsic physical descriptions for states in the special set.
Using properties of the geometry, we can guarantee that the description for the special set of pure states is accurate up to the equivalence relation -- a global unitary or anti-unitary transformation --  characterized by Theorem~\ref{prop:equiv}.
The construction of the special set of pure states is related to the proof of Wigner's theorem.

We denote the special collection of pure states as $\rho^{\mathrm{(basis)}}_i, \forall i \in \{1, \ldots, d\}$, $\rho^{\mathrm{(real)}}_{ij}, \forall i \neq j \in \{1, \ldots, d\}$, $\rho^{\mathrm{(imag)}}_{ij}, \forall i \neq j \in \{1, \ldots, d\}$, $\rho^{\mathrm{(triplet)}}_{ij}, \forall i \neq j \in \{2, \ldots, d\}$, $\rho^{\mathrm{(triplet, 12)}}_{j}, \forall j \in \{3, \ldots, d\}$ and $\rho^{\mathrm{(triplet, i)}}_{ij}, \forall i \neq j \in \{2, \ldots, d\}$.
The geometry of the states is given by the following equations.
\begin{align}
    \Tr(\rho^{\mathrm{(basis)}}_i \rho^{\mathrm{(basis)}}_j) &= \delta_{ij}, &\forall i, j \in \{1, \ldots, d\}, & \quad \mbox{\emph{(Fix the basis)}} \label{eq:first-geom} \\
    \Tr(\rho^{\mathrm{(basis)}}_i \rho^{\mathrm{(real)}}_{ij}) &= \frac{1}{2}, &\forall i \neq j \in \{1, \ldots, d\}, & \quad \mbox{\emph{(Fix absolute amplitude for $\rho^{\mathrm{(real)}}_{ij}$)}}\\
    \Tr(\rho^{\mathrm{(basis)}}_j \rho^{\mathrm{(real)}}_{ij}) &= \frac{1}{2}, &\forall i \neq j \in \{1, \ldots, d\}, & \quad \mbox{\emph{(Fix absolute amplitude for $\rho^{\mathrm{(real)}}_{ij}$)}}\\
    \Tr(\rho^{\mathrm{(basis)}}_1 \rho^{\mathrm{(triplet)}}_{ij}) &= \frac{1}{3}, &\forall i \neq j \in \{2, \ldots, d\}, & \quad \mbox{\emph{(Fix absolute amplitude for $\rho^{\mathrm{(triplet)}}_{ij}$)}}\\
    \Tr(\rho^{\mathrm{(basis)}}_i \rho^{\mathrm{(triplet)}}_{ij}) &= \frac{1}{3}, &\forall i \neq j \in \{2, \ldots, d\}, & \quad \mbox{\emph{(Fix absolute amplitude for $\rho^{\mathrm{(triplet)}}_{ij}$)}}\\
    \Tr(\rho^{\mathrm{(basis)}}_j \rho^{\mathrm{(triplet)}}_{ij}) &= \frac{1}{3}, &\forall i \neq j \in \{2, \ldots, d\}, & \quad \mbox{\emph{(Fix absolute amplitude for $\rho^{\mathrm{(triplet)}}_{ij}$)}}\\
    \Tr(\rho^{\mathrm{(real)}}_{1i} \rho^{\mathrm{(triplet)}}_{ij}) &= \frac{2}{3}, &\forall i \neq j \in \{2, \ldots, d\}, & \quad \mbox{\emph{(Transfer relative phase $+1$, a)}} \label{eq:tranfers+1A}\\
    \Tr(\rho^{\mathrm{real}}_{1j} \rho^{\mathrm{(triplet)}}_{ij}) &= \frac{2}{3}, &\forall i \neq j \in \{2, \ldots, d\}, & \quad \mbox{\emph{(Transfer relative phase $+1$, b)}} \label{eq:tranfers+1B}\\
    \Tr(\rho^{\mathrm{(real)}}_{ij} \rho^{\mathrm{(triplet)}}_{ij}) &= \frac{2}{3}, &\forall i \neq j \in \{2, \ldots, d\}, & \quad \mbox{\emph{(Transfer relative phase $+1$, c)}} \label{eq:tranfers+1C}\\
    \Tr(\rho^{\mathrm{(basis)}}_i \rho^{\mathrm{(imag)}}_{ij}) &= \frac{1}{2}, &\forall i \neq j \in \{1, \ldots, d\}, & \quad \mbox{\emph{(Fix absolute amplitude for $\rho^{\mathrm{(imag)}}_{ij}$)}}\\
    \Tr(\rho^{\mathrm{basis}}_j \rho^{\mathrm{(imag)}}_{ij}) &= \frac{1}{2}, &\forall i \neq j \in \{1, \ldots, d\}, & \quad \mbox{\emph{(Fix absolute amplitude for $\rho^{\mathrm{(imag)}}_{ij}$)}}\\
    \Tr(\rho^{\mathrm{real}}_{ij} \rho^{\mathrm{(imag)}}_{ij}) &= \frac{1}{2}, &\forall i \neq j \in \{1, \ldots, d\}, & \quad \mbox{\emph{(Partially fix the phase for $\rho^{\mathrm{(imag)}}_{ij}$)}}\\
    \Tr(\rho^{\mathrm{(basis)}}_1 \rho^{\mathrm{(triplet, 12)}}_{j}) &= \frac{1}{3}, &\forall j \in \{3, \ldots, d\}, & \quad \mbox{\emph{(Fix absolute amplitude for $\rho^{\mathrm{(triplet, 12)}}_{j}$)}}\\
    \Tr(\rho^{\mathrm{(basis)}}_2 \rho^{\mathrm{(triplet, 12)}}_{j}) &= \frac{1}{3}, &\forall j \in \{3, \ldots, d\}, & \quad \mbox{\emph{(Fix absolute amplitude for $\rho^{\mathrm{(triplet, 12)}}_{j}$)}}\\
    \Tr(\rho^{\mathrm{(basis)}}_j \rho^{\mathrm{(triplet, 12)}}_{j}) &= \frac{1}{3}, &\forall j \in \{3, \ldots, d\}, & \quad \mbox{\emph{(Fix absolute amplitude for $\rho^{\mathrm{(triplet, 12)}}_{j}$)}}\\
    \Tr(\rho^{\mathrm{(real)}}_{2j} \rho^{\mathrm{(triplet, 12)}}_{j}) &= \frac{2}{3}, &\forall j \in \{3, \ldots, d\}, & \quad \mbox{\emph{(Transfer relative phase $+1$, a')}}\\
    \Tr(\rho^{\mathrm{(imag)}}_{12} \rho^{\mathrm{(triplet, 12)}}_{j}) &= \frac{2}{3}, &\forall j \in \{3, \ldots, d\}, & \quad \mbox{\emph{(Transfer relative phase $+\rmi$, b')}}\\
    \Tr(\rho^{\mathrm{(imag)}}_{1j} \rho^{\mathrm{(triplet, 12)}}_{j}) &= \frac{2}{3}, &\forall j \in \{3, \ldots, d\}, & \quad \mbox{\emph{(Transfer relative phase $+\rmi$, c')}}\\
    \Tr(\rho^{\mathrm{(basis)}}_1 \rho^{\mathrm{(triplet, i)}}_{ij}) &= \frac{1}{3}, &\forall i \neq j \in \{2, \ldots, d\}, & \quad \mbox{\emph{(Fix absolute amplitude for $\rho^{\mathrm{(triplet, i)}}_{ij}$)}}\\
    \Tr(\rho^{\mathrm{(basis)}}_i \rho^{\mathrm{(triplet, i)}}_{ij}) &= \frac{1}{3}, &\forall i \neq j \in \{2, \ldots, d\}, & \quad \mbox{\emph{(Fix absolute amplitude for $\rho^{\mathrm{(triplet, i)}}_{ij}$)}}\\
    \Tr(\rho^{\mathrm{(basis)}}_j \rho^{\mathrm{(triplet, i)}}_{ij}) &= \frac{1}{3}, &\forall i \neq j \in \{2, \ldots, d\}, & \quad \mbox{\emph{(Fix absolute amplitude for $\rho^{\mathrm{(triplet, i)}}_{ij}$)}}\\
    \Tr(\rho^{\mathrm{(real)}}_{1i} \rho^{\mathrm{(triplet, i)}}_{ij}) &= \frac{2}{3}, &\forall i \neq j \in \{2, \ldots, d\}, & \quad \mbox{\emph{(Transfer relative phase $+1$, a'')}} \label{eq:tranfersa''}\\
    \Tr(\rho^{\mathrm{(imag)}}_{1j} \rho^{\mathrm{(triplet, i)}}_{ij}) &= \frac{2}{3}, &\forall i \neq j \in \{2, \ldots, d\}, & \quad \mbox{\emph{(Transfer relative phase $+\rmi$, b'')}} \label{eq:tranfersb''}\\
    \Tr(\rho^{\mathrm{(imag)}}_{ij} \rho^{\mathrm{(triplet, i)}}_{ij}) &= \frac{2}{3}, &\forall i \neq j \in \{2, \ldots, d\}, & \quad \mbox{\emph{(Transfer relative phase $+\rmi$, c'')}} \label{eq:tranfersc''}.
\end{align}
We comment on each geometric constraint; this makes it easier to refer to each constraint in the following analysis.
This geometry determines the description of the set of states.

\begin{lemma}[Geometry and states] \label{lem:geom-state}
The geometry of the states given in Eq.~\eqref{eq:first-geom} to Eq.~\eqref{eq:tranfersc''} is satisfied if and only if
\begin{align}
    \rho^{\mathrm{(basis)}}_i &= U \ketbra{i}{i} U^{-1}, &\forall i \in \{1, \ldots, d\}, \label{eq:first-state} \\
     \rho_{ij}^{\mathrm{(real)}} &= \frac{1}{2} U \left(\ket{i} +  \ket{j}\right)\left(\bra{i} + \bra{j}\right) U^{-1}, &\forall i \neq j \in \{1, \ldots, d\},\\
    \rho_{ij}^{\mathrm{(imag)}} &= \frac{1}{2} U \left(\ket{i} +  \rmi \ket{j}\right)\left(\bra{i} - \rmi \bra{j}\right) U^{-1}, &\forall i \neq j \in \{2, \ldots, d\}, \\
    \rho_{ij}^{\mathrm{(triplet)}} &= \frac{1}{3} U \left(\ket{1} +  \ket{i} + \ket{j}\right)\left(\bra{1} +  \bra{i} + \bra{j}\right) U^{-1} , &\forall i \neq j \in \{2, \ldots, d\}, \\
    \rho_{j}^{\mathrm{(triplet, 12)}} &= \frac{1}{3} U \left(\ket{1} + \rmi \ket{2} + \rmi \ket{j}\right)\left(\bra{1} - \rmi \bra{2} - \rmi \bra{j}\right) U^{-1} , &\forall j \in \{3, \ldots, d\},\\
    \rho_{ij}^{\mathrm{(triplet, i)}} &= \frac{1}{3} U \left(\ket{1} + \ket{i} + \rmi \ket{j}\right)\left(\bra{1} + \bra{i} - \rmi \bra{j}\right) U^{-1}, & \forall i \neq j \in \{2, \ldots, d\} \label{eq:last-state}
\end{align}
for a unitary or anti-unitary transformation $U$.
\end{lemma}
\begin{proof}
One can directly verify that the set of states given in Eq.~\eqref{eq:first-state} to Eq.~\eqref{eq:last-state} satisfies the geometry given in Eq.~\eqref{eq:first-geom} to Eq.~\eqref{eq:tranfersc''}.
For the other direction, we utilize the following steps.
The basic idea is to use the geometric constraints to gradually determine the descriptions for the states.
\begin{enumerate}
    \item The constraint \emph{(Fix the basis)} ensures that there exists a unitary transformation $U^{(0)}$ such that
    \begin{equation} \label{eq:basisUU}
        (U^{(0)})^{-1} \rho_i^{\mathrm{(basis)}} (U^{(0)}) = \ketbra{i}{i}, \,\, \forall i \in \{1, \ldots, d\}.
    \end{equation}
    \item The constraint \emph{(Fix absolute amplitude for $\rho^{\mathrm{(real)}}_{ij}$)}  ensures that
    \begin{equation}
        (U^{(0)})^{-1} \rho_{ij}^{\mathrm{(real)}} (U^{(0)}) = \frac{1}{2} \left(\ket{i} + \euler^{\rmi \phi_{ij}} \ket{j}\right)\left(\bra{i} + \euler^{-\rmi \phi_{ij}} \bra{j}\right), \forall i \neq j \in \{1, \ldots, d\}
    \end{equation}
    for some unknown phase $\phi_{ij} \in [0, 2\pi)$.
    Consider $U^{(1)} = U^{(0)} D$, where $D$ is a diagonal matrix with $D_{11} = 1$ and $D_{ii} = \euler^{\rmi \phi_{1i}}, \forall i \in \{2, \ldots, d\}$.
    We have
    \begin{align}
        (U^{(1)})^{-1} \rho_i^{\mathrm{(basis)}} (U^{(1)}) &= \ketbra{i}{i}, &\forall i \in \{1, \ldots, d\},\\
        (U^{(1)})^{-1} \rho_{1 i}^{\mathrm{(real)}} (U^{(1)}) &= \frac{1}{2} \left(\ket{1} + \ket{i}\right)\left(\bra{1} + \bra{i}\right), &\forall i \in \{2, \ldots, d\}, \label{eq:rho1ireal}\\
        (U^{(1)})^{-1} \rho_{ij}^{\mathrm{(real)}} (U^{(1)}) &= \frac{1}{2} \left(\ket{i} + \euler^{\rmi \phi'_{ij}} \ket{j}\right)\left(\bra{i} + \euler^{-\rmi \phi'_{ij}} \bra{j}\right), &\forall i \geq 2, i \neq j \in \{1, \ldots, d\},\label{eq:rhoijreal}
    \end{align}
    where $\phi'_{ij} \in [0, 2\pi)$ is some unknown phase.
    \item From the constraints \emph{(Fix absolute amplitude for $\rho^{\mathrm{(triplet)}}_{ij}$)}, \emph{(Transfer relative phase $+1$, a)} and \emph{(Transfer relative phase $+1$, b)} and Eq.~\eqref{eq:rho1ireal}, we have
    \begin{align}
    (U^{(1)})^{-1} \rho_{ij}^{\mathrm{(triplet)}} (U^{(1)}) = \frac{1}{3} \left(\ket{1} +  \ket{i} + \ket{j}\right)\left(\bra{1} +  \bra{i} + \bra{j}\right), \quad \forall i \neq j \in \{2, \ldots, d\}. \label{eq:rhotriplet}
    \end{align}
    \item The constraint \emph{(Transfer relative phase $+1$, c)}, Eq.~\eqref{eq:rhoijreal} and Eq.~\eqref{eq:rhotriplet} ensure that
    \begin{equation} \label{eq:rhoijreal-full}
        (U^{(1)})^{-1} \rho_{ij}^{\mathrm{(real)}} (U^{(1)}) = \frac{1}{2} \left(\ket{i} +  \ket{j}\right)\left(\bra{i} + \bra{j}\right), \forall i \neq j \in \{1, \ldots, d\}.
    \end{equation}
    The state $\rho_{ij}^{\mathrm{(triplet)}}$ serves as an intermediate point to transfer relative phases.
    \item The constraints \emph{(Fix absolute amplitude for $\rho_{ij}^{\mathrm{(imag)}}$)}, \emph{(Partially fix the phase for $\rho_{ij}^{\mathrm{(imag)}}$)}, and Eq.~\eqref{eq:rhoijreal-full} ensure that
    \begin{equation} \label{eq:rhoij-imag-partial}
        (U^{(1)})^{-1} \rho_{ij}^{\mathrm{(imag)}} (U^{(1)}) = \frac{1}{2} \left(\ket{i} + s_{ij} \rmi \ket{j}\right)\left(\bra{i} - s_{ij} \rmi \bra{j}\right), \forall i \neq j \in \{1, \ldots, d\},
    \end{equation}
    where $s_{ij} = \pm 1$ is an unknown phase.
    If $s_{12} = 1$, we define $U^{(2)} = U^{(1)}$.
    If $s_{12} = -1$, we define $U^{(2)} = U^{(1)} K$, where $K$ is the complex conjugation operation.
    We have $U^{(2)}$ is either a unitary or anti-unitary transformation.
    Using the newly defined $U^{(2)}$, we have
    \begin{align}
        (U^{(2)})^{-1} \rho_i^{\mathrm{(basis)}} (U^{(2)}) &= \ketbra{i}{i}, &\forall i \in \{1, \ldots, d\},\\
        (U^{(2)})^{-1} \rho_{ij}^{\mathrm{(real)}} (U^{(2)}) &= \frac{1}{2} \left(\ket{i} +  \ket{j}\right)\left(\bra{i} + \bra{j}\right), &\forall i \neq j \in \{1, \ldots, d\},\\
        (U^{(1)})^{-1} \rho_{ij}^{\mathrm{(triplet)}} (U^{(1)}) &= \frac{1}{3} \left(\ket{1} +  \ket{i} + \ket{j}\right)\left(\bra{1} +  \bra{i} + \bra{j}\right), &\forall i \neq j \in \{2, \ldots, d\},\\
        (U^{(2)})^{-1} \rho_{12}^{\mathrm{(imag)}} (U^{(2)}) &= \frac{1}{2} \left(\ket{1} + \rmi \ket{2}\right)\left(\bra{1} - \rmi \bra{2}\right), & \\
        (U^{(2)})^{-1} \rho_{ij}^{\mathrm{(imag)}} (U^{(2)}) &= \frac{1}{2} \left(\ket{i} + s'_{ij} \rmi \ket{j}\right)\left(\bra{i} - s'_{ij} \rmi \bra{j}\right), &\forall i \neq j \in \{1, \ldots, d\},
    \end{align}
    for some $s'_{ij} \in \{\pm 1\}$.
    \item From \emph{(Fix absolute amplitude for $\rho^{\mathrm{(triplet, 12)}}_{j}$)}, \emph{(Transfer relative phase $+1$, a')}, and \emph{(Transfer relative phase $+\rmi$, b')}, we have
    \begin{equation}
        (U^{(2)})^{-1} \rho_{j}^{\mathrm{(triplet, 12)}} (U^{(2)}) = \frac{1}{3} \left(\ket{1} + \rmi \ket{2} + \rmi \ket{j}\right)\left(\bra{1} - \rmi \bra{2} - \rmi \bra{j}\right), \quad \forall j \in \{3, \ldots, d\}.
    \end{equation}
    \item From \emph{(Transfer relative phase $+\rmi$, c')}, we have
    \begin{equation} \label{eq:rho1jimag}
        (U^{(2)})^{-1} \rho_{1j}^{\mathrm{(imag)}} (U^{(2)}) = \frac{1}{2} \left(\ket{1} +  \rmi \ket{j}\right)\left(\bra{1} - \rmi \bra{j}\right), \quad \forall j \in \{2, \ldots, d\},
    \end{equation}
    Similar to $\rho_{ij}^{\mathrm{(triplet)}}$, the state $\rho_{j}^{\mathrm{(triplet, 12)}}$ serves as an intermediate point to transfer relative phases.
    \item From \emph{(Fix absolute amplitude for $\rho^{\mathrm{(triplet, i)}}_{ij}$)}, \emph{(Transfer relative phase $+1$, a'')}, and \emph{(Transfer relative phase $+\rmi$, b'')}, we have
    \begin{equation}
        (U^{(2)})^{-1} \rho_{ij}^{\mathrm{(triplet, i)}} (U^{(2)}) = \frac{1}{3} \left(\ket{1} + \ket{i} + \rmi \ket{j}\right)\left(\bra{1} + \bra{i} - \rmi \bra{j}\right), \quad \forall i \neq j \in \{2, \ldots, d\}.
    \end{equation}
    \item From \emph{(Transfer relative phase $+\rmi$, c'')}, we have
    \begin{equation} \label{eq:rhoij-imag-full}
        (U^{(2)})^{-1} \rho_{ij}^{\mathrm{(imag)}} (U^{(2)}) = \frac{1}{2} \left(\ket{i} +  \rmi \ket{j}\right)\left(\bra{i} - \rmi \bra{j}\right), \quad \forall i \neq j \in \{2, \ldots, d\}.
    \end{equation}
    Here, the state $\rho_{ij}^{\mathrm{(triplet, i)}}$ is used as an intermediate point to transfer relative phases.
\end{enumerate}
By considering $U = U^{(2)}$, we have established the claim.
\end{proof}

As $\eta$ goes to zero, the learning algorithm can find a set of states that satisfy the geometry up to an arbitrarily small error. This implies that the description for the states would also be close to the true one up to an arbitrarily small error as $\eta$ goes to zero.
Using this basic, idea, we can see that Lemma~\ref{lem:geom-state} yields the corollary stated below.

\begin{corollary}
Given $\epsilon > 0$. For all sufficiently small $\eta$, there exists a unitary or anti-unitary transformation $U$, such that
\begin{align}
    \norm{\rho_i^{\mathrm{(basis)}} - U \ketbra{i}{i} U^{-1}} &< \epsilon, &\forall i \in \{1, \ldots, d\}, \label{eq:rhoibasis} \\
    \norm{\rho_{ij}^{\mathrm{(real)}} - \frac{1}{2} U \left(\ket{i} +  \ket{j}\right)\left(\bra{i} + \bra{j}\right) U^{-1}} &< \epsilon, &\forall i \neq j \in \{1, \ldots, d\}, \label{eq:rhoireal} \\
    \norm{\rho_{ij}^{\mathrm{(imag)}} - \frac{1}{2} U \left(\ket{i} +  \rmi \ket{j}\right)\left(\bra{i} - \rmi \bra{j}\right) U^{-1}} &< \epsilon, &\forall i \neq j \in \{2, \ldots, d\}. \label{eq:rhoimag}
\end{align}
We only need to focus on these three sets of states in the following discussion, but the claim also holds for the other sets of states.
\end{corollary}

\subsection{Quantum state/process/measurement tomography}
\label{sec:tomography--intr}

In the final step, the learning algorithm utilizes the learned descriptions of the states in the previous subsection to perform quantum state/process/measurement tomography.
For readers familiar with quantum tomography, the claim could be established easily.
For completeness, we present the detailed derivations in the following.

\paragraph{States:} $\forall x \in \mathcal{X}$, we can always write $\rho_x$ as
\begin{equation}
     \rho_x = U \sum_{ij} a_{ij} \ketbra{i}{j} U^{-1},
\end{equation}
where $(a_{ij})_{ij}$ is a Hermitian matrix. If we write
\begin{equation}
    \sum_{ij} a_{ij} \ketbra{i}{j} = \sum_{i} a_{ii} \ketbra{i}{i} + \sum_{i \neq j} \frac{a_{ij} + a_{ji}}{2} (\ketbra{i}{j} + \ketbra{j}{i}) + \sum_{i \neq j} \frac{a_{ij} - a_{ji}}{2} (\ketbra{i}{j} - \ketbra{j}{i}),
\end{equation}
and we assume that Eq.~\eqref{eq:rhoibasis},~\eqref{eq:rhoireal},~and~\eqref{eq:rhoimag} holds exactly, then we can learn the matrix $(a_{ij})_{ij}$ by noting the following identities
\begin{align}
    a_{ii} &= \Tr(\rho_x \rho_{i}^{\mathrm{(basis)}}), \,\, \forall i, \label{eq:aii}\\
    \frac{a_{ij} + a_{ji}}{2} &= \Tr(\rho_x \rho_{ij}^{\mathrm{(real)}}) - \frac{\Tr(\rho_x \rho_{i}^{\mathrm{(basis)}}) + \Tr(\rho_x \rho_{j}^{\mathrm{(basis)}})}{2}, \,\, \forall i \neq j, \label{eq:aij+aji}\\
    \frac{a_{ij} - a_{ji}}{2} &= \frac{1}{\rmi} \left[\Tr(\rho_x \rho_{ij}^{\mathrm{(imag)}}) - \frac{\Tr(\rho_x \rho_{i}^{\mathrm{(basis)}}) + \Tr(\rho_x \rho_{j}^{\mathrm{(basis)}})}{2}\right], \,\, \forall i \neq j. \label{eq:aij-aji}
\end{align}
The state overlap $\Tr(\rho_1 \rho_2)$ can be estimated using the procedure provided in Appendix~\ref{sec:state-overlap} based on approximate Haar random unitaries.
For $\eta > 0$, define $\tilde{\rho}_x$ to be the empirical estimate of $\sum_{ij} a_{ij} \ketbra{i}{j}$ based on the above equations.
Due to error in the estimation of the state overlap $\Tr(\rho_1 \rho_2)$ and the error in the states $\rho_i^{\mathrm{(basis)}}, \rho_{ij}^{\mathrm{(real)}}, \rho_{ij}^{\mathrm{(imag)}}$, we have $\tilde{\rho}_x \neq \sum_{ij} a_{ij} \ketbra{i}{j}$. Nevertheless, one can use basic inequalities to show that $\norm{\tilde{\rho}_x - \sum_{ij} a_{ij} \ketbra{i}{j}}_1 < e_a(\eta)$ with high probability. The error $e_a(\eta)$ can be made arbitrarily small when $\eta$ goes to zero.

\paragraph{POVMs:} $\forall z \in \mathcal{Z}, \forall b \in \mathcal{B}$, we can learn $M_{zb}$ similar to learning states.
We write $M_{zb}$ as
\begin{equation}
     M_{zb} = U \sum_{ij} b_{ij} \ketbra{i}{j} U^{-1},
\end{equation}
where $(b_{ij})_{ij}$ is a Hermitian matrix.
For any quantum state $\rho$, we can estimate $\Tr(M_{zb} \rho)$ by simply computing the proportion of counts that we see the outcome $b$ when we measure $\mathcal{M}_z$ on $\rho$.
Using this simple procedure, we can estimate $\Tr(M_{zb} \rho)$ to an error $\eta$ with high probability.
Then we learn the matrix $(b_{ij})_{ij}$ using the same formulas given in Eq.~\eqref{eq:aii}~\eqref{eq:aij+aji}~and~\eqref{eq:aij-aji}, but we replace $\rho_x$ with $M_{zb}$.
For $\eta > 0$, we define $\tilde{M}_{zb}$ to be the empirical estimate of $\sum_{ij} b_{ij} \ketbra{i}{j}$.
Due to error in the estimation of $\Tr(M_{zb} \rho)$ and the error in $\rho_i^{\mathrm{(basis)}}, \rho_{ij}^{\mathrm{(real)}}, \rho_{ij}^{\mathrm{(imag)}}$, $\tilde{M}_{zb}$ is not exactly equal to $\sum_{ij} b_{ij} \ketbra{i}{j}$. But $\tilde{M}_{zb}$ will be close to $\sum_{ij} b_{ij} \ketbra{i}{j}$. In particular, there exists an error function $e_b(\eta)$, such that $\lim_{\eta \rightarrow 0} e_b(\eta) = 0$ and $\norm{\tilde{M}_{zb} - \sum_{ij} b_{ij} \ketbra{i}{j}}_1 < e_b(\eta)$ with high probability for any $\eta$.

\paragraph{CPTP maps:} $\forall y \in \mathcal{Y}$, we can write $\mathcal{E}_y$ as
\begin{equation}
    \mathcal{E}_y(\cdot) = U \sum_{ijkl} c_{ijkl} \ketbra{k}{l} \Tr( \ketbra{i}{j} U^{-1} (\cdot) U ) U^{-1}.
\end{equation}
The coefficients $c_{ijkl}$ could be learned using the state overlap procedure given in Appendix~\ref{sec:state-overlap} and the states $\rho_i^{\mathrm{(basis)}}, \rho_{ij}^{\mathrm{(real)}}, \rho_{ij}^{\mathrm{(imag)}}$ in Eq.~\eqref{eq:rhoibasis},~\eqref{eq:rhoireal},~and~\eqref{eq:rhoimag}.
To achieve this, we gather a collection of data by preparing each state in $\rho_i^{\mathrm{(basis)}}, \rho_{ij}^{\mathrm{(real)}}, \rho_{ij}^{\mathrm{(imag)}}$, evolving under $\mathcal{E}_y$, and estimating the state overlap of the output state with every state in $\rho_i^{\mathrm{(basis)}}, \rho_{ij}^{\mathrm{(real)}}, \rho_{ij}^{\mathrm{(imag)}}$.
We can then use the collection of data to estimate $c_{ijkl}, \forall ijkl$.
If Eq.~\eqref{eq:rhoibasis},~\eqref{eq:rhoireal},~and~\eqref{eq:rhoimag} holds exactly, then for $i \neq j$ and $k = l$,
\begin{align}
    c_{ijkl} &= \Tr( \rho_k^{\mathrm{(basis)}} \mathcal{E}_y( \rho_{ij}^{\mathrm{(real)}} ) ) + \frac{1}{\rmi} \Tr( \rho_k^{\mathrm{(basis)}} \mathcal{E}_y( \rho_{ij}^{\mathrm{(imag)}} ) )\\
    &- \left(\frac{1}{2} + \frac{1}{2 \rmi}\right) \left(\Tr( \rho_k^{\mathrm{(basis)}} \mathcal{E}_y( \rho_{i}^{\mathrm{(basis)}} ) ) + \Tr( \rho_k^{\mathrm{(basis)}} \mathcal{E}_y( \rho_{j}^{\mathrm{(basis)}} ) ) \right)
\end{align}
For $i = j$ and $k = l$, we can obtain
\begin{align}
    c_{ijkl} &= \Tr( \rho_k^{\mathrm{(basis)}} \mathcal{E}_y( \rho_{i}^{\mathrm{(basis)}} ) ).
\end{align}
For $i = j$ and $k \neq l$, we see that
\begin{align}
    c_{ijkl} &= \Tr( \rho_{kl}^{\mathrm{(real)}} \mathcal{E}_y( \rho_{i}^{\mathrm{(basis)}} ) ) + \frac{1}{\rmi} \Tr( \rho_{kl}^{\mathrm{(imag)}} \mathcal{E}_y( \rho_{i}^{\mathrm{(basis)}} ) )\\
    &- \left(\frac{1}{2} + \frac{1}{2 \rmi}\right) \left(\Tr( \rho_k^{\mathrm{(basis)}} \mathcal{E}_y( \rho_{i}^{\mathrm{(basis)}} ) ) + \Tr( \rho_l^{\mathrm{(basis)}} \mathcal{E}_y( \rho_{i}^{\mathrm{(basis)}} ) ) \right)
\end{align}
For $i \neq j$ and $k \neq l$, we have
\begin{align}
    c_{ijkl} &= \Tr\Bigg( \left(\rho_{kl}^{\mathrm{(real)}} + \frac{1}{\rmi} \rho_{kl}^{\mathrm{(imag)}} - \left(\frac{1}{2} + \frac{1}{2 \rmi}\right) \left( \rho_k^{\mathrm{(basis)}} + \rho_l^{\mathrm{(basis)}} \right) \right) \\
    &\quad\quad\quad \mathcal{E}_y \left( \left(\rho_{ij}^{\mathrm{(real)}} + \frac{1}{\rmi} \rho_{ij}^{\mathrm{(imag)}} - \left(\frac{1}{2} + \frac{1}{2 \rmi}\right) \left( \rho_i^{\mathrm{(basis)}} + \rho_j^{\mathrm{(basis)}} \right) \right) \right) \Bigg).
\end{align}
Expanding the right hand side of the above equation gives a weighted sum of $\Tr(\rho_2 \mathcal{E}_y (\rho_1))$ for some states $\rho_1, \rho_2$.
We consider $\tilde{\mathcal{E}}_y(\cdot)$ to be the empirical estimate for $\sum_{ijkl} c_{ijkl} \ketbra{k}{l} \Tr( \ketbra{i}{j} (\cdot) )$.
There exists an error function $e_c(\eta)$ such that $\norm{\tilde{\mathcal{E}}_y(\cdot) - \sum_{ijkl} c_{ijkl} \ketbra{k}{l} \Tr( \ketbra{i}{j} (\cdot) )}_\diamond < e_c(\eta)$.
Furthermore, as $\eta$ approaches zero, $e_c(\eta)$ goes to zero.

\subsection{Putting everything together}

For each world model $\mathcal{W} \in \mathcal{Q}$, after finishing the tomography step in Appendix~\ref{sec:tomography--intr}, we can guarantee the following.
There is an error function $\epsilon(\eta)$. For $\eta > 0$, there exists a global unitary or anti-unitary transformation $U$, such that the learned descriptions
$\tilde{\rho}_x, \tilde{\mathcal{E}}_y, \tilde{M}_{z b}$ satisfies
\begin{align}
 \norm{\rho_x - U \tilde{\rho}_x U^{-1}}_1 &< \epsilon(\eta), \\
 \norm{M_{z b} - U \tilde{M}_{z b} U^{-1}}_1 &< \epsilon(\eta), \\
 \norm{ \mathcal{E}_y(\cdot) - U \tilde{\mathcal{E}}_y(U^{-1} (\cdot) U) U^{-1} }_\diamond &< \epsilon(\eta),
\end{align}
for all $x \in \mathcal{X}, y \in \mathcal{Y}, z \in \mathcal{Z}, b \in \mathcal{B}$.
As $\eta$ goes to zero, $\epsilon(\eta)$ goes to zero.
After running the above procedures with precision parameter $\eta$, the learning algorithm considers $\eta \leftarrow \eta / 2$ and repeatedly runs the previous steps to obtain more accurate descriptions.
Because $\epsilon(\eta)$ goes to zero as $\eta$ goes to zero, the learning algorithm can learn all the physical descriptions to arbitrarily small error up to a global unitary or anti-unitary transformation.
Hence, $\mathcal{Q}$ is learnable.
This concludes the proof of Theorem~\ref{thm:int-description-ditto}.

\section{Basic properties of learnability and unlearnability}
\label{sec:basiclearna}

In this section, we will present various basic results regarding the relationship between different classes of world models.
These results will be useful for proving what kinds of world models are learnable, and what kinds are not, in the following sections.

We begin with a basic property: If a model class $\mathcal{Q}$ is learnable, then any subset $\mathcal{Q}'$ of $\mathcal{Q}$ is also learnable. 
This property is expected because removing possible models from the class will not make it harder to learn which model is the correct one.
We note that $\mathcal{Q}'$ and $\mathcal{Q}$ are model classes, i.e., sets of potential world models. The relation of $\mathcal{Q}' \subseteq \mathcal{Q}$ is very different from the concept of extension given in Definition~\ref{def:extension}, which considers the relation between two world models.

\begin{proposition}[Monotonicity of (un)learnability] \label{prop:mono-learnability}
Given sets $\mathcal{X}, \mathcal{Y}, \mathcal{Z}, \mathcal{B}$ and two model classes $\mathcal{Q}, \mathcal{Q}'$ over $\mathcal{X}, \mathcal{Y}, \mathcal{Z}, \mathcal{B}$ such that $\mathcal{Q}' \subseteq \mathcal{Q}$.
If $\mathcal{Q}$ is learnable, then $\mathcal{Q}'$ is learnable. Equivalently, if $\mathcal{Q}'$ is unlearnable, then $\mathcal{Q}$ is unlearnable.
\end{proposition}
\begin{proof}
Every world model in $\mathcal{Q}'$ is in $\mathcal{Q}$.
Hence, if $\mathcal{Q}$ is learnable, then $\mathcal{Q}' \subseteq \mathcal{Q}$ is learnable.
This is equivalent to the contrapositive statement:
if $\mathcal{Q}' \subseteq \mathcal{Q}$ is unlearnable, then $\mathcal{Q}$ is unlearnable.
\end{proof}

Another important result states that a model class $\mathcal{Q}$ is unlearnable if $\mathcal{Q}$ contains world models $\mathcal{W}_1$ and $\mathcal{W}_2$ that are weakly indistinguishable, but are not equivalent. This follows because, by the definition of weakly indistinguishable, no experiment within the model class can tell $\mathcal{W}_1$ and $\mathcal{W}_2$ apart.
It may seem that this follows immediately from the definition, but there are subtleties arising from the fact that learning is probabilistic and allows arbitrarily small error.

\begin{proposition}[Weakly indistinguishability implies unlearnability]
\label{prop:indist->unlearn}
Given sets $\mathcal{X}, \mathcal{Y}, \mathcal{Z}, \mathcal{B}$ and a model class $\mathcal{Q} = \{\mathcal{W}\}$ for $d$-dimensional world models over $\mathcal{X}, \mathcal{Y}, \mathcal{Z}, \mathcal{B}$. If there exists $\mathcal{W}_1 \not\equiv \mathcal{W}_2 \in \mathcal{Q}$ such that $\mathcal{W}_1$ and $\mathcal{W}_2$ are weakly indistinguishable, then $\mathcal{Q}$ is unlearnable.
\end{proposition}
\begin{proof}
Assume that $\mathcal{W}_A \not\equiv \mathcal{W}_B \in \mathcal{Q}$ are weakly indistinguishable.
Because $\mathcal{W}_A \not\equiv \mathcal{W}_B$, for all unitary and anti-unitary $U$, there exists $x \in \mathcal{X}$ or $y \in \mathcal{Y}$ or $z \in \mathcal{Z}, b \in \mathcal{B}$ such that the corresponding physical operations are different, i.e., $\rho^B_x \neq U \rho^A_x U^{-1}$, $\mathcal{E}^{B}_y(\cdot) \neq U \mathcal{E}^{A}_y( U^{-1} (\cdot) U) U^{-1}$, or $M^B_{z b} \neq U M^A_{z b} U^{-1}$.
We define the minimum error $\tilde{\epsilon}$ over $U$ to be
\begin{equation}
    \tilde{\epsilon} = \min_{U} \sup_{\substack{x \in \mathcal{X}\\y \in \mathcal{Y}\\ z \in \mathcal{Z} \\ b \in \mathcal{B}}}\left( \norm{\rho^B_x - U \rho^A_x U^{-1}}_1, \norm{\mathcal{E}^{B}_y(\cdot) - U \mathcal{E}^{A}_y( U^{-1} (\cdot) U) U^{-1}}_\diamond, \norm{M^B_{z b} - U M^A_{z b} U^{-1}}_1 \right),
\end{equation}
where $U$ is a unitary or anti-unitary transformation. We can use minimum instead of infimum because unitary and anti-unitary transformations form a compact space.
If $\tilde{\epsilon} = 0$, then $\mathcal{W}_A \equiv \mathcal{W}_B$. Hence $\tilde{\epsilon} > 0$.
The quantity $\tilde{\epsilon}$ sets a lower bound on the error achieved by any learning algorithm.

Suppose that $\mathcal{Q}$ is learnable. Then there is an algorithm $\mathcal{A}$ and a unitary or anti-unitary $U_A$ such that, with probability at least $1-\delta$, for every action the output from the algorithm has an error of at most $\epsilon$ after transforming under $U_A$. We may choose $\delta=1/3$ and $\epsilon = \tilde{\epsilon} / 3$.
Similarly, for world model $\mathcal{W}_B$, there exists a unitary or anti-unitary $U_B$ such that, with probability at least $1-\delta$, for every action the output from the algorithm $\mathcal{A}$ has an error of at most $\epsilon$ after transforming under $U_B$.
From the definition of $\tilde{\epsilon}$, we have
\begin{equation}
    \sup_{\substack{x \in \mathcal{X}\\y \in \mathcal{Y}\\ z \in \mathcal{Z} \\ b \in \mathcal{B}}}\left( \norm{\rho^B_x - U_* \rho^A_x U_*^{-1}}_1, \norm{\mathcal{E}^{B}_y(\cdot) - U_* \mathcal{E}^{A}_y( U_*^{-1} (\cdot) U_*) U_*^{-1}}_\diamond, \norm{M^B_{z b} - U_* M^A_{z b} U_*^{-1}}_1 \right) \geq \tilde{\epsilon}
\end{equation}
where $U_* = U_B U_A^{-1}$.
Hence, there exists an action such that the error $> \tfrac{9}{10} \tilde{\epsilon}$. (Here we could have chosen any positive multiplicative factor less than one, and chose $\tfrac{9}{10}$ for convenience.)
Without loss of generality, assume that for some $x \in \mathcal{X}$, $\norm{\rho^B_x - U_* \rho^A_x U_*^{-1}}_1 > \tfrac{9}{10} \tilde{\epsilon}$.

Our randomized learning algorithm outputs an estimate for the state $\rho_x$, which we denote by $\tilde{\rho}_x$.
Because $\mathcal{W}_A, \mathcal{W}_B$ are weakly indistinguishable, this output is drawn from the same probability distribution whether the world model is $\mathcal{W}_A$ or $\mathcal{W}_B$, as follows immediately from Eq.~\eqref{eq:probeq} in Definition~\ref{def:weakly-indist}.
We will refer to this probability distribution as $p_{AB}(\tilde{\rho}_x)$.
From the learnability, we have $\norm{\rho^A_x - U_A \tilde{\rho}_x U_A^{-1}}_1 \leq \epsilon$ with probability $\geq 2/3$ and $\norm{\rho^B_x - U_B \tilde{\rho}_x U_B^{-1}}_1 \leq \epsilon$ with probability $\geq 2/3$.
The above probability statements are over the same probability distribution $p_{AB}(\tilde{\rho}_x)$. Equivalently,
\begin{align}
    P\left[ \norm{ \rho^A_x - U_A \tilde{\rho}_x U_A^{-1}}_1 > \epsilon \right] &< 1/3,\\
    P\left[ \norm{\rho^B_x - U_B \tilde{\rho}_x U_B^{-1}}_1 > \epsilon \right] &< 1/3,
\end{align}
Hence, by union bound, we have
\begin{equation}
    P\left[ \norm{\rho^A_x - U_A  \tilde{\rho}_x U_A^{-1} }_1 > \epsilon \,\, \mbox{or} \,\, \norm{\rho^B_x - U_B \tilde{\rho}_x U_B^{-1} }_1 > \epsilon \right] < 2/3.
\end{equation}
This is equivalent to the fact that $\norm{\rho^A_x - U_A  \tilde{\rho}_x U_A^{-1} }_1 \leq \epsilon$ and $\norm{\rho^B_x - U_B \tilde{\rho}_x U_B^{-1} }_1 \leq \epsilon$ with probability $\geq 1/3$.
Because the probability is greater than zero and the probability distribution is over the choice of $\tilde{\rho}_x$, there must exist some state $\tilde{\rho}_x \in \mathbb{C}^{d \times d}$ such that $\norm{\rho^A_x - U_A  \tilde{\rho}_x U_A^{-1} }_1 \leq \epsilon$ and $\norm{\rho^B_x - U_B \tilde{\rho}_x U_B^{-1} }_1 \leq \epsilon$.
Because trace norm $\norm{\cdot}_1$ is invariant under unitary or anti-unitary transformations, we have $\norm{ U_B U_A^{-1} \rho^A_x U_A U_B^{-1} - U_B  \tilde{\rho}_x U_B^{-1} }_1 \leq \epsilon$.
Recall that $U_* = U_B U_A^{-1}$.
By the triangle inequality, we have
\begin{equation}
     \frac{9}{10} \tilde{\epsilon} < \norm{U_* \rho^A_x U_*^{-1} - \rho^B_x}_1 \leq \norm{U_B U_A^{-1} \rho^A_x U_A U_B^{-1} - U_B \tilde{\rho}_x U_B^{-1}}_1 + \norm{\rho^B_x - U_B \tilde{\rho}_x U_B^{-1}}_1 \leq 2\epsilon = \frac{2}{3} \tilde{\epsilon}.
\end{equation}
This is a contradiction, hence $\mathcal{Q}$ is not learnable.
\end{proof}


We present a simple example where $\mathcal{W}_1$ and $\mathcal{W}_2$ are weakly indistinguishable, but not equivalent.
Consider a model class $\mathcal{Q}$ that contains two $d=2$-dimensional world models $\{\mathcal{W}_A, \mathcal{W}_B\}$ with the action space $\mathcal{X} = \{0\}, \mathcal{Y} = \{h, t\}, \mathcal{Z} = \{0\}$ and the outcome space $\mathcal{B} = \{0, 1\}$. 
We define the physical actions in $\mathcal{W}_A$ and $\mathcal{W}_B$ as
\begin{align}
    \rho^A_0 &= I/2, & \mathcal{E}^A_h(\rho) &= H\rho H^\dagger, & \mathcal{E}^A_t(\rho) &= T\rho T^\dagger, & \mathcal{M}^A_0 &= \{I/2, I/2\}, \\
    \rho^B_0 &= \ketbra{0}{0}, & \mathcal{E}^B_h(\rho) &= H\rho H^\dagger, & \mathcal{E}^B_t(\rho) &= T\rho T^\dagger, & \mathcal{M}^B_0 &= \{I/2, I/2\}.
\end{align}
$\mathcal{W}_A$ has an initial state that is maximally mixed, hence the state $ \rho^A_0$ has a purity $\Tr((\rho^A_0)^2)$ of $1/2$. But $\mathcal{W}_B$ has an initial state that is pure, so the state $ \rho^B_0$ has a purity of $1$. Theorem~\ref{prop:equiv} implies that the two world models are not equivalent. However, both of the POVMS $\mathcal{M}_0^A$ and $\mathcal{M}_0^B$ produce uniformly random outcomes in $\mathcal{B}$ when applied to any state. Therefore, $\mathcal{W}_A$ and $\mathcal{W}_B$ are weakly indistinguishable, and hence by Proposition~\ref{prop:indist->unlearn} the model class $\mathcal{Q}$ is unlearnable. In this example, both of the world models $\mathcal{W}_A$ and $\mathcal{W}_B$ have a useless measurement device that provides no information, so there is no way to learn which is which.


Monotonicity of learnability focuses on two model classes that have the same action spaces.
Here, we provide a basic proposition that considers two model classes with different action spaces.
The proposition holds because of the compositional nature in the design of an experiment --- we can compose different states, evolutions, and POVMs to form new states, evolutions, and POVMs.

\begin{proposition}[Learnability after adding composed states] \label{prop:comp-learnability-state}
Given sets $\mathcal{X}, \mathcal{Y}, \mathcal{Z}, \mathcal{B}$ and a model class $\mathcal{Q} = \{\mathcal{W}\}$ over $\mathcal{X}, \mathcal{Y}, \mathcal{Z}, \mathcal{B}$.
Consider a set $\Xi$, a constant $L \geq 1$, and a function $f$ that takes in $\xi \in \Xi$ and outputs $(x, y_1, \ldots, y_\ell)$ where $\ell \leq L, x \in \mathcal{X}, y_1, \ldots, y_\ell \in \mathcal{Y}$.
The model class $\mathcal{Q}'$ is over $\mathcal{X}' = \mathcal{X} \cup \Xi, \mathcal{Y}, \mathcal{Z}, \mathcal{B}$, and contains the world model
\begin{equation}
\left(\{\rho_x\}_{x \in \mathcal{X}} \cup \{\rho_{\xi} = \left(\mathcal{E}_{y_\ell} \circ \ldots \circ \mathcal{E}_{y_1}\right)(\rho_x)\}_{\substack{\xi \in \Xi, \\ f(\xi) = (x, y_1, \ldots, y_\ell)}}, \{\mathcal{E}_y\}_{y \in \mathcal{Y}}, \{\mathcal{M}_z\}_{z \in \mathcal{Z}} \right),
\end{equation}
for each world model $\mathcal{W} = \left(\{\rho_x\}_{x \in \mathcal{X}}, \{\mathcal{E}_y\}_{y \in \mathcal{Y}}, \{\mathcal{M}_z\}_{z \in \mathcal{Z}} \right)$ in the original model class $\mathcal{Q}$.
We have $\mathcal{Q}$ is learnable if and only if $\mathcal{Q}'$ is learnable.
\end{proposition}
\begin{proof}
We begin with two basic statements.
First, every experiment in $\mathcal{Q}$ can be simulated by an experiment in $\mathcal{Q}'$.
And every experiment in $\mathcal{Q}'$ can be simulated by an experiment in $\mathcal{Q}$.
The first statement immediately holds by noting that $\mathcal{Q}'$ contains all the actions in $\mathcal{Q}$.
The second statement is true because the new action added in $\mathcal{Q}'$ is composed of actions in $\mathcal{Q}$. Since each experiment $E$ is a composition of actions, we can compose the corresponding actions in $\mathcal{Q}$ to simulate an experiment in $\mathcal{Q}'$.
The proof of this proposition is simple given the knowledge of these two facts.
We separate the proof for the two directions of the statement into two paragraphs.

\paragraph{$\mathcal{Q}'$ is learnable implies $\mathcal{Q}$ is learnable:}
If $\mathcal{Q}'$ is learnable, then there is a learning algorithm, such that for every action $a$ in the action spaces $\mathcal{X}, \mathcal{Y}, \mathcal{Z}$, the algorithm uses actions in $\mathcal{X}', \mathcal{Y}, \mathcal{Z}$ to learn the intrinsic description of action~$a$.
We can simulate every added action $\xi$ in $\Xi$ with the actions $x \in \mathcal{X}, y_1, \ldots, y_\ell \in \mathcal{Y}$, where $f(\xi) = (x, y_1, \ldots, y_\ell)$. Hence, we have a learning algorithm using only actions in $\mathcal{X}, \mathcal{Y}, \mathcal{Z}$ to learn the description of actions in $\mathcal{X}, \mathcal{Y}, \mathcal{Z}$ under the model class $\mathcal{Q}$.
Together, $\mathcal{Q}$ is learnable if $\mathcal{Q}'$ is learnable.

\paragraph{$\mathcal{Q}$ is learnable implies $\mathcal{Q}'$ is learnable:}
If $\mathcal{Q}$ is learnable, then there is a learning algorithm that learns the physical operation associated with every actions in $\mathcal{X}, \mathcal{Y}, \mathcal{Z}$. To show that $\mathcal{Q}'$ is learnable, we need to show that all actions in $\mathcal{X}', \mathcal{Y}, \mathcal{Z}$ are learnable.
By simulating the experiments under $\mathcal{Q}$ using actions in $\mathcal{Q}'$, we can learn all the physical operations associated to actions in $\mathcal{X}, \mathcal{Y}, \mathcal{Z}$.
Now for all $\xi \in \Xi$, the initial state $\rho_{\xi}$ associated to the additional action $\xi$ with $f(\xi) = (x, y_1,\ldots, y_\ell)$ can be learned.
This follows from the facts that $\rho_{\xi}$ is equal to $\left(\mathcal{E}_{y_\ell} \circ \ldots \circ \mathcal{E}_{y_1}\right)(\rho_x)$, and each of $\rho_x, \cE_{y_1}, \ldots, \cE_{y_\ell}$ can be learned to arbitrarily high accuracy up to a global unitary or anti-unitary transformation.
Therefore, $\mathcal{Q}'$ is learnable.
\end{proof}

The same argument as for Proposition~\ref{prop:comp-learnability-state} can be used to establish the following other two propositions where we consider a model class with new composed CPTP maps or POVMs.

\begin{proposition}[Learnability after adding composed CPTP maps] \label{prop:comp-learnability-map}
Given sets $\mathcal{X}, \mathcal{Y}, \mathcal{Z}, \mathcal{B}$ and a model class $\mathcal{Q} = \{\mathcal{W}\}$ over $\mathcal{X}, \mathcal{Y}, \mathcal{Z}, \mathcal{B}$.
Consider a set $\Xi$, a constant $L \geq 2$, and a function $f$ that takes in an element $\xi$ in $\Xi$ and outputs $(y_1, \ldots, y_\ell)$ where $\ell \leq L, y_1, \ldots, y_\ell \in \mathcal{Y}$.
The model class $\mathcal{Q}'$ is over $\mathcal{X}, \mathcal{Y}' = \mathcal{Y} \cup \Xi, \mathcal{Z}, \mathcal{B}$, and contains the world model
\begin{equation}
\left(\{\rho_x\}_{x \in \mathcal{X}}, \{\mathcal{E}_y\}_{y \in \mathcal{Y}} \cup \{\mathcal{E}_{\xi} = \left(\mathcal{E}_{y_\ell} \circ \ldots \circ \mathcal{E}_{y_1}\right)\}_{\substack{\xi \in \Xi, \\ f(\xi) = (y_1, \ldots, y_\ell)}}, \{\mathcal{M}_z\}_{z \in \mathcal{Z}} \right),
\end{equation}
for each world model $\mathcal{W} = \left(\{\rho_x\}_{x \in \mathcal{X}}, \{\mathcal{E}_y\}_{y \in \mathcal{Y}}, \{\mathcal{M}_z\}_{z \in \mathcal{Z}} \right)$ in the original model class $\mathcal{Q}$.
We have $\mathcal{Q}$ is learnable if and only if $\mathcal{Q}'$ is learnable.
\end{proposition}

\begin{proposition}[Learnability after adding composed POVMs] \label{prop:comp-learnability-povm}
Given sets $\mathcal{X}, \mathcal{Y}, \mathcal{Z}, \mathcal{B}$ and a model class $\mathcal{Q} = \{\mathcal{W}\}$ over $\mathcal{X}, \mathcal{Y}, \mathcal{Z}, \mathcal{B}$.
Consider a set $\Xi$, a constant $L \geq 1$, and a function $f$ that takes in an element $\xi$ in $\Xi$ and outputs $(y_1, \ldots, y_\ell, z)$ where $\ell \leq L, y_1, \ldots, y_\ell \in \mathcal{Y}, z \in \mathcal{Z}$.
The model class $\mathcal{Q}'$ is over the spaces $\mathcal{X}, \mathcal{Y}, \mathcal{Z}' = \mathcal{Z} \cup \Xi, \mathcal{B}$, and contains the world model
\begin{equation}
\left(\{\rho_x\}_{x \in \mathcal{X}}, \{\mathcal{E}_y\}_{y \in \mathcal{Y}}, \{\mathcal{M}_z\}_{z \in \mathcal{Z}} \cup \{\mathcal{M}_{\xi} = \mathcal{M}_z \circ \left(\mathcal{E}_{y_\ell} \circ \ldots \circ \mathcal{E}_{y_1}\right)\}_{\substack{\xi \in \Xi, \\ f(\xi) = (y_1, \ldots, y_\ell, z)}} \right),
\end{equation}
for each world model $\mathcal{W} = \left(\{\rho_x\}_{x \in \mathcal{X}}, \{\mathcal{E}_y\}_{y \in \mathcal{Y}}, \{\mathcal{M}_z\}_{z \in \mathcal{Z}} \right)$ in the original model class $\mathcal{Q}$.
We have $\mathcal{Q}$ is learnable if and only if $\mathcal{Q}'$ is learnable.
\end{proposition}

Similar to the three propositions stated above, we can consider new actions that are convex combinations of existing actions. Adding the new actions does not affect the learnability.

\begin{proposition}[Learnability after adding mixtures of states] \label{prop:mix-learnability-state}
Given sets $\mathcal{X}, \mathcal{Y}, \mathcal{Z}, \mathcal{B}$ and a model class $\mathcal{Q} = \{\mathcal{W}\}$ over $\mathcal{X}, \mathcal{Y}, \mathcal{Z}, \mathcal{B}$.
Consider a set $\Xi$, a constant $L \geq 1$, and a function $f$ that takes in an element $\xi$ in $\Xi$ and outputs $((p_1, x_1) \ldots, (p_\ell, x_\ell))$ where $\ell \leq L, x_1, \ldots, x_\ell \in \mathcal{X}$ and $(p_1, \ldots, p_\ell)$ is a probability distribution.
The model class $\mathcal{Q}'$ is over the spaces $\mathcal{X}' = \mathcal{X} \cup \Xi, \mathcal{Y}, \mathcal{Z}, \mathcal{B}$, and contains the world model
\begin{equation}
\left(\{\rho_x\}_{x \in \mathcal{X}} \cup \left\{\rho_{\xi} = \sum_{\ell = 1}^L p_\ell \rho_{x_\ell} \right\}_{\substack{\xi \in \Xi, \\ f(\xi) = ((p_1, x_1) \ldots, (p_\ell, x_\ell))}}, \{\mathcal{E}_y\}_{y \in \mathcal{Y}}, \{\mathcal{M}_z\}_{z \in \mathcal{Z}} \right),
\end{equation}
for each world model $\mathcal{W} = \left(\{\rho_x\}_{x \in \mathcal{X}}, \{\mathcal{E}_y\}_{y \in \mathcal{Y}}, \{\mathcal{M}_z\}_{z \in \mathcal{Z}} \right)$ in the original model class $\mathcal{Q}$.
We have $\mathcal{Q}$ is learnable if and only if $\mathcal{Q}'$ is learnable.
\end{proposition}
\begin{proof}
Similar to the proof of Proposition~\ref{prop:comp-learnability-state}, every experiment in $\mathcal{Q}$ can be simulated by experiments in $\mathcal{Q}'$. And every experiment in $\mathcal{Q}'$ can be simulated by experiments in $\mathcal{Q}$.
The first statement is trivial as $\mathcal{Q}'$ contains all actions in $\mathcal{Q}$.
The second statement is true because we can simulate any experiment that begins with the action $\xi \in \Xi$ by randomly sampling $x_i$ from $(x_1, \ldots, x_\ell)$ according to the probability distribution $(p_1, \ldots, p_\ell)$ and running the experiment using $x_i \in \mathcal{X}$.
Using essential the same proof as for Proposition~\ref{prop:comp-learnability-state}, we can show that $\mathcal{Q}'$ is learnable implies $\mathcal{Q}$ is learnable and $\mathcal{Q}$ is learnable implies $\mathcal{Q}'$ is learnable.
\end{proof}

Using essentially the same proof as for Proposition~\ref{prop:mix-learnability-state}, we can also obtain the equivalence of learnability after adding mixtures of CPTP maps or POVMs.



\section{Gate-dependent Pauli noise is unlearnable with Clifford+T gates}
\label{sec:cliffordT}

\subsection{Statement and unlearnability of a target model class}

In quantum state/process tomography, it is well known that Clifford circuits are informationally complete.
For example, we can learn any quantum state with Clifford circuits and computational basis measurement. We can also learn any quantum process with Clifford circuits, all zero state preparation, and computational basis measurements.
In these works, it is often assumed that the Clifford circuits, the state preparations, and the measurements are perfect.
The situation changes dramatically when these physical operations are not perfect.

In this section, we show that when there is gate-dependent Pauli noise, Clifford circuits are fundamentally uncapable of learning the noisy processes.
Even more interestingly, adding the non-Clifford $T$ gate is still insufficient, even if the all zero state preparation and computational basis measurements are perfect.

\begin{theorem}[Restatement of Theorem~\ref{prop:Clif+T}; Gate-dependent Pauli noise is unlearnable with Clifford+T gates] \label{prop:Clif+T-detailed}
Given $\tfrac{1}{2} > \epsilon > 0$. Consider a qubit system. Suppose we can prepare the zero state $\ket{0}$ perfectly and any state $\rho$ with an unknown error $\leq \epsilon$, measure in the computational basis perfectly, and apply Clifford gates and $T$ gate, where each gate is followed by an unknown gate-dependent Pauli noise channel that is $\epsilon$-close to the identity channel.
It is impossible for any algorithm to learn the gate-dependent Pauli noise channels to arbitrarily small error.
\end{theorem}

To prove Theorem~\ref{prop:Clif+T-detailed}, we begin by stating the conditions in Theorem~\ref{prop:Clif+T-detailed} as a model class.
Consider the the action spaces $\mathcal{X} = \{x_{\ketbra{0}{0}}\} \cup \{x_{\sigma}\}_{\sigma: \mathrm{state}}, \mathcal{Y} = \{y_U\}_{U \in \mathcal{C} \cup \{T\} }, \mathcal{Z} = \{ 0 \}$, where $\sigma$ is a quantum state, $\mathcal{C}$ is the Clifford group, and $T$ is the $T$ gate. And consider the outcome space $\mathcal{B} = \{0, 1\}$.
Given $\epsilon > 0$, we define the model class $\mathcal{Q}^{\epsilon} = \{\mathcal{W}\}$ to be the set of world models $\mathcal{W} = \left(\{\rho_x\}_{x \in \mathcal{X}}, \{\mathcal{E}_y\}_{y \in \mathcal{Y}}, \{\mathcal{M}_z\}_{z \in \mathcal{Z}} \right)$ that satisfies the following conditions
\begin{align}
    \norm{\rho_{x_\sigma} - \sigma}_1 &\leq \epsilon, &\,\, \forall \sigma: \mathrm{state},\\
    \rho_{x_{\ketbra{0}{0}}} &= \ketbra{0}{0}, &\\
    \mathcal{E}_{y_U}(\rho) &= \mathcal{P}_{U} (U \rho U^\dagger), \, \norm{\mathcal{P}_{U} - \mathcal{I}}_\diamond \leq \epsilon, &\,\, \forall U \in \mathcal{C} \cup \{T\},\\
    \mathcal{M}_0 &= \{\ketbra{0}{0}, \ketbra{1}{1}\}. &
\end{align}
where $\mathcal{P}_U$ is a Pauli channel, i.e.,
\begin{equation}
\mathcal{P}_U(\rho) = p^U_I \rho + p^U_X X \rho X + p^U_Y Y \rho Y + p^U_Z Z \rho Z
\end{equation}
for some probability distribution $\left( p^U_I, p^U_X, p^U_Y, p^U_Z \right)$, or equivalently (in the Heisenberg picture)
\begin{equation}
\mathcal{P}_U(I) = I, \mathcal{P}_U(X) = \lambda^U_X X, \mathcal{P}_U(Y) = \lambda^U_Y Y, \mathcal{P}_U(Z) = \lambda^U_Z Z
\end{equation}
for some real value $\lambda^U_X, \lambda^U_Y, \lambda^U_Z$, $\mathcal{I}$ is the identity channel. $\norm{\cdot}_\diamond$ denotes the diamond norm.

\subsection{Unlearnability of $\mathcal{Q}^{\epsilon}$}
\label{sec:unlearnability-Qeps}

We provide more analysis of  $\mathcal{Q}^{\epsilon}$ in this subsection. First, we show that $\mathcal{Q}^{\epsilon}$ is not redundant: no distinct world models in $\mathcal{Q}^{\epsilon}$ are equivalent and hence describe the same physical reality.
From Theorem~\ref{prop:equiv}, two equivalent world models are related by a global unitary or anti-unitary transformation~$V$.
If two distinct world models in $\mathcal{Q}^{\epsilon}$ are equivalent, then there exists a global unitary or anti-unitary transformation $V \neq \euler^{\rmi \phi} I$, for any $\phi \in \mathbb{R}$, and Pauli channels $\mathcal{P}, \mathcal{P}', \mathcal{Q}, \mathcal{Q}'$ that are at most $\epsilon$-far from identity channel, such that
\begin{align}
    \ketbra{0}{0} &= V \ketbra{0}{0} V^{-1}, \label{eq:v00vinv-1}\\
    \ketbra{1}{1} &= V \ketbra{1}{1} V^{-1},\\
    \mathcal{P}(H \ketbra{0}{0} H) &= V \mathcal{P}'(H V^{-1} \ketbra{0}{0} V H) V^{-1},\\
    \mathcal{Q}((SH) \ketbra{0}{0} (SH)^{\dagger}) &= V \mathcal{Q}'((SH) V^{-1} \ketbra{0}{0} V (SH)^{\dagger}) V^{-1}, \label{eq:v00vinv-2}
\end{align}
where $H$ is the Hadamard gate and $S$ is the phase gate.
The first two equalities come from the condition for $\cM_0$ (using the condition on $\rho_{x_{\ketbra{0}{0}}}$ would only yield the first equality).
The latter two equalities come from the conditions for $\cE_{y_{U}}$.

We now review some basic facts about Pauli channels.
Recall that for two Pauli channels $\mathcal{P}, \mathcal{P}'$, $\norm{\mathcal{P} - \mathcal{P}'}_\diamond = |p_I - p_I'| + |p_X - p_X'| + |p_Y - p_Y'| + |p_Z - p_Z'|$ and
$\lambda_X = 1 - 2 p_Y - 2 p_Z, \lambda_Y = 1 - 2p_X - 2p_Z, \lambda_Z = 1 - 2p_X - 2p_Y$.
Therefore, the assumption $\norm{\mathcal{P} - \mathcal{I}}_\diamond < \epsilon \leq 1/2$ implies
\begin{equation} \label{eq:good-lambda}
 \lambda_X, \lambda_Y, \lambda_Z \geq 1/2.
\end{equation}
The above inequality is useful for bounding the eigenvalues $\lambda_X, \lambda_Y, \lambda_Z$ away from zero.
Using the basic properties of Pauli channels, we can rewrite Eq.~\eqref{eq:v00vinv-1} to Eq.~\eqref{eq:v00vinv-2} as follows.
\begin{align}
    \ketbra{0}{0} &= V \ketbra{0}{0} V^{-1}, \label{eq:v00vinv-1p}\\
    \ketbra{1}{1} &= V \ketbra{1}{1} V^{-1},\\
    (I + \lambda^\mathcal{P}_X X) &= V (I + \lambda^{\mathcal{P}'}_X X) V^{-1},\\
    (I + \lambda^\mathcal{Q}_Y Y) &= V (I + \lambda^{\mathcal{Q}'}_Y Y) V^{-1}, \label{eq:v00vinv-2p}
\end{align}
where we utilize the fact that $\ketbra{0}{0} = (I+Z) / 2$.
Hence, we have
\begin{equation} \label{eq:Pauli-transform-equiv}
Z = V Z V^{-1}, \lambda X = \lambda' V X V^{-1}, \Lambda Y = \Lambda' V Y V^{-1},
\end{equation}
where $\lambda, \lambda', \Lambda, \Lambda' \geq 1/2$ from Eq.~\eqref{eq:good-lambda}.
Because a unitary can be expressed as $V = \euler^{\rmi \phi} \exp(\rmi (aX + bY + cZ))$ and an anti-unitary can be expressed as $V = \euler^{\rmi \phi} \exp(\rmi (aX + bY + cZ)) K$ for $a, b, c, \phi \in \mathbb{R}$ and $K$ the complex conjugation, Eq.~\eqref{eq:Pauli-transform-equiv} implies that $V = \euler^{\rmi \phi} I$, which is a contradiction.
Therefore, no distinct world models in $\mathcal{Q}^{\epsilon}$ are equivalent.

The fact that $\mathcal{Q}^{\epsilon}$ is not redundant is useful for the following logical reasoning.
Suppose no algorithm can learn the gate-dependent Pauli noise channels to arbitrarily small error.
It means that there exists two distinct world models (with different gate-dependent Pauli noise) in $\mathcal{Q}$ such that all experiments have the same output distribution (i.e., the two world models are weakly indistinguishable).
From the non-redundancy of $\mathcal{Q}^{\epsilon}$, these two distinct world models are not equivalent to one another.
Hence from Proposition~\ref{prop:indist->unlearn}, the model class $\mathcal{Q}^{\epsilon}$ is unlearnable.

On the other hand, suppose there is an algorithm that can learn the gate-dependent Pauli noise channels to arbitrarily small error.
We can build on the Pauli channel learning algorithm to learn $\rho_{x_\sigma}$ for any state $\sigma$ using the following equation
\begin{align}
    \rho_{x_\sigma} &= \frac{I}{2} + \Tr(X \rho_{x_\sigma} ) \frac{X}{2} + \Tr(Y \rho_{x_\sigma}) \frac{Y}{2} + \Tr(Z \rho_{x_\sigma}) \frac{Z}{2}\\
    &= \frac{I}{2} + \frac{\Tr(Z \mathcal{E}_{y_H}(\rho_{x_\sigma}) )}{\lambda_Z^H} \frac{X}{2} - \frac{\Tr(Z \mathcal{E}_{y_H}(\mathcal{E}_{y_S}(\rho_{x_\sigma})) )}{\lambda_Z^H \lambda_X^S} \frac{Y}{2} + \Tr(Z \rho_{x_\sigma}) \frac{Z}{2},
\end{align}
and $\Tr(Z \rho) = \bra{0} \rho \ket{0} - \bra{1} \rho \ket{1}$.
Because every action in $\mathcal{Q}^{\epsilon}$ can be learned to arbitrarily high accuracy, the model class $\mathcal{Q}^{\epsilon}$ is learnable.
Together, Theorem~\ref{prop:Clif+T-detailed} is equivalent to stating that $\mathcal{Q}^{\epsilon}$ is unlearnable. In the following, we will prove that $\mathcal{Q}^{\epsilon}$ is unlearnable.

\subsection{Unlearnability of a simpler model class}

We combine the basic results proven in Appendix~\ref{sec:basiclearna} to establish the unlearnability of $\mathcal{Q}^{\epsilon}$
The model class $\mathcal{Q}^{\epsilon}$ is quite complicated.
So we will begin by proving that a simpler model class $\tilde{\mathcal{Q}}^{\epsilon}$ is unlearnable. We will then use a set of tools developed in Appendix~\ref{sec:basiclearna} to show the unlearnability of $\mathcal{Q}^\epsilon$ from the unlearnability of $\tilde{\mathcal{Q}}^{\epsilon}$.
The simpler model class $\tilde{\mathcal{Q}}^{\epsilon}$ is over a simpler action space $\tilde{\mathcal{X}} = \{0\}, \tilde{\mathcal{Y}} = \{h, s, t\}, \tilde{\mathcal{Z}} = \{0\}$ and the same outcome space $\mathcal{B} = \{0, 1\}$.
Here, the three actions $h, s, t$ represent Hadamard gate, Phase gate, and $T$ gate.
Each world model $\mathcal{W} = \left(\{\rho_x\}_{x \in \tilde{\mathcal{X}}}, \{\mathcal{E}_y\}_{y \in \tilde{\mathcal{Y}}}, \{\mathcal{M}_z\}_{z \in \mathcal{Z}} \right)$ in $\tilde{\mathcal{Q}}^{\epsilon}$
is fully specified by the following conditions,
\begin{align}
    \rho_{0} = \ketbra{0}{0}, &\\
    \mathcal{E}_{h}(\rho) = \mathcal{P}_{h} (H \rho H^\dagger), \, \norm{\mathcal{P}_{h} - \mathcal{I}}_\diamond \leq \epsilon, & \label{eq:ehrho}\\
    \mathcal{E}_{s}(\rho) = \mathcal{P}_{s} (S \rho S^\dagger), \, \norm{\mathcal{P}_{s} - \mathcal{I}}_\diamond \leq \epsilon, & \label{eq:esrho}\\
	\mathcal{E}_{t}(\rho) = \mathcal{P}_{t} (T \rho T^\dagger), \, \norm{\mathcal{P}_{t} - \mathcal{I}}_\diamond \leq \epsilon, &\\
    \mathcal{M}_0 = \{\ketbra{0}{0}, \ketbra{1}{1}\}, &
\end{align}
where $\mathcal{P}_h, \mathcal{P}_s, \mathcal{P}_t$ are Pauli channels.
The simpler model class $\tilde{\mathcal{Q}}^{\epsilon}$ considers perfect zero state preparation, perfect computational basis measurement, and noisy Hadamard, phase, and $T$ gates that are subject to Pauli noise channels.
And every world model in the simpler model class $\tilde{\mathcal{Q}}^{\epsilon}$ is fully determines by the three Pauli noise channels $\mathcal{P}_h, \mathcal{P}_s, \mathcal{P}_t$.
In the following lemma, we state the unlearnability of the simpler model class $\tilde{\mathcal{Q}}^{\epsilon}$.

\begin{lemma}[Gate-dependent Pauli noise is unlearnable with Hadamard+S+T gates and $\ket{0}$] \label{lem:simplemodelclass-unlearnable}
Given $\epsilon > 0$. Consider a qubit system. Suppose we can prepare a perfect zero state $\ket{0}$, measure in the computational basis perfectly, and apply Hadamard gate, phase gate, and $T$ gate, where each gate is followed by an unknown gate-dependent Pauli noise channel that is $\epsilon$-close to the identity channel.
It is impossible for any algorithm to learn the gate-dependent Pauli noise channels to arbitrarily small error. Equivalently, $\tilde{\mathcal{Q}}^{\epsilon}$ is unlearnable.
\end{lemma}

The proof of Lemma~\ref{lem:simplemodelclass-unlearnable} is based on Proposition~\ref{prop:indist->unlearn}.
The proposition states that when two world models in a model class are weakly indistinguishable and are not equivalent, then the model class is unlearnable.
We will identify two such world models in $\tilde{\mathcal{Q}}^{\epsilon}$, then apply Proposition~\ref{prop:indist->unlearn} to conclude the proof.
On a high level, the existence of two weakly indistinguishable world models arises from the fact that the geometric structure for the action of Hadamard gate, phase gate, and $T$ gate are well-aligned with Pauli noise channels, causing some noise to be indistinguishable from another.


\subsection{Actions of Clifford+T on Pauli operators}

We begin by illustrating the geometric structure in any quantum experiments one could perform.
Every experiment under the given action spaces $\tilde{\mathcal{X}}, \tilde{\mathcal{Y}}, \tilde{\mathcal{Z}}$ is $x = 0, y_1, \ldots, y_L \in \{h, s, t\}, z = 0$.
Now, we consider the actions of $\mathcal{E}_h$, $\mathcal{E}_s$, and $\mathcal{E}_t$ on Pauli operators,
\begin{align}
\mathcal{E}_h(X) &= \lambda^h_Z Z, & \mathcal{E}_s(X) &= \lambda^s_Y Y, & \mathcal{E}_t(X) &= \frac{1}{\sqrt{2}} \lambda^t_X X + \frac{1}{\sqrt{2}} \lambda^t_Y Y,  \\
\mathcal{E}_h(Y) &= -\lambda^h_Y Y,& \mathcal{E}_s(Y) &= -\lambda^s_X X, & \mathcal{E}_t(Y) &= -\frac{1}{\sqrt{2}} \lambda^t_X X + \frac{1}{\sqrt{2}} \lambda^t_Y Y,\\
\mathcal{E}_h(Z) &= \lambda^h_X X,& \mathcal{E}_s(Z) &= \lambda^s_Z Z, & \mathcal{E}_t(Z) &= \lambda^t_Z Z,
\end{align}
where $\lambda^h_X, \lambda^h_Y, \lambda^h_Z$ are the Pauli eigenvalues that define the Pauli noise channel for the Hadamard gate, $\lambda^s_X, \lambda^s_Y, \lambda^s_Z$ defines the Pauli noise channel for the phase gate, and $\lambda^t_X, \lambda^t_Y, \lambda^t_Z$ defines the Pauli noise channel for the $T$ gate.
Every world model in $\mathcal{Q}^\epsilon$ is specified by the nine real values $\lambda^h_X, \lambda^h_Y, \lambda^h_Z, \lambda^s_X, \lambda^s_Y, \lambda^s_Z, \lambda^t_X, \lambda^t_Y, \lambda^t_Z$.

For an experiment specified by $x = 0, y_1, \ldots, y_L \in \{h, t\}, z = 0$, the probability that the experimental outcome is~$0$ can be written as
\begin{equation}
\frac{1}{2} + \frac{1}{2} \Tr\left( Z (\mathcal{E}_{y_L} \circ \ldots \circ \mathcal{E}_{y_1})(Z) \right).
\end{equation}
This follows from the identities $\ketbra{0}{0} = (I+Z)/2, \ketbra{1}{1} = (I-Z)/2$.
The probability that we obtain~$1$ as the experimental outcome is equal to one minus the probability for obtaining $0$,
\begin{equation}
\frac{1}{2} - \frac{1}{2} \Tr\left( Z (\mathcal{E}_{y_L} \circ \ldots \circ \mathcal{E}_{y_1})(Z) \right).
\end{equation}
The actions of $\mathcal{E}_h$, $\mathcal{E}_s$, and $\mathcal{E}_t$ on Pauli operators are now useful for understanding the term $\Tr\left( Z (\mathcal{E}_{y_L} \circ \ldots \circ \mathcal{E}_{y_1})(Z) \right)$.

\subsubsection{Experiments as multiple particles traversing a graph}

We represent the action of the three unitaries $H, S, T$ on the three Pauli operators $X, Y, Z$ as a small graph with three nodes corresponding to $X, Y, Z$.
We can consider each experiment as particles traversing the graph.
At the start of the experiment, a single particle resides on the node $Z$ with an initial value of $1$.
Applying $\mathcal{E}_h$ corresponds to moving the particle from $Z \rightarrow X, Y \rightarrow Y, X \rightarrow Z$, and the value of the particle will be multiplied by $\lambda_X^h, -\lambda_Y^h, \lambda_Z^h$ accordingly.
Applying $\mathcal{E}_s$ corresponds to moving the particle from $Y \rightarrow X, X \rightarrow Y, Z \rightarrow Z$, and the value of the particle will be multiplied by $-\lambda_X^s, \lambda_Y^s, \lambda_Z^s$ accordingly.
Applying $\mathcal{E}_t$ corresponds to a more complicated action.
If the particle resides on $X$, then the particle will
split into two particles: one on $X$ and one on $Y$.
The value of the duplicated particle on $X$ and $Y$ will be equal to the value of the original particle multiplied by $\frac{1}{\sqrt{2}} \lambda^t_X$ and $\frac{1}{\sqrt{2}} \lambda^t_Y$ accordingly.
Similarly, if the particle resides on $Y$, the particle will
split into two particles on $X$ and $Y$.
The value of the duplicated particle on $X$ and $Y$ will be equal to the value of the original particle multiplied by $-\frac{1}{\sqrt{2}} \lambda^t_X$ and $\frac{1}{\sqrt{2}} \lambda^t_Y$ accordingly.
If the particle resides on $Z$, then the particle will stay at $Z$, and the value of the particle will be multiplied by $\lambda_Z^t$.
After many applications of the CPTP maps $\mathcal{E}_{y}$, there will be many particles moving on the three-node graph. The number of particles is exponential in the number of $T$ gates applied.
At the end of the experiment, we sum up the values of particles residing at $Z$ to obtain $\Tr\left( Z (\mathcal{E}_{y_L} \circ \ldots \circ \mathcal{E}_{y_1})(Z) \right)$.


\subsubsection{Polynomial forms and unlearnability}

By induction, we can show that the value of each particle residing at $X, Y, Z$ can be written as
\begin{align}
    \lambda^h_X (\lambda^h_Z \lambda^h_X)^{k_X} &\cdot f_X(\lambda^h_Y, \lambda^s_X, \lambda^s_Y, \lambda^s_Z, \lambda^t_X, \lambda^t_Y, \lambda^t_Z), \\
    \lambda^h_X (\lambda^h_Z \lambda^h_X)^{k_Y} &\cdot f_Y(\lambda^h_Y, \lambda^s_X, \lambda^s_Y, \lambda^s_Z, \lambda^t_X, \lambda^t_Y, \lambda^t_Z), \\
    (\lambda^h_Z \lambda^h_X)^{k_Z} &\cdot f_Z(\lambda^h_Y, \lambda^s_X, \lambda^s_Y, \lambda^s_Z, \lambda^t_X, \lambda^t_Y, \lambda^t_Z),
\end{align}
accordingly, where $k_X, k_Y, k_Z$ are non-negative integers, $f_X, f_Y, f_Z$ are monomials.
We perform induction according to the dynamics of the particles from the start of the experiment to finish.
In the base case, there are only one particle residing at $Z$ with a value of $1$, hence the claimed statement holds.
For each induction step, the traversal/duplication rules guarantee that the above form of values is preserved.
At the end of the experiment, when we sum up the values of particles residing at $Z$, we have the following form for the $Z$ expectation value,
\begin{align}
    \Tr\left( Z (\mathcal{E}_{y_L} \circ \ldots \circ \mathcal{E}_{y_1})(Z) \right) &= \sum_{k} (\lambda^h_Z \lambda^h_X)^{k} f_k(\lambda^h_Y, \lambda^s_X, \lambda^s_Y, \lambda^s_Z, \lambda^t_X, \lambda^t_Y, \lambda^t_Z)\\
    &= p(\lambda^h_Z \lambda^h_X, \lambda^h_Y, \lambda^s_X, \lambda^s_Y, \lambda^s_Z, \lambda^t_X, \lambda^t_Y, \lambda^t_Z),
\end{align}
where $k$ sums over non-negative integers, $f_k$ is a monomial, $p$ is a polynomial.
Together, we can see that the probability for every experimental outcome is a polynomial function in $\lambda^h_Z \lambda^h_X, \lambda^h_Y, \lambda^s_X, \lambda^s_Y, \lambda^s_Z, \lambda^t_X, \lambda^t_Y, \lambda^t_Z$.
Hence, when two world models have the same $\lambda^h_Z \lambda^h_X$ and the other $\lambda$'s are also equal, the two world models are weakly indistinguishable.

The key point is that a particle departs from $Z$ only by moving to $X$ through the action of the Hadamard gate, which multiplies its value by $\lambda^h_X$, and a particle arrives at $Z$ only by moving from $X$ through the action of a Hadamard gate, which multiplies its value by $\lambda^h_Z$. Therefore the value acquired by any particle that is initialized at $Z$ and returns to $Z$ just before the measurement must have the same number of factors of  $\lambda^h_Z$ as factors of $\lambda^h_X$. This same argument still applies if we replace the $T$ gate by any gate that is diagonal in the $Z$ basis, or it we add to the world model additional gates that are diagonal in the $Z$ basis.

From this observation it follows that there exist two distinct world models in the models class $\tilde{\mathcal{Q}}^{\epsilon}$ which are weakly indistinguishable.
Recall some basic properties of Pauli channels \cite{sacchi2005optimal}:
\begin{align}
    \mathcal{P}(\rho) &= p_I \rho + p_X X \rho X + p_Y Y \rho Y + p_Z Z \rho Z, \\
    &= \frac{1}{2}\left(\Tr(\rho) I + \lambda_X \Tr( X\rho) X + \lambda_Y \Tr( Y \rho) Y  + \lambda_Z \Tr( Z \rho) Z\right), \\
    \norm{\mathcal{P} - \mathcal{P}'}_\diamond &= |p_I - p_I'| + |p_X - p_X'| + |p_Y - p_Y'| + |p_Z - p_Z'|,
\end{align}
where $(p_I, p_X, p_Y, p_Z)$ is a probability distribution, $\lambda_X = 1 - 2 p_Y - 2 p_Z, \lambda_Y = 1 - 2p_X - 2p_Z, \lambda_Z = 1 - 2p_X - 2p_Y$, and $\mathcal{P}, \mathcal{P}'$ are two Pauli channels.
Each world model in $\tilde{\mathcal{Q}}^{\epsilon}$ is specified by $\lambda^h_X, \lambda^h_Y, \lambda^h_Z, \lambda^s_X, \lambda^s_Y, \lambda^s_Z, \lambda^t_X, \lambda^t_Y, \lambda^t_Z$, or equivalently $p^h_X, p^h_Y, p^h_Z, p^s_X, p^s_Y, p^s_Z, p^t_X, p^t_Y, p^t_Z$ as $p_I = 1 - p_X - p_Y - p_Z$.
We consider two world models $\mathcal{W}^A, \mathcal{W}^B$ in $\tilde{\mathcal{Q}}^{\epsilon}$ to be defined by
\begin{align}
    p^{h, A}_X &= 0, \,\, p^{h, A}_Z = \epsilon, \,\, p^{h, A}_Y = p^{s, A}_X = p^{s, A}_Y = p^{s, A}_Z = p^{t, A}_X = p^{t, A}_Y = p^{t, A}_Z = 0, \\
    p^{h, B}_X &= \epsilon, \,\, p^{h, B}_Z = 0, \,\, p^{h, B}_Y = p^{s, B}_X = p^{s, B}_Y = p^{s, B}_Z = p^{t, B}_X = p^{t, B}_Y = p^{t, B}_Z = 0.
\end{align}
The two world models correspond to having $Z$-error or having $X$-error after Hadamard gate. Using the basic properties, it is not hard to check that both $\mathcal{W}^A, \mathcal{W}^B$ belong to $\tilde{\mathcal{Q}}^{\epsilon}$.
Furthermore, $\lambda^{h, A}_X \lambda^{h, A}_Z = \lambda^{h, B}_X \lambda^{h, B}_Z,$ and all other $\lambda$'s are equal.
Because the two world models are distinct, using a similar proof as in Appendix~\ref{sec:unlearnability-Qeps}, we can show that $\mathcal{W}^A \not\equiv \mathcal{W}^B$, i.e., no unitary or anti-unitary transformation exists that relates $\mathcal{W}^A$ and $\mathcal{W}^B$.
Hence, $\mathcal{W}^A, \mathcal{W}^B$ are two world models that are weakly indistinguishable but are not equivalent.
Using Proposition~\ref{prop:indist->unlearn}, we conclude that $\tilde{\mathcal{Q}}^{\epsilon}$ is unlearnable.

\subsection{Unlearnability of simple model class implies unlearnability of target model class}

We are now ready to prove Theorem~\ref{prop:Clif+T-detailed}, i.e., the model class $\mathcal{Q}^\epsilon$ is unlearnable.
The basic structure of the proof is the following. We first utilize the equivalence of learnability after adding some actions, stated in Proposition~\ref{prop:comp-learnability-state},~\ref{prop:comp-learnability-map},~and~\ref{prop:mix-learnability-state}, to show that a model class $\mathcal{R}^\epsilon$, which is a subset of $\mathcal{Q}^\epsilon$, is unlearnable.
Then, we can use the monotonicity of unlearnability, stated in Proposition~\ref{prop:mono-learnability}, to show that $\mathcal{Q}^\epsilon$ is unlearnable because $\mathcal{R}^\epsilon \subseteq \mathcal{Q}^\epsilon$ is unlearnable.

\subsubsection{Composing all Clifford gates}

We begin with the equivalence of learnability after adding composite CPTP maps as stated in Proposition~\ref{prop:comp-learnability-map}.
We will compose the two actions $h, s \in \tilde{\mathcal{Y}} = \{h, s, t\}$, which corresponds to Hadamard and phase gate subject to gate-dependent Pauli noise; see $\mathcal{E}_h$ and $\mathcal{E}_s$ in Equation~\eqref{eq:ehrho}~and~\eqref{eq:esrho}.
Because Hadamard and phase gates form a universal gate set for the Clifford group, we can construct any Clifford unitary $C$ from a composition of Hadamard gate $H$ and phase gate $S$. Let $L^*$ be the required sequence length to generate every element in the Clifford unitary group.
For any Clifford unitary $C$, we know that $f_C(P) = C P C^\dagger$ is equal to a Pauli operator (up to a phase of $\pm 1$) for any Pauli operator $P \in \{X, Y, Z\}$. Furthermore, up to the phase, $f_C$ is a permutation function over $\{X, Y, Z\}$.
Hence, for any Clifford unitary $C$, Pauli channel $\mathcal{P}$, and quantum state $\rho$,
\begin{equation} \label{eq:commutation-pauli-clifford}
    C \mathcal{P}(\rho) C^\dagger = \mathcal{Q}(C \rho C^\dagger),
\end{equation}
where $\mathcal{Q}$ is a different Pauli channel satisfying $\norm{\mathcal{Q} - \mathcal{I}}_\diamond = \norm{\mathcal{P} - \mathcal{I}}_\diamond$.

Now, we compose new CPTP maps based on the model class $\tilde{\mathcal{Q}}^{\epsilon^*}$, where $\epsilon^*$ will be chosen to be small enough later. The consideration of $\epsilon^*$ is needed to ensure that all physical operations have a small enough error.
Consider a world model $\mathcal{W} = \left(\{\rho_x\}_{x \in \tilde{\mathcal{X}}}, \{\mathcal{E}_y\}_{y \in \tilde{\mathcal{Y}}}, \{\mathcal{M}_z\}_{z \in \mathcal{Z}} \right)$ in $\tilde{\mathcal{Q}}^{\epsilon^*}$.
For every Clifford unitary $C$, there exists $y_1, \ldots, y_L \in \{h, s\}$ with $L \leq L^*$, such that the composition of the Hadamard and phase gate generates the Clifford unitary $C$.
From the commutation relation between Pauli channel and Clifford unitary in Equation~\eqref{eq:commutation-pauli-clifford} we see that
\begin{equation}
    (\mathcal{E}_{y_L} \circ \ldots \circ \mathcal{E}_{y_1})(\rho) = (\mathcal{Q}_L \circ \ldots \circ \mathcal{Q}_1)(C\rho C^\dagger),
\end{equation}
where $\mathcal{Q}_1, \ldots, \mathcal{Q}_L$ are Pauli channels, and $\norm{\mathcal{Q}_\ell - \mathcal{I}} \leq \epsilon^*, \forall \ell \in \{1, \ldots, L\}$.
It is not hard to check that the composition of Pauli channel is still a Pauli channel, hence $(\mathcal{Q}_L \circ \ldots \circ \mathcal{Q}_1)$ is a Pauli channel.
Using a telescoping sum and the triangle inequality, we have
\begin{align}
    \norm{\mathcal{Q}_L \circ \ldots \circ \mathcal{Q}_1 - \mathcal{I}}_\diamond &\leq \norm{(\mathcal{Q}_L - \mathcal{I})\circ \mathcal{Q}_{L-1} \circ \ldots \circ \mathcal{Q}_1}_\diamond + \norm{(\mathcal{Q}_{L-1} \circ \ldots \circ \mathcal{Q}_1) - \mathcal{I}}_\diamond\\
    &\leq \norm{\mathcal{Q}_L - \mathcal{I}}_\diamond + \norm{(\mathcal{Q}_{L-1} \circ \ldots \circ \mathcal{Q}_1) - \mathcal{I}}_\diamond\\
    &\leq \epsilon^* + \norm{(\mathcal{Q}_{L-1} \circ \ldots \circ \mathcal{Q}_1) - \mathcal{I}}_\diamond\\
    &\leq 2 \epsilon^* + \norm{(\mathcal{Q}_{L-2} \circ \ldots \circ \mathcal{Q}_1) - \mathcal{I}}_\diamond \leq \ldots \leq L \epsilon^* \leq L^* \epsilon^*.
\end{align}

We now add in a new action for each element in the Clifford group, excluding the Hadamard and phase gate $H, S$.
After adding these new actions, our action space for CPTP maps has expanded from $\tilde{\mathcal{Y}} = \{h, s, t\}$ to $\mathcal{Y} = \mathcal{C} \cup \{T\}$, a union of the Clifford group and the $T$ gate $\{T\}$.
We call the new model class with these additional CPTP maps $\tilde{\mathcal{R}}^{\epsilon^*}$.
From Proposition~\ref{prop:comp-learnability-map}, we know that adding these new actions do not affect the learnability.
Hence $\tilde{\mathcal{R}}^{\epsilon^*}$ is unlearnable.
To recap, $\tilde{\mathcal{R}}^{\epsilon^*}$ is a model class over $\tilde{\mathcal{X}} = \{0\}, \mathcal{Y} = \{y_U\}_{U \in \mathcal{C} \cup \{T \}}, \mathcal{Z} = \{0\}, \mathcal{B} = \{0, 1\}$, where $\mathcal{C}$ is the Clifford group. Every world model $\mathcal{W} = \left(\{\rho_x\}_{x \in \tilde{\mathcal{X}}}, \{\mathcal{E}_y\}_{y \in \mathcal{Y}}, \{\mathcal{M}_z\}_{z \in \mathcal{Z}} \right)$ in $\tilde{\mathcal{R}}^{\epsilon^*}$ satisfies the following condition,
\begin{align}
    \rho_{0} &= \ketbra{0}{0}, \\
    \mathcal{E}_{y_U}(\rho) &= \mathcal{P}_{U} (U \rho U^\dagger), \, \norm{\mathcal{P}_{U} - \mathcal{I}}_\diamond \leq L^* \epsilon^*, &\,\, \forall U \in \mathcal{C} \cup \{T\},\\
    \mathcal{M}_0 &= \{\ketbra{0}{0}, \ketbra{1}{1}\},
\end{align}
where $\mathcal{P}_U, \forall U \in \mathcal{C} \cup \{T\}$ are Pauli channels.
However, $\tilde{\mathcal{R}}^{\epsilon^*}$ does not contain all the world models that satisfy the above conditions.


\subsubsection{Composing all quantum states}

The next step in the proof is to use the equivalence of learnability after adding composite initial states and adding mixtures of initial states as in  Proposition~\ref{prop:comp-learnability-state}~and~\ref{prop:mix-learnability-state}.
We compose the zero initial state $\ketbra{0}{0}$ with the Clifford+T gates to generate all pure states up to some small errors.
Then, we can add mixtures of the pure states to generate all quantum states up to some small errors.
We consider $L^{s}$ to be the minimum integer such that for all pure states $\ket{\psi}$, there exists a sequence $U_1, \ldots, U_\ell$ consisting of Clifford gates and $T$ gates with $\ell \leq L^s$ and
\begin{equation} \label{eq:approxSKthm}
    \norm{\ketbra{\psi}{\psi} - (U_\ell \ldots U_1) \ketbra{0}{0} (U_1 \ldots U_\ell)^\dagger}_1 \leq \epsilon / 2.
\end{equation}
From the Solovay-Kitaev theorem \cite{dawson2005solovay} and follow-up works \cite{harrow2002efficient, selinger2012efficient}, $L^{s} = \mathcal{O}(\log(1 / \epsilon))$.
From Proposition~\ref{prop:comp-learnability-state}, the model class with the additional states is constructed by specifying the set of additional actions $\Xi$, a constant $L$, and a function $f$ with $f(\xi) = (x, y_1, \ldots, y_\ell)$.
Here, we consider $\Xi = \{ \ket{\psi} \}_{\ket{\psi}: \mathrm{pure\,state}}$ to be the space of all pure states, $L = L^s$, and we define $f$ to map a pure state $\ket{\psi}$ to $(0, y_{U_1}, \ldots, y_{U_\ell})$, where $U_1, \ldots, U_\ell$ are the unitaries for approximating the pure state $\ket{\psi}$ according to Eq.~\eqref{eq:approxSKthm}.
For every world model $\mathcal{W} = \left(\{\rho_x\}_{x \in \tilde{\mathcal{X}}}, \{\mathcal{E}_y\}_{y \in \mathcal{Y}}, \{\mathcal{M}_z\}_{z \in \mathcal{Z}} \right)$ in $\tilde{\mathcal{R}}^{\epsilon^*}$,  we have
\begin{align}
    &\norm{\ketbra{\psi}{\psi} - (\mathcal{E}_{y_{U_\ell}} \circ \ldots \circ \mathcal{E}_{y_{U_1}}) (\ketbra{0}{0})}_1
    \leq \norm{\ketbra{\psi}{\psi} - (U_\ell \ldots U_1) \ketbra{0}{0} (U_1 \ldots U_\ell)^\dagger}_1\\
    &+ \norm{(\mathcal{E}_{y_{U_\ell}} \circ \ldots \circ \mathcal{E}_{y_{U_1}}) (\ketbra{0}{0}) - (U_\ell \ldots U_1) \ketbra{0}{0} (U_1 \ldots U_\ell)^\dagger}_1\\
    &\leq \epsilon/2 + \ell L^* \epsilon^* \leq \epsilon/2 + L^s L^* \epsilon^*, \label{eq:approxboundstate}
\end{align}
where the second-to-last inequality uses Eq.~\eqref{eq:approxSKthm}, a telescoping sum, and the triangle inequality, and the last inequality uses $\ell \leq L^s$.
We have now added actions associated to generating arbitrary pure states.
Proposition~\ref{prop:comp-learnability-state} shows that adding these pure states will maintain the unlearnability.
Now, for all mixed states, we only need to make use of the fact that any single-qubit mixed state $\rho$ can be written as $\rho = p\ketbra{\psi}{\psi} + (1 - p) \ketbra{\phi}{\phi}$ for some $0 < p < 1$ and two orthogonal pure states $\ket{\psi}, \ket{\phi}$.
Suppose $y_{U_1}, \ldots, y_{U_\ell}$ specifies the unitaries to approximate $\ket{\psi}$ and $y_{U_1}', \ldots, y_{U_{\ell'}}'$ specifies the unitaries to approximate $\ket{\phi}$. Then by triangle inequality and Eq.~\eqref{eq:approxboundstate},
\begin{equation}
    \norm{\rho - \left[ p (\mathcal{E}_{y_{U_\ell}} \circ \ldots \circ \mathcal{E}_{y_{U_1}}) (\ketbra{0}{0}) + (1 - p) (\mathcal{E}_{y_{U_{\ell'}}'} \circ \ldots \circ \mathcal{E}_{y_{U_1}'}) (\ketbra{0}{0}) \right] }_1 \leq \epsilon/2 + L^s L^* \epsilon^*.
\end{equation}
Using Proposition~\ref{prop:mix-learnability-state}, we can add all the mixed states without altering the unlearnability.
We have now create a model class $\mathcal{R}^{\epsilon^*}$ over $\mathcal{X} = \{x_\sigma\}_{\sigma: \mathrm{state}}, \mathcal{Y} = \{y_U\}_{U \in \mathcal{C} \cup \{T \}}, \mathcal{Z} = \{0\}, \mathcal{B} = \{0, 1\}$, where $\mathcal{C}$ is the Clifford group. Every world model $\mathcal{W} = \left(\{\rho_x\}_{x \in \mathcal{X}}, \{\mathcal{E}_y\}_{y \in \mathcal{Y}}, \{\mathcal{M}_z\}_{z \in \mathcal{Z}} \right)$ in $\mathcal{R}^{\epsilon^*}$ satisfies
\begin{align}
    \norm{\rho_{x_\sigma} - \sigma}_1 &\leq \epsilon/2 + L^s L^* \epsilon^*, &\,\, \forall \sigma: \mathrm{state},\\
    \rho_{x_{\ketbra{0}{0}}} &= \ketbra{0}{0}, &\\
    \mathcal{E}_{y_U}(\rho) &= \mathcal{P}_{U} (U \rho U^\dagger), \, \norm{\mathcal{P}_{U} - \mathcal{I}}_\diamond \leq L^* \epsilon^*, &\,\, \forall U \in \mathcal{C} \cup \{T\},\\
    \mathcal{M}_0 &= \{\ketbra{0}{0}, \ketbra{1}{1}\}. &
\end{align}
Furthermore, from the fact that $\tilde{\mathcal{R}}^{\epsilon^*}$ is unlearnable, and Proposition~\ref{prop:comp-learnability-state}~and~\ref{prop:mix-learnability-state}, we have $\mathcal{R}^{\epsilon^*}$ is unlearnable.
Now, if we choose $\epsilon^* = \epsilon / (2 L^s L^*)$, then every world model in $\mathcal{R}^{\epsilon^*}$ is in $\mathcal{Q}^{\epsilon}$.
Hence, using monotonicity of learnability in Proposition~\ref{prop:mono-learnability}, we have $\mathcal{Q}^{\epsilon}$ is unlearnable. This concludes the proof of Theorem~\ref{prop:Clif+T-detailed} by recalling the equivalence of Theorem~\ref{prop:Clif+T-detailed} and the unlearnability of $\mathcal{Q}^{\epsilon}$.

\begin{remark}
By tracing through the proof, one can see that the unavoidable noise floor for learning the gate-dependent Pauli noise channel is $\epsilon^*$, which is of order $\epsilon / \log(1 / \epsilon)$.
\end{remark}

\section{Noise and unlearnability}
\label{sec:noisystate}

We now give two examples of unlearnable model classes.
Consider $d$-dimensional world models.
We focus on the action spaces $\mathcal{X} = \{\sigma\}_{\sigma:\mathrm{state}}, \mathcal{Y} = \{U\}_{U \in \mathrm{SU}(d)}, \mathcal{Z} = \{0\}$.
Consider $\epsilon > 0$ and the model class $\mathcal{S}^{\epsilon}$ over $\mathcal{X}, \mathcal{Y}, \mathcal{Z}$ that consists of two world models $\mathcal{W}^A = \left(\{\rho^A_x\}_{x \in \mathcal{X}}, \{\mathcal{E}^A_y\}_{y \in \mathcal{Y}}, \{\mathcal{M}^A_z\}_{z \in \mathcal{Z}} \right)$, where
\begin{align}
    \rho^A_\sigma &= (1 - \epsilon) \sigma + \epsilon \frac{I}{d}, &\forall \sigma: \mathrm{state},\\
    \mathcal{E}^A_U(\rho) &= U \rho U^\dagger, &\forall U \in \mathrm{SU}(d),\\
    \mathcal{M}^A_0 &= \left\{ \ketbra{b}{b} \right\}_{b=1, \ldots, d},&
\end{align}
and $\mathcal{W}^B = \left(\{\rho^B_x\}_{x \in \mathcal{X}}, \{\mathcal{E}^B_y\}_{y \in \mathcal{Y}}, \{\mathcal{M}^B_z\}_{z \in \mathcal{Z}} \right)$, where
\begin{align}
    \rho^B_\sigma &= \sigma, &\forall \sigma: \mathrm{state},\\
    \mathcal{E}^B_U(\rho) &= U \rho U^\dagger, &\forall U \in \mathrm{SU}(d),\\
    \mathcal{M}^B_0 &= \left\{ (1-\epsilon)\ketbra{b}{b} + \epsilon \frac{I}{d} \right\}_{b=1, \ldots, d},&
\end{align}
Verbally, the model class $\mathcal{S}^{\epsilon}$ contains world models where we have perfect unitaries, but the initial state or the computational basis measurement is subject to a depolarization noise of strength~$\epsilon$.

We also consider another model class $\mathcal{S}^{\Omega}$ that encompasses world models where  states, unitaries, and the computational basis measurements are all noisy.
Formally, $\mathcal{S}^{\Omega}$ is over the same set of action spaces $\mathcal{X} = \{\sigma\}_{\sigma:\mathrm{state}}, \mathcal{Y} = \{U\}_{U \in \mathrm{SU}(d)}, \mathcal{Z} = \{0\}$.
And $\mathcal{S}^{\Omega}$ contains all world models $\mathcal{W} = \left(\{\rho_x\}_{x \in \mathcal{X}}, \{\mathcal{E}_y\}_{y \in \mathcal{Y}}, \{\mathcal{M}_z\}_{z \in \mathcal{Z}} \right)$, where $\rho_x$ is a quantum state, $\mathcal{E}_y$ is a CPTP map, and $\mathcal{M}_z$ is a POVM.
It is not hard to see that $\mathcal{S}^{\epsilon} \subset \mathcal{S}^{\Omega}$. Hence, the two world models $\mathcal{W}^A, \mathcal{W}^B$ are also in $\mathcal{S}^{\Omega}$.

We will now 
provide a formal proof showing that $\mathcal{S}^{\epsilon}$ is unlearnable.
For any experiments $E = (\sigma, U_1, \ldots, U_\ell, 0)$, where $\sigma \in \mathcal{X}, U_1, \ldots, U_\ell \in \mathcal{Y}, 0 \in \mathcal{Z}$, we have
\begin{align}
    \Tr\left( M^A_{0 b} (\mathcal{E}^A_{U_\ell} \circ \ldots \circ \mathcal{E}^A_{U_1})(\rho^A_\sigma) \right) &= (1-\epsilon) \bra{b} U_\ell \ldots U_1 \sigma U_1^\dagger \ldots U_\ell^\dagger \ket{b} + \frac{\epsilon}{d}, \,\, \forall b \in \{1, \ldots, d\},\\
    \Tr\left( M^B_{0 b} (\mathcal{E}^B_{U_\ell} \circ \ldots \circ \mathcal{E}^B_{U_1})(\rho^B_\sigma) \right) &= (1-\epsilon) \bra{b} U_\ell \ldots U_1 \sigma U_1^\dagger \ldots U_\ell^\dagger \ket{b} + \frac{\epsilon}{d}, \,\, \forall b \in \{1, \ldots, d\}.
\end{align}
Hence the two world models are weakly indistinguishable.
It is also easy to check that the two world models are not equivalent to one another, hence they describe different physical realities.
By Theorem~\ref{prop:equiv}, if the two world models are equivalent, then there exists a unitary or anti-unitary transformation $U$ such that $\rho^B_\sigma = U \rho^A_\sigma U^{-1}$ for any quantum state $\sigma$.
We can use this relation to deduce that $\Tr((\rho^B_\sigma)^2) = \Tr((\rho^A_\sigma)^2)$, i.e.,  the purity of the two states $\rho^A_\sigma, \rho^B_\sigma$ must be equal, for all state $\sigma$.
However, for $\epsilon > 0$ and pure state $\sigma = \ketbra{\psi}{\psi}$, the purity of $\rho^A_\sigma$ is less than one, but the purity of $\rho^B_\sigma$ is one.
Hence, the two world models $\mathcal{W}^A, \mathcal{W}^B$ are not equivalent.

Because $\mathcal{S}^{\epsilon}$ contains two weakly indistinguishable world models that are not equivalent to one another, $\mathcal{S}^{\epsilon}$ is unlearnable according to Proposition~\ref{prop:indist->unlearn}.
Then using Proposition~\ref{prop:mono-learnability}, the monotonicity of unlearnability, $\mathcal{S}^{\Omega}$ is unlearnable because $\mathcal{S}^{\epsilon} \subset \mathcal{S}^{\Omega}$.

\section{Learning under gate-independent noise on Clifford gates}
\label{sec:gate-indep-Clif}

We give the detailed proof for Theorem~\ref{thm:learn-gate-indp} in this appendix.

\subsection{Review on unitary design for Clifford gates}

We recall the following well-known facts which relate Clifford gates to unitary designs.
These two properties will be used through to design our learning algorithm and to prove its correctness.
Furthermore, we will be using a substantial amount of tensor network manipulation, where only the final simplified results are shown.
Readers unfamiliar with tensor network manipulations could refer to reviews in \cite{bridgeman2017hand, chen2021exponential}.

\begin{lemma}[Unitary design for Clifford gates \cite{dankert2009tdesign, webb2015clifford, zhu2016clifford}]
Consider $n > 0$ to be the number qubits. Let $\mathcal{C}$ be the set of all Clifford gates over $n$ qubits and let $d = 2^n$. We have
\begin{align}
    \frac{1}{|\mathcal{C}|} \sum_{C \in \mathcal{C}} C A C^\dagger &= \frac{I}{d} \Tr(A),& (\mbox{1-design})\\
    \frac{1}{|\mathcal{C}|} \sum_{C \in \mathcal{C}} C^{\otimes 2} B (C^\dagger)^{\otimes 2} &= \frac{1}{d^2-1}\left(  I \Tr(B) + S \Tr(S B) - \frac{1}{d} S \Tr(B)  - \frac{1}{d} I \Tr(S B) \right), & (\mbox{2-design})
\end{align}
where $A$ is an $2^n \times 2^n$ complex matrix, $B$ is a complex tensor living in the tensor product space of two $2^n \times 2^n$ complex matrices, $S$ is the swap operator over the two tensor product components.
\end{lemma}

\subsection{Learning noisy zero state, Clifford gate noise, and noisy basis measurement}

Two sets of randomized experiments are conducted to learn about the noisy initial state $\rho_0 \approx \ketbra{0^n}{0^n}$, the Clifford gate noise $\cN \approx \cI$, and the noisy computational basis measurement $\cM_0 = \{M_b\}_{b \in \{0, 1\}^n}$ with $M_b \approx \ketbra{b}{b}$.
The first set of $N_A$ experiments prepares $\rho_0$, evolves by $\cE_C$ for a random Clifford $C$, and measures $\cM_0$.
The second set of $N_B$ experiments prepares $\rho_0$, evolves by $\cE_{C_1}$ for a random Clifford $C_1$, evolves by $\cE_{C_2}$ for a second random Clifford $C_2$, and measures $\cM_0$.
We denote the two sets of experimental outcomes as
\begin{align}
    \left( C^{(A, i)}, b^{(A, i)} \in \{0, 1\}^n \right), &\quad \forall i = 1, \ldots, N_A,\\
    \left( C^{(B, i)}_1, C^{(B, i)}_2, b^{(B, i)} \in \{0, 1\}^n \right), &\quad \forall i = 1, \ldots, N_B.
\end{align}
We will also define the POVM $\cM' = \{ M'_b = \cN^\dagger(M_b) \}_{b \in \{0, 1\}^n}$ which is equivalent to applying the Clifford gate noise $\cN$ followed by performing
the noisy computational basis measurement $\cM_0$.

In the following, we will denote $d = 2^n$, $f = \bra{0^n}\rho_0 \ket{0^n}$, $g = \frac{1}{d}\sum_{b \in \{0, 1\}^n} \bra{b}M'_b\ket{b}$, $\mathbb{I}[A]$ to be the indicator function, i.e., $\mathbb{I}[A] = 1$ if $A$ is true and $\mathbb{I}[A] = 0$ if $A$ is false.
Because $f$ is unlearnable, we can set $f$ to whatever value we want.
One practical choice is to set $f = 1$ since $\rho_0 \approx \ketbra{0^n}{0^n}$.
Under the assumption that $\rho_0 \approx \ketbra{0^n}{0^n}$ and $M'_b = \cN^{\dagger}(M_b) \approx M_b \approx \ketbra{b}{b}$, we can consider
\begin{equation} \label{eq:boundednoise-eq}
 f = \bra{0^n}\rho_0 \ket{0^n} \in [0.9, 1], \quad    \bra{b}M'_b\ket{b} \in [0.9, 1.1], \quad g = \frac{1}{d}\sum_{b \in \{0, 1\}^n} \bra{b}M'_b\ket{b} \in [0.9, 1.1].
\end{equation}
This follows from the fact that the noise is bounded.

\subsubsection{Learning noisy basis measurement conflated with Clifford gate noise}
\label{sec:leanringnoisyPOVM}

We now construct various estimators that characterize the POVM $\cM' = \{M'_b\}_{b \in \{0, 1\}^n}$. We start with the simplest estimator,
\begin{equation}
\widehat{\Tr(M'_b)} = \frac{1}{N_A} \sum_{i=1}^{N_A} d \, \bI\left[ b^{(A, i)} = b\right].
\end{equation}
Using the unitary $1$-design property of Clifford gates, we have
\begin{equation}
    \E \left[ \widehat{\Tr(M'_b)} \right] = \frac{1}{|\mathcal{C}|} \sum_{C \in \mathcal{C}} d \Tr\left( M'_b C\rho_0 C^\dagger \right) = \Tr(M'_b).
\end{equation}
Hence, we can estimate $\Tr(M'_b)$ to arbitrarily small error with large enough $N_A$.
The next estimator is slightly more complicated and uses $\widehat{\Tr(M'_b)}$,
\begin{equation} \label{eq:Mprimeb}
\widehat{M'_b} = \frac{d^2 - 1}{f - \frac{1}{d}} \left[ \frac{1}{N_A} \sum_{i=1}^{N_A} \bI\left[ b^{(A, i)} = b\right] C^{(A, i)} \ketbra{0^n}{0^n} (C^{(A, i)})^\dagger \right] - \frac{1 - \frac{f}{d}}{{f - \frac{1}{d}}} \widehat{\Tr(M'_b)} I.
\end{equation}
Using the unitary $2$-design property of Clifford gates, we have
\begin{equation} \label{eq:expMprimeb}
    \E \left[ \widehat{M'_b} \right] = \frac{1}{|\mathcal{C}|} \sum_{C \in \mathcal{C}} \frac{d^2 - 1}{f - \frac{1}{d}} \Tr(M'_b C \rho_0 C^\dagger) C \ketbra{0^n}{0^n} C^\dagger - \frac{1 - \frac{f}{d}}{{f - \frac{1}{d}}} \Tr(M'_b) I = M'_b.
\end{equation}
Hence, we can estimate the POVM element $M'_b$ to arbitrarily small error with large enough $N_A$.

We will also utilize the following estimator to estimate $g = \frac{1}{d} \sum_{b \in \{0, 1\}^n} \bra{b}M'_b\ket{b}$.
\begin{equation} \label{eq:hatg-est}
    \widehat{g} = \frac{1}{d} \sum_{b \in \{0, 1\}^n} \bra{b} \widehat{M'_b} \ket{b} = \frac{d^2 - 1}{d f - 1} \left[ \frac{1}{N_A} \sum_{i=1}^{N_A} \bra{b^{(A, i)}} C^{(A, i)} \ketbra{0^n}{0^n} (C^{(A, i)})^\dagger \ket{b^{(A, i)}} \right] - \frac{d - f}{{df - 1}}.
\end{equation}
The second equality follows from $\sum_{b \in \{0, 1\}^n} \widehat{\Tr(M'_b)} = d$.
From the equality in Eq.~\eqref{eq:expMprimeb} and linearity of expectation, we have
\begin{equation}
    \E \left[ \widehat{g} \right] = \frac{1}{d} \sum_{b \in \{0, 1\}^n} \bra{b}M'_b\ket{b} = g.
\end{equation}
Hence, we can estimate the scalar value $g$ to arbitrarily small error with large enough $N_A$.

\subsubsection{Learning noisy zero state}
\label{sec:learn-zerostate}

If $g$ is perfectly known, then the estimator for $\rho_0$ is given as follows.
\begin{equation} \label{eq:eststate1}
\widehat{\rho_0}^{(g*)} = \frac{d^2 - 1}{d g - 1} \left[\frac{1}{N_A} \sum_{i=1}^{N_A} (C^{(A, i)})^\dagger \ketbra{b^{(A, i)}}{b^{(A, i)}} C^{(A, i)}\right] - \frac{d - \frac{g}{d}}{d g - 1} I.
\end{equation}
Notice that in the noiseless setting ($g = 1$), the above estimator is exactly equal to the classical shadow representation based on randomized Clifford measurements \cite{huang2020predicting, huang2022learning}.
Using the unitary $2$-design property of Clifford gates, we have
\begin{equation}
    \E\left[ \widehat{\rho_0}^{(g*)} \right] = \frac{1}{|\mathcal{C}|} \sum_{C \in \mathcal{C}} \frac{d^2 - 1}{d g - 1} \Tr(M'_b C \rho_0 C^\dagger) C^\dagger \ketbra{b}{b} C - \frac{d - \frac{g}{d}}{d g - 1} I = \rho_0. \label{eq:rho0}
\end{equation}
Hence, we can estimate $\rho_0$ to arbitrarily small error with large enough $N_A$ when $g$ is perfectly known.
However, since $g$ is estimated using $\hat{g}$, we will use the following estimator instead.
\begin{equation} \label{eq:eststate2}
\widehat{\rho_0} = \frac{d^2 - 1}{d \hat{g} - 1} \left[\frac{1}{N_A} \sum_{i=1}^{N_A} (C^{(A, i)})^\dagger \ketbra{b^{(A, i)}}{b^{(A, i)}} C^{(A, i)}\right] - \frac{d - \frac{\hat{g}}{d}}{d \hat{g} - 1} I.
\end{equation}
We use that fact that, with large enough $N_A$, $\hat{g}$ can be made arbitrarily close to $g$.
Because $g \in [0.9, 1.1]$ from Eq.~\eqref{eq:boundednoise-eq}, with large enough $N_A$, $\frac{d^2 - 1}{d \hat{g} - 1}$ and $\frac{d- \frac{\hat{g}}{d}}{d \hat{g} - 1}$ can be made arbitrarily close to $\frac{d^2 - 1}{d g - 1}$ and $\frac{d - \frac{g}{d}}{d g - 1}$, respectively.
Hence, with Eq.~\eqref{eq:rho0}, the estimator $\widehat{\rho_0}$ can be made arbitrarily close to $\rho_0$ with large enough $N_A$.

\subsubsection{Learning Clifford gate noise}
\label{sec:learning-Clif-gate}

So far, we have only used the first set of experiments.
We are now ready to learn the Clifford gate noise $\cN$ using the second set of experiments.
We will use the Choi matrix representation of a quantum channel.
Recall that the Choi matrix of a channel $\cN$ is given by
\begin{equation}
    \Phi_\cN \equiv (\cN \otimes \cI) ( \ketbra{\omega}{\omega} ),
\end{equation}
where $\ket{\omega} = \frac{1}{\sqrt{d}} \sum_{b \in \{0, 1\}^n} \ket{b} \otimes \ket{b}$ is the maximally entangled state over $2n$ qubits, and $\cI$ is the identity channel on $n$ qubits.
We first estimate the state $\cN(I/d)$.
The basic idea is to intentionally neglect $C_1^{(B, i)}$ in the data.
Then $\rho_0$ evolved under an unknown random Clifford gate $C_1$ followed by the gate noise $\cN$ is equal to the state $\cN(I/d)$ from the unitary $1$-design property of random Clifford gate.
Then we can essentially use the same estimator as Eq.~\eqref{eq:eststate1} to learn the state $\cN(I/d)$.
Assuming $g$ is known, then we can use the estimator,
\begin{equation}
    \widehat{\cN(I/d)}^{(g*)} = \frac{d^2 - 1}{d g - 1} \left[\frac{1}{N_B} \sum_{i=1}^{N_B} (C_2^{(B, i)})^\dagger \ketbra{b^{(B, i)}}{b^{(B, i)}} C_2^{(B, i)}\right] - \frac{d - \frac{g}{d}}{d g - 1} I.
\end{equation}
From unitary $1$-design and $2$-design property, we have
\begin{align}
    \E\left[ \widehat{\cN(I/d)}^{(g*)} \right] &= \frac{1}{|\mathcal{C}|^2} \sum_{C_1 \in \mathcal{C}} \sum_{C_2 \in \mathcal{C}} \frac{d^2 - 1}{d g - 1} \Tr(M'_b C_2 \cN( C_1 \rho_0 C_1^\dagger) C_2^\dagger) C_2^\dagger \ketbra{b}{b} C_2 - \frac{d - \frac{g}{d}}{d g - 1} I \\
    &= \cN(I/d). \label{eq:cNI/d}
\end{align}
Similar to learning $\rho_0$, we only have an estimate for $g$ given by $\hat{g}$.
Hence, the estimator we will use is
\begin{equation}
    \widehat{\cN(I/d)} = \frac{d^2 - 1}{d \hat{g} - 1} \left[\frac{1}{N_B} \sum_{i=1}^{N_B} (C_2^{(B, i)})^\dagger \ketbra{b^{(B, i)}}{b^{(B, i)}} C_2^{(B, i)}\right] - \frac{d - \frac{\hat{g}}{d}}{d \hat{g} - 1} I.
\end{equation}
With large enough $N_A$, $\hat{g}$ can be made arbitrarily close to $g$.
So, from Eq~\eqref{eq:cNI/d}, the estimator $\widehat{\cN(I/d)}$ can be made arbitrarily close to $\cN(I/d)$ with large enough $N_A$ and $N_B$.

We now present the estimator for the Choi matrix of $\cN$.
We will begin by assuming that $g$ is perfectly known, then we will approximate $g$ by $\widehat{g}$.
The estimator $\widehat{\Phi_\cN}^{(g*)}$ is defined as
\begin{align}
    \widehat{\Phi_\cN}^{(g*)} &= \frac{(d^2-1)^2}{(df - 1)(dg - 1)} \left[\frac{1}{N_B} \sum_{i=1}^{N_B} \left((C_2^{(B, i)})^\dagger \ketbra{b^{(B, i)}}{b^{(B, i)}} C_2^{(B, i)} \right) \otimes \left( C_1^{(B, i)} \ketbra{0^n}{0^n} (C_1^{(B, i)})^\dagger\right)^T\right]\\
    &- \frac{(d - \frac{1}{d})(d - \hat{g})}{(df - 1)(d g - 1)} \,\, \left[I \otimes I\right] - \frac{1 - \frac{f}{d}}{f - \frac{1}{d}} \,\, \left[\widehat{\cN(I/d)}^{(g*)} \otimes I \right].
\end{align}
Using unitary $2$-design property of Clifford gates and Eq.~\eqref{eq:cNI/d}, we have
\begin{align}
    \E\left[ \widehat{\Phi_\cN}^{(g*)} \right] &= \frac{1}{|\mathcal{C}|^2} \sum_{C_1 \in \mathcal{C}} \sum_{C_2 \in \mathcal{C}} \frac{(d^2-1)^2}{(df - 1)(g - 1)} \Tr(M'_b C_2 \cN( C_1 \rho_0 C_1^\dagger) C_2^\dagger)\\
    &\quad\quad\quad\quad\quad\quad\quad\quad\quad\quad \times \left(C_2^\dagger \ketbra{b}{b} C_2 \right) \otimes \left(C_1 \ketbra{0^n}{0^n} C_1^\dagger\right)^T \\
    &- \frac{(d - \frac{1}{d})(1 - \frac{g}{d^2})}{(df - 1)(g - 1)} I \otimes I - \frac{1 - \frac{f}{d}}{f - \frac{1}{d}} \cN(I/d) \otimes I \\
    &= \Phi_\cN. \label{eq:phicn}
\end{align}
Similar to before, we only have an estimate for $g$ given by $\hat{g}$, so we will instead use the following estimator, i.e., replacing all $g$ by $\hat{g}$ and $\widehat{\cN(I/d)}^{(g*)}$ by $\widehat{\cN(I/d)}$.
\begin{align}
    \widehat{\Phi_\cN} &= \frac{(d^2-1)^2}{(df - 1)(d \hat{g} - 1)} \left[\frac{1}{N_B} \sum_{i=1}^{N_B} \left((C_2^{(B, i)})^\dagger \ketbra{b^{(B, i)}}{b^{(B, i)}} C_2^{(B, i)} \right) \otimes \left( C_1^{(B, i)} \ketbra{0^n}{0^n} (C_1^{(B, i)})^\dagger\right)^T\right]\\
    &- \frac{(d - \frac{1}{d})(d - \hat{g})}{(df - 1)(d \hat{g} - 1)} \,\, \left[I \otimes I\right] - \frac{1 - \frac{f}{d}}{f - \frac{1}{d}} \,\, \left[\widehat{\cN(I/d)} \otimes I \right]. \label{eq:PhiN-est}
\end{align}
With large enough $N_A, N_B$, $\hat{g}$ and $\widehat{\cN(I/d)}$ can be made arbitrarily close to $g$ and $\cN(I/d)$.
So, from Eq~\eqref{eq:phicn}, the estimator $\widehat{\Phi_\cN}$ can be made arbitrarily close to $\Phi_\cN$ with large enough $N_A$ and $N_B$.
With the estimator for Choi matrix $\widehat{\Phi_\cN}$, we can obtain the estimator $\widehat{\cN}$ for the original CPTP map $\cN$ by applying the linear invertible transformation between Choi matrix and CPTP map.

\subsubsection{Learning noisy basis measurement}
\label{sec:leanringnoisyPOVM-2}

After learning $\Phi_\cN$ through the estimator $\widehat{\Phi_\cN}$, we can obtain the noisy computational basis measurement $\cM_0 = \{M_b\}_{b \in \{0, 1\}^n}$ by considering the following estimator.
\begin{equation}
    \widehat{M_b} = \left(\widehat{\cN}^\dagger \right)^{-1} \widehat{M'_b}, \quad \forall b \in \{0, 1\}^n.
\end{equation}
Note that $\widehat{M'_b}$ can be made arbitrarily close to $M'_b = \cN^\dagger(M_b)$ and $\widehat{\cN}$ can be made arbitrarily close to $\cN$.
Because $\cN$ is assumed to be close to the identity, the difference between $\left(\widehat{\cN}^\dagger \right)^{-1}$ and $\left(\cN^\dagger\right)^{-1}$ can be made arbitrarily small by increasing $N_A, N_B$.
Hence, $\widehat{M_b}$ can be arbitrarily close to $M_b$ with large enough $N_A, N_B$.

\subsection{Learning any state, CPTP map, and POVM}

After learning $\rho_0, \cN, \cM_0$, we present the learning of any physical operation in the system.
The procedures are very similar to the previous subsection.

\subsubsection{Learning POVM}

Given any POVM $\cM_z = \{ M_{zb} \}_{b \in \cB}$, where $\cB$ is a set denoting all the possible outcomes.
We can learn $M_{zb}$ by conducting the following randomized experiments for $N_z$ times: prepare $\rho_0$, evolve under $\cE_C$ for a random Clifford gate $C$, measure $\cM_z$.
We denote the sets of experimental outcomes as
\begin{align}
    \left( C^{(z, i)}, b^{(z, i)} \in \cB \right), &\quad \forall i = 1, \ldots, N_z.
\end{align}
The procedure is very similar to that for learning $\cM_0$.
First, we consider the POVM $\cM'_z = \{ M'_{zb} = \cN^{\dagger}(M_zb) \}_{b \in \cB}$, which conflates $\cM_z$ with the Clifford gate noise $\cN$.
We can learn $\cM'_z$ using the following estimators based on the data we obtained.
\begin{align}
\widehat{\Tr(M'_{zb})} &= \frac{1}{N_z} \sum_{i=1}^{N_z} d \, \bI\left[ b^{(z, i)} = b\right],\\
\widehat{M'_{zb}} &= \frac{d^2 - 1}{f - \frac{1}{d}} \left[ \frac{1}{N_z} \sum_{i=1}^{N_z} \bI\left[ b^{(z, i)} = b\right] C^{(z, i)} \ketbra{0^n}{0^n} (C^{(z, i)})^\dagger \right] - \frac{1 - \frac{f}{d}}{{f - \frac{1}{d}}} \widehat{\Tr(M'_{zb})} I.
\end{align}
Using the same analysis as Appendix~\ref{sec:leanringnoisyPOVM}, we can show that $\widehat{M'_{zb}}$ can be made arbitrarily close to $M'_{zb}$ with large enough $N_z$.
Then similar to Appendix~\ref{sec:leanringnoisyPOVM-2}, we can learn $M_{zb}$ using the following estimator after obtaining $\hat{\cN}$ from the algorithms presented in Appendix~\ref{sec:learning-Clif-gate}.
\begin{equation}
    \widehat{M_{zb}} = \left(\widehat{\cN}^\dagger \right)^{-1} \widehat{M'_{zb}}, \quad \forall b \in \cB.
\end{equation}
With large enough $N_A, N_B, N_z$, we can make $\widehat{M_{zb}}$ arbitrarily close to $M_{zb}$.

\subsubsection{Learning states}

Given any state $\rho_x$. We can learn $\rho_x$ by conducting the following randomized experiments $N_x$ times: prepare $\rho_x$, evolve under $\cE_C$ for a random Clifford gate $C$, measure $\cM_0$.
We denote the sets of experimental outcomes as
\begin{align}
    \left( C^{(x, i)}, b^{(x, i)} \in \{0, 1\}^n \right), &\quad \forall i = 1, \ldots, N_x.
\end{align}
We can obtain an estimate for $\rho_x$ using the following formula.
\begin{equation}
\widehat{\rho_x} = \frac{d^2 - 1}{d \hat{g} - 1} \left[\frac{1}{N_x} \sum_{i=1}^{N_x} (C^{(x, i)})^\dagger \ketbra{b^{(x, i)}}{b^{(x, i)}} C^{(x, i)}\right] - \frac{d - \frac{\hat{g}}{d}}{d \hat{g} - 1} I.
\end{equation}
With large enough $N_A, N_B, N_x$, we can make $\widehat{\rho_x}$ arbitrarily close to $\rho_x$.
The analysis is the same as that in Appendix~\ref{sec:learn-zerostate}.

\subsubsection{Learning CPTP maps}

Given any CPTP map $\cE_y$. We can learn $\cE_y$ by conducting the following randomized experiments for $N_y$ times: prepare $\rho_0$, evolve under $\cE_{C_1}$ for a random Clifford gate $C_1$, evolve under $\cE_y$, evolve under $\cE_{C_2}$ for a different random Clifford gate $C_2$, measure $\cM_0$.
We denote the sets of $N_y$ experimental outcomes as follows.
\begin{align}
    \left( C_1^{(y, i)}, C_2^{(y, i)}, b^{(y, i)} \in \{0, 1\}^n \right), &\quad \forall i = 1, \ldots, N_y.
\end{align}
We first learn the quantum process $\cE'_y = \cE_y \circ \cN$, which conflates $\cE_y$ with the Clifford gate noise $\cN$. Following the same analysis as Appendix~\ref{sec:learning-Clif-gate} but replacing $\cN$ by $\cE'_y$, we define the following estimators.
\begin{align}
    \widehat{\cE'_y(I/d)} &= \frac{d^2 - 1}{d \hat{g} - 1} \left[\frac{1}{N_y} \sum_{i=1}^{N_y} (C_2^{(u, i)})^\dagger \ketbra{b^{(y, i)}}{b^{(y, i)}} C_2^{(y, i)}\right] - \frac{d - \frac{\hat{g}}{d}}{d \hat{g} - 1} I. \\
    \widehat{\Phi_{\cE'_y}} &= \frac{(d^2-1)^2}{(df - 1)(d \hat{g} - 1)} \left[\frac{1}{N_y} \sum_{i=1}^{N_y} \left((C_2^{(y, i)})^\dagger \ketbra{b^{(y, i)}}{b^{(y, i)}} C_2^{(y, i)} \right) \otimes \left( C_1^{(y, i)} \ketbra{0^n}{0^n} (C_1^{(y, i)})^\dagger\right)^T\right]\\
    &- \frac{(d - \frac{1}{d})(d - \hat{g})}{(df - 1)(d \hat{g} - 1)} \,\, \left[I \otimes I\right] - \frac{1 - \frac{f}{d}}{f - \frac{1}{d}} \,\, \left[\widehat{\cN(I/d)} \otimes I \right].
\end{align}
With large enough $N_A, N_B, N_y$, we can make $\widehat{\Phi_{\cE'_y}}$ arbitrarily close to $\Phi_{\cE'_y}$, the Choi matrix for the CPTP map $\cE'_y$.
Hence, we can obtain the estimator $\widehat{\cE'_y}$ for $\cE'_y$ by transforming the Choi matrix back to the CPTP map.
Now we simply need to invert the conflation with $\cN$ by considering
\begin{equation}
    \widehat{ \cE_y } = \widehat{\cE'_y} \circ (\widehat{\cN})^{-1}.
\end{equation}
Because the Clifford gate noise $\cN$ is assumed to be close to the identity, the difference between $\left(\widehat{\cN}^\dagger \right)^{-1}$ and $\left(\cN^\dagger\right)^{-1}$ can be made arbitrarily small by increasing $N_A, N_B$.
Therefore, we can make the difference between $\widehat{\cE_y}$ and $\cE_y$ arbitrarily small with large enough $N_A, N_B, N_y$.

\subsection{Sample complexity for learning Clifford gate noise}

All the previous analyses could be equipped with rigorous convergence guarantee using concentration inequalities similar to quantum state/process tomography based on randomized measurements \cite{guta2020fast, surawy2021projected}.
As an example, we present the sample complexity for learning Clifford gate noise $\cN$ with $\cN(I) = I$.
The reconstruction of $\cN$ under the condition $\cN(I) = I$ was previously studied in \cite{kimmel2014robust, roth2018recovering, helsen2021estimating}, where learning algorithms based on interleaved randomized benchmarking \cite{magesan2012efficient} have been devised.
We will show that the proposed algorithm is much more efficient than the best known sample complexity of $\mathcal{O}(d^8)$ in \cite{helsen2021estimating} for learning the Choi matrix $\Phi_{\cN}$ up to a constant error in Hilbert-Schmidt norm.

Under the assumption that $\cN(I) = I$, we do not need to estimate $\cN(I/d)$ because $\cN(I/d) = I/d$.
Hence, the estimator in Eq.~\eqref{eq:PhiN-est} simplifies to
\begin{align}
    \widehat{\Phi_\cN} &= \frac{(d^2-1)^2}{(df - 1)(d \hat{g} - 1)} \left[\frac{1}{N_B} \sum_{i=1}^{N_B} \left((C_2^{(B, i)})^\dagger \ketbra{b^{(B, i)}}{b^{(B, i)}} C_2^{(B, i)} \right) \otimes \left( C_1^{(B, i)} \ketbra{0^n}{0^n} (C_1^{(B, i)})^\dagger\right)^T\right]\\
    &- \frac{(d - \frac{1}{d})(d - \hat{g})}{(df - 1)(d \hat{g} - 1)} \,\, \left[I \otimes I\right] - \frac{1 - \frac{f}{d}}{df - 1} \,\, \left[I \otimes I \right].
\end{align}
Furthermore, using the definition of $\hat{g}$ in Eq.~\eqref{eq:hatg-est}, we see that
\begin{equation}
    (df - 1)(d \hat{g} - 1) = (d^2-1) \left( d \left[ \frac{1}{N_A} \sum_{i=1}^{N_A} \bra{b^{(A, i)}} C^{(A, i)} \ketbra{0^n}{0^n} (C^{(A, i)})^\dagger \ket{b^{(A, i)}} \right] - 1 \right)
\end{equation}
is independent of $f$. Furthermore, we can simplify the coefficients for $I\otimes I$ using
\begin{equation}
    - \frac{(d - \frac{1}{d})(d - \hat{g})}{(df - 1)(d \hat{g} - 1)} - \frac{1 - \frac{f}{d}}{df - 1} = - \frac{(d - \frac{1}{d})^2}{(df - 1)(d \hat{g} - 1)} + \frac{1}{d^2},
\end{equation}
which is also independent of $f$ because $(df - 1)(d \hat{g} - 1)$ is independent of $f$.
We can now rewrite the estimator $\widehat{\Phi_\cN}$ as the following $f$-independent expression,
\begin{align}
    \widehat{\Phi_\cN} &= (d^2-1) \frac{\frac{1}{N_B} \sum_{i=1}^{N_B} \left((C_2^{(B, i)})^\dagger \ketbra{b^{(B, i)}}{b^{(B, i)}} C_2^{(B, i)} \right) \otimes \left( C_1^{(B, i)} \ketbra{0^n}{0^n} (C_1^{(B, i)})^\dagger\right)^T}{d \left[ \frac{1}{N_A} \sum_{i=1}^{N_A} \bra{b^{(A, i)}} C^{(A, i)} \ketbra{0^n}{0^n} (C^{(A, i)})^\dagger \ket{b^{(A, i)}} \right] - 1}\\
    &- (d^2-1) \frac{ (I / d) \otimes (I / d)}{d \left[ \frac{1}{N_A} \sum_{i=1}^{N_A} \bra{b^{(A, i)}} C^{(A, i)} \ketbra{0^n}{0^n} (C^{(A, i)})^\dagger \ket{b^{(A, i)}} \right] - 1} + (I / d) \otimes (I / d).
\end{align}
We will now present the algorithm for estimating $\Phi_\cN$ under the Pauli basis.
Consider two $n$-qubit Pauli operators $P, Q \in \{I, X, Y, Z\}^{\otimes n}$.
We would estimate $\Tr((P \otimes Q) \Phi_\cN)$ based on the above expression for $\widehat{\Phi_\cN}$.
The estimation procedure is separated into three cases.

\subsubsection{$P=I$ and $Q=I$}

This is the simplest case.
Because $\Phi_\cN$ is a quantum state, we have $\Tr(\Phi_\cN) = 1$.
Hence, we can obtain a perfect estimate for $\Tr((P \otimes Q) \Phi_\cN)$ as it is always one.

\subsubsection{Exactly one of $P$ and $Q$ is equal to $I$}

This is also a simple case. Recall that a non-identity Pauli operator has trace equal to zero. Then, using the definition of Choi matrix $\Phi_\cN$ and $\cN(I) = I$, we have
\begin{equation}
    \Tr((P \otimes Q) \Phi_\cN) = 0.
\end{equation}
Hence, we can obtain a perfect estimate for $\Tr((P \otimes Q) \Phi_\cN)$ as it is always zero.

\subsubsection{$P \neq I$ and $Q \neq I$}

Using the fact that a non-identity Pauli operator has trace equal to zero, we have
\begin{equation}
    \Tr((P \otimes Q) \widehat{\Phi_\cN}) = (d^2 - 1)\frac{\frac{1}{N_B} \sum_{i=1}^{N_B} \bra{b^{(B, i)}} C_2^{(B, i)} P (C_2^{(B, i)})^\dagger \ket{b^{(B, i)}} \bra{0^n} (C_1^{(B, i)})^\dagger Q^T C_1^{(B, i)} \ket{0^n}}{d \left[ \frac{1}{N_A} \sum_{i=1}^{N_A} \bra{b^{(A, i)}} C^{(A, i)} \ketbra{0^n}{0^n} (C^{(A, i)})^\dagger \ket{b^{(A, i)}} \right] - 1}.
\end{equation}
We will directly use the above formula to estimate $\Tr((P \otimes Q) \Phi_\cN)$.

To analyze the error in the above estimator, we separately consider the convergence of the numerator and the denominator.
We begin with the denominator,
\begin{align}
    \cY &= d \left[ \frac{1}{N_A} \sum_{i=1}^{N_A} \bra{b^{(A, i)}} C^{(A, i)} \ketbra{0^n}{0^n} (C^{(A, i)})^\dagger \ket{b^{(A, i)}} \right] - 1\\
    &= \frac{1}{N_A} \sum_{i=1}^{N_A} \left[ d \bra{b^{(A, i)}} C^{(A, i)} \ketbra{0^n}{0^n} (C^{(A, i)})^\dagger \ket{b^{(A, i)}} - 1 \right].
\end{align}

\begin{lemma}[Concentration for the denominator] \label{lem:con-den}
Fix $\epsilon > 0$.
Given $N_A = \mathcal{O}(1 / \epsilon^2)$. With probability at least $0.99$, we have
\begin{equation} \label{eq:cYconc}
     \left| \cY - \frac{(df - 1)(d g - 1)}{(d-1)(d+1)} \right| < \epsilon,
\end{equation}
where $\frac{(df - 1)(d g - 1)}{(d-1)(d+1)} \in [0.21, 1.21]$ from Eq.~\eqref{eq:boundednoise-eq}.
\end{lemma}
\begin{proof}
Let $\cY_i = d \bra{b^{(A, i)}} C^{(A, i)} \ketbra{0^n}{0^n} (C^{(A, i)})^\dagger \ket{b^{(A, i)}} - 1$.
From unitary $2$-design of random Clifford gate, we have the following identity,
\begin{equation}
    \E[\cY_i] = \frac{(df - 1)(d g - 1)}{(d-1)(d+1)}.
\end{equation}
Then, using the unitary $3$-design property of random Clifford gate \cite{webb2015clifford, zhu2016clifford} and the conditions in Eq.~\eqref{eq:boundednoise-eq}, we have
\begin{equation}
    \Var[\cY_i] \leq \frac{1}{|\cC|}\sum_{C \in \cC} \sum_{b \in \{0, 1\}^n} \Tr(M'_b C \rho_0 C^\dagger) d^2 (\bra{b} C \ketbra{0^n}{0^n} C^\dagger \ket{b} )^2 = \mathcal{O}(1).
\end{equation}
The claim then follows from $\cY = \frac{1}{N_A} \sum_{i=1}^{N_A} \cY_i$ and Chebyshev's inequality.
\end{proof}

Next, we can analyze the numerator,
\begin{align}
    \cX &= \frac{1}{N_B} \sum_{i=1}^{N_B} (d^2 - 1) \bra{b^{(B, i)}} C_2^{(B, i)} P (C_2^{(B, i)})^\dagger \ket{b^{(B, i)}} \bra{0^n} (C_1^{(B, i)})^\dagger Q^T C_1^{(B, i)} \ket{0^n}.
\end{align}

\begin{lemma}[Concentration for the numerator] \label{lem:con-num}
Fix $0.5 > \epsilon, \delta > 0$.
Given $N_B = \mathcal{O}(d^2 \log(1 / \delta) / \epsilon^2)$. With probability at least $1 - \delta$,
\begin{equation} \label{eq:cXconc}
     \left| \cX - \frac{(df - 1)(d g - 1)}{(d-1)(d+1)} \Tr((P \otimes Q) \Phi_\cN) \right| < \epsilon,
\end{equation}
where $\frac{(df - 1)(d g - 1)}{(d-1)(d+1)} \Tr((P \otimes Q) \Phi_\cN) \in [0.21, 1.21]$ from Eq.~\eqref{eq:boundednoise-eq}.
\end{lemma}
\begin{proof}
Let $\cX_i = (d^2 - 1) \bra{b^{(B, i)}} C_2^{(B, i)} P (C_2^{(B, i)})^\dagger \ket{b^{(B, i)}} \bra{0^n} (C_1^{(B, i)})^\dagger Q^T C_1^{(B, i)} \ket{0^n}$.
From unitary $2$-design of random Clifford gate and the fact that non-identity Pauli has trace zero, we have,
\begin{equation}
    \E[\cX_i] = \frac{(df - 1)(d g - 1)}{(d-1)(d+1)} \Tr((P \otimes Q) \Phi_\cN).
\end{equation}
Because $C_2^{(B, i)}$ is a random Clifford gate, $C_2^{(B, i)} P (C_2^{(B, i)})^\dagger$ is proportional to a random non-identity Pauli $\{I, X, Y, Z\}^{\otimes n} \setminus \{I^{\otimes n}\}$.
If $C_2^{(B, i)} P (C_2^{(B, i)})^\dagger$ is not proportional to a Pauli-Z operator $\{I, Z\}^{\otimes n}$, then $\cX_i = 0$.
Similarly $(C_1^{(B, i)})^\dagger Q^T C_1^{(B, i)}$ is proportional to a random non-identity Pauli.
And, $C_1^{(B, i)} Q^T (C_1^{(B, i)})^\dagger$ is not proportional to a Pauli-Z operator $\{I, Z\}^{\otimes n}$, then $\cX_i = 0$.
Because $C_1, C_2$ are independent random gates, we have $\cX_i \neq 0$ with probability at most
\begin{equation}
    \frac{(d - 1)}{(d^2 - 1)} \times \frac{(d - 1)}{(d^2 - 1)} = \frac{1}{(d+1)^2}.
\end{equation}
Furthermore, we have $|\cX_i| \leq (d^2 - 1)$ with probability one.
Therefore, we have
\begin{equation}
    \Var[\cX_i] \leq \E[\cX_i^2] \leq \frac{1}{(d+1)^2} \times (d^2 - 1)^2 = (d-1)^2 \leq d^2.
\end{equation}
From Bernstein's inequality and the definition that $\cX = \frac{1}{N_B} \sum_{i=1}^{N_B} \cX_i$, we conclude the proof.
\end{proof}

We can now establish the following statement.

\begin{lemma}[Combine numerator and denominator] \label{lem:con-den-num}
Given $0 < \epsilon < 0.1$.
Assume Eq.~\eqref{eq:cYconc} and \eqref{eq:cXconc} both hold. Then
\begin{equation}
    \left| \frac{\cX}{\cY} - \Tr((P \otimes Q) \Phi_\cN) \right| < 155 \epsilon.
\end{equation}
\end{lemma}
\begin{proof}
The proof follows from the following analysis.
\begin{align}
    \left| \frac{\cX}{\cY} - \Tr((P \otimes Q) \Phi_\cN) \right| &\leq \left| \frac{\cX}{\cY} - \frac{\cX}{\frac{(df - 1)(d g - 1)}{(d-1)(d+1)}} \right| + \left| \frac{\cX}{\frac{(df - 1)(d g - 1)}{(d-1)(d+1)}} - \Tr((P \otimes Q) \Phi_\cN) \right| \\
    &\leq \left(1.21 + \epsilon\right)\left|\frac{1}{\cY} - \frac{1}{\frac{(df - 1)(d g - 1)}{(d-1)(d+1)}}\right|\\
    &+ 5 \left| \cX - \frac{(df - 1)(d g - 1)}{(d-1)(d+1)} \Tr((P \otimes Q) \Phi_\cN) \right| \\
    &< 1.5 \times 100 \epsilon + 5 \epsilon = 155 \epsilon.
\end{align}
The first line uses triangle inequality.
The second inequality uses the fact that
\begin{equation} \label{eq:range-exp}
    \frac{(df - 1)(d g - 1)}{(d-1)(d+1)}, \frac{(df - 1)(d g - 1)}{(d-1)(d+1)} \Tr((P \otimes Q) \Phi_\cN) \in [0.21, 1.21].
\end{equation}
The third inequality uses Eq.~\eqref{eq:range-exp}, the condition that $ 0 < \epsilon < 0.1$, and
\begin{equation}
    \left| \frac{1}{a} - \frac{1}{b} \right| < \frac{1}{\min(a, b)^2} \times \left| a - b\right|, \forall a, b > 0.
\end{equation}
This concludes the proof.
\end{proof}

\subsubsection{Final result}

Recall the following Pauli basis representation of a $2n$-qubit quantum state,
\begin{equation}
    \Phi_\cN = \sum_{P, Q \in \{I, X, Y, Z\}^{\otimes n}} \Tr((P \otimes Q) \Phi_\cN) \frac{(P \otimes Q)}{d^2}.
\end{equation}
We can learn the coefficients $\Tr((P \otimes Q) \Phi_\cN)$ using the above strategy.
We denote the estimated coefficients as $\hat{o}(\Phi_\cN, P, Q) \in \mathbb{R}$.
We can then obtain a reconstruction for $\Phi_\cN$ as
\begin{equation}
    \tilde{\Phi}_\cN(N_A, N_B) = \sum_{P, Q \in \{I, X, Y, Z\}^{\otimes n}} \hat{o}(\Phi_\cN, P, Q) \frac{(P \otimes Q)}{d^2}.
\end{equation}
We can now combine the previous results to show the sample complexity to learn the Choi matrix of the Clifford gate noise $\Phi_\cN$ up to $\epsilon$ error in Hilbert Schmidt norm.

\begin{theorem}[Sample complexity for learning Clifford gate noise]
Given $0 < \epsilon < 0.1 d$. Assume the noise is bounded as stated in Eq.~\eqref{eq:boundednoise-eq}.
With $N_A = \mathcal{O}(d^2 / \epsilon^2)$ and $N_B = \mathcal{O}(d^4 \log(d) / \epsilon^2)$,
\begin{equation}
    \norm{\tilde{\Phi}_\cN(N_A, N_B) - \Phi_\cN}_{\mathrm{HS}} < \epsilon,
\end{equation}
with probability at least $0.99$, where $\norm{X}_{\mathrm{HS}} = \Tr(X^2)$ is the Hilbert-Schmidt norm (Frobenius norm).
\end{theorem}
\begin{proof}
Let $\epsilon' = \epsilon / d$.
We can employ union bound and Lemma~\ref{lem:con-den},~\ref{lem:con-num},~\ref{lem:con-den-num} to show that given $N_A = \mathcal{O}(1 / (\epsilon')^2) = \mathcal{O}(d^2 / \epsilon^2)$ and $N_B = \mathcal{O}(d^2 \log(d) / (\epsilon')^2) = \mathcal{O}(d^4 \log(d) / \epsilon^2)$ and $0 < \epsilon' < 0.1$, we have
\begin{equation}
    \left| \hat{o}(\Phi_\cN, P, Q) - \Tr((P \otimes Q) \Phi_\cN) \right| < \epsilon' = \epsilon / d,
\end{equation}
with probability at least $0.99$. Therefore, we have
\begin{equation}
    \norm{\tilde{\Phi}_\cN(N_A, N_B) - \Phi_\cN}_{\mathrm{HS}} = \sqrt{\frac{1}{d^2} \sum_{P, Q \in \{I, X, Y, Z\}^{\otimes n}} \left| \hat{o}(\Phi_\cN, P, Q) - \Tr((P \otimes Q) \Phi_\cN) \right|^2} < \epsilon
\end{equation}
with probability at least $0.99$.
This concludes the proof.
\end{proof}

\newpage
\section{Foundations for predicting extrinsic behavior}
\label{sec:behav-predict}

Recall that Definition~\ref{def:predictability} considers the predictability of properties in the world model.
Here, we refer to a property as a function that maps a world model to a real value.
Given action spaces $\mathcal{X}, \mathcal{Y}, \mathcal{Z}$ and an outcome space $\mathcal{B}$.
There is a class of properties that are guaranteed to be predictable: the probability of an outcome $b \in \mathcal{B}$ for an experiment specified by the sequence of actions $x \in \mathcal{X}, y_1, \ldots, y_\ell \in \cY, z \in \cZ$.
We refer to such properties as the extrinsic behavior of the world model.
Formally, this class of properties is given in the following definition.

\begin{definition}[Extrinsic behavior of world model]
Given action spaces $\mathcal{X}, \mathcal{Y}, \mathcal{Z}$ and an outcome space $\mathcal{B}$.
We consider extrinsic behavior to be a set of properties $\cF = \{f: \mathcal{W} \rightarrow \mathbb{R}\}$. For any $f \in \cF$, there is an experiment $E = (x \in \mathcal{X}, y_1 \in \mathcal{Y}, \ldots, y_\ell \in \mathcal{Y}, z \in \mathcal{Z})$, and an outcome $b \in \cB$, such that
\begin{equation}
    f(\mathcal{W}) = \Tr(M_{zb} (\mathcal{E}_{y_{\ell}} \circ \ldots \circ \mathcal{E}_{y_{1}})(\rho_x) ),
\end{equation}
for any world model $\cW =  \left(\{\rho_x\}_{x \in \mathcal{X}}, \{\mathcal{E}_y\}_{y \in \mathcal{Y}}, \{\mathcal{M}_z\}_{z \in \mathcal{Z}} \right)$.
We refer to $f$ in the set as $f^{\mathrm{(behav)}}_{x, y_1, \ldots, y_\ell, z, b}$.
\end{definition}

In the following, we give several fundamental results regarding the task of predicting extrinsic behavior. In particular, we will give a rigorous performance guarantee for existing gate-set tomography protocols formulated in terms of learning the extrinsic behavior of a world model.

\subsection{Hardness in predicting extrinsic behaviors}
\label{sec:hardness-behav}

The goal of this appendix is to prove the following theorem.

\begin{theorem}[Worst case complexity for predicting extrinsic behaviors; Restatement of Theorem~\ref{thm:worst-case-behav}]
Consider finite sets $\mathcal{X}, \mathcal{Y}, \mathcal{Z}, \mathcal{B}$ with $|\mathcal{B}| \geq 2$ and $|\mathcal{B}| = \mathcal{O}(\mathrm{poly}(|\mathcal{X}|, |\mathcal{Y}|, |\mathcal{Z}|))$.
Given $\epsilon, L > 0$.
For any model class $\mathcal{Q}$ over $\mathcal{X}, \mathcal{Y}, \mathcal{Z}, \mathcal{B}$, there is an algorithm that uses
\begin{equation}
    \widetilde{\mathcal{O}}\left(\frac{|\mathcal{X}| |\mathcal{Y}|^L |\mathcal{Z}|}{\epsilon^2} \right)
\end{equation}
experiments to accurately predict $f^{\mathrm{(behav)}}_{x, y_1, \ldots, y_L, z, b}(\mathcal{W})$ up to $\epsilon$-error for all $x \in \cX, y_1, \ldots, y_L \in \cY, z \in \cZ, b \in \cB$ with high probability.
Furthermore, there exists a model class $\mathcal{Q}$ over $\mathcal{X}, \mathcal{Y}, \mathcal{Z}, \mathcal{B}$ such that
\begin{equation}
    \Omega\left( \frac{|\mathcal{X}| |\mathcal{Y}|^L |\mathcal{Z}|}{\epsilon^2} \right)
\end{equation}
experiments are required to accurately predict $f^{\mathrm{(behav)}}_{x, y_1, \ldots, y_L, z, b}(\mathcal{W})$ up to $\epsilon$-error for all $x \in \cX, y_1, \ldots, y_L \in \cY, z \in \cZ, b \in \cB$ with high probability.
\end{theorem}

The first part of the theorem is straightforward.
We consider the algorithm that runs through all possible experiments with $L$ CPTP maps. For each experiment specified by $x \in \cX, y_1, \ldots, y_L \in \cY, z \in \cZ$, the algorithm performs the experiment a number of
\begin{equation}
K = \mathcal{O}\left(\frac{\log\left(|\mathcal{X}| |\mathcal{Y}|^L |\mathcal{Z}| |\cB| \right)}{\epsilon^2} \right)
\end{equation}
times.
For any world model $\mathcal{W} \in \mathcal{Q}$, using Hoeffding's inequality, the $K$ experimental outcomes can be used to estimate the probability of an outcome $b \in \mathcal{B}$ for the experiment $E = (x, y_1, \ldots, y_L, z)$, i.e., $f^{\mathrm{(behav)}}_{x, y_1, \ldots, y_L, z, b}(\mathcal{W})$, up to an error $\epsilon$ with probability $\geq 1 - 1 / (10 |\mathcal{X}| |\mathcal{Y}|^L |\mathcal{Z}| |\cB|)$.
Using the union bound, with probability $\geq 0.9$, we can estimate $f^{\mathrm{(behav)}}_{x, y_1, \ldots, y_L, z, b}(\mathcal{W})$ up to an error $\epsilon$ for all  $x \in \cX, y_1, \ldots, y_L \in \cY, z \in \cZ, b \in \cB$.
The total number of experiments used by the algorithm is
\begin{equation}
    |\mathcal{X}| |\mathcal{Y}|^L |\mathcal{Z}| K = \mathcal{O}\left( \frac{ |\mathcal{X}| |\mathcal{Y}|^L |\mathcal{Z}| \log\left(|\mathcal{X}| |\mathcal{Y}|^L |\mathcal{Z}| |\cB| \right)}{\epsilon^2} \right) =     \widetilde{\mathcal{O}}\left(\frac{|\mathcal{X}| |\mathcal{Y}|^L |\mathcal{Z}|}{\epsilon^2} \right),
\end{equation}
because of the assumption that $|\mathcal{B}| = \mathcal{O}(\mathrm{poly}(|\mathcal{X}|, |\mathcal{Y}|, |\mathcal{Z}|))$.

The second part of the theorem is more nontrivial.
For this part of the proof, we will only utilize two elements in $\cB$, which we denote as $b_0, b_1$.
We begin by constructing a model class $\mathcal{Q}$. The model class $\cQ$ contains a null world model $\mathcal{W}_0$ given by
\begin{align}
    \rho_{x} &= \ketbra{0}{0}, &\forall x \in \mathcal{X}, \\
    \cE_{y}(\rho) &= \ketbra{0}{0}, &\forall y \in \mathcal{Y}, \\
    M_{z b_0} &= \frac{1}{2} I, &\forall z \in \mathcal{Z}, \\
    M_{z b_1} &= \frac{1}{2} I, &\forall z \in \mathcal{Z}, \\
    M_{z b} &= 0, &\forall z \in \mathcal{Z}, \forall b \in \cB \setminus \{b_0, b_1\},
\end{align}
where $I$ is the identity operator.
The model class $\mathcal{Q}$ also contains another $2 |\mathcal{X}| |\mathcal{Y}|^L |\mathcal{Z}|$ world models.
Each of the world models is given by $\mathcal{W}_{x, y_1, \ldots, y_L, z, s}$ where $x \in \cX, y_1, \ldots, y_L \in \cY, z \in \cZ, s = \pm 1$.
Because $x$ could be one of $|\cX|$ elements, $y_\ell$ could be one of $|\cY|$ elements for any $\ell = 1, \ldots, L$, $z$ could be one of $|\cZ|$ elements, and $s$ could take one of the two values $\pm 1$, we have
\begin{equation}
    \left|\{\mathcal{W}_{x, y_1, \ldots, y_L, z, s}\}_{x \in \cX, y_1, \ldots, y_L \in \cY, z \in \cZ, s \in \{\pm 1\}} \right| = 2 |\mathcal{X}| |\mathcal{Y}|^L |\mathcal{Z}|
\end{equation}
The world model $\mathcal{W}_{x, y_1, \ldots, y_L, z, s}$ is defined by the following equalities
\begin{align}
    \rho_{x'} &= \ketbra{0}{0}, &\forall x' \in \mathcal{X} \setminus \{x\}, \label{eq:Wxyz-hardalt1} \\
    \rho_{x} &= \ketbra{1}{1}, & \\
    \cE_{y_\ell}(\ketbra{\ell}{\ell}) &= \ketbra{\ell+1}{\ell+1}, & \forall 1 \leq \ell \leq L,\\
    \cE_{y}(\ketbra{\ell}{\ell}) &= \ketbra{0}{0}, & \forall 0 \leq \ell \leq L+1, y \in \cY, \, \mbox{s.t.} \, \forall 1 \leq \ell' \leq L, y \neq y_{\ell'},\\
    M_{z b_0} &= \frac{1}{2} (I + 3s \epsilon \ketbra{L+1}{L+1}), &\forall z \in \mathcal{Z}, \\
    M_{z b_1} &= \frac{1}{2} (I - 3s \epsilon \ketbra{L+1}{L+1}), &\forall z \in \mathcal{Z}, \\
    M_{z b} &= 0, &\forall z \in \mathcal{Z}, \forall b \in \cB \setminus \{b_0, b_1\}, \label{eq:Wxyz-hardalt2}
\end{align}
For the null world model $\cW_0$, the outcome distribution is always a uniform distribution over $b_0, b_1$ for any experiments.
For world model $\mathcal{W}_{x, y_1, \ldots, y_L, z, s}$, only for one experiment, i.e., when we consider $E = (x, y_1, \ldots, y_L, z)$, the outcome distribution is a biased distribution over $b_0, b_1$, in particular, we see $b_0$ with probability $\frac{1}{2} + \frac{3}{2} s \epsilon$ and $b_1$ with probability $\frac{1}{2} - \frac{3}{2} s \epsilon$. For all other experiments, the outcome distribution is again a uniform distribution over $b_0, b_1$.
We will also denote $\mathcal{W}_{x, y_1, \ldots, y_L, z, s}$ as $\mathcal{W}_{E, s}$, where $E = (x, y_1, \ldots, y_L, z)$ is an experiment.

We consider the true world model to be $\mathcal{W}_0$ with probability $1/2$ and to be  $\mathcal{W}_{x, y_1, \ldots, y_L, z, s}$ for a particular choice of $x \in \cX, y_1, \ldots, y_L \in \cY, z \in \cZ, s = \pm 1$ with probability $1/ (4 |\mathcal{X}| |\mathcal{Y}|^L |\mathcal{Z}|)$.
If there is a learning algorithm that could accurately predict $f^{\mathrm{(behav)}}_{x, y_1, \ldots, y_L, z, b}(\mathcal{W})$ up to $\epsilon$-error for all $x \in \cX, y_1, \ldots, y_L \in \cY, z \in \cZ, b \in \cB$, then the algorithm can be used to check if the true world model is equal to $\mathcal{W}_0$.
Recall that for $\mathcal{W}_0$, all experiments yield an outcome probability distribution that is uniform over $\{b_0, b_1\}$.
For any world model in $\mathcal{Q} \setminus \{\mathcal{W}_0\}$, the probability for only one of the possible experiments specified by $x \in \cX, y_1, \ldots, y_L \in \cY, z \in \cZ$ will be a biased distribution.
Hence, we can determine whether the true world model is equal to $\mathcal{W}_0$ by checking if for all $x \in \cX, y_1, \ldots, y_L \in \cY, z \in \cZ$, the probability to obtain both $b_0$ and $b_1$ are close to $1/2$.
This enables us to map the learning problem to a two-hypothesis testing problem.
And the question becomes:
\begin{center}
How many experiments are required to test if the true world model is equal to $\mathcal{W}_0$?
\end{center}
We will utilize the proof techniques given in \cite{huang2021information, chen2021exponential} to answer the above question.

A learning algorithm is represented by a tree where each node represents the collection of data received from all prior experiments.
Each edge is labeled by an experiment $E$ and an outcome $b \in \{b_0, b_1\}$ from that experiment.
We only consider outcomes $\{b_0, b_1\}$ because the other outcomes happen with zero probability by construction.
The probability for traversing that edge is the product of the probability that the experiment $E$ is performed at this stage of the learning algorithm and the probability for seeing the outcome $b$ for the experiment $E$.
The depth $T$ of the tree is the total number of experiments conducted by the learning algorithm.
After conducting $T$ experiments, we will arrive at a leaf node of the tree.
The probability to arrive at a leaf node depends on the true world model.

We consider two events: when the true world model is $\mathcal{W}_0$, and when the true world model is not equal to $\mathcal{W}_0$.
We distinguish the two events based on the data we collected from the $T$ experiments, i,e., which leaf node we arrived at.
In order to successfully distinguish between the two events with a constant probability, we need the total variation distance of the probability distribution over the leaf nodes to be of $\Omega(1)$. Formally, this is known as LeCam's two point method.

Each leaf node $\ell$ is specified by the path from root to the leaf node, which is a sequence of $T$ experiments and their corresponding outcomes.
Hence, we write each leaf node as $\ell_{(E_1, \beta_1) \ldots, (E_T, \beta_T)}$, where $\beta_t \in \{b_0, b_1\} \subseteq \mathcal{B}$.
For the null world model $\mathcal{W}_0$, the probability of a leaf node is given by
\begin{equation}
    p_{\mathcal{W}_0}(\ell_{(E_1, \beta_1) \ldots, (E_T, \beta_T)}) = \prod_{t=1}^T \left( \frac{1}{2} p_{(E_1, \beta_1) \ldots, (E_{t-1}, \beta_{t-1})}(E_t) \right),
\end{equation}
where $p_{(E_1, \beta_1) \ldots, (E_{t-1}, \beta_{t-1})}(E_t)$ is the probability that the algorithm would perform the experiment $E_t$ when the algorithm have run $t-1$ experiments $E_1, \ldots, E_{t-1}$ and the corresponding outcomes are $\beta_1, \ldots, \beta_{t-1} \in \{b_0, b_1\}$.
Recall that we denote the world model $\mathcal{W}_{x, y_1, \ldots, y_L, z, s}$ defined in Eq.~\eqref{eq:Wxyz-hardalt1} to Eq.~\eqref{eq:Wxyz-hardalt2} as $\mathcal{W}_{E, s}$, where $E$ is a tuple $(x, y_1, \ldots, y_L, z)$.
Note that an experiment $E_t$ as defined in Def.~\ref{def:experiment} is also a tuple.
We will define the delta function $\delta_{E, E'}$ to be one if the two tuples $E, E'$ are exactly the same, and zero otherwise.
From Eq.~\eqref{eq:Wxyz-hardalt1} to Eq.~\eqref{eq:Wxyz-hardalt2}, the probability of a leaf node for the world model $\mathcal{W}_{E, s}$ is given by
\begin{equation}
    p_{\mathcal{W}_{E, s}}(\ell_{(E_1, \beta_1) \ldots, (E_T, \beta_T)}) = \prod_{t=1}^T \left( \frac{1}{2} \left( 1 + 3s \epsilon \delta_{E_t, E} \sign(\beta_t) \right) p_{(E_1, \beta_1) \ldots, (E_{t-1}, \beta_{t-1})}(E_t) \right),
\end{equation}
where $\sign(\beta_t) = 1$ if $\beta_t = b_0$ and $\sign(\beta_t) = -1$ if $\beta_t = b_1$.

We can now write down the total variation distance between the probability distribution over the leaf nodes under the two events (true world model $= \mathcal{W}_0$ or true world model $\neq \cW_0$).
In the event that the true world model $\neq \cW_0$, we have the true world model is equal to $\mathcal{W}_{E, s}$ for a particular choice of $E = (x, y_1, \ldots, y_L, z), x \in \cX, y_1, \ldots, y_L \in \cY, z \in \cZ, s = \pm 1$ with probability $1/ (2 |\mathcal{X}| |\mathcal{Y}|^L |\mathcal{Z}|)$.
The total variation distance is given by
\begin{align}
    \mathrm{TV} &= \frac{1}{2} \sum_{\ell} \left| p_{\mathcal{W}_0}(\ell) - \E_{E, s} p_{\mathcal{W}_{E, s}}(\ell) \right| = \sum_{\ell: \,\, p_{\mathcal{W}_0}(\ell) \geq \E_{E, s} p_{\mathcal{W}_{E, s}}(\ell)} \left( p_{\mathcal{W}_0}(\ell) - \E_{E, s} p_{\mathcal{W}_{E, s}}(\ell) \right)\\
    &\leq \sum_{\ell: \,\, p_{\mathcal{W}_0}(\ell) \geq \E_{E, s} p_{\mathcal{W}_{E, s}}(\ell)} p_{\mathcal{W}_0}(\ell) \eta \leq \eta,
\end{align}
where the parameter $\eta$ satisfies
\begin{equation}
    \frac{\E_{E, s} p_{\mathcal{W}_{E, s}}(\ell)}{p_{\mathcal{W}_0}(\ell)} \geq 1 - \eta, \forall \, \mbox{leaf node } \ell.
\end{equation}
The expectation $\E_{E, s}$ with $E = (x, y_1, \ldots, y_L, z)$ is a uniform distribution over $x \in \cX, y_1, \ldots, y_L \in \cY, z \in \cZ, s = \pm 1$.
We will now find a parameter $\eta$ that satisfies the above condition.
Given $\ell = \ell_{(E_1, \beta_1) \ldots, (E_T, \beta_T)}$, we have
\begin{align}
    \frac{\E_{E, s} p_{\mathcal{W}_{E, s}}(\ell)}{p_{\mathcal{W}_0}(\ell)} &= \E_{E, s} \prod_{t=1}^T \left( 1 + 3s \epsilon \delta_{E_t, E} \sign(\beta_t) \right) \\
    &= \E_{E, s} \exp\left[\sum_{t=1}^T \log\left( 1 + 3s \epsilon \delta_{E_t, E} \sign(\beta_t) \right)\right] \\
    &\geq \exp\left[\E_{E, s} \sum_{t=1}^T \log\left( 1 + 3s \epsilon \delta_{E_t, E} \sign(\beta_t) \right)\right] \\
    &= \exp\left[\frac{1}{2} \sum_{t=1}^T \E_{E} \log\left( 1 - 9 \epsilon^2 \delta_{E_t, E} \right)\right] \\
    &\geq \exp\left[\frac{1}{2} \sum_{t=1}^T \frac{1}{|\mathcal{X}| |\mathcal{Y}|^L |\mathcal{Z}|} \log\left( 1 - 9 \epsilon^2 \right)\right] \\
    &\geq \exp\left[- \sum_{t=1}^T \frac{1}{|\mathcal{X}| |\mathcal{Y}|^L |\mathcal{Z}|} 9 \epsilon^2 \right] \\
    &\geq 1 - \frac{9T \epsilon^2}{|\mathcal{X}| |\mathcal{Y}|^L |\mathcal{Z}|}.
\end{align}
The third line is Jensen's inequality.
The fourth line uses the fact that $s = \pm 1$ uniformly.
The fifth line uses the fact that $E = (x, y_1, \ldots, y_L, z)$ is distributed uniformly over $x \in \cX, y_1, \ldots, y_L \in \cY, z \in \cZ$.
The second-to-last line uses $\log(1-x) \geq -2x, \forall x \in [0, 0.79]$ which is satisfied given $\epsilon < 0.29$.
The last line uses $\exp(x) \geq 1 + x, \forall x \in \mathbb{R}$.
Together, we can choose $\eta = \frac{9T \epsilon^2}{|\mathcal{X}| |\mathcal{Y}|^L |\mathcal{Z}|}$ and we have established
\begin{equation}
    \Omega(1) \leq \mathrm{TV} \leq \frac{9T \epsilon^2}{|\mathcal{X}| |\mathcal{Y}|^L |\mathcal{Z}|}.
\end{equation}
We have thus proved that the number of experiments must be at least $T = \Omega(|\mathcal{X}| |\mathcal{Y}|^L |\mathcal{Z}| / \epsilon^2)$.

\section{A general theorem for predicting extrinsic behavior}
\label{sec:thm-extrinsic-behav}

In order to avoid the worst-case complexity proved in the previous subsection, we need to make a stronger assumption about the true world model.
We present one such assumption that is closely related to the assumption used in existing gate set tomography protocols \cite{greenbaum2015introduction, blume2017demonstration, nielsen2020gate, brieger2021compressive}.
Intuitively, the assumption is that we can efficiently find a complete set of states that span the set of states we can generate using the world model, and a complete set of POVM elements whose span includes the complete set of states.
In the worst case, such as in the world models we constructed in the previous subsection, we cannot find a complete set of states and POVM elements efficiently.
As a result, we saw that the optimal complexity scales very badly (exponentially in $L$).

Before presenting the assumption, we will recall some basic concepts.
Given a world model $\mathcal{W} = \left(\{\rho_x\}_{x \in \mathcal{X}}, \{\mathcal{E}_y\}_{y \in \mathcal{Y}}, \{\mathcal{M}_z\}_{z \in \mathcal{Z}} \right)$.
We can compose a state $\rho$ and a CPTP map $\mathcal{E}$ to prepare a new state $\mathcal{E}(\rho)$.
Similarly, we can compose a POVM $\cM$ and a CPTP map $\cE$ to compose a new POVM $\cM \circ \cE$ that is equivalent to first applying $\cE$ then measure using $\cM$.
We refer to the states that could be prepared by composing a finite sequence of CPTP maps $\mathcal{E}_y$ acting on some initial state $\rho_x$ as states that could be generated in $\cW$.

\subsection{Assumptions}

We assume that we have found a linearly independent set of composed states $\{ \rho_1, \ldots, \rho_{K_1} \}$ such that $\mathrm{span}(\rho_1, \ldots, \rho_{K_1})$ contains all states that can be generated in $\cW$, where $\mathrm{span}(\ldots)$ means all linear combinations with real coefficients.
This is equivalent to stating that for any state $\rho$ that can be generated in $\cW$, there exists a unique set of coefficients $\alpha_{1}, \ldots, \alpha_{K_1} \in \mathbb{R}$ satisfying
\begin{equation} \label{eq:rho-K1-span}
    \rho = \sum_{k_1 = 1}^{K_1} \alpha_{k_1} \rho_{k_1}.
\end{equation}
Because $\rho$ and $\rho_{k_1}$ for all $k_1 \in \{1, \ldots, K_1\}$ are quantum states, we have $\sum_{k_1} \alpha_{k_1} = 1$.
We consider $R_1 > 0$ to be a constant such that
\begin{equation} \label{eq:rho-K1-span-bound}
\sum_{k_1=1}^{K_1} |\alpha_{k_1}| \leq R_1
\end{equation}
for any quantum state $\rho$ in $\mathrm{span}(\rho_1, \ldots, \rho_{K_1})$.
\begin{remark}
$R_1 < \infty$ because a finite-dimensional quantum state space is compact.
\end{remark}

We also assume that we have found a set of composed POVM elements $\{ M_1, \ldots, M_{K_2} \}$ such that $\mathrm{span}(\rho_1, \ldots, \rho_{K_1}) \subseteq \mathrm{span}(M_1, \ldots, M_{K_2})$.
This assumption implies that for any $\alpha, \alpha' \in \mathbb{R}^{K_1}$, if
\begin{equation}
    \forall k_2 \in \{1, \ldots, K_2\}, \Tr\left(M_{k_2} \sum_{k_1=1}^{K_1} \alpha_{k_1} \rho_{k_1}\right) = \Tr\left(M_{k_2} \sum_{k_1=1}^{K_1} \alpha_{k_1}' \rho_{k_1}\right),
\end{equation}
then $\alpha = \alpha'$.
In particular, we consider $R_2 > 0$ to be a constant such that
\begin{equation} \label{eq:R2bound-MK}
    \norm{\alpha - \alpha'}_1 \leq R_2 \sum_{k_2=1}^{K_2} \left| \Tr\left(M_{k_2} \sum_{k_1=1}^{K_1} \alpha_{k_1} \rho_{k_1}\right) - \Tr\left(M_{k_2} \sum_{k_1=1}^{K_1} \alpha_{k_1}' \rho_{k_1}\right)\right|
\end{equation}
for any $\alpha, \alpha' \in \mathbb{R}^{K_1}$.
\begin{remark}
$R_2 < \infty$ because the homogeneity of Eq.~\eqref{eq:R2bound-MK} enables maximization to find $R_2$ over the compact space $\norm{\alpha - \alpha'}_1 = 1$.
\end{remark}

\subsection{Learning a frame}

We are now ready to present an efficient learning algorithm that can be used to predict the extrinsic behavior of the world model $\mathcal{W}$.
For each $k_2 = 1, \ldots, K_2$, we consider the POVM element $M_{k_2}$ to be a unit vector in a $K_2$-dimensional vector space,
\begin{equation}
    M_{k_2} \rightarrow \hat{e}_{k_2}.
\end{equation}
For each $k_1 = 1, \ldots, K_1$, we consider the state $\rho_{k_1}$ to be a vector $w_{k_1} \in \mathbb{R}^{K_2}$,
\begin{equation}
    \rho_{k_1} \rightarrow w_{k_1} \equiv \sum_{k_2 = 1}^{K_2} w_{k_1, k_2} \hat{e}_{k_2},
\end{equation}
where $w_{k_1, k_2}$ is an estimate for $\Tr(M_{k_2} \rho_{k_1})$, such that
\begin{equation} \label{eq:probw_k1k2}
    \mathrm{Pr}\left[\,\, |w_{k_1, k_2} - \Tr(M_{k_2} \rho_{k_1})| > \tilde{\epsilon}_w \,\, \right] \leq \delta.
\end{equation}
We can obtain an estimate for $w_{k_1, k_2}$ with the above guarantee using $\mathcal{O}(\log(1/\delta) / \tilde{\epsilon}_w^2)$ experiments. We simply prepare $\rho_{k_1}$ and measure the POVM associated to $M_{k_2}$, then compute the average of the indicator function on whether the POVM element $M_{k_2}$ is the outcome.
Because the indicator function of an event is a bounded random variable, Hoeffding's inequality gives us the above rigorous guarantee.

\subsection{Learning states}

For the initial states $\rho_x$ with $x \in \mathcal{X}$, we represent $\rho_x$ as a $K_1$-dimensional real vector
\begin{equation} \label{eq:learned-rhox}
    \rho_x \rightarrow v_x \equiv \sum_{k_1 = 1}^{K_1} v_{x, k_1} \hat{e}_{k_1},
\end{equation}
where $v_x$ is an optimum of the following optimization problem
\begin{equation} \label{eq:optv-x}
    \mathrm{OPT}(v_x) = \min_{\substack{\alpha \in \mathbb{R}^{K_1}, \norm{\alpha}_1 \leq R_1,\\ \sum_{k_1=1}^{K_1} \alpha_{k_1} = 1}} \norm{\sum_{k_1 = 1}^{K_1} \alpha_{k_1} w_{k_1} - \sum_{k_2=1}^{K_2} v'_{x, k_2} \hat{e}_{k_2}}_1,
\end{equation}
and $v'_{x, k_2}$ is an estimate for $\Tr(M_{k_2} \rho_x)$ by doing additional experiments, such that
\begin{equation} \label{eq:probv_xk2}
    \mathrm{Pr}\left[\,\, |v'_{x, k_2} - \Tr(M_{k_2} \rho_{x})| > \tilde{\epsilon}_v \,\, \right] \leq \delta.
\end{equation}
Because the above optimization is a convex optimization, one could solve for $\alpha_x$ efficiently.
The purpose of the optimization problem is to project the vector $\sum_{k_2=1}^{K_2} v'_{x, k_2} \hat{e}_{k_2}$ onto the space spanned by $v_1, \ldots, v_{K_1}$.
For a fixed $\tilde{\epsilon}_v, \delta > 0$, the total number of experiments for learning representations of the states is $\mathcal{O}(K_2 |\mathcal{X}| \log(1 / \delta) / \tilde{\epsilon}^2_v)$.

\subsection{Learning CPTP maps}

For the CPTP maps $\cE_y$ with $y \in \mathcal{Y}$, we represent $\cE_y$ as a $K_1 \times K_1$ real matrix of the form
\begin{equation} \label{eq:learned-Ey}
    \cE_y \rightarrow A_y \equiv \sum_{k_1 = 1}^{K_1} \sum_{k_1' = 1}^{K_1} A_{y, k_1, k_1'} \hat{e}_{k_1} \hat{e}_{k_1'}^T,
\end{equation}
where for each fixed value of $k_1'$, the $K_1$-dimensional vector $A_{y, (\cdot), k_1'}$ is an optimum of the following convex optimization problem
\begin{equation} \label{eq:optA-y}
    \mathrm{OPT}(A_{y, (\cdot), k_1'}) = \min_{\substack{\alpha \in \mathbb{R}^{K_1}, \norm{\alpha}_1 \leq R_1,\\ \sum_{k_1=1}^{K_1} \alpha_{k_1} = 1}} \norm{\sum_{k_1 = 1}^{K_1} \alpha_{k_1} w_{k_1} - \sum_{k_2=1}^{K_2} A'_{y, k_2, k_1'} \hat{e}_{k_2}}_1,
\end{equation}
and $A'_{y, k_2, k_1'}$ is an estimate for $\Tr(M_{k_2} \cE_{y}(\rho_{k_1'}))$ by doing additional experiments, such that
\begin{equation} \label{eq:probA_yk2}
    \mathrm{Pr}\left[\,\, |A'_{y, k_2, k_1'} - \Tr(M_{k_2} \cE_{y}(\rho_{k_1'}))| > \tilde{\epsilon}_A \,\, \right] \leq \delta.
\end{equation}
Similar to before, the purpose of the optimization is to project the vector $\sum_{k_2=1}^{K_2} A'_{y, k_2, k_1'} \hat{e}_{k_2}$ to the space formed by $v_1, \ldots, v_{K_1}$.
For a fixed $\tilde{\epsilon}_A, \delta > 0$, the total number of experiments for learning representations of the CPTP maps is $\mathcal{O}(K_1 K_2 |\mathcal{Y}| \log(1 / \delta) / \tilde{\epsilon}^2_A)$.

\subsection{Learning POVMs}

For the POVM elements $M_{zb}$ with $z \in \mathcal{Z}, b \in \mathcal{B}$, we represent $M_{zb}$ as a $K_1$-dimensional real vector
\begin{equation}
    M_{zb} \rightarrow u_{zb} \equiv \sum_{k_1=1}^{K_1} u_{z, b, k_1} \hat{e}_{k_1},
\end{equation}
where $u_{z, b, k_1}$ is an estimate for $\Tr(M_{zb} \rho_{k_1})$ by doing additional experiments, such that
\begin{equation} \label{eq:probu_zbk1}
    \mathrm{Pr}\left[\,\, |u_{z, b, k_1} - \Tr(M_{z b} \rho_{k_1})| > \tilde{\epsilon}_u \,\, \right] \leq \delta.
\end{equation}
For a fixed $\tilde{\epsilon}_u, \delta > 0$, the total number of experiments for learning representations of the POVMs is $\mathcal{O}(K_1 |\mathcal{Z}| \log(1 / \delta) / \tilde{\epsilon}_u^2)$.
The reason that we don't need an extra factor of $|\cB|$ is because when we measure the POVM $\mathcal{M}_z$, we can simultaneously estimate $c_{z, b, k_1}$ for all $b \in \mathcal{B}$.

\subsection{Prediction procedure and rigorous guarantee}

During the prediction phase, we predict the probability for obtaining an outcome $b \in \mathcal{B}$ after running the experiment $E = (x, y_1, \ldots, y_L, z)$ to be
\begin{equation}
    u_{z b}^T P A_{y_L} P \ldots P A_{y_1} v_x,
\end{equation}
where $P$ is the projection to the convex set $\{ \alpha \in \mathbb{R}^{K_1} | \sum_{k_1} \alpha_{k_1} = 1, \norm{\alpha}_1 \leq R_1 \}$.
We now give a rigorous performance guarantee for this algorithm.

\begin{theorem}[Predicting extrinsic behaviors; Restatement of Theorem~\ref{thm:ext-behavior}] \label{thm:ext-behavor-ditto}
Assume that we have found a complete set of linearly independent states and POVMs.
Using the proposed algorithm, we can predict $\Tr(M_{zb} (\cE_{y_L} \circ \ldots \circ \cE_{y_1})(\rho_x) )$ to $\epsilon$ error for all $x \in \cX, y_1, \ldots, y_L \in \cY, z \in \cZ, b \in \cB$ using a total of
\begin{equation}
\widetilde{\mathcal{O}}\left( \frac{ |\cX| + L^2 |\cY| + |\cZ| }{\epsilon^2} \right)
\end{equation}
experiments, where $\widetilde{\mathcal{O}}(\cdot)$ neglects logarithmic factors and considers $K_1, K_2, R_1, R_2$ to be constant.
\end{theorem}

\subsection{Proof of Theorem~\ref{thm:ext-behavor-ditto} --- Step A. Representations of target outputs}

From Eq.~\eqref{eq:rho-K1-span}~and~\eqref{eq:rho-K1-span-bound}, $\forall x \in \mathcal{X}$, we can write $\rho_x$ as
\begin{equation} \label{eq:rhoxv*}
    \rho_x = \sum_{k_1=1}^{K_1} v^*_{x, k_1} \rho_{k_1},
\end{equation}
for some $K_1$-dimensional vector $v^*_x$ with $\norm{v^*_x}_1 \leq R_1$.
Similarly, Eq.~\eqref{eq:rho-K1-span}~and~\eqref{eq:rho-K1-span-bound} shows that for any $k_1' \in \{1, \ldots, K_1\}$, we can write $\cE_y(\rho_{k_1'})$ as
\begin{equation} \label{eq:EyA*}
    \cE_y(\rho_{k_1'}) = \sum_{k_1=1}^{K_1} A^*_{y, k_1, k_1'} \rho_{k_1},
\end{equation}
for some $K_1$-dimensional vector $A^*_{y, (\cdot), k_1'}$ with $\norm{A^*_{y, (\cdot), k_1'}}_1 \leq R_1$.
Using this representation, for an experiment specified by $x \in \cX, y_1, \ldots, y_L \in \cY, z \in \cZ$, the probability to obtain the measurement outcome $b \in \cB$ can be written as
\begin{equation}
    \Tr(M_{zb} (\cE_{y_L} \circ \ldots \circ \cE_{y_1})(\rho_x) ) = \sum_{k_1 = 1}^{K_1} (A^*_{y_L}\ldots A^*_{y_1} v^*_x)_{k_1} \Tr( M_{zb} \rho_{k_1} ) = (u^*_{zb})^T A^*_{y_L}\ldots A^*_{y_1} v^*_x,
\end{equation}
where we let $u^*_{z, b, k_1} \equiv \Tr( M_{zb} \rho_{k_1} )$ and $u^*_{zb} = \sum_{k_1} u^*_{z, b, k_1} \hat{e}_{k_1}$.

\subsection{Proof of Theorem~\ref{thm:ext-behavor-ditto} --- Step B. Error analysis for the learned representations}

We begin by comparing the two vectors $v_{x}^*$ and $v_{x}$, where $v_x$ is the representation learned from experiments; see Eq.~\eqref{eq:learned-rhox}. We can bound the difference as follows,
\begin{align}
    \norm{v_{x}^* - v_x}_1 &\leq R_2 \sum_{k_2=1}^{K_2} \left|\Tr(M_{k_2} \rho_x) - \sum_{k_1=1}^{K_1} v_{x, k_1} \Tr(M_{k_2} \rho_{k_1})\right| \\
    &\leq R_2 \sum_{k_2=1}^{K_2} \left(\left|v'_{x, k_2} - \sum_{k_1=1}^{K_1} v_{x, k_1} w_{k1, k2}\right| + \tilde{\epsilon}_v + \norm{v_{x}}_1 \tilde{\epsilon}_w\right) \label{eq:R2errorbound} \\
    &= R_2 \mathrm{OPT}(v_x) + R_2 K_2 \left( \tilde{\epsilon}_v + \norm{v_{x}}_1 \tilde{\epsilon}_w \right),\\
    &\leq R_2 K_2 \left( \tilde{\epsilon}_v + \norm{v^*_{x}}_1 \tilde{\epsilon}_w \right) + R_2 K_2 \left( \tilde{\epsilon}_v + \norm{v_{x}}_1 \tilde{\epsilon}_w \right),\\
    &\leq 2 R_2 K_2 \left( \tilde{\epsilon}_v + R_2 \tilde{\epsilon}_w \right),
\end{align}
with probability at least $1 - (K_2 + 1)\delta$.
The first line follows from Eq.~\eqref{eq:R2bound-MK}~and~\eqref{eq:rhoxv*}.
The second line follows from Eq.~\eqref{eq:probw_k1k2},~\eqref{eq:probv_xk2}, and the union bound.
The third line follows from Eq.~\eqref{eq:optv-x}.
The fourth line follows from considering $\alpha = v^*_x$ in the optimization problem given at Eq.~\eqref{eq:optv-x}, and utilizing the following bounds,
\begin{align}
    &\norm{\sum_{k_1 = 1}^{K_1} v^*_{x, k_1} w_{k_1} - \sum_{k_2=1}^{K_2} v'_{x, k_2} \hat{e}_{k_2}}_1\\
    &\leq \sum_{k_2=1}^{K_2} \left|\sum_{k_1 = 1}^{K_1} v^*_{x, k_1} \Tr(M_{k_2} \rho_{k_1}) - \Tr(M_{k_2} \rho_x) \right| + K_2 \left( \tilde{\epsilon}_v + \norm{v^*_{x}}_1 \tilde{\epsilon}_w \right)\\
    &= \sum_{k_2=1}^{K_2} \left|\Tr(M_{k_2} \rho_x) - \Tr(M_{k_2} \rho_x) \right| + K_2 \left( \tilde{\epsilon}_v + \norm{v^*_{x}}_1 \tilde{\epsilon}_w \right) = K_2 \left( \tilde{\epsilon}_v + \norm{v^*_{x}}_1 \tilde{\epsilon}_w \right).
\end{align}
The last line follows from $\norm{v_x^*}_1, \norm{v_x}_1 \leq R_1$.

We can compare $A^*_{y, k_1, k_1'}$ and $A_{y, k_1, k_1'}$ by employing the same analysis and replacing Eq.~\eqref{eq:rhoxv*} with \eqref{eq:EyA*}, Eq.~\eqref{eq:learned-rhox} with \eqref{eq:learned-Ey}, Eq.~\eqref{eq:probv_xk2} with \eqref{eq:probA_yk2}, and Eq.~\eqref{eq:optv-x} with \eqref{eq:optA-y}.
For all $y \in \mathcal{Y}, k_1' \in \{1, \ldots, K_1\}$, the analysis shows that the following,
\begin{equation}
    \norm{A^*_{y, (\cdot), k_1'} - A_{y, (\cdot), k_1'}}_1 \leq 2 R_2 K_2 \left( \tilde{\epsilon}_A + R_2 \tilde{\epsilon}_w \right),
\end{equation}
happens with probability at least $1 - (K_2 + 1)\delta$.
And recall from Eq.~\eqref{eq:probu_zbk1}, we have $|u_{z, b, k_1} - u^*_{z, b, k_1}| \leq \tilde{\epsilon}_u$ with probability at least $1 - \delta$.
Together, with probability at least $1 - |\cX| (K_2+1) \delta - |\cY| K_1 (K_2 + 1) \delta - |\cZ| |\cB| \delta = 1 - ((|\cX| + K_1 |\cY|) (K_2+1) + |\cZ| |\cB|) \delta$, we have
\begin{align}
    \sum_{k_1 = 1}^{K_1} \left| v^*_{x, k_1} - v_{x, k_1} \right| &\leq 2 R_2 K_2 \left( \tilde{\epsilon}_v + R_2 \tilde{\epsilon}_w \right), & \forall x \in \cX, \label{eq:rhoX-behav}\\
    \sum_{k_1 = 1}^{K_1} \left| A^*_{y, k_1, k_1'} - A_{y, k_1, k_1'}\right| &\leq 2 R_2 K_2 \left( \tilde{\epsilon}_A + R_2 \tilde{\epsilon}_w \right), & \forall y \in \cY, \, \forall k_1' \in \{1, \ldots, K_1\}, \label{eq:Ey-behav}\\
    \left|  u^*_{z, b, k_1} - u_{z, b, k_1} \right| &\leq \tilde{\epsilon}_u, & \forall z \in \cZ, b \in \cB, \, \forall k_1 \in \{1, \ldots, K_1\}.\label{eq:Mz-behav}
\end{align}
This provides a set of error bounds for the learned representations.

\subsection{Proof of Theorem~\ref{thm:ext-behavor-ditto} --- Step C. Error analysis for the prediction}

We will now analyze the difference between the prediction and the true answer.
We first define a linear function mapping a $K_1$-dimensional vector $\alpha$ to a matrix
\begin{equation}
    \rho(\alpha) = \sum_{k_1 = 1}^{K_1} \alpha_{k_1} \rho_{k_1}.
\end{equation}
Because $\rho_{k_1}$ is a quantum state, we have
\begin{equation} \label{eq:norm-rho-to-no-rho}
    \norm{\rho(\alpha)}_1 \leq \sum_{k_1 = 1}^{K_1} \left|\alpha_{k_1}\right| \norm{\rho_{k_1}}_1 = \norm{\alpha}_1.
\end{equation}
Note that $\norm{\rho(\alpha)}_1$ is the trace norm for a matrix, while $\norm{\alpha}_1$ is a vector one-norm.
For an experiment specified by $x \in \cX, y_1, \ldots, y_L \in \cY, z \in \cZ$, the difference between the probability for obtaining $b \in \cB$ and the predicted probability is
\begin{align}
    &\left|u_{z b}^T P A_{y_L} P \ldots P A_{y_1} v_x - (u^*_{zb})^T A^*_{y_L}\ldots A^*_{y_1} v^*_x \right|\\
    &\leq \max_{k_1} \left|u_{z, b, k_1} - u^*_{z, b, k_1}\right| \norm{ P A_{y_L} P \ldots P A_{y_1} v_x }_1 \\
    &+ | \Tr\left(M_{z b} \rho(P A_{y_L} P \ldots P A_{y_1} v_x) \right) - \Tr\left( M_{z b} \rho(A^*_{y_L}\ldots A^*_{y_1} v^*_x) \right) |\\
    &\leq \tilde{\epsilon}_u R_1 + \norm{\rho(P A_{y_L} P \ldots P A_{y_1} v_x) - \rho(A^*_{y_L}\ldots A^*_{y_1} v^*_x) }_1.
\end{align}
The first inequality is a telescoping sum with a triangle inequality.
The second inequality follows from Eq.~\eqref{eq:Mz-behav} and the fact that $\norm{M_{zb}}_\infty \leq 1$.

We will now analyze the second term in the above equation. We will prove that
\begin{align}
    &\norm{\rho(P A_{y_L} P \ldots P A_{y_1} v_x) - \rho(A^*_{y_L}\ldots A^*_{y_1} v^*_x) }_1\\ &\leq 2 R_2 K_2 \left( \tilde{\epsilon}_v + R_2 \tilde{\epsilon}_w \right) + 2 L R_1 R_2 K_2 \left( \tilde{\epsilon}_A + R_2 \tilde{\epsilon}_w \right)
\end{align}
by induction on $L$.
For the base case $L = 0,$ from Eq.~\eqref{eq:rhoX-behav} and Eq.~\eqref{eq:norm-rho-to-no-rho}, we see that
\begin{equation}
    \norm{\rho(v_x) - \rho(v^*_x) }_1 \leq \norm{v_x - v^*_x}_1 \leq 2 R_2 K_2 \left( \tilde{\epsilon}_v + R_2 \tilde{\epsilon}_w \right).
\end{equation}
Suppose that the claim holds for $L - 1$. Then
\begin{align}
    &\norm{\rho(P A_{y_L} P \ldots P A_{y_1} v_x) - \rho(A^*_{y_L}\ldots A^*_{y_1} v^*_x) }_1 \\
    &\leq \norm{\rho(P A_{y_L} (P A_{y_{L-1}} \ldots P A_{y_1} v_x)) - \rho(A^*_{y_L} (P A_{y_{L-1}} \ldots P A_{y_1} v_x)) }_1\\
    &+ \norm{\rho(A^*_{y_L} (P A_{y_{L-1}} \ldots P A_{y_1} v_x)) - \rho(A^*_{y_L} (A^*_{y_{L-1}} \ldots A^*_{y_1} v^*_x)) }_1\\
    &\leq \norm{P A_{y_L} (P A_{y_{L-1}} \ldots P A_{y_1} v_x) - A^*_{y_L} (P A_{y_{L-1}} \ldots P A_{y_1} v_x) }_1\\
    &+ \norm{\mathcal{E}_{y_L}\big(\rho(P A_{y_{L-1}} \ldots P A_{y_1} v_x - A^*_{y_{L-1}} \ldots A^*_{y_1} v^*_x ) \big) }_1\\
    &\leq \sum_{k_1 = 1}^{K_1} \left( \max_{k_1'} \left|P A_{y_L, k_1, k_1'} - A^*_{y_L, k_1, k_1'}\right| \sum_{k_1'=1}^{K_1} \left|(P A_{y_{L-1}} \ldots P A_{y_1} v_x)_{k_1'}\right| \right)\\
    &+ \norm{\rho(P A_{y_{L-1}} \ldots P A_{y_1} v_x - A^*_{y_{L-1}} \ldots A^*_{y_1} v^*_x ) }_1\\
    &\leq 2 R_1 R_2 K_2 \left( \tilde{\epsilon}_A + R_2 \tilde{\epsilon}_w \right)\\
    & + 2 R_2 K_2 \left( \tilde{\epsilon}_v+ R_2 \tilde{\epsilon}_w \right) + 2 (L-1) R_1 R_2 K_2 \left( \tilde{\epsilon}_A + R_2 \tilde{\epsilon}_w \right)\\
    &= 2 R_2 K_2 \left( \tilde{\epsilon}_v + R_2 \tilde{\epsilon}_w \right) + 2 L R_1 R_2 K_2 \left( \tilde{\epsilon}_A + R_2 \tilde{\epsilon}_w \right).
\end{align}
The first inequality follows from the triangle inequality and the linearity of $\rho(\alpha)$.
The second inequality follows from the action of $\mathcal{E}_y$ given in Eq.~\eqref{eq:EyA*}.
The third inequality uses two basic inequalities: $\norm{(A - A')x} \leq \sum_{i} (\max_{j} |(A - A')_ij| \sum_j x_j), \forall A, A' \in \mathbb{R}^{k \times k}, x \in \mathbb{R}^k$, and $\norm{\mathcal{E}(X)}_1 \leq \norm{X}_1$ for CPTP map $\mathcal{E}$ and Hermitian matrix $X$.
The fourth inequality uses the induction hypothesis, Eq.~\eqref{eq:Ey-behav}, and the fact that $\sum_{k} | (P x)_k | \leq R_1$ for any $x \in \mathbb{R}^{K_1}$.

\subsection{Proof of Theorem~\ref{thm:ext-behavor-ditto} --- Step D. Putting everything together}

Together, we consider the following parameter choices,
\begin{align}
    \frac{1}{\delta} &= 100 ((|\cX| + K_1 |\cY|) (K_2+1) + |\cZ| |\cB|), \\
    \frac{1}{\tilde{\epsilon}_w} &= \frac{16 L R_1 R_2^2 K_2}{\epsilon}, \\
    \frac{1}{\tilde{\epsilon}_v} &= \frac{8 R_2 K_2}{\epsilon}, \\
    \frac{1}{\tilde{\epsilon}_A} &= \frac{8 L R_1 R_2 K_2}{\epsilon}, \\
    \frac{1}{\tilde{\epsilon}_u} &= \frac{4 R_1}{\epsilon},
\end{align}
to ensure that with probability at least $0.99$, for any experiment specified by $x \in \cX, y_1, \ldots, y_L \in \cY, z \in \cZ$, the difference between the actual probability for obtaining the measurement outcome $b \in \cB$ and the predicted probability is  bounded above by $\epsilon$,
\begin{equation}
\left|u_{z b}^T P A_{y_L} P \ldots P A_{y_1} v_x -  \Tr(M_{zb} (\cE_{y_L} \circ \ldots \circ \cE_{y_1})(\rho_x) ) \right| \leq \epsilon.
\end{equation}
By aggregating the number of experiments for learning a frame, states, maps, and POVMs,
the total number of experiments is of order
\begin{align}
   & \frac{ L^2 R^2_1 R_2^4 K^2_2 + |\cX| R_2^2 K_2^3 + L^2 |\cY| K_1 K^3_2 R^2_1 R^2_2 + |\cZ| R_1^2 K_1 }{\epsilon^2} \log\left( K_2 |\cX| + K_1 K_2 |\cY| + |\cZ| |\cB| \right)\\
   & = \widetilde{\mathcal{O}}\left( \frac{ |\cX| + L^2 |\cY| + |\cZ| }{\epsilon^2} \right),
\end{align}
where $\widetilde{\mathcal{O}}(\cdot)$ neglects logarithmic contributions and considers $K_1, K_2, R_1, R_2$ to be constant.

\newpage
\section{Related work}
\label{sec:related-work}

In this appendix, we present the connection of the theory developed in this work and existing works. We will refer to basic concepts developed in some of the previous appendices. In particular, Appendix~\ref{sec:qworld} on the definition of world models, Appendix~\ref{sec:learning_theory_foundations} on the definition of model classes and learning intrinsic descriptions, and Appendix~\ref{sec:behav-predict} on the definition of predicting extrinsic behaviors.

\subsection{Gate set tomography}

The most relevant literature to the theory of world models developed in this work is gate set tomography.
Here we provide a brief review on gate set tomography.
We refer the readers to two comprehensive reviews on gate set tomography \cite{greenbaum2015introduction, nielsen2020gate}. Ref.~\cite{greenbaum2015introduction} gives the basic concepts and \cite{nielsen2020gate} provides both the fundamental ideas and recent progress on gate set tomography. An experimental demonstration of gate set tomography is given in \cite{blume2017demonstration}. An open-sourced software for gate set tomography has been developed in \cite{erik_2021_5546759}.

The main goal of gate set tomography (GST) is to characterize how quantum processes and logical gates affect the qubits in the device.
This is closely related to the task of quantum process tomography (QPT).
However, in quantum process tomography, one assumes perfect state preparations and perfect POVM measurements.
The key differences between GST and QPT are: (1) the lack of assumption on perfect states and measurements; (2) the need to benchmark multiple quantum processes (gates) at once.
Because GST does not assume the ability to prepare perfect states and measurements, existing protocols learn the relation between different gates instead of the intrinsic physical description of each gate. This problem is referred to as gauge freedom in GST.

When we vectorize state and POVM element as $\mathrm{vec}(\rho_x)$ and $\mathrm{vec}(M_{zb})$ and write the CPTP maps $\cE_y$ as a matrix $A_y$, then we have
\begin{equation}
    \Tr(M_{zb} \cE_y(\rho_x)) = \mathrm{vec}(M_{zb})^T A_y \mathrm{vec}(\rho_x) = (M^{-1} \mathrm{vec}(M_{zb}))^T (M A_y M^{-1}) (M \mathrm{vec}(\rho_x)).
\end{equation}
The following transformation is known as a gauge transformation
\begin{align}
    \mathrm{vec}(\rho_x) &\rightarrow M \mathrm{vec}(\rho_x),\\
    A_y &\rightarrow M A_y M^{-1}, \\
    \mathrm{vec}(M_{zb}) &\rightarrow M^{-1} \mathrm{vec}(M_{zb}),
\end{align}
which is parameterized by an invertible matrix $M$.
Most existing GST protocols are designed to learn the vectorization and matricization up the the gauge freedom, specified by the matrix $M$.
Due to the gauge $M$, it has been difficult to provide a fully rigorous statistical analysis of GST. The key problem is that the gauge freedom makes it challenging to define errors in GST.
And we are not aware of a rigorous proof showing the required number of experiments to perform GST up to certain error.

We begin with a summary of the difference between the theory on learning world models developed in this work and the existing literature on gate set tomography \cite{nielsen2020gate}.
According to the review \cite{nielsen2020gate}, GST is tomography of a novel entity, which is not the individual description of each gate, but some form of relations between the gates.
Based on our theory, we can formalize this novel entity as a non-physical model capable of predicting the extrinsic behaviors of the quantum device under the following assumption: one can (efficiently) find a complete basis of states and POVMs by composing existing states, gates, and measurements.
Without this assumption, Theorem \ref{thm:worst-case-behav} shows that learning the extrinsic behavior can be extremely inefficient. And under this assumption, Theorem \ref{thm:ext-behavior} provides a rigorous algorithm, that shares many common aspects with GST, for predicting extrinsic behaviors that improves significantly upon the worst case complexity given in Theorem \ref{thm:worst-case-behav}.

In contrast, learning intrinsic physical descriptions of the operations in the device is significantly more challenging than performing GST. In certain scenarios, it is impossible to learn the intrinsic descriptions to arbitrary accuracy as shown in Theorem \ref{prop:Clif+T}.
However, we show that it is possible in many natural scenarios.
We can learn the intrinsic description up to any accuracy when the world model has an unknown pure state and a universal set of unknown gates (Theorem \ref{thm:int-description}).
Hence, being ``informationally complete'' for GST is easier than being ``informationally complete'' for learning intrinsic descriptions.

We now present several questions that illustrate the difference between the theory on learning world models developed in this work and gate set tomography.

\paragraph{Is equivalence up to a gauge transformation in GST the same as weak indistinguishability?}
No. When two world models are related by the gauge transformation defined in GST, it is true that the two world models are weakly indistinguishable, i.e., no experiments can distinguish the two world models.
However, there are pairs of weakly indistinguishable world models that are not related by a gauge transformation.
One such example has been presented in the main text. Consider the following two distinct physical realities in a single-qubit system,
\begin{align}
    \mathcal{W}^A_{\mathrm{HT}}: \quad \rho^A_0 &= I/2, & \mathcal{E}^A_H(\rho) &= H\rho H^\dagger, & \mathcal{E}^A_T(\rho) &= T\rho T^\dagger, & \mathcal{M}^A_0 &= \{\ketbra{0}{0}, \ketbra{1}{1}\}, \\
    \mathcal{W}^B_{\mathrm{HT}}: \quad \rho^B_0 &= I/2, & \mathcal{E}^B_H(\rho) &= I/2, & \mathcal{E}^B_T(\rho) &= I/2, & \mathcal{M}^B_0 &= \{\ketbra{0}{0}, \ketbra{1}{1}\},
\end{align}
with actions $\mathcal{X} = \{0\}, \mathcal{Y} = \{H, T\}, \mathcal{Z} = \{0\}$.
Because $\mathcal{E}^A_H \neq \mathcal{E}^A_T$ but $\mathcal{E}^B_H = \mathcal{E}^B_T$, it is impossible that the two world models are related by a gauge transformation.
However, as we have discussed in the main text, the two world models cannot be distinguished by any experiment, hence they are weakly indistinguishable.
Together, we see that gauge equivalence as defined in GST implies weak indistinguishability, but two devices that are weakly indistinguishable need not be gauge equivalent in GST.
This means the set of relations defined by gauge equivalence in GST is a subset of the relations defined by weakly indistinguishable world models.

\paragraph{Can we always learn a world model up to gauge transformation?}
No. While GST removes the assumption on perfect state preparations and measurements, GST still require assumptions to work. It is not true that we can always learn a world model up to a gauge transformation using existing GST protocols. This is already evident in the above example. The two world models $\mathcal{W}^A_{\mathrm{HT}}, \mathcal{W}^B_{\mathrm{HT}}$ are not related by a gauge transformation, but they are weakly indistinguishable.
Hence by Proposition~\ref{prop:indist->unlearn}, we cannot learn to distinguish the two world models from experiments.

\paragraph{Is equivalence up to a gauge transformation in GST the same as strong indistinguishability (i.e. equivalence in our framework)?}
No. As Theorem~\ref{prop:equiv} has shown, equivalent world models are related by unitary or anti-unitary transformation. But two world models related by the gauge transformation defined in GST may not be related by unitary or anti-unitary transformation.
For example, consider two $d$-dimensional world models over the action spaces $\mathcal{X} = \{\sigma\}_{\sigma:\mathrm{state}}, \mathcal{Y} = \{U\}_{U \in \mathrm{SU}(d)}, \mathcal{Z} = \{0\}$ and outcome space $\cB = \{1, \ldots, d\}$:
$\mathcal{W}^A = \left(\{\rho^A_x\}_{x \in \mathcal{X}}, \{\mathcal{E}^A_y\}_{y \in \mathcal{Y}}, \{\mathcal{M}^A_z\}_{z \in \mathcal{Z}} \right)$, where
\begin{align}
    \rho^A_\sigma &= (1 - \epsilon) \sigma + \epsilon \frac{I}{d}, &\forall \sigma: \mathrm{state},\\
    \mathcal{E}^A_U(\rho) &= U \rho U^\dagger, &\forall U \in \mathrm{SU}(d),\\
    \mathcal{M}^A_0 &= \left\{ \ketbra{b}{b} \right\}_{b=1, \ldots, d},&
\end{align}
and $\mathcal{W}^B = \left(\{\rho^B_x\}_{x \in \mathcal{X}}, \{\mathcal{E}^B_y\}_{y \in \mathcal{Y}}, \{\mathcal{M}^B_z\}_{z \in \mathcal{Z}} \right)$, where
\begin{align}
    \rho^B_\sigma &= \sigma, &\forall \sigma: \mathrm{state},\\
    \mathcal{E}^B_U(\rho) &= U \rho U^\dagger, &\forall U \in \mathrm{SU}(d),\\
    \mathcal{M}^B_0 &= \left\{ (1-\epsilon)\ketbra{b}{b} + \epsilon \frac{I}{d} \right\}_{b=1, \ldots, d},&
\end{align}
We can see that $\mathcal{W}^A$ has depolarized initial states, while $\mathcal{W}^B$ has depolarized measurements.
The two world models are not related by a unitary or anti-unitary transformation, hence they describe distinct physical realities.
However, it is not hard to show that the two world models are related by a gauge transformation $M$.

\paragraph{Are two gauge-equivalent devices always physically the same?}
No. We can have two noisy devices with depolarization noise happening in the states or in the measurements; see the above example of $\mathcal{W}^A$ and $\mathcal{W}^B$. The two devices are gauge equivalent. But these two noise processes are not physically the same.
And hence the two devices are not physically the same.
The inability to distinguish the two world models $\mathcal{W}^A$ and $\mathcal{W}^B$ is due to the lack of sufficiently informative actions.
This is similar to the case when one performs quantum state tomography with computational basis measurements.
A coherent $\ketbra{+}{+}$ state is indistinguishable from a maximally mixed state $I/2$ under computational basis measurements. However, the two states are not physically the same.
The indistinguishability arises not from being physically the same, but from the lack of sufficiently useful actions.

\paragraph{How does gauge equivalence defined in GST relate to our theory of world models?}
Suppose that we have a $d$-dimensional world model $\mathcal{W}$, and assume we can identify a set of (composed) states and a set of (composed) POVM elements, such that both sets span the entire $d$-dimensional quantum state space.
Then every world model $\mathcal{W}'$ that is weakly indistinguishable from $\mathcal{W}$ is related by a gauge transformation.
Hence, under this assumption, we can equate weak indistinguishability in our framework and gauge equivalence in GST.
The assumption is critical and has been explicitly or implicitly assumed in existing GST protocols \cite{nielsen2020gate}.
A similar assumption is made in Theorem~\ref{thm:ext-behavior} (restated in Theorem~\ref{thm:ext-behavor-ditto}) to establish the existence of an efficient algorithm for learning extrinsic behaviors.

\subsection{Quantum state/process/measurement tomography}

The settings studied in quantum state, process, and measurement tomography are special cases of our theory.
In quantum state tomography, we are learning a world model such that the CPTP maps and POVM measurements are perfect.
Quantum state tomography focuses on learning an unknown quantum state $\rho$ to high accuracy, usually in trace norm or fidelity --- see e.g.~\cite{hradil1997quantum,gross2010quantum,blume2010optimal,banaszek2013focus} and references therein.
Recently, there have also been works investigating how learning theory can be used to provide new insights into quantum state tomography \cite{aaronson2018online, quek2021private, buadescu2021improved, huang2020predicting}.
In quantum measurement tomography \cite{d2003quantum}, we are learning a world model such that the state preparations and the the CPTP maps are perfect.
The goal of quantum measurement tomography is to learn the descriptions of the POVMs.
Finally, in quantum process tomography \cite{mohseni2008quantum}, we are learning a world model such that the state preparations and the POVM measurements are perfect.
The purpose of quantum process tomography is to learn the full description of some unknown quantum processes.
A subset of quantum processes that have been actively studied recently are the Pauli channels \cite{harper2020efficient, flammia2020efficient, flammia2021pauli}, which are often considered to be a good description of the noise in quantum computers.

\newpage
\section{Quantum advantage with noisy devices that cannot be fully learned}
\label{sec:qadv-learn-incomplete}

\subsection{Setting}
\label{sec:setting-super-noisy-device}

We are given an $(2n)$-qubit device with unknown operations.

\subsubsection{State preparation and measurements}

We can prepare an unknown product state given by
\begin{equation}
    \rho_1 \otimes \ldots \otimes \rho_{2n}.
\end{equation}
We can perform an unknown product measurement given by a POVM with $2^{2n}$ outcomes,
\begin{equation}
    \left\{ \bigotimes_{i=1}^{2n} M^{(i)}_{B_i} \right\}_{B \in \{0, 1\}^{2n}},
\end{equation}
where $M^{(i)}_{B_i}, \forall i = 1, \ldots, 2n$ are $2 \times 2$ positive-semidefinite matrices.

\subsubsection{Operations}

In the device, we can apply layers of non-overlapping single- and two-qubit gates.
We consider the $2n$ qubits to be represented by an $n \times 2$ grid.
The grid coordinate of the $i$-th qubit ($i$ from $1$ to $n$) is $(\lceil i/2 \rceil, i \,\, \mathrm{mod} \,\, 2)$.
Single-qubit gates can be applied at every grid point.
But two-qubit gates can only be applied at edges on the $n \times 2$ grid.
Each single-qubit gate is an unknown CPTP map on the corresponding single qubit.
And each two-qubit gate is an unknown CPTP map on the corresponding two qubits.
The unknown CPTP map implemented by each gate may depend on the presence or absence of all the other gates applied at the same layer.

Finally, we assume that we can load some unknown $n$-qubit quantum state $\rho$ on the $n$ qubits in the left hand side of the $n \times 2$ grid.
Suppose $\sigma$ is the quantum state on the $2n$ qubits.
After loading $\rho$, the state on the $2n$ qubits becomes $\rho \otimes \Tr_{\mathrm{left\,\, side}}(\sigma)$.

\subsection{Algorithm for partially learning the device}
\label{sec:algorithm-partial-learn}

In this section, we will partially learn some of the operations in the unknown quantum device. We will obtain some descriptions that are useful for using the quantum device to achieve advantage in performing entangled data analysis.
The descriptions do not fully characterize the device.
But the lack of full characterization enables us to work under a more general setting.

\subsubsection{Experiments and loss functions: Single-qubit}

We perform an optimization to find nine single-qubit gates for every qubit, denoted as $g_{i, k}$ for all $i = 1,\ldots, 2n, k = 0, \ldots, 8$, by minimizing a loss function.
In particular, after choosing $g_{i, k}$, we conduct the following experiments.

We consider the same single-qubit gates on the left hand side and the right hand side of the $n\times 2$ qubit grid.
For each of $k_1, k_1' \in \{0, \ldots, 5\}, k_2, k_2' \in \{6, 7, 8\}$, we prepare the unknown product state $\rho_1 \otimes \ldots \otimes \rho_{2n}$, apply a layer of single-qubit gates, followed by another layer of single-qubit gates, then measure using the unknown product measurement.
In the first layer of the single-qubit gates, we apply $g_{i, k_1}$ for qubit $i$ on the left and apply $g_{i, k_1'}$ for qubit $i$ on the right.
In the second layer of the single-qubit gates, we apply $g_{i, k_2}$ for qubit $i$ on the left and apply $g_{i, k_2'}$ for qubit $i$ on the right.
For each $k_1, k_1', k_2, k_2'$ and each qubit $i$, we estimate the probability for obtaining the outcome $b \in \{0, 1\}$, denoted as $\hat{p}_{k_1, k_1', k_2, k_2', i, b}$.
From a total of $\mathcal{O}(\log(n / \delta) / \eta_0^2)$ such experiments, the estimate $\hat{p}_{k_1, k_1', k_2, k_2', i, b}$ is equal to the true probability up to $\eta_0$ error for all $k_1, k_1', k_2, k_2', i, b$ with a probability at least $1 - \delta$.
We now define a loss function based on the estimated values,
\begin{equation}
    \eta_1 = \max_{k_1, k_1', k_2, k_2', i, b} \left| \hat{p}_{k_1, k_1', k_2, k_2', i, b} - f(k_1, k_1', k_2, k_2', i, b) \right|.
\end{equation}
The function $f(k_1, k_1', k_2, k_2', i, b)$ is defined as follows.
For qubit $i$ on the left, we have
\begin{equation}
    f(k_1, k_1', k_2, k_2', i, b) = \begin{cases}
    1/2 &\lfloor k_1 / 2\rfloor \neq (k_2 - 6)\\
    1 & \lfloor k_1 / 2\rfloor = (k_2 - 6), \,\, k_1 \equiv b \, (\mathrm{mod} 2)\\
    0 & \lfloor k_1 / 2\rfloor = (k_2 - 6), \,\, k_1 \not\equiv b \, (\mathrm{mod} 2)\\
    \end{cases}
\end{equation}
For qubit $i$ on the right, we have
\begin{equation}
    f(k_1, k_1', k_2, k_2', i, b) = \begin{cases}
    1/2 &\lfloor k_1' / 2\rfloor \neq (k_2' - 6)\\
    1 & \lfloor k_1' / 2\rfloor = (k_2' - 6), \,\, k_1' \equiv b \, (\mathrm{mod} 2)\\
    0 & \lfloor k_1' / 2\rfloor = (k_2' - 6), \,\, k_1' \not\equiv b \, (\mathrm{mod} 2)\\
    \end{cases}
\end{equation}
We optimize over the selection of gates such that the loss function $\eta_1$ is as small as possible.

\subsubsection{Rigorous guarantee: Single-qubit}

Given an estimation error $\eta_0$ and the loss function $\eta_1$, we can approximately learn the following states and measurements.
We combine the first layer of single-qubit gates determined by $k_1, k_1'$ and the unknown product state to create a new set of product states, denoted as
\begin{equation}
    \rho^{(k_1, k_1')}_1 \otimes \ldots \otimes \rho^{(k_1, k_1')}_{2n},
\end{equation}
for all $k_1, k_1' \in \{0, 1, 2, 3, 4, 5\}$.
We also combine the last layer of single-qubit gates determined by $k_2, k_2'$ and the unknown product measurement to create a new set of product measurements,
\begin{equation}
    \left\{ \bigotimes_{i=1}^{2n} M^{(i, k_2, k_2')}_{B_i} \right\}_{B \in \{0, 1\}^{2n}},
\end{equation}
for all $k_2, k_2' \in \{6, 7, 8\}$.
We have the following characterization.

\begin{lemma}[Single-qubit stabilizer states] \label{lem:learn-stab}
For any $i = 1, \ldots, 2n$, there exists a unitary or anti-unitary transformation $U_i$ on qubit $i$, such that the following holds for any $k_1, k_1', k_2, k_2', b$.
For qubit $i$ on the left, we have
\begin{align}
    \norm{\rho_i^{(k_1, k_1')} - U_{i} \sigma_{k_1} U_{i}^{-1}}_1 \leq \mathcal{O}(\eta_0 + \eta_1),\\
    \norm{M^{(i, k_2, k_2')}_{b} - U_{i} \sigma_{2 (k_2-6) + b} U_{i}^{-1}}_1 \leq \mathcal{O}(\eta_0 + \eta_1).
\end{align}
For qubit $i$ on the right, we have
\begin{align}
    \norm{\rho_i^{(k_1, k_1')} - U_{i} \sigma_{k_1'} U_{i}^{-1}}_1 \leq \mathcal{O}(\eta_0 + \eta_1),\\
    \norm{M^{(i, k_2, k_2')}_{b} - U_{i} \sigma_{2 (k_2' - 6) + b} U_{i}^{-1}}_1 \leq \mathcal{O}(\eta_0 + \eta_1),
\end{align}
where the pure states $\sigma_x$ for $x = 0, \ldots, 5$ are given by
\begin{equation}
\ketbra{0}{0}, \,\, \ketbra{1}{1}, \,\,\ketbra{+}{+}, \,\,\ketbra{-}{-}, \,\,\ketbra{y+}{y+}, \,\,\ketbra{y-}{y-}.
\end{equation}
\end{lemma}
\begin{proof}
Without loss of generality, we consider qubit $i$ on the left hand side of the $n \times 2$ grid. The proof is the same for qubits on the right side except that we need to swap $k_1 \leftrightarrow k_1', k_2 \leftrightarrow k_2'$.
From the estimate $\hat{p}_{k_1, k_1', k_2, k_2', i, b}$ and the definition of the loss function $\eta_1$,
\begin{equation} \label{eq:error-geometry}
    \left|\Tr\left(M^{(i, k_2, k_2')}_{b} \rho_i^{(k_1, k_1')} \right) - f(k_1, k_1', k_2, k_2', i, b) \right| \leq \eta_0 + \eta_1,
\end{equation}
for all $k_1, k_1', k_2, k_2', b$. The above defines an approximate geometry that can be used to infer the underlying operations. This technique is used in Appendix~\ref{sec:learning-special-set} to learn general $d$-dimensional systems.

For a fixed choice of $k_1', k_2'$, we have six matrices for $\rho_i^{(k_1, k_1')}$ and six matrices for $M^{(i, k_2, k_2')}_{b}$ corresponding to different $k_1, k_2$ and $b$.
Now, given $q = 0, 1, 2$.
We denote the two matrices $\rho_i^{(k_1, k_1')}$ associated to $k_1 = 2q + 0, 2q + 1$ as $\rho_0, \rho_1$, denote the two matrices $M^{(i, k_2, k_2')}_{b}$ associated to $k_2 = q$ and $b = 0, 1$ as $M_0, M_1$, and define $\eta = \eta_0 + \eta_1$.
We have the following inequalities from Eq.~\eqref{eq:error-geometry},
\begin{equation}
    \Tr(M_0 \rho_0) \geq 1 - \eta, \,\,\, \Tr(M_0 \rho_1) \leq \eta, \,\,\, \Tr(M_1 \rho_0) \leq \eta, \,\,\, \Tr(M_1 \rho_1) \geq 1- \eta.
\end{equation}
Using the fact that $M_0, M_1, \rho_0, \rho_1$ are positive-semidefinite, $M_0 + M_1 = I,$ and  $\Tr(\rho_0) = \Tr(\rho_1) = 1$, there exists two orthogonal pure states $\ketbra{\psi_0}{\psi_0}, \ketbra{\psi_1}{\psi_1}$ and a constant $C > 0$, such that
\begin{align}
    \norm{M_0 - \ketbra{\psi_0}{\psi_0}}_1 &\leq C \eta, & \,\, \norm{M_1 - \ketbra{\psi_1}{\psi_1}}_1 &\leq C \eta, \\
    \norm{\rho_0 - \ketbra{\psi_0}{\psi_0}}_1 &\leq C \eta, & \,\, \norm{\rho_1 - \ketbra{\psi_1}{\psi_1}}_1 &\leq C \eta.
\end{align}
Hence, we know that the six matrices for $\rho_i^{(k_1, k_1')}$ are approximately pure states.
Then we can use the approximate geometry given in Eq.~\eqref{eq:error-geometry} over pairs of distinct $q$ to show that there exists a constant $C' > 0$ and a unitary or anti-unitary transformation $U_{i, k_1', k_2'}$ such that
\begin{align}
    \norm{\rho_i^{(k_1, k_1')} - U_{i, k_1', k_2'} \sigma_{k_1} U_{i, k_1', k_2'}^{-1}}_1 \leq C' \eta, \label{eq:rho-same-gauge}\\
    \norm{M^{(i, k_2, k_2')}_{b} - U_{i, k_1', k_2'} \sigma_{2 (k_2-6) + b} U_{i, k_1', k_2'}^{-1}}_1 \leq C' \eta, \label{eq:M-same-gauge}
\end{align}
where the pure states $\sigma_{x}$ for $x = 0, \ldots, 5$ are given by
\begin{equation}
    \sigma_{x} = \ketbra{0}{0}, \ketbra{1}{1}, \ketbra{+}{+}, \ketbra{-}{-}, \ketbra{y+}{y+}, \ketbra{y-}{y-},
\end{equation}
respectively.

We see that $U_{i, k_1', k_2'}$ depends on $k_1', k_2'$. The last step is to show that we can actually choose a single unitary or anti-unitary transformation $U_{i}$.
Consider $U_i = U_{i, k_1' = 0, k_2' = 6}$. For all $x = 0, \ldots, 5$, we have
\begin{align}
    &\norm{U_{i, k_1', k_2'} \sigma_{x} U_{i, k_1', k_2'}^{-1} - U_i \sigma_x U_i^{-1}}_1\\
    &\leq \norm{U_{i, k_1', k_2'} \sigma_{x} U_{i, k_1', k_2'}^{-1} - U_{i,0, k_2'} \sigma_x U_{i, 0, k_2'}^{-1}}_1 + \norm{U_{i, 0, k_2'} \sigma_{x} U_{i, 0, k_2'}^{-1} - U_{i, 0, 6} \sigma_x U_{i, 0, 6}^{-1}}_1\\
    &\leq \norm{U_{i, k_1', k_2'} \sigma_{x} U_{i, k_1', k_2'}^{-1} - \rho_i^{(x, k_1')}}_1 + \norm{\rho_i^{(x, k_1')} - U_{i,0, k_2'} \sigma_x U_{i, 0, k_2'}^{-1}}_1\\
    &+ \norm{U_{i, 0, k_2'} \sigma_{x} U_{i, 0, k_2'}^{-1} - M^{(i, k_2, \lfloor x/2\rfloor + 6)}_{x \, \mathrm{mod}\, 2}}_1 + \norm{M^{(i, k_2, \lfloor x/2\rfloor + 6)}_{x \, \mathrm{mod}\, 2} - U_{i, 0, 6} \sigma_x U_{i, 0, 6}^{-1}}_1\\
    &\leq 2 C' \eta + 2 C' \eta = 4 C'\eta.
\end{align}
The second-to-last inequality follows from Eq.~\eqref{eq:rho-same-gauge}~and~\eqref{eq:M-same-gauge}.
Together, we have
\begin{align}
    \norm{\rho_i^{(k_1, k_1')} - U_{i} \sigma_{k_1} U_{i}^{-1}}_1 \leq \mathcal{O}(\eta_0 + \eta_1),\\
    \norm{M^{(i, k_2, k_2')}_{b} - U_{i} \sigma_{2 (k_2-6) + b} U_{i}^{-1}}_1 \leq \mathcal{O}(\eta_0 + \eta_1),
\end{align}
This concludes the proof.
\end{proof}

\subsubsection{Experiments and loss functions: Two-qubit}

We also perform another layer of optimization to find two two-qubit gates for every pair of corresponding qubits on the left and right, i.e., qubit $i = 2 \ell - 1, 2\ell$, where $\ell = 1, \ldots, n$.
We denote the two two-qubit gates as $g^{(2)}_{\ell, s}$ and $g^{(2)}_{\ell, e}$ for each $\ell = 0, \ldots, n$.
Similar to before, after choosing $g^{(2)}_{\ell, s}$ and $g^{(2)}_{\ell, e}$, we conduct the following set of experiments.

For each of $k_1, k_1' \in \{0, \ldots, 5\}, k_2, k_2' \in \{6, 7, 8\}, x \in \{s, e\}$, we prepare the unknown product state $\rho_1 \otimes \ldots \otimes \rho_{2n}$, apply a layer of single-qubit gates, a layer of two-qubit gates, followed by another layer of single-qubit gates, then measure using the unknown product measurement.
In the first layer of the single-qubit gates, we apply $g_{i, k_1}$ for qubit $i$ on the left and apply $g_{i, k_1'}$ for qubit $i$ on the right.
In the middle layer of two-qubit gates, we apply $g^{(2)}_{\ell, x}$ on qubit $2\ell-1$ on the left and qubit $2\ell$ on the right.
In the other layer of the single-qubit gates, we apply $g_{i, k_2}$ for qubit $i$ on the left and apply $g_{i, k_2'}$ for qubit $i$ on the right.
For each $k_1, k_1', k_2, k_2'$ and each pair of qubits $2\ell - 1, 2\ell$, we estimate the probability for obtaining the two bits $b, b' \in \{0, 1\}$ as the outcome, denoted as $\hat{p}^{(2)}_{k_1, k_1', k_2, k_2', x, \ell, b, b'}$.
From a total of $\mathcal{O}(\log(n / \delta) / \eta_0^2)$, the estimate $\hat{p}^{(2)}_{k_1, k_1', k_2, k_2', x, \ell, b, b'}$ is equal to the true probability up to $\eta_0$ error for all $k_1, k_1', k_2, k_2', x, \ell, b, b'$ with a probability at least $1 - \delta$.
We now define another loss function based on the estimated values,
\begin{equation} \label{eq:eta2loss}
    \eta_2 = \max_{k_1, k_1', k_2, k_2', x, \ell, b, b'} \left| \hat{p}^{(2)}_{k_1, k_1', k_2, k_2', x, \ell, b, b'} - h(k_1, k_1', k_2, k_2', x, b, b') \right|.
\end{equation}
The function $h(k_1, k_1', k_2, k_2', x, b, b')$ is defined as follows.
\begin{align}
    h(k_1, k_1', k_2, k_2', x, b, b') &= \Tr\left( (\rho^{\mathrm{(out)}}_{k_2, b} \otimes \rho^{\mathrm{(out)}}_{k_2', b'}) U_x (\rho^{\mathrm{(in)}}_{k_1} \otimes \rho^{\mathrm{(in)}}_{k_1'}) U_x^\dagger \right),\\
    \rho^{\mathrm{(in)}}_{k_1} &= \sigma_{k_1},\\
    \rho^{\mathrm{(out)}}_{k_2, b} &= \sigma_{2 (k_2-6) + b},\\
    \sigma_{0}, \ldots, \sigma_5 &= \ketbra{0}{0}, \ketbra{1}{1}, \ketbra{+}{+}, \ketbra{-}{-}, \ketbra{y+}{y+}, \ketbra{y-}{y-},\\
    U_s, U_e &= \mathrm{SWAP}, \mathrm{BELL},
\end{align}
where $\mathrm{SWAP}$ and $\mathrm{BELL}$ are two-qubit unitaries, $\mathrm{SWAP}$ swaps the left and right qubits, $\mathrm{BELL} = (H \otimes I)\,\, \mathrm{CNOT}$, $\mathrm{CNOT}$ is controlled on the same qubit that the Hadamard $H$ acts on.
The entangling operation $\mathrm{BELL}$ followed by computational basis measurement is one way to perform Bell measurement.

\subsubsection{Rigorous guarantee: Two-qubit}

Next, using the bound on the estimation error $\eta_0$ and the loss function $\eta_2$, we can approximately learn the entangling operations.
We can write the middle layer of two-qubit gates as
\begin{equation}
    \cE_{\ell = 1, x} \otimes \ldots \otimes \cE_{\ell = n, x},
\end{equation}
for $x \in \{s, e\}$, where $\cE_{\ell, x}$ acts on two qubits, $2\ell-1$ and $2\ell$.

\begin{lemma}[SWAP and Bell measurement]
For each qubit $i = 1, \ldots, 2n$, consider the unitary or anti-unitary transformation $U_i$ on qubit $i$ given in Lemma~\ref{lem:learn-stab}.
\begin{align}
&\norm{\cE_{\ell, s}(\cdot) - (U_{2\ell-1} \otimes U_{2\ell}) \mathrm{SWAP} (U_{2\ell-1} \otimes U_{2\ell})^{-1} (\cdot) (U_{2\ell-1} \otimes U_{2\ell}) \mathrm{SWAP} (U_{2\ell-1} \otimes U_{2\ell})^{-1} }_{1\rightarrow 1}\\
&\leq \mathcal{O}(\eta_0 + \eta_1 + \eta_2),\\
&\norm{\cE_{\ell, e}(\cdot) - (U_{2\ell-1} \otimes U_{2\ell}) \mathrm{BELL} (U_{2\ell-1} \otimes U_{2\ell})^{-1} (\cdot) (U_{2\ell-1} \otimes U_{2\ell}) \mathrm{BELL}^\dagger (U_{2\ell-1} \otimes U_{2\ell})^{-1} }_{1\rightarrow 1}\\
&\leq \mathcal{O}(\eta_0 + \eta_1 + \eta_2),
\end{align}
where $\norm{\cE}_{1\rightarrow 1} = \max_{\norm{X}_1 = 1} \norm{\cE(X)}_1$.
\end{lemma}
\begin{proof}
We prove this lemma for $\mathrm{SWAP}$. The proof for $\mathrm{BELL}$ is basically the same.
We begin by defining the following notations.
\begin{align}
    \rho^{(in)} &=  (\sigma_{k_1} \otimes \sigma_{k_1')}), \\
    \rho^{(out)} &= (\sigma_{2(k_2 - 6) + b} \otimes \sigma_{2(k_2' - 6) + b'}), \\
    \rho^{(in, U)} &= (U_{2\ell-1} \sigma_{k_1} U_{2\ell-1}^{-1}) \otimes (U_{2\ell} \sigma_{k_1'} U_{2\ell}^{-1}),\\
    \rho^{(out, U)} &= (U_{2\ell-1} \sigma_{2(k_2 - 6) + b} U_{2\ell-1}^{-1}) \otimes (U_{2\ell} \sigma_{2(k_2' - 6) + b'} U_{2\ell}^{-1}),\\
    \sigma_{0}, \ldots, \sigma_5 &= \ketbra{0}{0}, \ketbra{1}{1}, \ketbra{+}{+}, \ketbra{-}{-}, \ketbra{y+}{y+}, \ketbra{y-}{y-}.
\end{align}
We can use triangle inequality in Eq.~\eqref{eq:eta2loss} to show that
\begin{align}
    &\left|\Tr\left( \left( M^{2\ell-1, k_2, k_2'}_{b} \otimes M^{2\ell, k_2, k_2'}_{b'} \right) \cE_{\ell, s}\left( \rho^{k_1, k_1'}_{2\ell-1} \otimes \rho^{k_1, k_1'}_{2\ell} \right) \right) - \Tr\left( \rho^{(out)} \mathrm{SWAP}\left( \rho^{(in)} \right) \mathrm{SWAP} \right) \right|\\
    &\leq \left|\Tr\left( \left( M^{2\ell-1, k_2, k_2'}_{b} \otimes M^{2\ell, k_2, k_2'}_{b'} \right) \cE_{\ell, s}\left( \rho^{k_1, k_1'}_{2\ell-1} \otimes \rho^{k_1, k_1'}_{2\ell} \right) \right) - \hat{p}^{(2)}_{k_1, k_1', k_2, k_2', x, \ell, b, b'} \right|\\
    &+ \left| \hat{p}^{(2)}_{k_1, k_1', k_2, k_2', x, \ell, b, b'} - \Tr\left( \rho^{(out)} \mathrm{SWAP}\left( \rho^{(in)} \right) \mathrm{SWAP} \right) \right|\\
    &\leq \eta_0 + \eta_2.
\end{align}
Then from triangle inequality and Lemma~\ref{lem:learn-stab}, we have
\begin{align}
    &\left|\Tr\left( \left( M^{2\ell-1, k_2, k_2'}_{b} \otimes M^{2\ell, k_2, k_2'}_{b'} \right) \cE_{\ell, s}\left( \rho^{k_1, k_1'}_{2\ell-1} \otimes \rho^{k_1, k_1'}_{2\ell} \right) \right) - \Tr\left( \rho^{(out, U)} \cE_{\ell, s}\left( \rho^{(in, U)} \right) \right) \right|\\
    &\leq \left|\Tr\left( \left( M^{2\ell-1, k_2, k_2'}_{b} \otimes M^{2\ell, k_2, k_2'}_{b'} \right) \cE_{\ell, s}\left( \rho^{k_1, k_1'}_{2\ell-1} \otimes \rho^{k_1, k_1'}_{2\ell} \right) \right) - \Tr\left( \rho^{(out, U)}  \cE_{\ell, s}\left( \rho^{k_1, k_1'}_{2\ell-1} \otimes \rho^{k_1, k_1'}_{2\ell} \right) \right) \right| \\
    &+ \left|\Tr\left( \rho^{(out, U)} \cE_{\ell, s}\left( \rho^{k_1, k_1'}_{2\ell-1} \otimes \rho^{k_1, k_1'}_{2\ell} \right) \right) - \Tr\left( \rho^{(out, U)} \cE_{\ell, s}\left( \rho^{(in, U)} \right) \right) \right| \\
    &\leq \norm{\left( M^{2\ell-1, k_2, k_2'}_{b} \otimes M^{2\ell, k_2, k_2'}_{b'} \right) - \rho^{(out, U)}}_1 + \norm{\rho^{k_1, k_1'}_{2\ell-1} \otimes \rho^{k_1, k_1'}_{2\ell} - \rho^{(in, U)}}_1\\
    &\leq \mathcal{O}(\eta_0 + \eta_1).
\end{align}
Hence, by another triangle inequality, we obtain
\begin{align}
   &\left|\Tr\left( \rho^{(out, U)} \cE_{\ell, s}\left( \rho^{(in, U)} \right) \right) - \Tr\left( \rho^{(out)} \mathrm{SWAP}\left( \rho^{(in)} \right) \mathrm{SWAP} \right) \right|  \leq \mathcal{O}(\eta_0 + \eta_1 + \eta_2). \label{eq:rhoinout-diff}
\end{align}
Because every $(2\times 2) \times (2\times 2)$ Hermitian matrix can be written as a linear combination of $\sigma_{i} \otimes \sigma_j, \forall i, j \in \{0, 1, \ldots, 5\}$, we have the following bound by triangle inequality,
\begin{align}
&\norm{\cE_{\ell, s}(\cdot) - (U_{2\ell-1} \otimes U_{2\ell}) \mathrm{SWAP} (U_{2\ell-1} \otimes U_{2\ell})^{-1} (\cdot) (U_{2\ell-1} \otimes U_{2\ell}) \mathrm{SWAP} (U_{2\ell-1} \otimes U_{2\ell})^{-1} }_{1 \rightarrow 1}\\
&\leq \mathcal{O}(\eta_0 + \eta_1 + \eta_2).
\end{align}
This concludes the proof.
\end{proof}

\begin{remark}
If $\eta_0 + \eta_1 + \eta_2$ is smaller than some constant, then $U_{2\ell - 1}$ and $U_{2\ell}$ must be both unitary or both anti-unitary, which can be shown by a proof by contradiction.
\end{remark}

\subsubsection{Putting everything together}

When the quantum device is designed perfectly, the loss functions vanish, i.e., $\eta_1 = \eta_2 = 0$.
However, when the quantum device is subject to some unknown noise, the loss functions $\eta_1, \eta_2$ will be small but non-zero.
Here, we show that when we perform sufficiently many experiments to estimate the loss functions and we find the loss functions to be small, then we can guarantee that the underlying physical operations satisfy a certain form.

From now on, we assume that $\eta_0 + \eta_1 + \eta_2$ are small enough such that there exists $\epsilon > 0$ and a set of unitary or anti-unitary transformations $U_i, \forall i = 1, \ldots, 2n$ satisfying the following constraints.
For all $\ell = 1, \ldots, n$, $k_1, k_1' \in \{0, \ldots, 5\}$, and $k_2, k_2' \in \{6, 7, 8\}$, we have
\begin{align}
    \norm{\rho_{2\ell-1}^{(k_1, k_1')} \otimes \rho_{2\ell}^{(k_1, k_1')} - (U_{2\ell-1} \sigma_{k_1} U_{2\ell}^{-1}) \otimes (U_{2\ell} \sigma_{k_1'} U_{2\ell}^{-1})}_1 &\leq \epsilon, \label{eq:state-left} \\
    \sum_{b, b' \in \{0, 1\}} \norm{M^{(2\ell-1, k_2, k_2')}_{b} \otimes M^{(2\ell, k_2, k_2')}_{b'} - (U_{2\ell-1} \sigma_{2 (k_2-6) + b} U_{2\ell}^{-1}) \otimes (U_{2\ell} \sigma_{2 (k_2'-6) + b'} U_{2\ell}^{-1})}_1 &\leq \epsilon, \label{eq:povm-left}
\end{align}
where the pure states $\sigma_x$ for $x = 0, \ldots, 5$ are given by
\begin{equation}
\ketbra{0}{0}, \,\, \ketbra{1}{1}, \,\,\ketbra{+}{+}, \,\,\ketbra{-}{-}, \,\,\ketbra{y+}{y+}, \,\,\ketbra{y-}{y-}.
\end{equation}
And we also have
\begin{align}
&\norm{\cE_{\ell, s}- (U_{2\ell-1} \otimes U_{2\ell}) \mathrm{SWAP} (U_{2\ell-1} \otimes U_{2\ell})^{-1} (\cdot) (U_{2\ell-1} \otimes U_{2\ell}) \mathrm{SWAP} (U_{2\ell-1} \otimes U_{2\ell})^{-1} }_{1\rightarrow 1} \leq \epsilon, \label{eq:SWAP-approx} \\
&\norm{\cE_{\ell, e} - (U_{2\ell-1} \otimes U_{2\ell}) \mathrm{BELL} (U_{2\ell-1} \otimes U_{2\ell})^{-1} (\cdot) (U_{2\ell-1} \otimes U_{2\ell}) \mathrm{BELL}^\dagger (U_{2\ell-1} \otimes U_{2\ell})^{-1} }_{1\rightarrow 1} \leq \epsilon,\label{eq:BELL-approx}
\end{align}
where $\mathrm{SWAP}$ and $\mathrm{BELL}$ are two-qubit unitaries, $\mathrm{SWAP}$ swaps the left and right qubits, $\mathrm{BELL} = (H \otimes I)\,\, \mathrm{CNOT}$, and $\norm{\cE}_{1\rightarrow 1} = \max_{\norm{X}_1 = 1} \norm{\cE(X)}_1$.

\subsection{Task description and quantum advantage}

We focus on the task of predicting many incompatible properties in unknown physical systems studied in \cite{huang2021information, chen2021exponential, huang2021quantum}.
We consider an unknown physical system described by an $n$-qubit separable state $\rho$. Recall that a separable state is a classical probability mixture over product states.
The goal is to learn to predict properties about $\rho$.

We compare two experimental settings: conventional experiments and quantum-enhanced experiments.
In \emph{conventional experiments}, the physicist could perform any POVM measurement on the unknown state $\rho$ to gather classical data. Based on the measurement data, the physicist could adaptively choose the next measurement on $\rho$ to obtain more data.
After many rounds of measurements, the physicist combines all measurement outcomes to form a model about the unknown state $\rho$.
In \emph{quantum-enhanced experiments}, the physicist could load multiple copies of the unknown state $\rho$ to a quantum computing system. The physicist can then use the quantum computer to perform quantum data analysis to learn a model about the unknown state $\rho$.
After learning a model of $\rho$, we will ask the physicist to predict properties of $\rho$.

We assume that the conventional experiments are perfect. All the POVM measurement could be chosen arbitrarily and are not subject to any noise.
In contrast, we consider quantum-enhanced experiments to base only on the noisy and unknown quantum device that we partially learned in Appendix~\ref{sec:algorithm-partial-learn}.
Despite the imperfection of quantum-enhanced experiments, we can still demonstrate a large quantum advantage.

\begin{theorem}[Advantage with noisy quantum device; Restatement of Theorem~\ref{thm:qadv-noisy}] \label{thm:qadv-noisy-ditto}
Given $n > 0$. Suppose $\epsilon$ is the error on each two-qubit operation in the noisy quantum device. There exists a distribution over unknown $n$-qubit separable states $\rho$ and properties, such that quantum-enhanced experiments using the noisy quantum device only require $N_{\mathrm{Q}} = \mathcal{O}((1/(1 - 4 \epsilon))^{2n})$ experiments to predict accurately, while noiseless conventional experiments require $N_{\mathrm{C}} = \Omega(2^n)$ experiments. This yields a separation of $N_{\mathrm{C}} = \Omega(N_{\mathrm{Q}}^{a}),$ where $a = - \log(2) / (2 \log(1 - 4\epsilon)) = \mathcal{O}(\epsilon)$.
\end{theorem}

\subsection{A class of states and properties}

For certain states and properties, conventional experiments can be very powerful.
For example, classical shadow tomography \cite{huang2020predicting, huang2022learning} is a class of conventional experiments based on randomized measurements that can be used to make accurate prediction for many properties.
Here, we give an example where there is a large separation between conventional experiments with perfect measurements and noisy quantum-enhanced experiments.
We consider a distribution over the unknown $n$-qubit separable states $\rho$ and the properties as follows.

With probability $1 / (4^n - 1)$, we sample a Pauli observable $P \in \{I, X, Y, Z\}^{\otimes n} \setminus \{I^{\otimes n}\}$.
Then with probability $1/2$, we consider the unknown state $\rho$ to be the maximally mixed state $I / 2^n$.
With probability $1/4$, we consider the unknown state $\rho$ to be
\begin{equation}
    \left(\bigotimes_{\ell=1}^{n} U_{2 \ell - 1}\right) \left(\frac{I+ 0.9 P}{2^n}\right) \left(\bigotimes_{\ell=1}^{n} U_{2 \ell - 1}^{-1}\right).
\end{equation}
With probability $1/4$, we consider the unknown state $\rho$ to be
\begin{equation}
    \left(\bigotimes_{\ell=1}^{n} U_{2 \ell - 1}\right) \left(\frac{I - 0.9 P}{2^n}\right) \left(\bigotimes_{\ell=1}^{n} U_{2 \ell - 1}^{-1}\right).
\end{equation}
The property we would like to predict is the absolute value of the expectation value of
\begin{equation}
    O = \left(\bigotimes_{\ell=1}^{n} U_{2 \ell - 1}\right) P \left(\bigotimes_{\ell=1}^{n} U_{2 \ell - 1}^{-1}\right).
\end{equation}
We tell the learning algorithm both the Pauli observable $P$ and the above observable $O$.
When $\rho = I/2^n$, the property is equal to zero.
However, when $\rho$ is an alternative state that is not $I/2^n$, the property is equal to $0.9$.
Hence, making accurate prediction in this task is equivalent to distinguishing whether $\rho$ is the maximally mixed state $I/2^n$.

\subsubsection{Characterization of the probability distribution}
\label{sec:separablestaterho}

Both the states $I / 2^n$ and $(I \pm 0.9 P) / 2^n$ are separable states and can be represented as a classical probability distribution over tensor products of the single-qubit stabilizer states
\begin{equation}
    S = \big\{ \ketbra{0}{0}, \, \ketbra{1}{1}, \,\ketbra{+}{+}, \,\ketbra{-}{-}, \,\ketbra{y+}{y+}, \,\ketbra{y-}{y-} \big\}.
\end{equation}
The physical system $\rho$ is a classical probability mixture over tensor products of $U_{2\ell-1} \sigma_{\ell} U_{2\ell-1}^{-1}$ for $\ell = 1, \ldots, n$, where $\sigma_{\ell}$ is a single-qubit stabilizer state. Hence, we have
\begin{equation} \label{eq:prob-rho}
    \rho = \sum_{\sigma_{\ell} \in S, \forall \ell = 1, \ldots, n} p(\sigma_{1}, \ldots, \sigma_n) (U_{1} \sigma_1 U_{1}^{-1}) \otimes \ldots \otimes (U_{2n-1} \sigma_n U_{2n-1}^{-1}).
\end{equation}
There are multiple distinct probability distributions that lead to the same state $\rho$.
Here, we consider the following classical distribution based on the chosen Pauli observable $P \in \{I, X, Y, Z\}^{\otimes n} \setminus \{I^{\otimes n}\}$ for the ease of analysis.
\begin{itemize}
    \item $\rho$ is the maximally mixed state: For each $\ell = 1, \ldots, n$, we consider the following. If $P_\ell = I, X, Y$, we choose $\sigma_\ell$ from the uniform distribution over $\ketbra{0}{0}, \ketbra{1}{1}$. If $P_\ell = Z$, we choose $\sigma_\ell$ from the uniform distribution over $\ketbra{+}{+}, \ketbra{-}{-}$.
    \item $\rho$ is a locally rotated $(I\pm 0.9 P) / 2^n$: The state $(I\pm 0.9 P) / 2^n$ is equal to the uniform mixture of the maximally mixed state and the state $(I\pm P)/2^n$. With probability $0.1$, we prepare $\rho$ the same way as the maximally mixed state given above. With probability $0.9$, we prepare $\rho$ as $(I\pm P)/2^n$.
    For each $\ell = 1, \ldots, n$, we consider the following. If $P_\ell = I$, we choose a $\sigma_\ell$ to be a uniform distribution over $\ketbra{0}{0}, \ketbra{1}{1}$.
    If $P_\ell \neq I$, we choose $\sigma_\ell$ to be one of the eigenstates of $P_\ell$.
    When $P_\ell$ is not the last non-identity Pauli operator, we choose $\sigma_\ell$ uniformly at random.
    When $P_\ell$ is the last non-identity Pauli operator, we choose $\sigma_\ell$ deterministically based on the choice of $\sigma_{\ell'}$ for $\ell' < \ell$ where $P_{\ell'}$ is not an identity. The determinisitic choice satisfies a parity constraint given by the state $(I\pm P)/2^n$.
\end{itemize}
Other choices of the classical probability distribution that give rise to the same state would yield exactly the same result but the analysis could be slightly more complex.

\subsection{Upper bound for noisy quantum-enhanced experiments}

We provide a sample complexity upper bound for quantum-enhanced experiments using the noisy quantum device that we partially learned.
The quantum-enhanced experiment we implement loads in two copies of the physical system $\rho$ and performs an entangled measurement across the two copies.

\subsubsection{Detailed procedure}

In the noisy quantum-enhanced experiments, we utilize the noisy quantum device and repeat the following for $N_{\mathrm{Q}} / 2$ times.
\begin{enumerate}
    \item Prepare the initial product state followed by a layer of single-qubit gates $g_{i, 0}$ for all qubit $i = 1, \ldots, 2n$.
    \item Load the physical system $\rho$ into the left hand side of the $n \times 2$ grid.
    \item Apply the entangling layer $\cE_{\ell = 1, s} \otimes \ldots \otimes \cE_{\ell = n, s}$, which is approximately equal to applying SWAP gates between left and right side of the $n \times 2$ grid from Eq.~\eqref{eq:SWAP-approx}.
    \item Load the physical system $\rho$ into the left hand side of the $n \times 2$ grid again.
    \item Apply the entangling layer $\cE_{\ell = 1, e} \otimes \ldots \otimes \cE_{\ell = n, e}$, which is approximately equal to rotating the pairs of qubits into a Bell basis from Eq.~\eqref{eq:BELL-approx}.
    \item Apply a layer of single-qubit gates $g_{i, 6}$ for all qubit $i= 1, \ldots, 2n$. Then measure using the unknown product measurement.
    \item Store the measurement outcome as $b_{t, i}, \forall i = 1, \ldots, 2n$ for the $t$-th experiment.
\end{enumerate}
After the $N_{\mathrm{Q}} / 2$ experiments using $N_{\mathrm{Q}}$ copies of the physical system $\rho$, when we are given a Pauli observable $P \in \{I, X, Y, Z\}^{\otimes n}$, we compute the following,
\begin{equation} \label{eq:compute-PP}
     \Xi \equiv \frac{2}{N_{\mathrm{Q}}} \sum_{t=1}^{N_{\mathrm{Q}} / 2} \prod_{\ell=1}^n \bra{\beta(b_{t, 2 \ell -1 }, b_{t, 2 \ell })} P_{\ell} \otimes P_{\ell} \ket{\beta(b_{t, 2 \ell -1 }, b_{t, 2 \ell })},
\end{equation}
where $\ket{\beta(x, y)}$ is the Bell state,
\begin{equation}
    \ket{\beta(x, y)} = \left( \frac{\ket{0, y} + (-1)^x \ket{1, 1-y}}{\sqrt{2}} \right).
\end{equation}
To understand what is happening, consider $U_i = I, \forall i$.
When all the operations are perfect, the quantum-enhanced experiment is equivalent to storing two copies of $\rho$ in the $2n$ qubits and measuring every pair of corresponding qubits in the Bell basis.
The Bell basis simultaneously diagonalizes $P \otimes P$ for all $P \in \{I, X, Y, Z\}^{\otimes n}$.
Hence, we can simultaneously predict $\Tr((P \otimes P)(\rho \otimes \rho)) = \Tr(P \rho)^2$ for all the Pauli observables.

\subsubsection{Noise analysis}

We analyze how noise affects the quantum-enhanced experiments.
Suppose the first sample of $\rho$ is
\begin{equation}
    (U_{1} \sigma^{A}_1 U_{1}^{-1}) \otimes \ldots \otimes (U_{2n-1} \sigma^{A}_n U_{2n-1}^{-1}),
\end{equation}
and the second sample of $\rho$ is
\begin{equation}
    (U_{1} \sigma^{B}_1 U_{1}^{-1}) \otimes \ldots \otimes (U_{2n-1} \sigma^{B}_n U_{2n-1}^{-1}).
\end{equation}
Let us focus on a pair of qubits $2 \ell - 1, 2 \ell$ for $\ell = 1, \ldots, n$.
From Eq.~\eqref{eq:state-left}, after the first two steps, the pair of qubits is in a state $\rho^{(2\ell-1, 2\ell), a}$ with
\begin{equation}
\norm{\rho^{(2\ell-1, 2\ell), a} - (U_{2\ell-1} \sigma^A_\ell U_{2\ell-1}^{-1}) \otimes (U_{2\ell} \ketbra{0}{0} U_{2\ell}^{-1})}_1 \leq \epsilon.
\end{equation}
After the third and fourth step, the pair of qubits is now in a state $\rho^{(2\ell-1, 2\ell), b}$ with
\begin{equation}
\norm{\rho^{(2\ell-1, 2\ell), b} - (U_{2\ell-1} \sigma^B_\ell U_{2\ell-1}^{-1}) \otimes (U_{2\ell} \sigma^A_\ell U_{2\ell}^{-1})}_1 \leq 2\epsilon
\end{equation}
using Eq.~\eqref{eq:SWAP-approx}.
After the fifth step, the two-qubit state is now $\rho^{(2\ell-1, 2\ell), c}$.
From Eq.~\eqref{eq:BELL-approx}, we have
\begin{equation}
\norm{\rho^{(2\ell-1, 2\ell), c} - (U_{2\ell-1} \otimes U_{2\ell}) \mathrm{BELL} (\sigma^B_\ell \otimes \sigma^A_\ell) \mathrm{BELL}^\dagger (U_{2\ell-1}^{-1} \otimes U_{2\ell}^{-1}) }_1 \leq 3\epsilon.
\end{equation}
In the sixth and seventh step, we measure the two-qubit state $\rho^{(2\ell-1, 2\ell), c}$ with a two-qubit product POVM. The two-qubit product POVM is given by
\begin{equation}
\{ M_{0}^{(2\ell-1, 6, 6)} \otimes M_{0}^{(2\ell, 6, 6)}, M_{0}^{(2\ell-1, 6, 6)} \otimes M_{1}^{(2\ell, 6, 6)}, M_{1}^{(2\ell-1, 6, 6)} \otimes M_{0}^{(2\ell, 6, 6)}, M_{1}^{(2\ell-1, 6, 6)} \otimes M_{1}^{(2\ell, 6, 6)} \},
\end{equation}
where the approximation error is given in Eq.~\eqref{eq:povm-left}.

We now combine with the classical post-processing in Eq.~\eqref{eq:compute-PP}.
Given a Pauli observable $P \in \{I, X, Y, Z\}^{\otimes n}$.
For each experiment, we can show that we are measuring the two-qubit state $\rho^{(2\ell-1, 2\ell), c}$ with an observable
\begin{equation}
O^{(2\ell-1, 2\ell)} \equiv \sum_{x, y \in \{0, 1\}} \bra{\beta(x, y)} P_\ell \otimes P_\ell \ket{\beta(x, y)} M_{x}^{(2\ell-1, 6, 6)} \otimes M_{y}^{(2\ell, 6, 6)}.
\end{equation}
The observable $O^{(2\ell-1, 2\ell)}$ differs from the ideal observable
\begin{equation}
    \sum_{x, y \in \{0, 1\}} \bra{\beta(x, y)} P_\ell \otimes P_\ell \ket{\beta(x, y)} (U_{2\ell-1} \ketbra{x}{x} U_{2\ell-1}^{-1}) \otimes (U_{2\ell} \ketbra{y}{y} U_{2\ell}^{-1})
\end{equation}
by at most $\epsilon$ error in $\norm{\cdot}_\infty$.
Using the triangle inequality, we have
\begin{equation}
    \left| \Tr\left(O^{(2\ell-1, 2\ell)} \rho^{(2\ell-1, 2\ell), c} \right) - \sum_{x, y} \bra{\beta(x, y)} P_\ell \otimes P_\ell \ket{\beta(x, y)} \bra{\beta(x, y)} (\sigma^B_\ell \otimes \sigma^A_\ell) \ket{\beta(x, y)} \right| \leq 4 \epsilon.
\end{equation}
Because the Bell basis simultaneously diagonalizes $\{I \otimes I, X\otimes X, Y\otimes Y, Z\otimes Z\}$, we have
\begin{equation}
    \sum_{x, y} \bra{\beta(x, y)} P_\ell \otimes P_\ell \ket{\beta(x, y)} \bra{\beta(x, y)} (\sigma^B_\ell \otimes \sigma^A_\ell) \ket{\beta(x, y)} = \Tr(P_\ell \sigma^B_\ell) \Tr(P_\ell \sigma^A_\ell).
\end{equation}
Recall that there exists a CPTP map $\cE_\ell$ such that $\rho^{(2\ell-1, 2\ell), c} = \cE_\ell(\sigma^A_\ell \otimes \sigma^B_\ell)$. Hence, we have
\begin{equation} \label{eq:error-sigmaAB}
     \left| \Tr\left(O^{(2\ell-1, 2\ell)} \cE_\ell(\sigma^A_\ell \otimes \sigma^B_\ell) \right) - \Tr(P_\ell \sigma^A_\ell) \Tr(P_\ell \sigma^B_\ell) \right| \leq 4 \epsilon.
\end{equation}
Furthermore, when $P_{\ell} = I$, we have
\begin{equation} \label{eq:noerr-sigmaAB}
     \Tr\left(O^{(2\ell-1, 2\ell)} \cE_\ell(\sigma^A_\ell \otimes \sigma^B_\ell) \right) = \Tr(P_\ell \sigma^A_\ell) \Tr(P_\ell \sigma^B_\ell) = 1.
\end{equation}
The two characterizations above will be crucial for the sample complexity analysis.

\subsubsection{Sample complexity analysis}

We use the probability distribution given in Appendix~\ref{sec:separablestaterho} to analyze the expectation value of $\Xi$ in Eq.~\eqref{eq:compute-PP}.
Combining with Eq.~\eqref{eq:prob-rho}, we have
\begin{align}
    \E[\Xi] = \sum_{\substack{\sigma^A_\ell \in S,\\ \forall \ell = 1, \ldots, n}} \sum_{\substack{\sigma^B_\ell \in S,\\ \forall \ell = 1, \ldots, n}} p(\sigma^A_1, \ldots, \sigma^A_n) p(\sigma^B_1, \ldots, \sigma^B_n) \prod_{\ell=1}^n \Tr\left(O^{(2\ell-1, 2\ell)} \cE_\ell(\sigma^A_\ell \otimes \sigma^B_\ell) \right).
\end{align}
We separate the analysis into two cases.
\begin{itemize}
    \item $\rho$ is the maximally mixed state: We only consider $\sigma^A, \sigma^B$ that appear with nonzero probability. For all $\ell = 1, \ldots, n$, we consider whether $P_\ell$ is equal to $I$. If $P_\ell \neq I$, then we have $\Tr(P_\ell \sigma^A_\ell) \Tr(P_\ell \sigma^B_\ell) = 0$, hence
    \begin{equation}
    |\Tr\left(O^{(2\ell-1, 2\ell)} \cE_\ell(\sigma^A_\ell \otimes \sigma^B_\ell) \right)| \leq 4 \epsilon
    \end{equation}
    from Eq.~\eqref{eq:error-sigmaAB}. If $P_\ell = I$, from Eq.~\eqref{eq:noerr-sigmaAB}, we have
    \begin{equation}
    \Tr\left(O^{(2\ell-1, 2\ell)} \cE_\ell(\sigma^A_\ell \otimes \sigma^B_\ell) \right) = 1.
    \end{equation}
    Together, we have
    \begin{equation}
        \left| \prod_{\ell=1}^n \Tr\left(O^{(2\ell-1, 2\ell)} \cE_\ell(\sigma^A_\ell \otimes \sigma^B_\ell) \right) \right| \leq (4 \epsilon)^{\#(P_\ell \neq I)},
    \end{equation}
    where $\#(P_\ell \neq I)$ is the number of $P_\ell, \forall \ell = 1, \ldots, n$ that is not equal to identity $I$. Therefore we conclude that
    \begin{equation} \label{eq:rhoismm}
        |\E[\Xi]| \leq (4 \epsilon)^{\#(P_\ell \neq I)}.
    \end{equation}
    \item $\rho$ is a locally rotated $(I\pm 0.9 P) / 2^n$: There is a probability of $0.19$ such that at least one of $\sigma^A, \sigma^B$ is sampled according to the distribution for the maximally mixed state. The same analysis as the case for maximally mixed state shows that we have
    \begin{equation} \label{eq:partial-mmstate}
        \left| \prod_{\ell=1}^n \Tr\left(O^{(2\ell-1, 2\ell)} \cE_\ell(\sigma^A_\ell \otimes \sigma^B_\ell) \right) \right| \leq (4\epsilon)^{\#(P_\ell \neq I)}.
    \end{equation}
    For a probability of $0.81$, we have both  $\sigma^A$ and $\sigma^B$ are sampled according to the probability distribution for $(I\pm P) / 2^n$ defined in Appendix~\ref{sec:separablestaterho}.
    We focus on $\sigma^A$ and $\sigma^B$ that occur with non-zero probability.
    We consider all $\ell = 1, \ldots, n$. We again separate into two cases: $P_\ell \neq I$ and $P_\ell = I$.
    \begin{itemize}
        \item If $P_\ell \neq I$, then we have $\Tr(P_\ell \sigma^A_\ell) \Tr(P_\ell \sigma^B_\ell) \in \{1, -1\}$. If $\Tr(P_\ell \sigma^A_\ell) \Tr(P_\ell \sigma^B_\ell) = 1$, then
        \begin{equation}
        \Tr\left(O^{(2\ell-1, 2\ell)} \cE_\ell(\sigma^A_\ell \otimes \sigma^B_\ell) \right) \geq 1 - 4 \epsilon
        \end{equation}
        from Eq.~\eqref{eq:error-sigmaAB}.
        In contrast, if $\Tr(P_\ell \sigma^A_\ell) \Tr(P_\ell \sigma^B_\ell) = -1$, then
        \begin{equation}
        \Tr\left(O^{(2\ell-1, 2\ell)} \cE_\ell(\sigma^A_\ell \otimes \sigma^B_\ell) \right) \leq -1 + 4 \epsilon.
        \end{equation}
        \item If $P_\ell = I$, from Eq.~\eqref{eq:noerr-sigmaAB}, we have
        \begin{equation}
        \Tr\left(O^{(2\ell-1, 2\ell)} \cE_\ell(\sigma^A_\ell \otimes \sigma^B_\ell) \right) = 1.
        \end{equation}
    \end{itemize}
    Furthermore, the parity constraint in the probability distribution over $\sigma^A, \sigma^B$ given in Appendix~\ref{sec:separablestaterho} shows that
    \begin{equation}
        \prod_{\ell=1}^n \Tr(P_\ell \sigma^A_\ell) \Tr(P_\ell \sigma^B_\ell) = 1.
    \end{equation}
    Hence, we have
    \begin{equation}
        \prod_{\ell=1}^n \Tr\left(O^{(2\ell-1, 2\ell)} \cE_\ell(\sigma^A_\ell \otimes \sigma^B_\ell) \right) \geq (1 - 4 \epsilon)^{\#(P_\ell \neq I)}.
    \end{equation}
    Combining with Eq.~\eqref{eq:partial-mmstate}, we can conclude that
    \begin{equation} \label{eq:rhoisnotmm}
        \E[\Xi] \geq 0.81 (1 - 4 \epsilon)^{\#(P_\ell \neq I)} - 0.19 (4\epsilon)^{\#(P_\ell \neq I)}.
    \end{equation}
\end{itemize}
Recall from Eq.~\eqref{eq:compute-PP}, $\Xi$ is the average over $N_{\mathrm{Q}} / 2$ independent random variables bounded between $[-1, 1]$. Hence, by Hoeffding's inequality, we need
\begin{equation}
    N_{\mathrm{Q}} = \mathcal{O}(\log(1 / \delta) / \tilde{\epsilon}^2)
\end{equation}
to estimate $\E[\Xi]$ to $\tilde{\epsilon}$ error with probability at least $1 - \delta$.
In order to distinguish between $\rho$ being the maximally mixed state or not, we need to estimate $\E[\Xi]$ to an error of at most
\begin{equation}
    0.81 (1 - 4 \epsilon)^{\#(P_\ell \neq I)} - 0.19 (4\epsilon)^{\#(P_\ell \neq I)} - (4\epsilon)^{\#(P_\ell \neq I)}.
\end{equation}
For $\epsilon$ less than a constant and $n$ sufficiently large, the above function is minimized at $\#(P_\ell \neq I) = n$.
In order to predict accurately with a probability at least $0.99$, the sample complexity for the noisy quantum-enhanced experiment is
\begin{equation}
    N_{\mathrm{Q}} = \mathcal{O}\left(\frac{1}{(1 - 4 \epsilon)^{2n}}\right).
\end{equation}
This concludes the proof for the first part of Theorem~\ref{thm:qadv-noisy-ditto}.


\subsection{Lower bound for noiseless conventional experiments}

We give a sample complexity lower bound for conventional experiments based on adaptive POVM measurements.
We do not assume the presence of any noise in conventional experiments.
The proof uses techniques proposed in \cite{huang2021information, huang2021quantum, chen2021exponential}.
In particular, the proof is closely related to one of the proofs in \cite{huang2021quantum} up to minor changes. We present a concise proof here for completeness.

A learning algorithm using noiseless conventional experiments is a rooted tree.
At every node, we perform a POVM on $\rho$.
Based on the POVM outcome, we move to a child node of the node.
Because a rank-$1$ POVM $\{ w_b \ketbra{\psi_b}{\psi_b} \}_b$ is always at least as powerful as general POVM \cite{huang2021information, chen2021exponential}, we will only consider rank-$1$ POVMs.
After $N_{\mathcal{C}}$ experiments, we arrive at a leaf node of the tree at depth $N_{\mathcal{C}}$.
Depending on the unknown physical state $\rho$, the probability distribution over the leaf nodes will be different. We write the probability distribution as
\begin{equation}
    p_{\rho}(\ell) = \prod_{t=1}^{N_{\mathrm{C}}} w_t \bra{\psi_t} \rho \ket{\psi_t}, \forall \ell: \mbox{leaf nodes},
\end{equation}
where $w_1 \ketbra{\psi_1}{\psi_1}, \ldots, w_{N_{\mathrm{C}}} \ketbra{\psi_{\mathrm{C}}}{\psi_{\mathrm{C}}}$ are the POVM elements associated to the outcomes of the $N_{\mathrm{C}}$ measurements that ends up at the leaf node $\ell$.

After the experiments, we obtain a Pauli observable $P \in \{I, X, Y, Z\}^{\otimes n} \setminus \{I^{\otimes n}\}$ and the assocaited observable $O = \left(\bigotimes_{\ell=1}^{n} U_{2 \ell - 1}\right) P \left(\bigotimes_{\ell=1}^{n} U_{2 \ell - 1}^{-1}\right).$
Suppose we can use the conventional experiments to classify between maximally mixed state
\begin{equation}
\rho_{\mathrm{mm}} = I/2^n
\end{equation}
and the alternative states
\begin{equation}
\rho_{S, P} = \left(\bigotimes_{\ell=1}^{n} U_{2 \ell - 1}\right) \left(\frac{I + S 0.9 P}{2^n}\right) \left(\bigotimes_{\ell=1}^{n} U_{2 \ell - 1}^{-1}\right)
\end{equation}
under the knowledge of $P$ and $U_{i}, \forall i$, where $S \in \{\pm 1\}, P \in \{I, X, Y, Z\}^{\otimes n} \setminus \{I^{\otimes n}\}$.
Then the average total variation distance between the leaf node distribution of the maximally mixed state and the alternative states $\rho_{s, P}$ must be greater than a constant,
\begin{equation} \label{eq:lowerbound-TV}
    \E_{\substack{P \in \{I, X, Y, Z\}^{\otimes n}} \setminus \{I^{\otimes n}\}} \left[ \, \frac{1}{2} \sum_{\ell: \mathrm{leaf}} \left| p_{\ell}(\rho_{\mathrm{mm}}) - \E_{S \in \{\pm 1\}} p_{\ell}\left( \rho_{S, P}\right) \right| \, \right] = \Omega(1).
\end{equation}
The expectation over $S \in \{\pm 1\}$ is in the inside because the knowledge of $S$ is not revealed.
On the other hand, the expectation over $P$ is on the outside because the knowledge of $P$ is revealed.

We now upper bound the average total variation distance with $N_{\mathcal{C}}$.
\begin{align}
    &\E_{\substack{P \in \{I, X, Y, Z\}^{\otimes n}\setminus \{I^{\otimes n}\} }} \left[ \, \frac{1}{2} \sum_{\ell: \mathrm{leaf}} \left| p_{\ell}(\rho_{\mathrm{mm}}) - \E_{S \in \{\pm 1\}} p_{\ell}\left( \rho_{S, P}\right) \right| \, \right]\\
    &= \E_{\substack{P \in \{I, X, Y, Z\}^{\otimes n} \setminus \{I^{\otimes n}\}}} \left[ \, \sum_{\ell: \mathrm{leaf}} \max\left(0, \,\,  p_{\ell}(\rho_{\mathrm{mm}}) - \E_{S \in \{\pm 1\}} p_{\ell}\left( \rho_{S, P}\right) \right) \, \right]\\
    &= \E_{\substack{P \in \{I, X, Y, Z\}^{\otimes n} \setminus \{I^{\otimes n}\}}} \left[ \, \sum_{\ell: \mathrm{leaf}}  p_{\ell}(\rho_{\mathrm{mm}}) \max\left(0, \,\, 1 - \E_{S \in \{\pm 1\}} \frac{p_{\ell}\left( \rho_{S, P}\right)}{p_{\ell}(\rho_{\mathrm{mm}})} \right) \, \right].
\end{align}
We lower bound the following term,
\begin{align}
    \E_{S \in \{\pm 1\}} \frac{p_{\ell}\left( \rho_{S, P}\right)}{p_{\ell}(\rho_{\mathrm{mm}})} &= \E_{S \in \{\pm 1\}}\prod_{t=1}^{N_{\mathrm{C}}} \bra{\psi_t} \left(\bigotimes_{\ell=1}^{n} U_{2 \ell - 1}\right) \left(I + S 0.9 P \right) \left(\bigotimes_{\ell=1}^{n} U_{2 \ell - 1}^{-1}\right)  \ket{\psi_t}\\
    &= \E_{S \in \{\pm 1\}}\prod_{t=1}^{N_{\mathrm{C}}} \left(1 + 0.9 S \bra{\tilde{\psi}_t} P \ket{\tilde{\psi}_t} \right)\\
    &\geq \exp\left[ \sum_{t=1}^{N_{\mathrm{C}}} \E_{S \in \{\pm 1\}} \log\left(1 + 0.9 S \bra{\tilde{\psi}_t} P \ket{\tilde{\psi}_t} \right) \right]\\
    &= \exp\left[ \frac{1}{2} \sum_{t=1}^{N_{\mathrm{C}}} \log\left(1 - 0.81 \bra{\tilde{\psi}_t} P \ket{\tilde{\psi}_t}^2 \right) \right]\\
    &= \prod_{t=1}^{N_{\mathrm{C}}} \sqrt{1 - 0.81 \bra{\tilde{\psi}_t} P \ket{\tilde{\psi}_t}^2 }.
\end{align}
The second line is a definition of $\ket{\tilde{\psi}_t} = \left(\bigotimes_{\ell=1}^{n} U_{2 \ell - 1}^{-1}\right)  \ket{\psi_t}$.
The third line uses Jensen's inequality.
Hence, we can upper bound the average total variation distance as
\begin{align}
    &\E_{\substack{P \in \{I, X, Y, Z\}^{\otimes n}\setminus \{I^{\otimes n}\} }} \left[ \, \frac{1}{2} \sum_{\ell: \mathrm{leaf}} \left| p_{\ell}(\rho_{\mathrm{mm}}) - \E_{S \in \{\pm 1\}} p_{\ell}\left( \rho_{S, P}\right) \right| \, \right]\\
    &\leq \E_{\substack{P \in \{I, X, Y, Z\}^{\otimes n} \setminus \{I^{\otimes n}\}}} \left[ \, \sum_{\ell: \mathrm{leaf}}  p_{\ell}(\rho_{\mathrm{mm}}) \max\left(0, \,\, 1 - \prod_{t=1}^{N_{\mathrm{C}}} \sqrt{1 - 0.81 \bra{\tilde{\psi}_t} P \ket{\tilde{\psi}_t}^2 } \right) \, \right]\\
    &= \sum_{\ell: \mathrm{leaf}}  p_{\ell}(\rho_{\mathrm{mm}}) \left( 1 - \E_{\substack{P \in \{I, X, Y, Z\}^{\otimes n} \setminus \{I^{\otimes n}\}}} \prod_{t=1}^{N_{\mathrm{C}}} \sqrt{1 - 0.81 \bra{\tilde{\psi}_t} P \ket{\tilde{\psi}_t}^2 } \right).
\end{align}
We can remove the $\max(0, \cdot)$ because $1 - \prod_{t=1}^{N_{\mathrm{C}}} \sqrt{1 - 0.81 \bra{\tilde{\psi}_t} P \ket{\tilde{\psi}_t}^2 } \geq 0$.
We bound the term,
\begin{align}
    &\E_{\substack{P \in \{I, X, Y, Z\}^{\otimes n} \setminus \{I^{\otimes n}\}}} \prod_{t=1}^{N_{\mathrm{C}}} \sqrt{1 - 0.81 \bra{\tilde{\psi}_t} P \ket{\tilde{\psi}_t}^2} \\
    &\geq \exp\left[ \frac{1}{2} \sum_{t=1}^{N_{\mathrm{C}}} \E_{\substack{P \in \{I, X, Y, Z\}^{\otimes n} \setminus \{I^{\otimes n}\}}}  \log\left(1 - 0.81 \bra{\tilde{\psi}_t} P \ket{\tilde{\psi}_t}^2 \right) \right] \\
    &\geq \exp\left[ -1.215 \sum_{t=1}^{N_{\mathrm{C}}} \E_{\substack{P \in \{I, X, Y, Z\}^{\otimes n} \setminus \{I^{\otimes n}\}}} \bra{\tilde{\psi}_t} P \ket{\tilde{\psi}_t}^2 \right] \\
    &\geq 1 -1.215 \sum_{t=1}^{N_{\mathrm{C}}} \E_{\substack{P \in \{I, X, Y, Z\}^{\otimes n} \setminus \{I^{\otimes n}\}}} \bra{\tilde{\psi}_t} P \ket{\tilde{\psi}_t}^2 \\
    &= 1 -1.215 \sum_{t=1}^{N_{\mathrm{C}}} \frac{1}{2^n + 1} = 1 - \frac{1.215 N_{\mathrm{C}}}{2^n + 1}.
\end{align}
The second line uses Jensen's inequality.
The third line uses $\log(1-x) \geq -3x, \forall x \in [0, 0.94]$.
The fourth line uses $\exp(x) \geq 1 + x, \forall x \in \mathbb{R}$.
The last line uses
\begin{equation}
    \E_{\substack{P \in \{I, X, Y, Z\}^{\otimes n} \setminus \{I^{\otimes n}\}}} P \otimes P = \frac{2^n \mathrm{SWAP} - I \otimes I}{4^n - 1},
\end{equation}
hence $\E_{\substack{P \in \{I, X, Y, Z\}^{\otimes n} \setminus \{I^{\otimes n}\}}} \bra{\tilde{\psi}_t} P \ket{\tilde{\psi}_t}^2 = \frac{2^n - 1}{4^n - 1} = \frac{1}{2^n + 1}$.
Combining with Eq.~\eqref{eq:lowerbound-TV}, we have
\begin{align}
    \Omega(1) = &\E_{\substack{P \in \{I, X, Y, Z\}^{\otimes n}\setminus \{I^{\otimes n}\} }} \left[ \, \frac{1}{2} \sum_{\ell: \mathrm{leaf}} \left| p_{\ell}(\rho_{\mathrm{mm}}) - \E_{S \in \{\pm 1\}} p_{\ell}\left( \rho_{S, P}\right) \right| \, \right] \\
    &\leq \sum_{\ell: \mathrm{leaf}}  p_{\ell}(\rho_{\mathrm{mm}}) \frac{1.215 N_{\mathrm{C}}}{2^n + 1} = \frac{1.215 N_{\mathrm{C}}}{2^n + 1}.
\end{align}
Thus we arrive at the desired lower bound,
\begin{equation}
    N_{\mathrm{C}} = \Omega(2^n),
\end{equation}
which concludes the proof of Theorem~\ref{thm:qadv-noisy-ditto}.

\end{document}